%% file: RecoveryTransitionOperators.tex
\documentclass[lettersize,journal]{IEEEtran}
\usepackage{amsmath,amsthm,amssymb,url,amscd}
\usepackage{mathtools,lipsum, nccmath}
\usepackage{cuted}
\usepackage{algorithmic}
\usepackage{algorithm}
\usepackage{array}
\usepackage{textcomp}
\usepackage{stfloats}
\usepackage{url}
\usepackage{verbatim}
\usepackage{graphicx}
\usepackage{cite}
\usepackage{amsfonts}
\usepackage{graphicx}
\usepackage{algorithm, algorithmic}
\usepackage[ruled,vlined,algo2e]{algorithm2e}
\usepackage{multirow}
\usepackage{enumerate}
\usepackage[alphabetic,initials]{amsrefs}
\usepackage{caption}
\usepackage{subcaption}
\usepackage[colorlinks=true, linkcolor = blue, citecolor=tumred]{hyperref}
\usepackage{cleveref}
\usepackage{xcolor}
\usepackage{esvect}
\usepackage{graphicx}
\usepackage{setspace}
\usepackage{tikz} 
\usepackage{pgfplots}
\usepackage{nicematrix}
\usepackage{stmaryrd}
\usepackage{bbold}
\usepackage{graphicx}
\usepackage[export]{adjustbox}
\newlength\figureheight
\newlength\figurewidth

\DeclareMathOperator{\e}{\mathrm{e}}

\DeclareMathOperator{\Ran}{Ran}

\DeclareMathOperator{\Id}{Id}

\DeclareMathOperator*{\argmin}{arg\,min}

\DeclareMathOperator{\rank}{rank}

\DeclareMathOperator{\R}{\mathbb{R}}
\DeclareMathOperator{\N}{\mathbb{N}}

\DeclareMathOperator{\diag}{diag}
\DeclareMathOperator{\dg}{dg}

\DeclareMathOperator{\ran}{ran}

\newcommand{\f}[1]{\mathbf{#1}}
\DeclareMathOperator{\PTX}{\y{P}_{T_{\f{H}}}}
\newcommand{\C}{\mathbb{C}}

\newcommand\hkm{^{(k-1)}}
\newcommand\hk{^{(k)}}
\newcommand\hkk{^{(k+1)}}
\newcommand\Mn{\mathcal{M}_n}
\newcommand\MnOplus{\Mn^{\oplus T}}
\newcommand\Rdd{\mathcal{M}_{d_1,d_2}}
\newcommand\Rddn{\mathcal{M}_{d_1 n,d_2 n}}
\newcommand\Rnn{\R^{n \times n}}

\newcommand{\noteC}[1]{\footnote{{\color{blue} \textbf{CK:} #1}}}

\newcommand{\HQTA}{\y{H}(\y{Q}_{T}(\f{A}))}
\newcommand{\HA}{\f{H}_{\f{A}}}
\newcommand{\Or}{\mathbb{O}}

\newcommand\bb[1]{\mathbb{#1}}
\newcommand\y[1]{\mathcal{#1}}

\def\Ex{\mathbb{E}}
\crefformat{equation}{(#2#1#3)}
\crefmultiformat{equation}{(#2#1#3)}
{ and~(#2#1#3)}{, (#2#1#3)}{ and~(#2#1#3)}

\xdefinecolor{tumblue}      {RGB}{  0,101,189}
\xdefinecolor{tumgreen}    {RGB}{162,173,  0}
\xdefinecolor{tumred}        {RGB}{229, 52, 24}
\xdefinecolor{tumivory}      {RGB}{218,215,203}
\xdefinecolor{tumorange}  {RGB}{227,114, 34}
\xdefinecolor{tumlightblue}{RGB}{152,198,234}

\newtheorem{theorem}{Theorem}[section]
\newtheorem{lemma}{Lemma}[section]
\newtheorem{proposition}{Proposition}[section]
\newtheorem{property}{Property}[section]
\newtheorem{corollary}{Corollary}[section]
\newtheorem{remark}{Remark}[section]

\newtheorem{definition}{Definition}[section]
\newtheorem{problem}{Problem}[section]

\usepackage{xcolor}

\begin{document}

\title{Learning Transition Operators From Sparse Space-Time Samples}

\author{Christian~K\"ummerle,~
        Mauro~Maggioni,~
        and~Sui~Tang
\thanks{C. K\"ummerle is with the Department
of Computer Science, University of North Carolina at Charlotte, Charlotte,
NC, 28223, USA, e-mail: kuemmerle@uncc.edu.}
\thanks{M. Maggioni is with the Department of Mathematics and the Department of Applied Mathematics \& Statistics, Johns Hopkins University, Baltimore, MD, 21218, e-mail: mauromaggionijhu@icloud.com.}
\thanks{S. Tang is with the Department of Mathematics, University of California, Santa Barbara, Isla Vista, CA 93117, USA, e-mail: suitang@ucsb.edu.}}

\maketitle

\begin{abstract}
We consider the nonlinear inverse problem of learning a transition operator $\f{A}$ from {partial observations at different times}, in particular from {sparse observations} of entries of its powers $\f{A},\f{A}^2,\cdots,\f{A}^{T}$. This \emph{Spatio-Temporal Transition Operator Recovery} problem is motivated by the recent interest in learning time-varying graph signals that are driven by graph operators depending on the underlying graph topology. We address the nonlinearity of the problem by embedding it into a higher-dimensional space of suitable block-Hankel matrices, where it becomes a low-rank matrix completion problem, even if $\f{A}$ is of full rank.
For both a uniform and an adaptive random space-time sampling model, we quantify the recoverability of the transition operator via suitable measures of incoherence of these block-Hankel embedding matrices. For graph transition operators these measures of incoherence depend on the interplay between the dynamics and the graph topology. 
We develop a suitable non-convex iterative reweighted least squares (IRLS) algorithm, establish its quadratic local convergence, and show that, in optimal scenarios, no more than $\mathcal{O}(rn \log(nT))$ space-time samples are sufficient to ensure accurate recovery of a rank-$r$ operator $\f{A}$ of size $n \times n$. This establishes that {spatial samples} can be substituted by a comparable number of {space-time samples}.
We provide an efficient implementation of the proposed IRLS algorithm with space complexity of order $O(r n T)$ and per-iteration time complexity {linear} in $n$. Numerical experiments for transition operators based on several graph models confirm that the theoretical findings accurately track empirical phase transitions, and illustrate the applicability and scalability of the proposed algorithm.
\end{abstract}


\section{Introduction}

Signals that arise from social, biological or transport networks are typically interconnected and structured, and can be modeled as residing on graphs. In many modern applications, the graph signals are time-varying and driven by graph  operators that are dependent on the underlying graph topology. For example,  the traffic flow on the road network is changing during a day; spatial temperature profiles measured by a sensor network vary at different time instances.    Estimating such graph signals and  dynamical processes from sparse observations is a research topic of wide interest (see for example \cite{ioannidis2018semi,pasdeloup2016characterization,ioannidis2018inference,dong2016learning,thanou2017learning,segarra2017network,mateos2019connecting,pasdeloup2017characterization} and references therein). 

We consider dynamical processes on graphs in the form of (discrete) linear dynamical systems:
\begin{align} \label{system:eq0}
x_{t+1} =\f{A}x_t: t=1,\ldots, 
\end{align}
where $x_t$ is the graph signal and the transition operator $\f{A}$  is typically a function of an algebraic descriptor of the graph structure (e.g. the adjacency matrix of the graph).  Examples for such transition operators include the random walk over a graph and its variations, heat operators, and other averaging processes. Powers of $\f{A}$ lead to both multiscale analyses on graphs \cite{CMDiffusionWavelets}, and eigenvectors of $\f{A}$ often play an important role in many machine learning applications such as dimension reduction and clustering \cite{DiffusionPNAS}. These models are in part motivated by applications that include modeling traffic in transportation networks \cite{deri2016new}, spatially-distributed atmospheric variables (e.g. temperature, pressure) measured by sensor networks \cite{thanou2017learning}, and neural activity in different regions of the brain \cite{sporns2010networks}.

In this paper, we are interested in learning $\f{A}$ from partial space-time observations of a temporal evolution. Let $\mathcal{X}_0 \in \mathbb{R}^{n\times m}$ be a set of $m$ initial states, one per column. We observe a discrete time series of length $T$ satisfying
\begin{align}\label{system:eq1}
\mathcal{X}_{t+1} = \f{A} \mathcal{X}_t\,, \qquad t\in[T]:=1,\ldots,T
\end{align} via spatio-temporal samples
\begin{equation} \label{eq:spatial:samples}
\mathcal{Y}_t=S_t(\mathcal{X}_t)\,,
\end{equation}
for $t\in[T]$ and $S_t:\mathbb{R}^{n\times m} \rightarrow \mathbb{R}^{m_t}$ representing a linear subsampling operator, typically returning a subset of the entries of its input.

{If we set aside the time evolution in \cref{system:eq1}, the task of recovering $\mathcal{X}_t$ at a fixed time $t$ from the observations $\mathcal{Y}_t$ may be considered as a \emph{completion problem}}, and may be tackled by \emph{low-rank matrix completion} techniques \cite{CR09,Candes10,Davenport16,ChiLuChen19}, even when the number of observations $m_t$ is much less than the number of entries in $\mathcal{X}_t$, if $\mathcal{X}_t$ has low rank. When $S_t \equiv \Id$, the identity map, related problems are also pursued in the model reduction community, to extract dominant eigenvectors of $\f{A}$, for example via the dynamic mode decomposition (DMD) \cite{schmid2010dynamic,KutzBruntonProctor16}.

In this paper, however, we are interested in situations where $T>1$ and both $\y{X}_t$ and $\f{A}$ are  \emph{not} necessarily  low rank.  At any single time $t$, we are not able to recover $\f{A}$ from $\mathcal{Y}_t$ and $\mathcal{Y}_{t+1}$, as in practice we often have $m_t, m_{t+1}\lesssim n$, due to application-specific constraints.  We seek to compensate the insufficient spatial samples at a single given time $t$ by leveraging the temporal dependent observations across time,  and tackle the fundamental problem of recovering $\f{A}$ from \textit{space-time} samples $\{\mathcal{Y}_t: t\in[T]\}$.  We restrict our attention, for simplicity, to recovering $\f{A}$ from \emph{partial observations} of $\f{A}, \f{A}^2,\cdots, \f{A}^T$, which, in the notation above, corresponds to $\mathcal{X}_0=\Id \in \R^{n \times n}$ in \eqref{system:eq1}.  Already in this setting, learning $\f{A}$ is a nonlinear inverse problem (as long as $T>1$), that is beyond the scope of regular matrix completion, dynamic mode decomposition, or subspace-based techniques \cite{coutino2020state}.

\subsection{Spatio-Temporal Transition Operator Recovery} \label{sec:problem:setup}
In the context of dynamical systems, it is expected that the recoverability of a transition operator $\f{A} \in \R^{n \times n}$ depends significantly on the structure of the space-time sampling as well as on spectral properties of $\f{A}$.
Denote the sampling locations at  $t \in [T]$ by $\Omega_t \subseteq [n] \times [n]$, and let the corresponding subsampling operator $S_t: \mathbb{R}^{n \times n} \to \R^{\Omega_{t}}$ be
\begin{align} \label{eq:def:sampling:St}
S_{t}(\f{M}):&= \left(\langle E_{i,j}, \f{M} \rangle\right)_{(i,j) \in \Omega_{t}} = \left(\f{M}_{i,j}\right)_{(i,j) \in \Omega_{t}}\,,
\end{align} 
where $E_{i,j}$ is the matrix with $1$ in its $(i,j)$-th entry and $0$ elsewhere, for $i,j \in [n]$. In applications, $\Omega_t$ may correspond to an observation model with mobile sensors that are moved to different locations at different times. 

Denoting the set of all possible space-time sampling locations by $I:= [n] \times [n] \times [T]$, we define the \emph{sampling set}
\begin{equation} \label{eq:Omega:def}
\Omega:= (\Omega_1 \times \{1\}) \cup \dots\cup (\Omega_T \times \{T\}) \subset I\,.
\end{equation}
Let $\y{M}_{n_1,n_2}$ denote the set of real $n_1 \times n_2$ matrices, abbreviated as $\Mn$ if $n=n_1=n_2$. 
We define the nonlinear monomial operator $\mathcal{Q}_T: \Mn \to  \MnOplus$ as
\begin{equation} \label{eq:matrix:powers}
\mathcal{Q}_T(\f{A}):= \f{A} \oplus \f{A}^2 \oplus \f{A}^3 \oplus \ldots \oplus \f{A}^T \in \MnOplus\,,
\end{equation}
and the sampling operator $P_{\Omega}: \MnOplus \to \R^{|\Omega|}$ as
\begin{equation}\label{eq:def:Somega} 
\begin{split}
\f{\widetilde{X}} := &\f{X}_1 \oplus \f{X}_2 \oplus \ldots  \oplus \f{X}_T  \mapsto \\
&  P_{\Omega}(\f{\widetilde{X}}) = \begin{bmatrix} S_1(\f{X}_1), & S_2(\f{X}_2), & \ldots, & S_T(\f{X}_T) \end{bmatrix}
\end{split}
\end{equation} 
where $S_{t}$ is as in \cref{eq:def:sampling:St}, for $t\in[T]$. 
We consider the following:
\begin{problem}[Spatio-Temporal Transition Operator Recovery] \label{problem}
Given a space-time sampling set $\Omega \subset I$, recover $\f{A} \in \Mn$ from the space-time samples
\begin{equation}\label{observations}
\f{y} = P_{\Omega}(\mathcal{Q}_T(\f{A})) = \left[S_1(\f{A}^1), S_2(\f{A}^2), \ldots, S_T(\f{A}^T) \right]\,,
\end{equation}
or from noisy space-time samples $\f{y} = P_{\Omega}(\mathcal{Q}_T(\f{A})) + \eta$, where $\eta$ is an (unknown) additive noise vector.
\end{problem}
We focus on two questions arising naturally in \Cref{problem}:
\begin{itemize}
\item Under which conditions on $\f{A}$ and on the distribution and size of the sampling set $\Omega$ can we guarantee to accurately estimate $\f{A}$ in \Cref{problem}?
\item Is there a computationally efficient recovery method to estimate $\f{A}$ in \Cref{problem} in these cases?
\end{itemize}

It is well-known from the literature on other structured inverse problems such as sparse vector recovery and matrix completion that deterministic sampling sets may not enable recovery from a {minimal} size of the sampling set $|\Omega|$ \cite{FR13,Davenport16}. For this reason, we focus on {random} space-time sampling schemes, and consider two types of {random} models: 
\begin{enumerate}
\item \textbf{Uniform sampling}: for $m \leq n^2 T$, a sampling set $\Omega$ consists of $m$ spatio-temporal samples in $[n] \times [n] \times [T]$ sampled uniformly at random without replacement;
\item \textbf{Adaptive sampling}: for each space-time index $(i,j,t) \in [n] \times [n] \times [T]$, let $p_{i,j,t}\in[0,1]$. An {adaptive sampling set} $\Omega \subseteq [n] \times [n] \times [T]$ consists of triplets $(i,j,t)$ drawn from i.i.d. Bernoulli trials with success probabilities $p_{i,j,t}$. The expected total number of samples is therefore $m_{\text{exp}} := \mathbb{E}[|\Omega|] = \sum_{i,j=1}^{n}\sum_{t=1}^T p_{i,j,t}$.
\end{enumerate}

While uniform sampling is conceptually simple as it only has one free parameter, $m= | \Omega |$, adaptive sampling is more flexible, in particular because its sampling probabilities $\{p_{i,j,t}\}_{(i,j,t) \in I}$ can be tuned to include prior information about a specific instance of \Cref{problem}.
\subsection{Our Contribution}
We tackle the spatio-temporal transition operator recovery problem by first applying an embedding into a structured {block Hankel matrix space}, under which the nonlinear relationship between different powers $\f{A}, \f{A}^2,\ldots,\f{A}^T$ of a matrix $\f{A}$ is mapped to a {low-rank} property of the block Hankel matrix $\HQTA$, essentially in a one-to-one manner (\Cref{thm:blkHankel:lowrank}). 
The block Hankel operator $\mathcal{H}: \MnOplus \to \Rddn$, with parameters $d_1, d_2 \in \N$ s.t. $T=d_1+d_2-1$, is defined as:
\begin{equation} \label{eq:BlockHankel:operatordef}
 \mathcal{H}\left(\f{X}_1\!\oplus\!\ldots\!\oplus\!\f{X}_T\right)\!:=\!
 \begin{bmatrix}
 	\f{X}_1\!\!\! & \f{X}_2\!\!\! & \f{X}_3\!\!\! & \reflectbox{$\ddots$} & \f{X}_{d_2}\!\!\! \\
 	\f{X}_2\!\!\! & \f{X}_3\!\!\! & \reflectbox{$\ddots$} & \reflectbox{$\ddots$} & \reflectbox{$\ddots$}   \\
    \f{X}_3\!\!\! &  \reflectbox{$\ddots$} & \reflectbox{$\ddots$} & \reflectbox{$\ddots$} &  \reflectbox{$\ddots$} \\
	\reflectbox{$\ddots$} &  \reflectbox{$\ddots$} & \reflectbox{$\ddots$} & \reflectbox{$\ddots$} &  \f{X}_{T-1}\!\!\! \\
    \f{X}_{d_1}\!\!\! & \reflectbox{$\ddots$} & \reflectbox{$\ddots$}  & \f{X}_{T-1}\!\!\!  &  \f{X}_{T}\!\!\!
 \end{bmatrix}\!\!
 \end{equation}
Related embeddings have been used to design computational methods for the solution of classical problems in signal processing and  system identification, see \Cref{scalarhankel} for details.

We then deploy an efficient low-rank optimization algorithm based on Iteratively Reweighted Least Squares (IRLS), which combines an iterative minimization of quadratic majorizing functions with an appropriate smoothing strategy for a log-determinant objective \cite{Daubechies10,Mohan10,KMV21} and, at the same time, respects the block Hankel structure. 
We address the connection between the choice of a space-time sampling set $\Omega$ and the identifiability of $\f{A}$ by proving a local convergence result for the proposed algorithm, called Transition Operator IRLS (\texttt{TOIRLS}), which shows that the operator $\f{A}$ can be efficiently computed from a number of spatio-temporal samples that is comparable to the sample complexity of using only spatial samples at time $T=1$,
for random sampling models based on either uniform or adaptive sampling. 
In particular, we show in \Cref{thm:convergence:mobilesensors} that in the noiseless case, with high probability, \texttt{TOIRLS} exhibits locally quadratic convergence if initialized close enough to the ground truth block matrix $\HQTA$, as soon as {only $\Omega(\mu_0 r n \log(n T))$ uniform or adaptive samples are provided}. An informal version of this result may be stated as follows:

\begin{theorem}[Local Convergence of \texttt{TOIRLS}, informal version] \label{thm:main:informal}
Let $\f{A} \in \Mn$ be a transition operator of rank $r$, and let $\HA := \HQTA$, with $\y{Q}_{T}(\f{A})$ and $\y{H}$ as in \cref{eq:matrix:powers} and \eqref{eq:BlockHankel:operatordef}, respectively.
Assume that either
\begin{enumerate}
\item[(i)] $\Omega$ is a space-time sampling set drawn by {uniform sampling} of cardinality 
\begin{equation} \label{eq:sample:complexity:uniform:informal}
m \gtrsim \mu_0 r n \log(n T),
\end{equation}
where $\mu_0$ is the incoherence factor (see \Cref{def:incoherence}) of $\HA$, or 
\item[(ii)] $\Omega$ is obtained by {adaptive sampling} with Bernoulli parameters
\[
p_{i,j,t} \gtrsim \min\left( \mu_{i,j,t}  \frac{r}{n T} \log(n T) ,1\right),
\]
where $\mu_{i,j,t}$ is a local incoherence (see \Cref{def:incoherence}) of $\HA$.
\end{enumerate}
If, additionally, an iterate of \texttt{TOIRLS} (\Cref{algo:IRLS:graphcompletion}), with observations $\f{y}= P_{\Omega}(\y{Q}_{T}(\f{A}))$, is close enough to $\y{Q}_{T}(\f{A})$, then, with high probability, the subsequent iterates converge to $\y{Q}_{T}(\f{A})$ with a quadratic convergence rate.
\end{theorem}
In \Cref{sec:numerical:experiments}, we provide numerical experiments that illustrate that the order of convergence in \Cref{thm:main:informal} captures the empirical behavior of \texttt{TOIRLS} for several transition operators on random graphs, and that its behavior appears robust to additive noise in the observations.

This result implies that despite the fact that samples are taken across $T$ different powers of $\f{A}$, not only from $\f{A}$ itself, \texttt{TOIRLS} recovers $\f{A}$ from {essentially as few uniformly random samples} as in the classical low-rank matrix completion setting \cite{CR09,koren_bell_volinsky,CandesTao10,Chen15,ChiLuChen19}, where $O( \nu_0 r n \log(n))$ samples are necessary in a uniform sampling model for the unique recovery of $\f{A}$ by any algorithm, where $\nu_0$ is the standard incoherence of $\f{A}$ \cite{Chen15}. 

We also analyze the incoherence of $\f{H}_{\f{A}}$ by relating it to that of $\f{A}$ for several families of transition operators, see \Cref{sec:incoherence:estimates}. In particular, we show that if $\f{A}$ is an orthogonal matrix or a projection, the incoherence $\mu_0$ of $\f{H}_{\f{A}}$ {coincides} with the incoherence $\nu_0$ of $\f{A}$, implying that the {same order} $\Omega( \nu_0 r n \log(n))$ of samples as in conventional matrix completion is sufficient in our setting, at least when $T \lesssim n$.

Unlike in conventional matrix completion, our results are nontrivial also when $\f{A}$ is of full rank, i.e. $r=\rank(\f{A})=n$: in this case, \Cref{thm:main:informal} implies that $O(\mu_0 n^2 \log(n T))$ samples, scattered over the $T$ observation times, are sufficient to ensure local convergence of \texttt{TOIRLS}, i.e. we pay a multiplicative oversampling factor of $O(\mu_0\log(n T))$ over the $n^2$ degrees of freedom of $\f{A}$. 
In particular, recovery is possible with a budget of only $O(n\log(n))$ sensors and $T=n$ observation times.

Finally, our results, such as \Cref{thm:main:informal}.2, on local incoherence-based sampling of specific space-time locations, inform adaptive sampling schemes that can be more data-efficient than uniform sampling.

Our results are presented in \Cref{sec:results:mobile:sensors} in the context of the spatio-temporal transition operator recovery problem. However, we point out that they are valid more generally for the problem of recovering a low-rank block Hankel matrix via \Cref{algo:IRLS:graphcompletion} if its output is chosen to be the entire matrix $\widetilde{\f{X}}^{(K)} \in \MnOplus$ instead of its restriction to its first block $\f{A}^{(K)}$. 

\begin{remark}
In this work, we focus on the algorithmic scheme \texttt{TOIRLS} optimizing non-convex surrogates in order to explore the fundamental information-theoretic properties of the underlying problem instead of a more traditional convex approach used for other low-rank optimization problems \cite{CR09,Chen15,DingChen20,chen_chi14,ye_kim_jin_lee} (see also \Cref{sec:matrix:completion,sec:recovery:structured:signals} below) for two reasons: 
\begin{itemize} 
\item we are able to ensure fast, albeit local, convergence, with high probability, under minimal assumptions on the sample complexity \cref{eq:sample:complexity:uniform:informal}. Using nuclear norm minimization on block Hankel matrices, we surmise that it is possible to also obtain an exact recovery result, albeit with possibly worse dependence on the sample complexity, with additional logarithmic factors in $r$, $\mu_0$ and potentially $T$, when using techniques such as a dual certificates or a leave-one-out analysis \cite{Chen15,DingChen20};
\item using a nuclear norm approach does not by itself lead to a scalable algorithm, as nuclear norm minimization is equivalent to a semidefinite program (SDP) with matrix variables of size $O(n T) \times O(n T)$. While some recent approximate solvers for large-scale SDPs require space of order only $O(r nT + m)$ \cite{Yurtsever-ScalableSDP2021,Ding-OptimalStorageSDP2021}, these methods do not find high-accuracy solutions. On the other hand, methods which provably solve the original SDP (e.g., interior-point methods \cite{Alizadeh-PrimalDualIntPointSDP1998} or augmented Lagrangian methods \cite{Sun-SDPNAL2020}) have storage requirements of $O(n^2 T^2)$ or larger.

We show in \Cref{thm:TOIRLS:computationalcost:Xkk} that \texttt{TOIRLS} is a scalable algorithm with space complexity of  $O(r nT + m)$, a per-iteration time complexity linear in $n$, and quickly leads to high-accuracy solutions thanks to the guaranteed local quadratic convergence rate.
\end{itemize}
\end{remark}

\subsection{Related Work} \label{s:relatedwork}
The transition operator recovery problem and the proposed low-rank modeling have connections to several different fields, which we briefly discuss.

\subsubsection{Low-Rank Matrix Completion} \label{sec:matrix:completion}
pioneered by \cite{Fazel02,CR09,CandesTao10,Gross-2011,Chen15} and popularized by applications in recommender systems \cite{ZhouWilkinsonNetflix08,koren_bell_volinsky}, the problem of recovering a low-rank matrix  from a subset of its entries or from underdetermined linear observations has been analyzed using both convex \cite{Recht10,CandesTao10,Chen15} and non-convex formulations \cite{MontanariKO10,SunL15,ChiLuChen19}. The a minimal sufficient condition for global convergence in the case of uniform samples is due to \cite{DingChen20}, where it was shown that $\Omega( \nu_0 r n \log(n) \log(\nu_0 r) )$ uniform samples are sufficient for the convex nuclear norm minimization approach to succeed with high probability if $\nu_0$ is the incoherence factor of \cite{Chen15}, $n$ the dimensionality and $r$ the rank of the matrix to be recovered. Local quadratic convergence in the presence of only $\Omega( \nu_0 r n \log(n))$ random observations was established for low-rank completion for a method similar to \texttt{TOIRLS} in \cite{KMV21} and in \cite{ZilberNadler-SIAMMDS2022} for a Gauss-Newton method, improving previous works on related algorithms \cite{Mohan10,Fornasier11,KS18} and \cite{bauch2021rank}, respectively. Low-rank matrix completion is a special case of the transition operator recovery problem \Cref{problem} corresponding to $T=1$; for $T>1$, however, \Cref{problem} is nonlinear in the transition operator.

A nonlinear generalization of the matrix completion problem that is different from ours was considered in \cite{Ongie2017algebraic,Ongie21}, where the low-rank properties of \emph{tensorized} data matrices are leveraged. While these problems also involve polynomial dependencies on a ground truth matrix, these dependencies are \emph{columnwise} and do not comprise the rich algebraic structure of matrix polynomials present in \Cref{problem}. The adaptive sampling model of \Cref{sec:problem:setup} had been considered for $T=1$ in the works \cite{ChenBhoSangWard15,EftekhariWakinWard-2018}.

\subsubsection{Dynamical Sampling}
here the aim is to recover a linear dynamical system from the union of coarse spatial samples at multiple time instances.   A  mathematical theoretical framework was proposed  in \cite{aldroubi2013dynamical,aldroubi2017dynamical} for linear systems of the form \eqref{system:eq0},  motivated by the pioneering work of \cite{lu2009spatial} that considered the space-time sampling of bandlimited diffusion fields over the real line. Several works \cite{tang2017universal, lai2019undersampled,aldroubi2019frames, aldroubi2017iterative,ulanovskii2021reconstruction} focus on the case where the transition operator $\f{A}$ in \eqref{system:eq0} is known, and the goal is to obtain sampling theorems ensuring  exact recovery of the initial state. For the case where $\f{A}$ is unknown, it has been shown that the eigenvalues of the matrix $\f{A}$ can be recovered from the space-time samples of a single trajectory, see \cite{aldroubi2014krylov, tang2017system, cheng2021estimate}. It is typically assumed that the observation operator $\mathcal{S}_t$ is deterministic and independent of $t$. Our paper is the first one, to our knowledge, to provide results for estimating $\f{A}$ from random space-time samples, i.e., for random subsampling operators $\mathcal{S}_t$ varying over the time $t$.

\subsubsection{System Identification}\label{scalarhankel}   
consider a linear time-invariant dynamical system
\begin{equation} 
\begin{aligned}
x_{t+1}&=\f{A} x_t+\f{B} u_t\\
y_t&=\f{C} x_t+\f{D} u_t \label{LIS2}
\end{aligned} 
\end{equation}
where $u_t$ is the input vector and $y_t$ is the output vector. The parameter estimation problem considered in control theory aims to recover the system matrices $\f{A}, \f{B} , \f{C}, \f{D}$ from the input-output pairs $(u_t,y_t)$. Classical results show that  a necessary condition to ensure identifiability is that $\f{C}$ is full rank \cite{bellman1970structural}. In general, this problem is ill-posed and the focus is to learn system matrices up to similarity transformations (see subspace identification methods \cite{ljung1998system,qin2006overview}) or the impulse response function (also called the Markov parameters) that determines the input-output map, both from a single trajectory \cite{fattahi2021learning,oymak2019non,sarkar2019near} and from multiple trajectories in \cite{zheng2020non,sun2020finite,tu2017non}. In the case of $\f{B}=\f{D}=0$ and $\f{C}=\f{I}$, a sufficient and necessary condition for the identifiability of $\f{A}$ from a single trajectory with a fixed initial condition is that $\f{A}$ has only one Jordan block for each of its eigenvalues, together with {certain constraints} on the initial condition \cite{stanhope2014identifiability,duan2020identification}. The low-rankness of a block-Hankel embedding of suitable powers and products of the matrices $\f{A},\f{B},\f{C}$ and $\f{D}$ similar to \cref{eq:BlockHankel:operatordef} is known to underlie the Kalman-Ho \cite{ho1966effective,oymak2019non} method for finding a \emph{realization} of the system, and has been explicitly used as optimization objective in \cite{fazel2013hankel,markovsky2013structured,markovsky,Grussler2018low}. However, the observation matrix $\f{C}$ is fixed in all the works the authors are aware of within this line of research, whereas in our setting, $\f{C}$ is random and varies over time $t$.

\subsubsection{Recovery of Structured Signals} \label{sec:recovery:structured:signals}
many structured signal recovery problems can be represented in the following abstract form: for a normed vector space $V$ over $\mathbb{C}$ and a \textit{known} linear operator $\f{A}: V \to V$, one is interested in recovering a signal $f\in V$ that is $M$-sparse in terms of eigenfunctions $\{v_j\}$ of $\f{A}$, i.e., $f=\sum_{j\in J}c_jv_j$, where $\{c_j\}_j$ is a set of coefficients in $\C$ and $|J|=M$. The goal is to recover $\{c_j\}_{j \in J}$ and $\{v_j\}_{j \in J}$ from observations $\mathcal{F}(\f{A}^{\ell} f)$ for $\ell=0,1,\ldots,L$ where $\mathcal{F}:V\rightarrow \mathbb{C}$ is a linear functional. For example, let $V$ be a vector space consisting of continuous functions on the real line and $\f{A}$ be a shift operator $(\f{A} f)(x) =f(x+1)$.
If one takes $v_j=\e^{a_jx}$ for $j\in[M]$ where $\{a_j\}_{j=1}^{M}$ are distinct complex numbers and $\mathrm{Im}(a_j)\in[-\pi,\pi)$, $\mathcal{F}(f)=f(x_0)$ for some $x_0\in \mathbb{R}$ and $L=2M-1$, then the recovery problem corresponds to the classical estimation of a sparse exponential sum, called harmonic retrieval. The other instances of this problem include super-resolution, blind deconvolution, recovery of signals with finite rate innovation; we refer to \cite{peter2013generalized,heckel2017generalized} for more details. Many well-known algorithmic approaches for these estimation problems are related to the Hankel matrix formed from the samples $\{\mathcal{F}(\f{A}^{\ell} f)\}_{\ell=1}^{L}$, including Prony's method \cite{potts2010parameter}, matrix pencil methods \cite{hua_sarkar}, and the algorithms MUSIC \cite{liao2016music} and ESPRIT \cite{roy1989esprit}. A generalization to irregular sets of samples has been considered in \cite{chen_chi14,jin_lee_ye,cai_wang_wei19,kuemmerle2019completion}, by formulating the signal recovery problem as a low-rank Hankel matrix completion, relating the problem to techniques discussed in \Cref{sec:matrix:completion} above. Multidimensional versions of these setting have also been considered in these works, leveraging low-rank properties of suitable multilevel Hankel or Toeplitz matrices, as well as in \cite{Yang16}. Specific to the $2D$ case, block Hankel matrices with Hankel blocks arise in these applications. In contrast, here the blocks of the block Hankel matrices we leverage are in general not Hankel, but a general matrix related to a linear dynamical system.
Finally, we remark that while the above works focus on recovering \emph{scalar signals}, with the shift operator $\f{A}$ being known, here we address the problem of estimating the unknown $\f{A}$, from partial observations along different trajectories.

\subsubsection{Graph Learning in Signal Processing}
there have been significant research efforts for inferring graph topology from observations of graph signals. This includes \cite{NIPS2010_8df707a9}, where the graph topology is estimated from full observations of a solution of a system of linear SDEs on the graph, via a regularized least squares approach, in particular focusing on the length of time the system needs to be observed in order to estimate the graph topology, as well as on the role of sparsity of the graph topology. Other existing approaches leverage a model based on graph filters \cite{segarra2016blind, segarra2017network}, or enforce sparsity \cite{maretic2017graph} or smoothness \cite{kalofolias2016learn,dong2016learning,thanou2017learning,egilmez2018graph} of signals using a penalized likelihood approach. Only a few works consider the graph signals as states of an underlying dynamical system, evolving according to the topology of the graph, e.g., \cite{coutino2020state}, in which case the single-trajectory states are observed via a fixed observation matrix that is static over time. In all cases, the identifiability of graph topology remains to be a challenging problem, as do the recovery algorithms, which lack theoretical guarantees.

\subsection{Outline}
The paper is organized as follows. In \Cref{section:lowrankoptimization}, we present in which sense the transition operator recovery problem is equivalent to a rank minimization problem over block Hankel matrices constrained to an affine space. In \Cref{section:IRLS}, we introduce \texttt{TOIRLS}, or Transition Operator Iteratively Reweighted Least Squares, to solve the resulting structured rank minimization problem. We introduce incoherence notions of block Hankel matrices in \Cref{sec:main:theoretical:results}, and present \Cref{thm:convergence:mobilesensors}, our main result that establishes local quadratic convergence of \texttt{TOIRLS} for respective sample complexities under both uniform and adaptive space-time sampling models. In \Cref{sec:computational:considerations}, we elaborate on computational considerations for \texttt{TOIRLS}, before presenting extensive numerical explorations in \Cref{sec:numerical:experiments}. In \Cref{sec:proof:vandermonde}, we provide the proof of the main theorem of \Cref{section:lowrankoptimization} and in \Cref{section:proof:main:theorem} a proof outline of \Cref{thm:convergence:mobilesensors}.
We conclude the main part of the paper in \Cref{section:conclusion}. Finally, we present useful incoherence estimates in \Cref{sec:coherence:estimates}, present a complete proof of \Cref{thm:convergence:mobilesensors} in \Cref{sec:proofs:appendix} and detail a practical implementation of \texttt{TOIRLS} in \Cref{sec:computational:details}.

\subsection{Notation}
We briefly summarize some notational conventions we use in this paper. The set of \emph{orthogonal matrices} of dimension $d$ is denoted by $\mathbb{O}^{d} = \big\{ \f{M} \in \Mn: \f{M}^*\f{M} = \Id \big\}$, while $ \Id$ is the identity matrix (omitting its dimension whenever suitable). If $\y{M}_{d_1 \times d_2}$ and $\f{v} \in \R^{d}$ is an arbitrary vector of dimension  $d := \min(d_1,d_2)$, the operator $\dg: \R^{d} \rightarrow \y{M}_{d_1 \times d_2}$ maps $\f{v}$ to the (generalized) diagonal matrix $\dg(\f{v}) \in \y{M}_{d_1 \times d_2}$ with $\dg(\f{v})_{ij} = \f{v}_{i}$ if $i=j$ and $\dg(\f{v})_{ij} = 0$ otherwise. For any matrix $\f{H}$, we denote its spectral norm by $\|\f{H}\| := \sigma_1(\f{H})$ and define the spectral norm ball with radius $\xi > 0$ around $\f{H}$ as $\y{B}_{\f{H}}(\xi) : = \left\{ \f{M} \in \Rddn: \left\| \f{M} - \f{H} \right\|  \leq \xi \right\}$.

\section{Recovering Transition Operators from Space-Time Samples using Low-Rank Optimization}
\label{section:lowrankoptimization}
In this section, we detail an approach to solve the transition operator recovery problem introduced in \Cref{sec:problem:setup}. 
A fundamental issue is the \emph{nonlinearity} of the operator $\mathcal{Q}_T$: $\f{A} \mapsto \mathcal{Q}_T(\f{A}) := \f{A} \oplus \f{A}^2 \oplus \f{A}^3 \oplus \ldots \oplus \f{A}^T$. We \emph{linearize} this nonlinearity by the transformation into the \emph{structured subspace} $\operatorname{Im} \y{H} \subset \Rddn$, where $\y{H}:\MnOplus \to \Rddn$ is the block Hankel operator of \cref{eq:BlockHankel:operatordef} with parameters $d_1, d_2 \in \N$. $\mathcal{H}$ maps a direct sums of $(n \times n)$ matrices to a \emph{block Hankel matrix} with $d_1$ block rows and $d_2$ block columns. In the remainder of the paper, we call $d_1$ and $d_2$ satisfying $T= d_1+d_2 -1$ the (first and second) \emph{pencil parameter} of $\y{H}$, in accordance with \cite{hua_sarkar,chen_chi14}.

The block Hankel operator $\y{H}$ enables us to recover the operator $\f{A}$ and its powers $\f{A}^2, \f{A}^3, \ldots \f{A}^T$ from a block Hankel matrix that is \emph{low rank} (\Cref{thm:blkHankel:lowrank}), an observation which lies at the core of our approach.
We use a dedicated low-rank optimization to recover a block Hankel-structured low-rank matrix $\y{H}(\widetilde{\f{X}}^*)$ compatible with the samples $P_{\Omega}(\mathcal{Q}_T(\f{A}))$  taken at space-time locations $\Omega$. If a sufficient number of random samples $\Omega$ from a sampling model in \Cref{sec:problem:setup} are provided, the hope is there is a unique generator $ \widetilde{\f{X}}^* = \mathcal{Q}_T(\f{A})$ for the Hankel matrix, from which the transition operator $\f{A}$ can then be directly inferred.

\subsection{Rank Minimization over Block Hankel Matrices} \label{sec:rankmin:blockhankel}
As a justification for our search for low-rank matrices in the subspace of block Hankel structured matrices, we establish in \Cref{thm:blkHankel:lowrank} the strong relationship between the rank of a block Hankel matrix $\f{H} \in \Rddn$ and the rank of an underlying ``generator'' matrix $\f{A}$. We say that $\overset{\square}{\f{H}} \in \R^{Dn}$ is a \emph{square extension} of $\f{H}$ if it is a block Hankel matrix with pencil parameters $D$ and $D$ whose first $T$ block anti-diagonals coincide with the $T$ anti-diagonals of $\f{H}$, and can otherwise have arbitrary entries in the last $2D-d_1-d_2$ blocks. 

 \begin{theorem} \label{thm:blkHankel:lowrank}
 	Recall the monomial operator $\mathcal{Q}_T: \Mn \to  \MnOplus, \f{A} \mapsto \f{A} \oplus \f{A}^2 \oplus \f{A}^3 \oplus \ldots \oplus \f{A}^T$  from \cref{eq:matrix:powers}, and the block Hankel operator $\mathcal{H}: \MnOplus \to \Rddn $ from \cref{eq:BlockHankel:operatordef}, with pencil parameters $d_1,d_2$. Then:
	\begin{enumerate}
	\item for any $\f{A} \in \Mn$,
 	\[
 	 \rank \left(\HQTA \right) =  \rank(\f{A})\,;
 	\]
	\item for any block Hankel matrix $\f{H} \in \Rddn$ with $(n \times n)$ blocks, that has a positive semidefinite square extension $\overset{\square}{\f{H}} \in \y{M}_{Dn,Dn}$, $D=\max(d_1,d_2)$, which has its first block $\f{H}_1 \in \y{M}_{n}$ of rank $r$, and at least one other block $\f{H}_{j} \in \y{M}_n$, $j>1$, of rank $r$,
	there exists a pair of matrices $(\f{Y},\f{M})$, where $\f{Y} \in \y{M}_{n,r}$ and $\f{M} \in \y{M}_{r}$ is symmetric, with $\rank(\f{H}) = \rank(\f{Y}) = \rank(\f{M}) = r$ such that
	\begin{equation} \label{eq:Hpsd:repr:Thm}
	\f{H} = \y{H}\big( \f{Y} \f{Y}^* \oplus \f{Y}\f{M}\f{Y}^* \oplus \ldots \oplus \f{Y}\f{M}^{T-1} \f{Y}^*\big).
	\end{equation}
	\end{enumerate}
 \end{theorem}
We refer to \Cref{sec:proof:vandermonde} for a proof of \Cref{thm:blkHankel:lowrank}. Related results have appeared in \cite{FeldmannHeinig96,Yang16,Andersson-2017Structure}; however, \Cref{thm:blkHankel:lowrank} does not follow from these results.

\Cref{thm:blkHankel:lowrank} implies a close relationship between a \emph{low-rank property} of block Hankel matrices $\y{H}(\widetilde{\f{X}})$ as in \cref{eq:BlockHankel:operatordef} and the existence of an operator $\f{A}$ such that $\mathcal{Q}_T(\f{A}) = \widetilde{\f{X}}$. While \Cref{thm:blkHankel:lowrank}.1 implies that the rank of generator matrix $\f{A}$ is inherited by its block Hankel image, \Cref{thm:blkHankel:lowrank}.2 is a statement in the other direction, i.e. about the existence of an underlying rank-$r$ generator matrix $\f{M}$ of a matrix semigroup associated to a rank-$r$ block Hankel matrix. 
We show the latter statement only if a positive semidefinite extension exists, noting that similar statements apply if additional, typically weak algebraic constraints are imposed on a general low-rank block Hankel matrix $\f{H} \in \y{H}_{d_1 n, d_2 n}$, see \cite{Tismenetsky92, FeldmannHeinig96}.
In particular, note that is $\f{A} \in \Mn$ if of full rank $n$, \Cref{thm:blkHankel:lowrank} implies that $\y{H}(\mathcal{Q}_T(\f{A}))$ is also of rank $n$. This is, however, \emph{low-rank} if we choose $\min(d_1,d_2) >1$, as the maximal rank of a $(d_1 n \times d_2 n)$ matrix is $\min(d_1,d_2)n$, and \emph{not} $n$. 

At a high level, \Cref{thm:blkHankel:lowrank} illustrates that the nonlinear relationship between the matrix powers $\f{A},\f{A}^2,\ldots,\f{A}^T$ is \emph{translated} to \emph{linear dependences} of the blocks in an associated block Hankel matrix in $\operatorname{Im} \y{H}$, motivating the pursuit of an optimization approach that aims to find a completion of a block Hankel matrix $\y{H}(\widetilde{X})$ that is \emph{both low-rank} and \emph{compatible} with the spatio-temporal measurements parametrized by the sampling operator $P_{\Omega}$ from \cref{eq:def:Somega}. This suggest a block Hankel structured \emph{rank minimization problem} 
\begin{equation} \label{eq:Hankel:rank:min}
	\min_{\widetilde{\f{X}} \in \MnOplus} \rank\Big(\mathcal{H}(\widetilde{\f{X}})\Big) \text{ s.t. } P_{\Omega}(\widetilde{\f{X}}) =  \f{y}
\end{equation}
where $\f{y} = P_{\Omega}(\y{Q}_T(\f{A}))$ the subset of observed entries of $\y{Q}_T(\f{A}) = \f{A} \oplus \f{A}^2 \oplus \f{A}^3 \oplus \ldots \oplus \f{A}^T$, indexed by $\Omega $.
More generally, for a given regularization parameter $\lambda \geq 0$, we define the \emph{data fitting function} $G_{\Omega,\f{y}}^{\lambda}: \MnOplus \to \R$
 \begin{equation} \label{eq:data:fit:fct}
 G_{\Omega,\f{y}}^{\lambda}(\widetilde{\f{X}}) =
 \begin{cases}
 	\iota_{P_{\Omega}^{-1}(\f{y})}(\widetilde{\f{X}}), & \text{ if } \lambda = 0, \\
 	\frac{1}{\lambda} \left\| P_{\Omega}(\f{\widetilde{X}}) - \f{y}\right\|_2^2, & \text{ if } \lambda > 0,
 \end{cases}
 \end{equation}
where $\iota_{P_{\Omega}^{-1}(\f{y})}: \MnOplus \to \R \cup \{\infty\}$ is $0$ if $P_{\Omega}(\widetilde{\f{X}}) = \f{y}$, and $\infty$ otherwise.
We then formulate the rank minimization problem
\begin{equation} \label{eq:Hankel:rank:min:Glambda}
	\min_{\widetilde{\f{X}} \in \MnOplus} \rank\Big(\mathcal{H}(\widetilde{\f{X}})\Big) +  G_{\Omega,\f{y}}^{\lambda}(\widetilde{\f{X}})\,,
\end{equation}
which reduces to \cref{eq:Hankel:rank:min} for $\lambda = 0$. In the presence of inexact measurements with additive noise such that $\f{y} = P_{\Omega}(\y{Q}_T(\f{A})) + \eta$ for some $\eta \in \R^{m}$, it can be beneficial to choose a positive regularization parameter $\lambda > 0$ \cite{Bunea-ReducedRank2011,Klopp-RankPenalized2011}.

Rank minimization problems such as \cref{eq:Hankel:rank:min} and \cref{eq:Hankel:rank:min:Glambda} are well-known to be NP-hard in general \cite{Recht10}, and different convex and non-convex reformulations of such problems have been studied for \emph{unstructured} problems, i.e., for the case that $\widetilde{\f{X}}$ itself is low-rank \cite{CandesTao10,MontanariKO10,recht,Vandereycken13,Park16,MaWangChiChen19}; see \cite{Davenport16,ChiLuChen19} for recent surveys.

While \cref{eq:Hankel:rank:min} enables us to formulate or problem in the language of optimization and to relate it to a common algorithmic paradigm in machine learning and signal processing, it poses several challenges from an optimization perspective. First, the rank objective is a non-convex and and non-smooth function, so that it is non-trivial to use derivative-based algorithms. Furthermore, unlike most low-rank optimization problems, the search space of \cref{eq:Hankel:rank:min} is the strict subspace $\left\{ \y{H}(\widetilde{\f{X}}): \widetilde{\f{X}}\in \MnOplus \right\}$ of $\Rddn$, making the problem a \emph{structured low-rank optimization problem} \cite{fazel2013hankel,markovsky,chen_chi14,Cai-2018Spectral}. Lastly, due to the large dimensionality of the ambient space $\Rddn$ even for moderate $n$ and $T$, only computationally efficient methods can be used for transfer operators of non-trivial size.

\section{Our Approach: Iteratively Reweighted Least Squares}\label{section:IRLS}
In several works in the literature, rank minimization problems have been tackled by designing optimization algorithms that optimize non-convex, \emph{smoothed} objective functions whose minimizers are designed to coincide with those of the rank objective in many cases. It was observed in \cite{FHB03,Candes13,KMV21} that optimizing a log-determinant objective often leads to solutions of underdetermined linear systems of very low-rank even in the presence of relatively few samples. Similarly, objective functions based on the Schatten-$p$ quasi-norm \cite{Mohan10,Lu2015,Ongie17,KS18,Giampouras20} and the smoothed clipped absolute deviation (SCAD) of the singular values \cite{Mazumder20} have been used to derive competitive algorithms for a variety of low-rank matrix recovery problems in signal processing and statistics. 

\begin{figure}
    \centering
    \resizebox{0.3\textwidth}{!}{
    \begin{tikzpicture}
    \begin{axis}[legend style={at={(.8,1)}, anchor=west, cells={anchor=west}, font=\footnotesize, rounded corners=2pt}, scale=1, xmin=-1.5, xmax=1.5, ymin=-5, ymax=3, axis x line=middle,thick, axis y line=middle,thick,yticklabels={,,},xticklabels={,,}]
        \addplot[red, thick, restrict expr to domain={(x>=-0.4)&&(x<=-0.000001)||(x>=0.000001)&&(x<=0.4)}{1:1}, samples=1000]{ln abs(x)};
        \label{graphic1}
        \addlegendentry{$f(\sigma)$}
        \addplot[blue, ultra thick, dashed, restrict expr to domain={(x<=-.4)||(x>=.4)}{1:1}, samples=1000]{ln abs(x)};
        \label{graphic2}
        \addlegendentry{$f_{\varepsilon}(\sigma)$}
        \addplot[red, thick, restrict expr to domain={(x>=-3)&&(x<=-.4)||(x>=.4)&&(x<=5)}{1:1}, samples=1000]{ln abs(x)};
        \addplot[blue, ultra thick, dashed] expression[domain=-.4:.4, samples=1000]{3.125*x^2-.5+ln(.4)};
         \addplot[black,mark=*, dashed] plot coordinates{(.4,-0.916) (.4,0)} node [above] {$\varepsilon$};
        \addplot[black,mark=*, dashed] plot coordinates{(-.4,-0.916)(-.4,0)} node [above] {$-\varepsilon$};
    \end{axis}
    \end{tikzpicture}
    }
\caption{Illustration of the smoothing $f_{\varepsilon}(\sigma)$ of $f(\sigma)=\log|\sigma|$.}
\label{figure_majorization}
\end{figure}
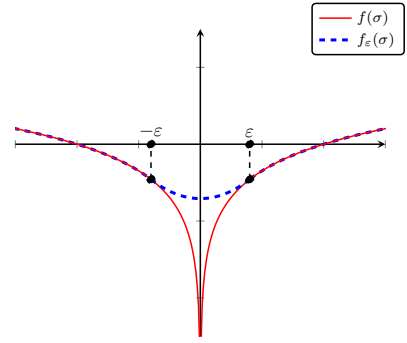

We propose an algorithm that adapts these ideas to the block Hankel rank minimization problem \eqref{eq:Hankel:rank:min:Glambda} as presented in \Cref{sec:rankmin:blockhankel}, which can be interpreted as an \emph{Iteratively Reweighted Least Squares (IRLS)} \cite{holland_welsch,Daubechies10,Mohan10,Fornasier11,Ongie17,KS18} strategy. 
Instead of optimizing the rank objective \cref{eq:Hankel:rank:min:Glambda} directly, let $\varepsilon >0$ be a  {smoothing parameter} and define the {smoothed log-deteterminant objective} $F_{\varepsilon}: \Rddn \to \R$ as 
\begin{equation} \label{eq:smoothing:Fpeps}	
F_{\varepsilon}(\f{M}) := \sum_{i=1}^{d n} f_{\varepsilon}(\sigma_i(\f{M})),\end{equation} where $d = \min(d_1, d_2)$ and
\begin{equation}
\quad f_{\varepsilon}(\sigma) =
\begin{cases}
 \log|\sigma|, & \text{ if } \sigma \geq \varepsilon, \\
 \log(\varepsilon) + \frac{1}{2}\Big( \frac{\sigma^2}{\varepsilon^2}-1\Big), & \text{ if } \sigma < \varepsilon,
 \end{cases}
 \end{equation}
which is continuously differentiable. If $J_{\varepsilon}: \MnOplus \to \R \cup \{\infty\}$ is the \emph{$\varepsilon$-smoothed surrogate objective} defined as
 \begin{equation} \label{eq:eps:smoothed:surrogate}
 J_{\varepsilon}(\widetilde{\f{X}}) = F_{\varepsilon}(\y{H}(\widetilde{\f{X}})) +  G_{\Omega,\f{y}}^{\lambda}(\widetilde{\f{X}}),
 \end{equation}
for a matrix $\widetilde{\f{X}} \in \MnOplus$, the steps of an iteration of IRLS can be understood as, first, the minimization of a quadratic model function $Q_{\varepsilon}(\cdot | \f{M}): \Rddn \to \Rddn$ that is an appropriate, {global upper bound} of $ J_{\varepsilon}(\cdot)$, leading to a {weighted least squares problems} and, second, as an update of the smoothing parameter $\varepsilon$ and refinement of the quadratic model function $Q_{\varepsilon}$ using the solution of the last weighted least squares problem.
The quadratic model functions $Q_{\varepsilon}(\cdot | \f{M})$ can be defined implicitly using {weight operators}, with which we are then able to formulate \texttt{TOIRLS}, an IRLS algorithm for transition operator learning (\Cref{algo:IRLS:graphcompletion}).

\begin{definition}[{see also \cite{K19,KFMMV21}}] \label{def:weightoperator:MatrixIRLS}
	Let $\f{M} \in \Rddn$ be a matrix with singular value decomposition $\f{M} = \f{U} \dg(\sigma) \f{V}^{*}$, where  $\f{U} \in \mathbb{O}^{d_1 n }$,
	$\f{V} \in \mathbb{O}^{d_2 n}$, and $\varepsilon >0$.
\begin{enumerate}
	\item The \emph{optimal weight operator} $W_{\f M}: \Rddn \to \Rddn$ associated to $\f{M}$ and $\varepsilon$ is the linear operator
\begin{equation} \label{eq:def:W}
	W_{\f M}(\f{Z}) = \f{U}  \Sigma_{\varepsilon,d_1}^{-1} \f{U}^{*} \f{Z} \f{V} \Sigma_{\varepsilon,d_2}^{-1} \f{V}^{*},
\end{equation}
where $\Sigma_{\varepsilon,d_1}  \in \y{M}_{d_1 n}$ and $\Sigma_{\varepsilon,d_2}  \in \y{M}_{d_2 n}$ are diagonal with $(\Sigma_{\varepsilon,d_1})_{ii} = \max(\sigma_i,\varepsilon)$ for $i \in [d_1 n]$ and $(\Sigma_{\varepsilon,d_2})_{jj} = \max(\sigma_j,\varepsilon)$ for $j \in [d_2 n]$.\footnote{with the convention that $\sigma_i = 0$ for $\min(d_1,d_2) n < i \leq \max(d_1,d_2) n$.}
\item Let $\mathcal{H}: \MnOplus \to \Rddn$ be the block Hankel operator of \cref{eq:BlockHankel:operatordef}. We define the \emph{effective weight operator} $\widetilde{W}_{\f M}: \MnOplus \to \MnOplus$ as the linear operator
\[
\widetilde{W}_{\f M} (\widetilde{\f{Z}}):= \mathcal{H}^{\;*} W_{\f M} \mathcal{H}(\widetilde{\f{Z}})\,.
\]
\end{enumerate}
\end{definition}

The choice of $W_{\f M}$ in \Cref{def:weightoperator:MatrixIRLS} in the weighted least squares problem \cref{eq:MatrixIRLS:Xdef} of \Cref{algo:IRLS:graphcompletion} can be regarded as the \emph{geometric operator mean} of the one-sided weight operator notions of the first IRLS papers considering rank optimization \cite{Mohan10,Fornasier11}. While a detailed discussion is beyond the scope of this paper, we note that in \cite{K19,KFMMV21} it is shown that the associated quadratic model function $Q_{\varepsilon}(\cdot | \f{M})$ not only majorizes the $ \varepsilon$-smoothed surrogate objective $ J_{\varepsilon}(\cdot)$ of \cref{eq:eps:smoothed:surrogate} pointwise, but also is \emph{optimal} in the sense that any smaller weight operator does not lead to majorizing quadratic model functions. Using the pointwise majorization, it is possible to show that the iterates $(\widetilde{\f{X}}\hk)_{k \geq 1}$ of \Cref{algo:IRLS:graphcompletion} lead to a monotonically decreasing sequence $\big( J_{\varepsilon_k}(\widetilde{\f{X}}\hk) \big)_{k \geq 1}$, and that each accumulation point of $(\widetilde{\f{X}}\hk)_{k \geq 1}$ is a stationary point of $J_{\overline{\varepsilon}}(\cdot)$, where $\overline{\varepsilon}:= \lim_{k \to \infty} \varepsilon_k$ \cite{KMV21}. 

\begin{algorithm}[h]
\caption{\texttt{TOIRLS} Transition Operator Iteratively Reweighted Least Squares} \label{algo:IRLS:graphcompletion}
\begin{algorithmic}
\STATE{\bfseries Input:} Indices $\Omega \subset I$, observations $\f{y} \in \R^m$, rank estimate $\widetilde{r} \leq n$, regularization parameter $\lambda \geq 0$, first pencil parameter $1 \leq d_1 \leq n$. 
\STATE Set $\varepsilon^{(0)}=\infty$ and $\widetilde{W}_{\f{H}_0} = \mathcal{H}^{\;*} \mathcal{H}$ with $\mathcal{H}: \MnOplus \to \Rddn$ as in \cref{eq:BlockHankel:operatordef} where $d_2 = T-d_1 +1$.

\FOR{$k=1$ to $K$}
\STATE \textbf{Solve weighted least squares problem} 
\vspace*{-2mm}
\begin{equation} \label{eq:MatrixIRLS:Xdef}
\widetilde{\f{X}}\hk :=\argmin\limits_{\widetilde{\f{X}} \in \MnOplus} \left\{\langle \widetilde{\f{X}}, \widetilde{W}_{\f{H}_{k-1}}(\widetilde{\f{X}}) \rangle + G_{\Omega}^{\lambda}(\widetilde{\f{X}}) \right\}\,,
\end{equation}
where $G_{\Omega}^{\lambda}: \MnOplus \to \R $ is the data fitting function of \cref{eq:data:fit:fct} and $\widetilde{W}_{\widetilde{\f{X}}^{(k-1)}}$ is the effective weight operator of \Cref{def:weightoperator:MatrixIRLS}.
\STATE \textbf{Update smoothing:} \label{eq:MatrixIRLS:bestapprox} Compute 	$(\widetilde{r}+1)$-st singular value of $\f{H}_{k} = \y{H}(\widetilde{\f{X}}\hk)$ to update

\begin{equation} \label{eq:MatrixIRLS:epsdef}
\varepsilon_k := \min\left(\varepsilon_{k-1},\sigma_{\widetilde{r}+1}\left(\f{H}_{k}\right)\right).
\end{equation}

\STATE \textbf{Update weight operator:} For $r_k := \big|\big\{i \in [d n]: \sigma_i\big(\f{H}_{k}\big) > \varepsilon_k\big\}\big|$, compute reduced rank-$r_k$ singular value decomposition of of $\f{H}_{k}$ to obtain leading $r_k$ singular values $\sigma_i\big(\f{H}_{k}\big)$, $i=1,\ldots,r_k$ and matrices $\f{U}\hk \in \R^{n d_1 \times r_k}$ and $\f{V}\hk \in \R^{n d_2 \times r_k}$, use this to update $\widetilde{W}_{\f{H}_k}$ as defined in \Cref{def:weightoperator:MatrixIRLS}.
\STATE $k = k + 1$.
\ENDFOR 
\STATE Extract the first block $\f{A}^{(K)}:= \left[\widetilde{\f{X}}^{(K)} \right]_{1:n,1:n}$ of $\widetilde{\f{X}}^{(K)}$.
\STATE{\bfseries Output:} $\f{A}^{(K)}$.
\end{algorithmic}
\end{algorithm}

While the domain of the weighted least squares step \cref{eq:MatrixIRLS:Xdef} of \Cref{algo:IRLS:graphcompletion} is $\MnOplus$, by the definition of the effective weight operator $\widetilde{W}_{\f{M}}$, a spectral reweighting \emph{in the subspace of block Hankel matrices} is applied \emph{implicitly}. As initialization for $k=1$, the weight operator $W_{\f{H}_0}$ \cref{eq:def:W} is chosen to be the identity operator, implying that the \emph{effective} weight operator $\widetilde{W}_{\f{H}_0} = \mathcal{H}^{\;*} \mathcal{H} = \y{D}^2$ is a diagonal operator that is constant for each summand of $\MnOplus$, and which amounts to the multiplicity of each block in the block Hankel image \cref{eq:BlockHankel:operatordef} defined via the operator $\y{H}$; cf. \cref{eq:DEijt:def} in \Cref{sec:coherence:estimates}.

\paragraph{Choice of regularization parameter $\lambda$}
the parameter $\lambda \geq 0$ in \Cref{algo:IRLS:graphcompletion} determines which surrogate objective $J_{\varepsilon}(\cdot)$ is optimized by \texttt{TOIRLS} and which underlying rank objective \cref{eq:Hankel:rank:min:Glambda} is chosen. As described in \Cref{sec:rankmin:blockhankel}, the choice of $\lambda = 0$, which imposes an affine constraint defined by the sampling operator $P_{\Omega}$ and the observation vector $\f{y} \in \R^{m}$, is appropriate if exact space-time samples are provided to the algorithm. While an optimal choice might correspond to some $\lambda > 0$ in the presence of \emph{inexact} space-time samples that depends on the order of magnitude of the noise, it turns out that $\lambda = 0$ is surprisingly robust to noise in practice, as explored in \Cref{sec:numerical:noisy:observations}. Theoretically, this observation is related to the so-called \emph{quotient property} of the measurement operator, which has been used to establish robust guarantees for equality-constrained low-rank and sparse recovery methods \cite{Wojtaszczyk10,Candes.2011,Liu.2011,KKM22}. 

\paragraph{Choice of rank estimate $\widetilde{r}$}
if the rank $r=\rank(\f{A})$ of the transition operator $\f{A}$ to be recovered is known, one should choose $\widetilde{r} = r$. If $r$ is unknown, or if only a vague estimate is available, it is advisable to \emph{overestimate} the true rank, i.e. to choose $\widetilde{r} \geq r$. While exact recovery of the transition operator might need more samples in that case, \Cref{algo:IRLS:graphcompletion} seems to be often able to good estimates for $\f{A}$ in that case.

\paragraph{Update rule for smoothing parameter $\varepsilon_k$}
after each weighted least squares step, the smoothing parameter $\varepsilon_k$ is updated, cf. \cref{eq:MatrixIRLS:epsdef}, in a non-increasing manner. This distinguishes IRLS from a conventional majorize-minimize (MM) method \cite{Lange16} for the smoothed surrogate objective $J_{\varepsilon}(\cdot)$ for a fixed $\varepsilon$. Similarly to related IRLS methods \cite{Daubechies10,Fornasier11,voronin_daubechies,KS18,KMV21}, the choice of the update rule quantifies the distance to a matrix of target rank $\widetilde{r}$ that is compatible with the observations $\f{y}$, playing a crucial role in the design of the algorithm due to the non-convexity of $F_{\varepsilon_k}(\cdot)$. If $\varepsilon_k$ is large, $J_{\varepsilon_k}(\cdot)$ will possess much fewer non-global minima than if $\varepsilon_k$ is small, in which case, however, $F_{\varepsilon_k}(\cdot)$ resembles much more the concave log-determinant objective that is known to constitute a powerful surrogate for the rank function \cite{Foucart18}.

For a complexity analysis and implementation details, we refer to  \Cref{sec:computational:considerations}. 
\section{Main Results} \label{sec:main:theoretical:results}
In this section, we present a convergence theory for \Cref{algo:IRLS:graphcompletion} for the problem of recovering transition operators from sparse time-space samples. 

It has been an open problem to establish global convergence of similar IRLS methods to minimizers of non-smooth, non-convex surrogate objectives such as \cref{eq:eps:smoothed:surrogate} underlying the respective problems \cite{Daubechies10,Mohan10,KS18,KMV21}, despite it being observed numerically in simulations. For this reason, we restrict the convergence analysis for \texttt{TOIRLS} to a \emph{local} one, which is based on the assumption we are given an iterate $\widetilde{\f{X}}^{(k)} \in \MnOplus$ that is close to a ground truth which is an image $\mathcal{Q}_T(\f{A})$ of a transition operator $\f{A}$. We quantify this using the set
\begin{equation} \label{eq:def:BAxi}
\y{B}_{\HA}(\xi) : = \left\{ \f{H} \in \Rddn: \left\| \f{H} - \f{H}_{\f{A}}\right\|  \leq \xi \right\}
\end{equation}
that contains matrices close to the block Hankel matrix  $\f{H}_{\f{A}} := \y{H}(\mathcal{Q}_T(\f{A}))$.

With \Cref{thm:convergence:mobilesensors} in \Cref{sec:results:mobile:sensors} below, we show sufficient conditions on the number of space-time samples, under either the uniform and adaptive sampling model, that, with high probability, guarantee the local convergence of \texttt{TOIRLS} to the ground truth, and therefore the recovery of $\f{A}$.

\subsection{Incoherence for Block Hankel Matrices} \label{sec:incoherence}
Due to the coordinate-wise nature of  either of our sampling models, even for a fixed dimensionality $n$ and fixed rank $r$, it cannot be expected that each transition operator $\f{A}$ will require a similar number of samples for successful recovery. In particular, a more \emph{localized} transition operator with a non-zero pattern that is not very distributed will not benefit from space samples at locations associated to its zero coordinates.

In order to quantify which transition operators can be recovered by either of our sampling models, we therefore introduce a notion of \emph{incoherence} for the block Hankel embedding matrix $\f{H}_{\f{A}}$ of a transition operator $\f{A}$. This extends the fundamental ideas in low-rank matrix completion \cite{CR09,recht,Chen15}, where the difficulty of a completion problem is measured by the \emph{incoherence} of a low-rank matrix with respect to the standard basis.
We also introduce local incoherence quantities, to be used to guide the adaptive sampling scheme. 

Let $\f{T}_{\f{Z}}$ be the {tangent space to the manifold of rank-$r$ matrices $\y{M}_r = \{ \f{X} \in \Rddn: \rank(\f{X}) = r \}$ at $\f{Z} \in \Rddn$,} where $r \in \N$ and $\f{Z} \in \Rddn$ is a rank-$r$ matrix with compact singular value decomposition $\f{Z} = \f{U} \Sigma \f{V}^*$ with $\f{U} \in \R^{n d_1 \times r}$ and $\f{V} \in \R^{n d_2 \times r}$ with orthonormal columns, and $\f{\Sigma} \in \R^{r \times r}$ the diagonal matrix of non-increasing singular values of $\f{Z}$. By \cite{Vandereycken13},
\begin{equation} \label{eq:tangent:space:def}	
\f{T}_{\f{Z}}\!:=\!\!\{ \f{U} \f{M}_1^* + {\f{M}}_2 \f{V}^*\!: \!\f{M}_1 \in \R^{n d_2 \times r}, {\f{M}}_2 \in \R^{n d_1 \times r} \}.\!
\end{equation} 

\begin{definition}\label{def:incoherence}
Let $\f{Z} \in \Rddn$ be a rank-$r$ matrix. Let $\left\{ \f{B}_{i,j,t}: (i,j,t) \in I \right\}$ be the standard basis of the space of block Hankel matrices.\footnote{See \Cref{lemma:Hankel:action} in \Cref{sec:coherence:estimates} for an explicit representation.} 
Let $d_1,d_2$ be the pencil parameters of the block Hankel operator $\y{H}$, and $c_s:=\frac{T(T+1)}{d_1 d_2}$.
\begin{enumerate}
\item For $1 \leq i, j \leq n$ and $1 \leq t \leq T$, we define the \emph{local incoherence at space-time index $(i,j,t)$} of $\f{Z}$ as
\begin{equation} \label{eq:mu:ijk}
\mu_{i,j,t} :=  \frac{n T}{c_s r} \|\mathcal{P}_{\f{T}_{\f{Z}}} (\f{B}_{i,j,t})\|_F^2\,.
\end{equation}
\item We say that $\f{Z}$ is \emph{$\mu_0$-incoherent} if there exists a constant $\mu_0 \geq 1$ such that
\begin{equation}
\label{eq:incoherence}
\begin{aligned}
\max_{1 \leq i, j \leq n, 1 \leq t \leq T} \|\mathcal{P}_{\f{T}_{\f{Z}}} (\f{B}_{i,j,t})\|_F
 \leq \sqrt{ \mu_0 c_s \frac{r}{n T}}\,,
\end{aligned} 
\end{equation}
i.e., if $\max_{1 \leq i, j \leq n, 1 \leq t \leq T} \mu_{i,j,t} \leq \mu_0$. We call the smallest $\mu_0$ satisfying \cref{eq:incoherence} the \emph{incoherence parameter} of $\f{Z}$. 
\end{enumerate}
\end{definition}
Intuitively, a rank-$r$ matrix $\f{Z}$ is $\mu_0$-incoherent with small $\mu_0$ if the projections of all elements of the {standard basis of the space of block Hankel matrices} $\left\{ \f{B}_{i,j,t}\right\}$ onto the tangent space $T_{\f{Z}}$ associated to $\f{Z}$ are {small}. In order to use an incoherence notion that is adequate for our purposes of understanding the fundamental difficulty of an instance of \Cref{problem}, we follow the notion of \cite[Definition 3.3.1]{K19} in \cref{eq:incoherence}, which is a slightly weaker notion than the notions used in the context of structured low-rank matrices \cite[(27)]{chen_chi14} and \cite{cai_wang_wei19}. In fact, $\mu_0$ in \cref{eq:incoherence} can be upper bounded by the incoherence parameter of \cite{recht,chen_chi14} (see also \cite[Remark B.1.]{KMV21}, the discussion of \Cref{sec:incoherence:estimates} and \Cref{lemma:incoherence:bound:traditional}).

\subsection{Local Quadratic Convergence of \texttt{TOIRLS}}\label{sec:results:mobile:sensors}
We are now ready to state local convergence guarantees of \texttt{TOIRLS} (\Cref{algo:IRLS:graphcompletion}) for the recovery of transition operators from space-time samples. 
\begin{theorem}[Local Quadratic Convergence of \texttt{TOIRLS}] \label{thm:convergence:mobilesensors}
There exist absolute constants $\widetilde{c}_0,C$ such that the following holds. Let $\f{A} \in \Mn$ be a rank-$r$ transition operator, 
let $\f{H}_{\f{A}} = \y{H}(\mathcal{Q}_T(\f{A}))$ be the block Hankel matrix associated to the first $T$ time scales of $\f{A}$, where $\y{H}: \MnOplus \to \Rddn$ is the block Hankel embedding map \cref{eq:BlockHankel:operatordef} with pencil parameters $d_1,d_2$. Let $\widetilde{\f{X}}\hk$ be the $k$-th iterate of \Cref{algo:IRLS:graphcompletion} with inputs $\Omega$, $\f{y} = P_{\Omega}(\y{Q}_T(\f{A}))$ and $\widetilde{r} = r$, assume that the smoothing parameter \cref{eq:MatrixIRLS:epsdef} satisfies $\varepsilon_k = \sigma_r(\y{H}(\widetilde{\f{X}}\hk))$.  Let $\kappa:=\sigma_1(\f{H}_{\f{A}})/\sigma_r(\f{H}_{\f{A}})$ denote the condition number of $\f{H}_{\f{A}}$.

Suppose that one of the following statements holds:
\begin{enumerate}
\item \emph{[Uniform sampling]} $\f{H}_{\f{A}}$ is $\mu_0$-incoherent and that $\Omega$ is a random subset of cardinality $m$ uniformly drawn without replacement in the set of space-time samples $I= [n] \times [n] \times [T]$ , with
\begin{equation} \label{eq:mainthm:uniform:samplecomplexity}
m = \Omega( c_s \mu_0 r n \log(n T) ),
\end{equation}
and, furthermore, $\y{H}(\widetilde{\f{X}}\hk) \in  \y{B}_{\HA}(R_0)$ with\footnote{recall that $d:=\min(d_1,d_2)$}
\begin{equation} \label{eq:proximity:assumption:mainthm:1}
R_0:=\widetilde{c}_0 \bigg(\frac{\mu_0}{nT}\bigg)^{3/2} \frac{r^{1/2}}{\kappa (dn-r)^{1/2}}  \sigma_r(\HA) \,.
\end{equation}
\item \emph{[Adaptive sampling]} With $(\mu_{i,j,t})_{(i,j,t) \in I}$ being the local incoherences \cref{eq:mu:ijk} of $\f{H}_{\f{A}} $, $\Omega$ consists of random index triplets $(i,j,t) \in I$ that are independently observed according to Bernoulli distributions with probabilities $p_{i,j,t}$ each satisfying
\begin{equation}  \label{eq:sample:complexity:condition:localcoherence}
p_{i,j,t} \geq \min\left( C c_s  \frac{\mu_{i,j,t}\log(n T)}{n T} r ,1\right)\,,
\end{equation}
and, furthermore, $\y{H}(\widetilde{\f{X}}\hk)\!\in\!\y{B}_{\HA}(R_0)$ with
\begin{equation} \label{eq:proximity:assumption:mainthm:2}
R_0:=\widetilde{c}_0 \!\min_{(i,j,t) \in I}\! \bigg(\frac{\mu_{i,j,t}\log(nT)}{nT}\bigg)^{3/2}\!\! \frac{r^{1/2}  \sigma_r(\HA)}{\kappa (dn-r)^{1/2}}\,.
\end{equation}
\end{enumerate}
Then, with probability of  of at least $1 - 2 n^{-2}$, the subsequent iterates of \texttt{TOIRLS} (\Cref{algo:IRLS:graphcompletion}) converge to $\f{H}_{\f{A}}$, i.e. $\y{H}(\widetilde{\f{X}}^{(k+\ell)}) \xrightarrow{\ell \to \infty} \f{H}_{\f{A}}$, with quadratic convergence rate: for a dimension-dependent constant $\nu$.\footnote{See \Cref{sec:thm:convergence:mobilesensors:proof} for a possible choices for $\nu$.}
\begin{align*}
&\|\y{H}(\widetilde{\f{X}}^{(k+\ell+1)}) - \f{H}_{\f{A}}\|  \\ &\leq \min(\nu \|\y{H}(\widetilde{\f{X}}^{(k+\ell)})-\f{H}_{\f{A}}\|^{2} , \|\y{H}(\widetilde{\f{X}}^{(k+\ell)})-\f{H}_{\f{A}}\|)\,.
\end{align*}
\end{theorem}

\Cref{thm:convergence:mobilesensors} justifies that the spatio-temporal transition operator recovery problem can be solved efficiently using \texttt{TOIRLS} given a number of random samples that is, up to constants, only logarithmically larger than the $r (2n -r) = O(r n)$ free parameters that are required to describe a rank-$r$ transition operator $\f{A} \in \Mn$. In the case of adaptive sampling, the condition \cref{eq:sample:complexity:condition:localcoherence} can be translated into a bound on the number of expected samples $m_{\text{exp}}$ since $m_{\text{exp}} = \mathbb{E}[|\Omega|] = \sum_{(i,j,t) \in I} p_{i,j,t} \geq 
C c_s \frac{r}{n T} \log(n T) \sum_{(i,j,t) \in I} \mu_{i,j,t} $ (if the constants in \cref{eq:sample:complexity:condition:localcoherence} are small enough to attain the minimum in the first argument). See \Cref{sec:incoherence:estimates} for further discussion.

The proximity assumptions \cref{eq:proximity:assumption:mainthm:1,eq:proximity:assumption:mainthm:2}, which ensure that the spectral norm error of subsequent iterates of \texttt{TOIRLS} {decreases with a quadratic convergence rate}, are comparably restrictive due to their dependence on the $n$, $d$ and $T$, which makes it hard to find an initialization that satisfy the conditions for large-scale problems. However, extending the convergence radius of IRLS methods remains an open problem even for simpler problems such as sparse vector and unstructured low-rank matrix recovery if a non-convex objective such as \cref{eq:smoothing:Fpeps} is used \cite{Daubechies10,KS18,KMV21}. In \Cref{section:numerical:recoverability}, we provide numerical experiments illustrating that in practice, exact recovery of transition operators is observed empirically with an empirical probability of essentially $1$ once enough samples are provided, even if \texttt{TOIRLS} is initialized with the natural, data-agnostic weights of $\widetilde{W}_{\f{H}_0} = \mathcal{H}^{\;*} \mathcal{H}$ as in \Cref{algo:IRLS:graphcompletion}. 

The statements of \Cref{thm:convergence:mobilesensors} address the case of exact observations $\f{y} = P_{\Omega}(\y{Q}_T(\f{A}))$ that are not perturbed by any noise. In \Cref{sec:numerical:noisy:observations}, we provide numerical experiments suggesting that \texttt{TOIRLS} is in practice robust in the presence of noisy observations (including with a choice of the regularization parameter $\lambda =0$).

The proof strategy for \Cref{thm:convergence:mobilesensors} is outlined in \Cref{section:proof:main:theorem} and detailed in \Cref{sec:thm:convergence:mobilesensors:proof}. 

We note that while \Cref{sec:results:mobile:sensors} focuses on the the behavior of \Cref{algo:IRLS:graphcompletion} in the context of the recovery of a transition operator $\f{A}$ from space-time samples, it \emph{is} possible to extend the applicability of \Cref{thm:convergence:mobilesensors} to \Cref{algo:IRLS:graphcompletion} recovering---more generally---rank-$r$ block Hankel matrices by choosing as output the entire matrix $\widetilde{\f{X}}^{(K)} \in \MnOplus$ instead of its restriction to its first block $\f{A}^{(K)}$. In particular, in this setting, the ground truth $\HA$ can be substituted by any ground truth $\y{H}(\widetilde{\f{X}}_0)$ with $\rank(\y{H}(\widetilde{\f{X}}_0)) = r$, using the same notions of (local) incoherence as in the presented results.

\begin{remark}
We recall that the pencil parameter $d_1$ is a free parameter in \Cref{algo:IRLS:graphcompletion}. \Cref{thm:convergence:mobilesensors} also has implications for the choice of $d_1$: for uniform sampling, the sampling complexity \cref{eq:mainthm:uniform:samplecomplexity} is minimized if we choose $d_1$ such that $c_s \cdot \mu_0$ (both $c_s$ and $\mu_0$ depend on $d_1$) is minimized. The factor $c_s = \frac{T(T+1)}{d_1 d_2} = \frac{T(T+1)}{d_1 (T-d_1+1)}$ is minimized for $d_1 = \left\lfloor (T+1)/2\right\rfloor$, yielding a roughly square block Hankel matrix. Such a choice is observed to be favorable also for other problems using structured low-rank optimization \cite{chen_chi14,Cai-2022structured}.

A priori, the dependence of $\mu_0$ on $d_1$ is unclear; however, numerical experiments conducted in \Cref{sec:numerical:pencildependence} suggest that this choice of $d_1$ also minimizes the product $c_s  \mu_0$ at least in some of situations we consider.
\end{remark}

\subsection{Examples and Discussion of Sample Complexity} \label{sec:incoherence:estimates}
We now attempt to better understand the implications of \Cref{sec:results:mobile:sensors} and, in particular, the sample complexity conditions \cref{eq:mainthm:uniform:samplecomplexity} and \cref{eq:sample:complexity:condition:localcoherence} for uniform and adaptive sampling schemes. We provide sufficient conditions on the sample complexity by providing bounds on the incoherence parameter $\mu_0$ and local incoherences $\mu_{i,j,t}$, respectively, in various examples.

It is instructive to relate $\mu_0$, the incoherence of the block Hankel matrix  $\f{H}_{\f{A}} = \y{H}(\mathcal{Q}_T(\f{A}))$, with the now-classical incoherence notion \cite{CandesTao10,Chen15} of the transition operator $\f{A}$ (which coincides with $\f{H}_{\f{A}}$ in the static case of $T=1$). While in general there is no direct relationship between these two notions, as the singular vectors of $\f{H}_{\f{A}}$  may not be always expressed in terms of the singular vectors of $\f{A}$,
In two particular cases, when $\f{A}$ is an orthogonal matrix or a positive semi-definite matrix, it is possible to establish a simple relationship between these incoherence notions.
\paragraph{Orthogonal matrices}  If $\f{A} \in \Or^{n} := \{ \f{X} \in \Mn: \f{X}^* \f{X} = \Id\}$, it holds that $\mathrm{rank} ({\f{A}})=n$. In this case, the incoherence parameter $\mu_0$ of $\f{H}_{\f{A}}$ satisfies 
\[
\mu_0 \leq 1 =: \widetilde{\mu}_0
\]
and, furthermore, the local incoherences $\mu_{i,j,t}$ of  $\f{H}_{\f{A}}$ satisfy
\[
\sum_{(i,j,t) \in I} \mu_{i,j,t} \leq T n^2,
\] 
see \Cref{coherence:unitary} for details. The two parts of \Cref{thm:convergence:mobilesensors} therefore imply that for both uniform and adaptive sampling, $\Theta( c_s n^2 \log(Tn))$ space-time samples are sufficient to establish local convergence of IRLS with high probability. 
These results are consistent with the intuition that a dynamical system driven by an orthogonal transition operator is energy-preserving, and from the bound $\Theta( c_s n^2 \log(Tn))$, we see that up to a logarithmic factor of $\log(Tn)$, space-time samples contain a comparable amount of information to that of static samples. As the resulting sample complexity bound is of the same order in both cases, we expect adaptive sampling and uniform sampling to exhibit similar behavior for orthogonal transition operators. We refer to \Cref{sec:numerical:Tdependence} for numerical experiments.

\paragraph{Positive semi-definite matrices} Let $d_1 \leq d_2$ without loss of generality. If the transition operator is a positive semidefinite matrix $\f{A} = \sum_{i=1}^r \lambda_i \f{u}_i\f{u}_i^*$ with (positive) eigenvalues $\lambda_i$ and corresponding eigenvectors $\f{u}_i$, we show that the incoherence parameter $\mu_0$ of $\mathcal{H}\big(\mathcal{Q}_T(\f{A})\big)$ satisfies
\begin{equation} \label{eq:mu:tilde:mobile:uniform}
\mu_0 \leq \max_{1\leq i\leq n}{\sum_{\ell=1}^{r}\frac{nd_2 (\f{u}_{\ell})_{i}^2}{r(\sum_{s=0}^{d_1-1}\lambda_{\ell}^{2s})}} =: \widetilde{\mu}_0.
\end{equation}  
In particular, if $\f{A}$ is a rank-$r$ projection and if $d_1 = d_2$, this bound becomes (with $e_i$ denoting the $i$-th canonical basis vector)
\[
\widetilde{\mu}_0  = \max_{1 \leq i \leq n} \frac{n}{r} \left\| \f{U}^* e_i  \right\|_2^2 =:  \nu_0\,,
\]
which coincides with the incoherence constant of $\f{A}$ as defined in the low-rank matrix completion literature \cite{CandesTao10,Chen15}. In this case, we obtain a space-time sampling bound $ \Theta(c_s \nu_ 0 r n \log(n T))$, which is just slightly more than the necessary condition of $\Theta ( c_s \nu_ 0 r n \log(n))$ for exact recovery by any method under a uniform sampling model \cite{CandesTao10,Chen15}. For adaptive sampling,  we show an upper bound for $\sum_{i,j,t} p_{i,j,t}$ as $\Theta (rn\log(nT)\log(T))$ ($T\geq 3$), and this bound can be improved to $\Theta(r n \log(nT))$ if $\f{A}$ is a rank $r$ projection. Our bound is comparable with the one obtained in \cite[Theorem 2]{ChenBhoSangWard15} for Bernoulli sampling for low-rank matrix completion, namely, $\Theta(rn\log^2(n))$ for the case $T=1$. We refer to \Cref{coherence:psd} for proofs of the presented estimates.

In general, $\widetilde{\mu}_0 $ could be larger or smaller than  $\nu_0$, depending on the interplay of the spectrum of $\f{A}$ with the coherence of the eigenvectors.  For very \emph{spiky} operators $\f{A}$ with large incoherence $\nu_0$ and quickly decaying spectrum, however, the best estimate we obtain from \cref{eq:mu:tilde:mobile:uniform} is $\widetilde{\mu}_0 \leq d_2 \nu_0$. This implies that in such a setting, our estimates lead to a sufficient condition of $\Omega( c_s \nu_0 r n T \log(nT))$ required samples, which is rather pessimistic.

\section{Computational Considerations} \label{sec:computational:considerations}

If $\f{H}_{k-1} = \y{H}(\widetilde{\f{X}}\hkm)$ is the block Hankel matrix at iteration $k-1$, the solution $\widetilde{\f{X}}\hk$ of the weighted least squares \cref{eq:MatrixIRLS:Xdef} can be written as (see \Cref{lemma:explicit:weighted:leastsquares} in \Cref{sec:computational:details})
\begin{equation*}
\widetilde{\f{X}}\hk  = \widetilde{W}_{\f{H}_{k-1}}^{-1}  P_{\Omega}^* \left( \lambda\Id +  P_{\Omega} \widetilde{W}_{\f{H}_{k-1}}^{-1} P_{\Omega}^* \right)^{-1}(\f{y})\,.
\end{equation*}

However, using this formula directly can be impractical as we have no explicit representation of the inverse $\widetilde{W}_{\f{H}_{k-1}}^{-1}: \MnOplus \to \MnOplus$ of the effective weight operator $\widetilde{W}_{\f{H}_{k-1}}: \MnOplus \to \MnOplus$, unlike in the case of \emph{unstructured} low-rank optimization, where the optimization domain is not restricted to a strict linear subspace such as $\y{H}(\MnOplus) \subset \Rddn$ and for which a related IRLS method was studied in \cite{KMV21}. A space and memory-efficient implementation of the weighted least squares step leveraging an underlying ``low-rank plus diagonal'' structure of $\widetilde{W}_{\f{H}_{k-1}}$ can still be achieved, as can be seen in \Cref{thm:TOIRLS:computationalcost:Xkk}.

\begin{theorem} \label{thm:TOIRLS:computationalcost:Xkk}
Let $\widetilde{\f{X}}\hkm \in \Rdd$ be the $(k-1)$-st iterate of \texttt{TOIRLS} (\Cref{algo:IRLS:graphcompletion}) for an observation vector $\f{y} \in \R^m$ with $m = |\Omega|$, $\widetilde{r}=r$, and $\lambda \geq 0$. Assume that $r_{k-1} = r$. Then an approximation of the $k$-th iterate $\widetilde{\f{X}}\hk$ of \texttt{TOIRLS} can be computed within $N_{\text{CG\_inner}}$ steps of a conjugate method solving a $O(r n T) \times O(r n T)$ linear system with space complexity of $O(r n T + m)$ and in $O(N_{\text{CG\_inner}} r T ( m + n \log T + n r T) )$ time.
\end{theorem}
\Cref{thm:TOIRLS:computationalcost:Xkk} follows directly from \Cref{lemma:Algo:implementation} in \Cref{sec:computational:details}, using the implementation outlined in \Cref{algo:TOIRLS:implementation}. The linear systems solved within \Cref{lemma:Algo:implementation} can be shown to be well-conditioned under reasonable assumptions, in which case a constant number $N_{\text{CG\_inner}}$ of CG iterations is sufficient to obtain an \emph{accurate} approximation of $\widetilde{\f{X}}\hk$.

As stated in the weight operator update step of \Cref{algo:IRLS:graphcompletion}, the action of the effective weight operator $ \widetilde{W}_{\f{H}_{k-1}}^{-1}$ only uses information about the $r_{k-1}$ leading singular vector pairs and singular values of $\y{H}(\widetilde{\f{X}}\hkm)$. In particular, if for all iterations where the smoothing update \cref{eq:MatrixIRLS:epsdef} is such that  $\varepsilon_k = \sigma_{\widetilde{r}+1}\left(\f{H}_{k}\right)$, it holds that $r_k=\widetilde{r}$. This means that, in this case, only $\widetilde{r}$ singular values and singular vector pairs of $\f{H}_{k}$ need to be computed in the weight update step of \Cref{algo:IRLS:graphcompletion}, and these can be computed up to high precision using matrix-matrix multiplications with a randomized block Krylov method in \cite{MuscoMusco15,yuan2018superlinear} in $O( mT r +  r T ( \log T+ r T) n  +  Tn r^2 )$ time (using fast multiplication with block circulant matrices, see also proof of \Cref{lemma:Algo:implementation}).

We conclude that one iteration of \texttt{TOIRLS} can be computed with a time complexity that is \emph{linear} in the dimension $n$ of the transition operator $\f{A}$, at least if it is of rank $r= O(1)$. For example, if $|\Omega| = m = \Theta( rn \log (n T))$ space-time samples of an $O(1)$-incoherent ground truth $\f{A}$ are provided uniformly at random, one full \texttt{TOIRLS} iteration using \Cref{lemma:Algo:implementation} takes $O( nT^2 \log(n T) )$ time.

\section{Numerical Experiments} \label{sec:numerical:experiments}
In this section we explore the numerical performance of \texttt{TOIRLS}, \Cref{algo:IRLS:graphcompletion} for estimating transition operators from sparse space-time samples. We consider operators $\f{A} \in \R^{n \times n}$ associated with random graph models, as well as orthogonal matrices $\f{A}$. These experiments are meant to shed light on the sharpness of our sampling complexity results \Cref{thm:convergence:mobilesensors}, and verify they are consistent with the empirically observed behavior. While there are no dedicated computational approaches to our recovery problem available in the literature, we include also comparisons with the interior-point algorithm \cite{waltz2006interior} used in the nonlinear optimization wrapper \texttt{fmincon} of MATLAB, minimizing the objective $f: \Mn \to \R$
\begin{equation} \label{eq:nonlinear:leastsquares}
f(\f{B}) := \left\| P_{\Omega}(\mathcal{Q}_T(\f{B}))  -  P_{\Omega}(\mathcal{Q}_T(\f{A})) \right\|_2^2
\end{equation}
using finite difference gradient approximations. 

\paragraph{Numerical setup} since the number of degrees of freedom is $r (2n -r)$ for a rank-$r$ matrix $\f{A} \in \Mn$ and $r(n-(r-1)/2)$
for a symmetric $(n \times n)$ rank-$r$ matrix $\f{A}$, we define, for $m_{total} = |\Omega|$ space-time samples, the oversampling factor $\rho$ as, respectively, 
\[
\rho = \frac{m_{total}}{r (2n -r)}
\qquad\text{and}\qquad
\rho = \frac{m_{total}}{r(n-(r-1)/2)}\,.
\]
The average number of spatial samples taken at each time instance is
$
{m_{single}}={m_{total} }/{T}
$.
We use $m_{1}$ to denote the number of samples taken at $T=1$.  We will use both the uniform and adaptive schemes described in \Cref{sec:problem:setup}.

In the numerical experiments, we use \texttt{TOIRLS} as outlined in \Cref{algo:IRLS:graphcompletion} using the implementation described in \Cref{sec:computational:details} for computing the tangent spaces, and solving the linear systems associated to the weighted least squared problems with a conjugate gradient method.\footnote{For problem instances with relatively large ambient dimension $n$, such as the Minnesota road network graph of \Cref{sec:dep:graphtopology}, we use a MATLAB implementation that follows closely the steps outlined in the proof of \Cref{lemma:Algo:implementation}.2. For problems with larger number of time steps $T$, we used matrix-vector multiplications in \Cref{algo:TOIRLS:implementation} that include antiaveraging of block Hankel matrices instead of block-wise fast Fourier transforms as these turned out to be faster for the problem dimensions we were interested in.} Unless stated otherwise, we use \Cref{algo:IRLS:graphcompletion} with stopping criterion combining a maximal number of iterations $N_0 = 250$, a tolerance $\text{tol} = 10^{-11}$ with respect to the relative change in Frobenius norm, and $\text{tol\_CG}=10^{-13}$ for the conjugate gradient step. We provide the true $\rank(\f{A})$ of the ground truth as the rank estimate $\widetilde{r} = \rank(\f{A})$ to the algorithm. If not stated otherwise, we provide \texttt{TOIRLS} with the pencil parameter $d_1 = \left\lfloor (T+1)/2\right\rfloor$, leading an (approximately) square dimensionality of the block Hankel embedding space $\ran(\y{H})$. 

\paragraph{Evaluation metrics}  we define the recovery error of an estimator $\f{\hat A}$ of $\f{A}$ as  $$\text{Rec}_{\f{A}}:={\|\f{\hat A} -\f{A}\|_F}/{\|\f{A}\|_F}.$$
For a random model, unless stated otherwise, we run 10 independent trials  and report the mean and standard deviation of the recovery errors. 

\paragraph{Graph Topology-Induced Transition Operators}
We consider operators representing dynamics on different graphs and random graph models. Let $\mathcal{G}=(V,E,\f{W})$ be an \emph{undirected weighted graph}  with $n$ vertices $V=\{v_1,\cdots, v_n\}$, edges $ E \subset\mathcal{ V\times V}$ and adjacency matrix $\f{W} \in \Rnn$, i.e., $\f{W}_{ij} = 1$ if $(i,j) \in E$ and $\f{W}_{ij} = 0$ otherwise. The \emph{degree} of a vertex $v_i \in \mathcal{V}$ is $\mathrm{deg}(v_i)=\sum_{j=1}^{n} \f{W}_{ij}$.
Given a graph $\mathcal{G}$, a variety of associated transition operators can be defined, encoding structural information about the graph \cite{chung1997spectral} and associating to the graph certain dynamical processes on it.

\begin{definition} \label{def:graph:operators}
The \emph{normalized diffusion operator} of a graph $\mathcal{G}=(V,E,\f{W})$ is $\f{A} :=   (\f{D}^{-1})^{\frac{1}{2}}\f{W}  (\f{D}^{-1})^{\frac{1}{2}}$, where $\f{D} :=\mathrm{diag}(\mathrm{deg}(v_i))_{v_i\in V}$ and $\f{D}^{-1}$ denotes its pseudo-inverse. The \emph{normalized graph Laplacian operator} is $\f{L} = \Id -\f{A}$. The \emph{random walk matrix}  $\f{P}$ is $ \f{D}^{-1}\f{W}$, and the \emph{heat diffusion operator} for time parameter $\tau >0$ is $\exp(-\tau \f{L})$.
\end{definition}

\subsection{Recoverability in the Noiseless Setting} \label{section:numerical:recoverability}
We first investigate the empirical recoverability of transition operators $\f{A}$ by \Cref{algo:IRLS:graphcompletion} from spatio-temporal samples $\Omega$ given different numbers of time steps $T$, sampling schemes and different sample complexities. Furthermore, we consider different types of transition operators that include both full and low-rank matrices, symmetric and non-symmetric matrices, orthogonal matrices and operators associated to the topology of graphs.

\subsubsection{Dependence on Number of Time Steps $T$} \label{sec:numerical:Tdependence}
For a first experiment, we fix the number $m = |\Omega|$ of uniformly sampled space-time samples from $\mathcal{Q}_T(\f{A}):= \f{A} \oplus \f{A}^2 \oplus \f{A}^3 \oplus \ldots \oplus \f{A}^T$ and consider different choices of $T$.

\paragraph{Random orthogonal matrices} we consider random orthogonal matrices $\f{A} \in \Rnn$, with $n=50$, sampled from the Haar measure on the orthogonal group $\mathcal{O}(n) = \{\f{A} \in \Rnn: \f{A}\f{A}^T = \f{A}^T \f{A} = \mathbf{I} \}$. In this case, $\f{A}$ has $n^2 = 2500$ degrees of freedom,\footnote{In fact, an orthogonal matrix has only $n(n-1)/2$ degrees of freedom; however, as reconstruction method is oblivious to the orthogonality constraints, we neglect these in our calculation.} and we first fix the total number of space-time samples (uniform sampling) to $m_{\text{total}} = 7500$, corresponding to an oversampling factor of $\rho = 3$. We investigate the performance of the proposed approach when $T$ is between $10$ and $50$, i.e., the average spatial samples per time instance ranges from $750$ to $150$. We report the results in \Cref{table:rom}: our approach is able to recover $\f{A}$ accurately after about $30$ IRLS iterations, even when $T$ grows larger, increasing the apparent nonlinearity of the problem.
This is consistent with our theoretical analysis: in short, in an energy-preserving system, one can trade spatial samples for an equal amount of temporal samples without loss of information. 

\begin{table}[H]
\begin{center}
\small{
\begin{tabular}{|c|c|c|c|c|c|c|cc|||}
\hline
$T$ & $\text{Rec}_{\f{A}}$ & $m_{\text{single}}$ & $\operatorname{dof}(\f{A})$ & $\rho$ & Iterations \\ \hline
10    &        $(2.3 \pm 0.2)\cdot 10^{-14} $             &     750                 & 2500  & 3 &  $20.9 \pm 0.3$       \\ \hline
  20 &         $(5.7  \pm 0.8)\cdot 10^{-14} $              &           375        & 2500     &  3   &   $25.1 \pm 0.6$         \\ \hline
   30 &          $(1.1 \pm 0.2)\cdot 10^{-13} $              &            250     & 2500      & 3    &  $27.5 \pm 0.8$     \\ \hline
      40 &          $(1.1 \pm 0.1)\cdot 10^{-13} $              &             187.5     & 2500        &  3 &       $30.2 \pm 2$     \\ \hline
      50 &          $(1.8 \pm 0.3)\cdot 10^{-13} $              &            150      & 2500        &   3  &  $32.8 \pm 2$     \\ \hline

\end{tabular}
}
\end{center}
\caption{The estimation errors for random orthogonal matrices  of size $50\times 50$ using uniform sampling with replacement.  }\label{table:rom}
\end{table}

\begin{figure}[b]
\begin{center}
    \setlength\figureheight{24mm} 
    \setlength\figurewidth{30mm}
\input{experiment_Orthogonal_PT_uniform.tex}
\caption{Phase transition plot for orthogonal matrices, with oversampling factor $\rho$ on the $x$-axis and time steps $T$ on $y$-axis. Yellow corresponds to exact recovery for all random instances of the problem, blue corresponds to no recovery. Red line: $1 + 0.21 \log(nT)$.}
\label{fig:OrthoMat:experiment}
\end{center}
\end{figure}
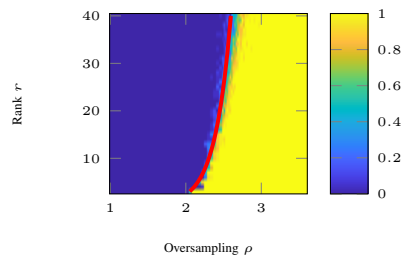
In \Cref{fig:OrthoMat:experiment}, we report on an experiment with the same data and sampling model, but where we vary both the number of time steps $T = 1,\ldots, 40$ and the oversampling factor $\rho = 1,\ldots, 3.5$. For $24$ random instances, we visualize the empirical probability of exact recovery (defined as a relative Frobenius error of $\text{Rec}_{\f{A}} < 10^{-4}$). We observe the existence of a sharp phase transition between no recovery and exact recovery for all instances, at an oversampling factor between $\rho=2$ and $\rho = 2.7$, depending weakly on $T$. This is consistent with \Cref{thm:main:informal}, which predicts exact recovery from $\rho n^2 \gtrsim n^2 \log(n T)$ samples, since here $\mu_0=1$, cf. \Cref{sec:incoherence:estimates}. In fact, the phase transition in \Cref{fig:OrthoMat:experiment} occurs at around $\rho \approx 1+0.21 \log(nT)$ for the tested parameters.

\paragraph{Erd\H{o}s-R\'enyi Graphs.}
In the next experiment, we consider graph matrices $\f{A}$ associated with Erd\H{o}s-R\'enyi graphs \cite{ErdoesRenyi-PMD1959,Gilbert-AnnalsMS1959} with $n = 60$ nodes with connectivity probability of $p = 0.8$. With $\mathcal{T}_{r}(\cdot)$ the map from a matrix to its best rank-$r$ approximation, for $r$ between $1$ and $60$, we create rank-truncated continuous-time heat diffusion operators $\f{A} = \mathcal{T}_{r}\left( \exp(-\tau \f{L})\right) \in \Rnn$, with $\f{L}$ the normalized graph Laplacian as in \Cref{def:graph:operators}, $\tau = 0.4$, on an instantiation of an Erd\H{o}s-R\'enyi graph. In \Cref{fig:ErdoesRenyi:groundtruth:rank20}, we depict $\mathcal{Q}_T(\f{A})$ for such a transition operator, with $r=20$ and $T=7$ time steps. 
\begin{figure*}[b]
\includegraphics[width=1\textwidth]{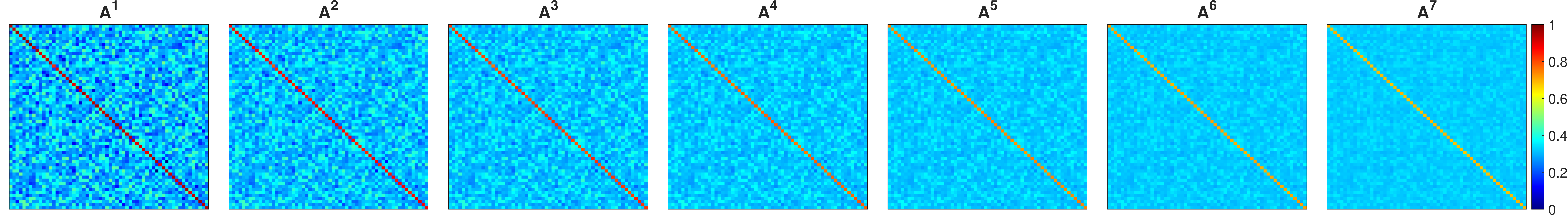}
\caption{Left column: Realization of Erd\H{o}s-R\'enyi graph with $n=60$ and $p=0.8$. Other columns: Aggregation of matrix of powers $\mathcal{Q}_T(\f{A})$ of rank-$20$ truncation $\f{A}$ of heat diffusion operator (color scheme normalized across powers, log-scale).}
\label{fig:ErdoesRenyi:groundtruth:rank20}
\end{figure*}

As they are symmetric, such matrices have $\operatorname{dof}(\f{A}) = r(n-(r-1)/2)$ degrees of freedom. In \Cref{fig:ER:T_dependence}, we visualize the recovery performance of \texttt{TOIRLS} for varying numbers of samples $m = |\Omega|$, for three different numbers of time steps $T$. In the left column of \Cref{fig:ER:T_dependence}, we see that the phase transition for $T=1$ occurs extremely close to the information theoretical limit---in this case, the setting coincides with low-rank matrix completion via \texttt{MatrixIRLS} as described in \cite{KMV21}. For $T=4$, the transition occurs at around $m = 1.5 r n$. Apart from the fact it is expected that generally, the phase transition will occur at larger sample complexities than for $T=1$ due its the logarithmic dependence on $T$, it is remarkable that the quadratic dependence of $\operatorname{dof}(\f{A})$ on $r$ is not reflected in the empirical transition curve. However, this is still compatible with \Cref{thm:convergence:mobilesensors}, as dependence on $r$ in the sufficient condition is linear. As expected due to the logarithmic dependence on $T$, we observe a similar, but slightly deteriorated transition curve for $T=7$.

\begin{figure*}
\begin{subfigure}[c]{0.25\textwidth}
\vspace*{-37.5mm}
\includegraphics[width=1\textwidth]{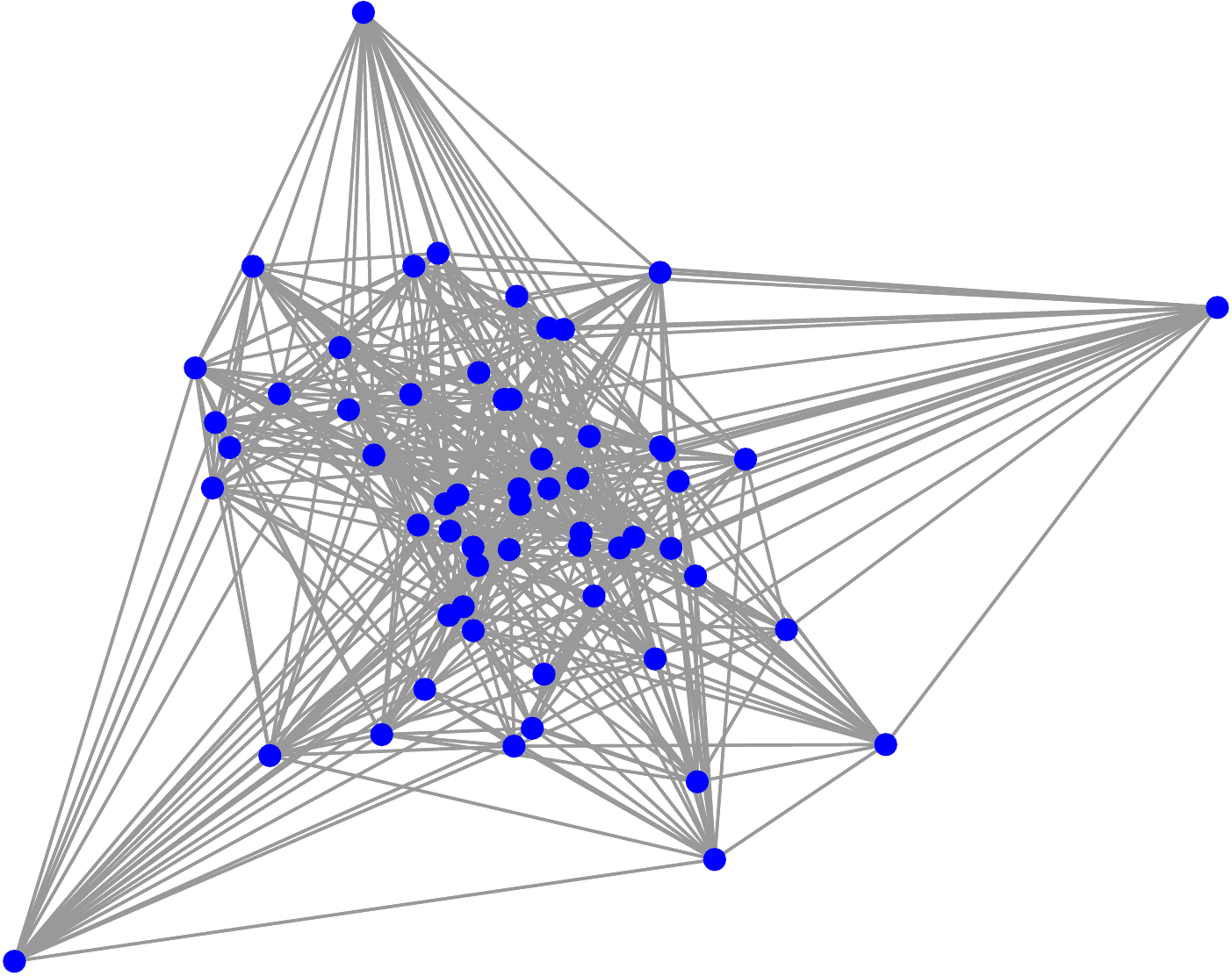}
\end{subfigure}
\begin{subfigure}[b]{0.216\textwidth}
\vspace*{-5mm}
    \setlength\figureheight{24.0mm}
    \setlength\figurewidth{30mm}
    \hspace*{5mm}
\input{experiment_GC_ER_good_PT_uniform_T_1.tex}
\end{subfigure}
\begin{subfigure}[b]{0.256\textwidth}
\hspace*{4mm}
    \setlength\figureheight{24.6mm} 
    \setlength\figurewidth{30mm}
\input{experiment_GC_ER_good_PT_uniform_T_4_pencil_3.tex}
\end{subfigure}
\begin{subfigure}[b]{0.232\textwidth}
\hspace*{2mm}
    \setlength\figureheight{24.6mm}
    \setlength\figurewidth{30mm}
\input{experiment_GC_ER_good_PT_uniform_T_7.tex}
\end{subfigure}
\caption{Left: Erd\H{o}s-R\'enyi graph with $n=60$, $p=0.8$. Center to right: Phase transition plots for rank $r$-approximations of heat diffusion operators of Erd\H{o}s-R\'enyi graphs, \emph{uniform sampling}. Increasing rank on $y$-axis, increasing number of samples $m$ on $x$-axis.
Center left: $T = 1$. Center right: $T=4$. Right: $T=7$. Red curved line: Number of degrees of freedom $\operatorname{dof}(\f{A}) = r(n-(r-1)/2)$ of operator $\f{A}$; Pink dotted line: $1.5 r n$.}
\label{fig:ER:T_dependence}
\end{figure*}

\subsubsection{Choice of Pencil Parameter $d_1$.} \label{sec:numerical:pencildependence}
In the experiments of \Cref{sec:numerical:Tdependence}, we always chose the first pencil parameter $d_1$ so that block Hankel matrices $\y{H}(\f{\widetilde{X}})$ are square or as square as possible, i.e., such that $d_1 = \left\lfloor (T+1)/2\right\rfloor$. 

Revisiting the experiments of \Cref{sec:numerical:Tdependence} for the Erd\H{o}s-R\'enyi graph model and $T=7$ time steps, we now explore the sensitivity of the problem to the choice of $d_1$. In \Cref{fig:ER:pencil_dependence}, we observe that for $d_1=1$/$d_2=7$, the phase transition occurs only for significantly more samples $m$ than for the square choice of $d_1 = d_2 = 4$; for example, it can be seen that for $r=20$, the transition is at $m = 3600$ or $\rho \approx 3.56$ for $d_1=1$, whereas it is at $m=2500$ or $\rho \approx 2.47$ for $d_1=4$. For $d_1=1$, the recovery problem becomes impossible if the rank of $\f{A}$ satisfies $r=60$ due to a lack of any low-rank property of the embedding matrix $\y{H}\big(\y{Q}_T(\f{A})\big)$, and the experiment indicates that even for lower ranks $r \ll 60$, this choice of $d_1$ is disadvantageous. For $d_1=2$ and $d_1=3$, the behavior is quite similar to the square case in this example with a just slightly worse phase transition.

Furthermore, we illustrate in the last column of \Cref{fig:ER:pencil_dependence} the values of the $d_1$-dependent product $c_s \mu_0$, where $c_s = \frac{T(T+1)}{d_1 d_2}$ is the constant of \Cref{def:incoherence} and $\mu_0$ is the incoherence parameter \cref{eq:incoherence} of $\HQTA$ for a given choice of the pencil parameter $d_1$. The values are illustrated with a one standard deviation confidence interval across $24$ realizations of the Erd\H{o}s-R\'enyi model. We observe that $c_s \mu_0$ is minimal for $d_4=1$ for essentially all ranks $r$, indicating that our sample complexity bound \cref{eq:mainthm:uniform:samplecomplexity} in \Cref{thm:convergence:mobilesensors} indeed justifies a square choice for the pencil parameter such that $d_1 = \left\lfloor (T+1)/2\right\rfloor$. 
\begin{figure*}[h]
\begin{subfigure}[b]{0.20\textwidth}
\vspace*{1mm}
    \setlength\figureheight{20.5mm}
    \setlength\figurewidth{25mm}
\input{experiment_GC_ER_good_PT_uniform_T_7_pencil_1.tex}
\end{subfigure}
\begin{subfigure}[b]{0.15\textwidth}
    \setlength\figureheight{20.5mm}
    \setlength\figurewidth{25mm}
\input{experiment_GC_ER_good_PT_uniform_T_7_pencil_2.tex}
\end{subfigure}
\begin{subfigure}[b]{0.15\textwidth}
    \setlength\figureheight{20.5mm}
    \setlength\figurewidth{25mm}
\input{experiment_GC_ER_good_PT_uniform_T_7_pencil_3.tex}
\end{subfigure}
\begin{subfigure}[b]{0.15\textwidth}
    \setlength\figureheight{20.5mm}
    \setlength\figurewidth{25mm}
\input{experiment_GC_ER_good_PT_uniform_T_7_pencil_4.tex}
\end{subfigure}
\begin{subfigure}[b]{0.32\textwidth}
    \setlength\figureheight{20.5mm}
    \setlength\figurewidth{45.5mm}
\input{experiment_GC_ER_T_7_coherences.tex}
\end{subfigure}
\caption{First columns: Experiment as in \Cref{fig:ER:T_dependence} for $T=7$, for different pencil parameters $d_1$. From left to right: $d_1 = 1$, $d_1 = 2$, $d_1 = 3$ and $d_1=4$. Last column: Value of $c_s \mu_0$ of $\HQTA$ for $d \in \{1,\ldots,4\}$.}
\label{fig:ER:pencil_dependence}
\end{figure*}
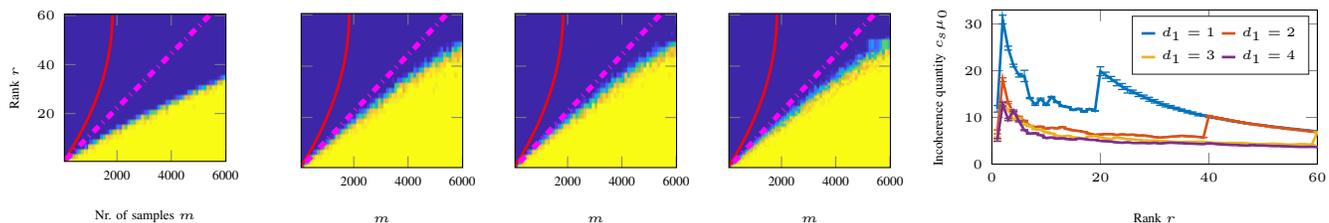

\subsubsection{Uniform Sampling vs. Adaptive Sampling} \label{sec:uniform:adaptive:experiments}
Next, we explore the empirical benefits of \emph{adaptive sampling} compared to uniform sampling for the recovery of transition operators $\f{A}$. In particular, we assume that we have knowledge about the local incoherences $\mu_{i,j,t}$ of $\f{A}$ for all $(i,j,t) \in I$, see \cref{eq:mu:ijk} in \Cref{sec:incoherence}, and design an adaptive sampling scheme with probabilities $p_{i,j,t} = c \mu_{i,j,t}$ for all $(i,j,t) \in I$, where $c>0$ is chosen such that the expected number of samples $m_{\text{exp}} = \mathbb{E}[|\Omega|] = \sum_{(i,j,t) \in I} p_{i,j,t}$. We vary then $m_{\text{exp}}$ in a similar manner as $m$ above for uniform sampling. We note that this is not a very realistic sampling scheme, since local incoherences are not immediately accessible as they require the knowledge of $\f{A}$. An implementable approximation of the ideal adaptive sampling scheme was proposed in \cite{ChenBhoSangWard15} for the related low-rank matrix completion problem; however, an application to our setting is beyond the scope of this paper.

\begin{figure*}[t]
\includegraphics[width=1\textwidth]{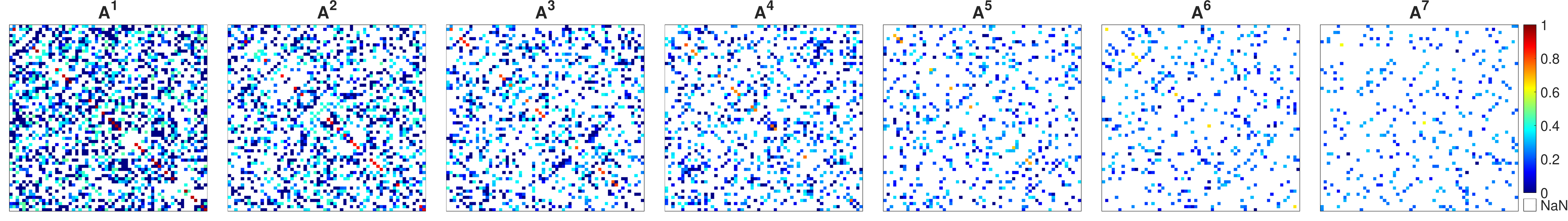}
\caption{Adaptive space-time samples $P_{\Omega}(\mathcal{Q}_T(\f{A}))$ of rank-$20$ truncation $\f{A}$ of heat diffusion operator ($1010$ degrees of freedom, log-scale) with $m_{\text{exp}} = 3000$.}
\label{fig:ErdoesRenyi:sampling:rank20}
\end{figure*}
In \Cref{fig:ErdoesRenyi:sampling:rank20}, we illustrate a realization of an expected number of $m_{\text{exp}} = 3000$ adaptive samples in the Erd\H{o}s-R\'enyi setting of \Cref{sec:numerical:Tdependence}, corresponding to an oversampling factor $\rho \approx 2.97$. Applying \texttt{TOIRLS} to the recovery of heat diffusion operators associated with Erd\H{o}s-R\'enyi graphs from adaptive sampling, we report the results of the experiment of \Cref{sec:numerical:Tdependence} in \Cref{fig:ER:T:adaptivesampling}. It can be seen that for $T=1$ the phase transition is very similar to the one corresponding to uniform sampling (see \Cref{fig:ER:T_dependence}), as it was already close to the information theoretic threshold $\rho = 1$ (red curve). For the dynamic cases $T=4,7$, we see that we obtain a modest improvement compared to uniform sampling, with the phase transition exceeding the line $m_{\text{exp}} = 1.5 rn$ especially for large $r$ and $T=4$, and achieving recoverability of full rank operators from around $6,000$ samples for $T=7$, unlike in the uniform sampling case.

\begin{figure*}[t]
\begin{subfigure}[b]{0.33\textwidth}
    \setlength\figureheight{25.6mm}
    \setlength\figurewidth{25mm}
    \hspace*{10mm}
\input{experiment_GC_ER_good_PT_adaptive_T_1.tex}
\end{subfigure}
\begin{subfigure}[b]{0.33\textwidth}
\hspace*{8mm}
    \setlength\figureheight{25.6mm}
    \setlength\figurewidth{30mm}
\input{experiment_GC_ER_good_PT_adaptive_T_4_pencil_3.tex}
\end{subfigure}
\begin{subfigure}[b]{0.33\textwidth}
\hspace*{7mm}
    \setlength\figureheight{25.6mm}
    \setlength\figurewidth{30mm}
\input{experiment_GC_ER_good_PT_adaptive_T_7.tex}
\end{subfigure}
\caption{Experiment as in \Cref{fig:ER:T_dependence} for \emph{adaptive sampling} with $m_{\text{exp}}$ expected space-time samples. Left: $T = 1$. Center: $T=4$. Right: $T=7$. Red curved line: Number of degrees of freedom $\operatorname{dof}(\f{A}) = r(n-(r-1)/2)$ of operator $\f{A}$; Pink dotted line: $1.5 r n$.}
\label{fig:ER:T:adaptivesampling}
\end{figure*}
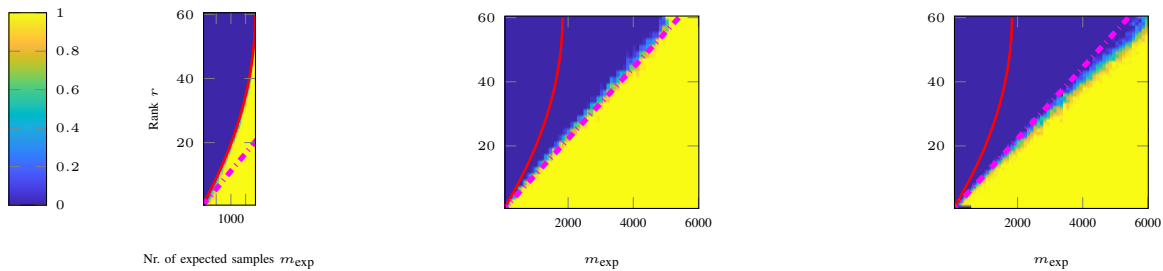

\paragraph{Community Graphs}
\begin{figure*}[t]
\includegraphics[width=1\textwidth]{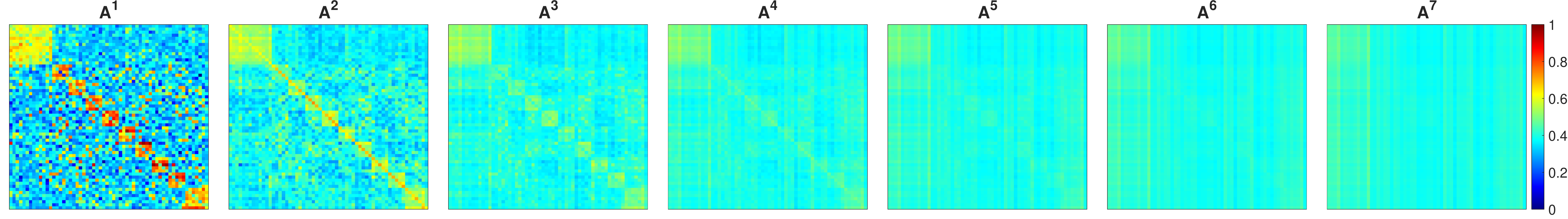}
\includegraphics[width=1\textwidth]{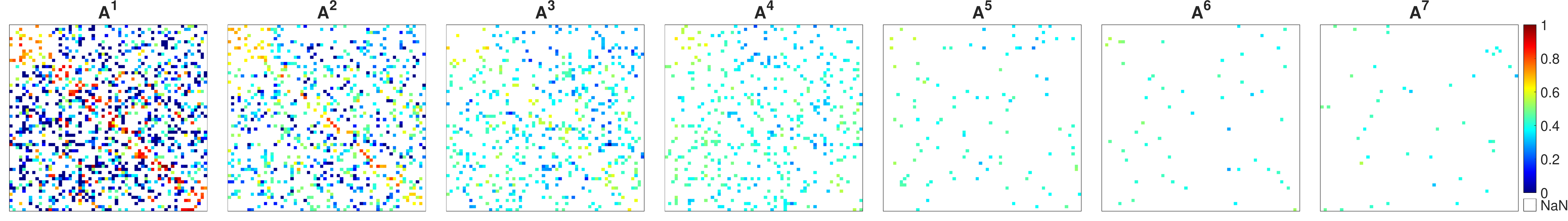}
\caption{First row: Rank-$r$ truncation $\f{A}$ of random walk transition operator of random community graph with $60$ nodes for $r=20$, aggregation $\mathcal{Q}_T(\f{A})$ of $T=7$ matrix powers; Second row: Adaptive space-time samples $P_{\Omega}(\mathcal{Q}_T(\f{A}))$ of rank-$20$ truncation $\f{A}$ (log-scale) with $m_{\text{exp}} = 3000$}
\label{fig:Community:groundtruth}
\end{figure*}
the improved efficiency of the adaptive sampling scheme has been rather modest for the heat diffusion operator based on a Erd\H{o}s-R\'enyi graph in our parameter setting due to the relatively benign spectral decay of $\f{A}$. We now consider a \emph{community graph} with $n=60$ vertices and $10$ communities (eight of size $5$, one of size $13$ and one of size $7$), with dense connections within a community, and independent random inter-community edges with probability $1/10$. We use the Graph Signal Processing (GSP) toolbox \cite{perraudin2014} to create such graphs, and define the associated transition matrix as the truncated random walk matrix $\f{A} = \mathcal{T}_{r}\left(\f{P}\right) = \mathcal{T}_{r}\left(\f{D}^{-1}\f{W}\right)$, cf. \Cref{def:graph:operators}. Note that this matrix is in general asymmetric. For $r=20$ and $T=7$, we visualize $\mathcal{Q}_T(\f{A})$, i.e., $\f{A}$ and its powers $\f{A}^2, \ldots, \f{A}^7$ in \Cref{fig:Community:groundtruth}, together with an example of adaptive samples for this setting with $m_{\text{exp}}=3,000$, computed based on local incoherences. Comparing \Cref{fig:Community:groundtruth} with the adaptive sampling pattern for the Erd\H{o}s-R\'enyi heat diffusion model (\Cref{fig:ErdoesRenyi:sampling:rank20}), we note that the sampling density for larger time steps such as $t=5,6,7$ is {smaller} for community graphs, indicating that the adaptive sampling focuses now more on smaller time scales than for the Erd\H{o}s-R\'enyi model. This is expected, since the spectrum of the (untruncated) transition matrix decays faster than for the Erd\H{o}s-R\'enyi heat diffusion operator above, indicating that sampling large time steps is less informative than sampling earlier time steps, see also \Cref{fig:Community:groundtruth}.

Empirically, this is confirmed in the experiment of \Cref{fig:Community:T7}, where we report on the phase transition for both adaptive sampling and uniform sampling, considering $T=7$ steps of a random walk. Unlike for the Erd\H{o}s-R\'enyi transition operators (\Cref{fig:ER:T_dependence} and \Cref{fig:ER:T:adaptivesampling}), we observe a significant difference between adaptive and uniform sampling for this model, as the uniform sampling scheme requires approximately the double amount of samples to obtain exact recovery, with this phase transition being located at around $m = 4.8 r n$ (uniform sampling) and $m = 2.4 rn $ (adaptive sampling), respectively. 

\begin{figure*}[t]
\begin{subfigure}[c]{0.25\textwidth}
\vspace*{-37.5mm}
\includegraphics[width=1\textwidth]{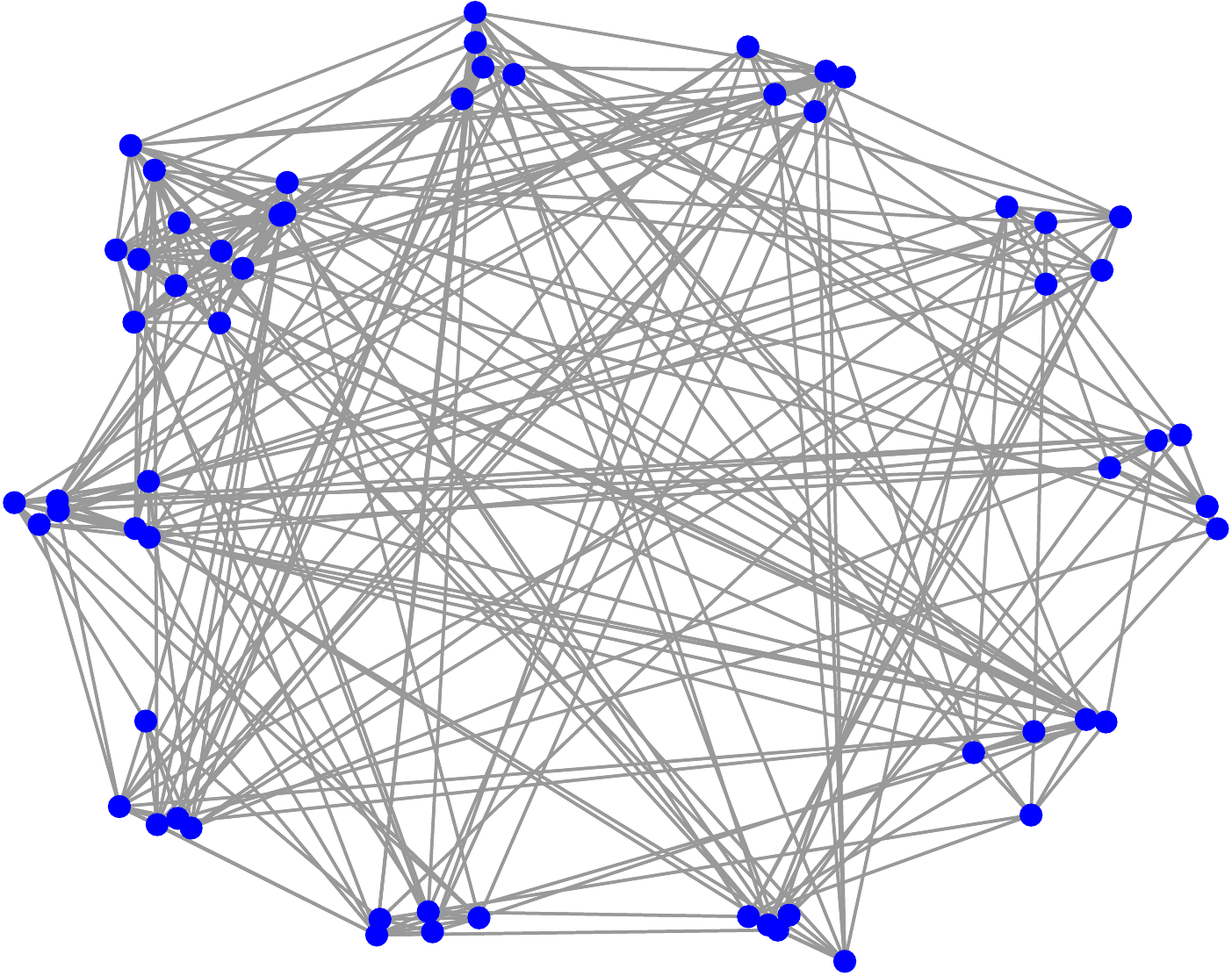}
\end{subfigure}
\begin{subfigure}[b]{0.33\textwidth}
\hspace*{2mm}
    \setlength\figureheight{24.6mm}
    \setlength\figurewidth{30mm}
\input{experiment_GC_Community_uniform_T_7.tex}
\end{subfigure}
\begin{subfigure}[b]{0.33\textwidth}
\hspace*{2mm}
    \setlength\figureheight{24.6mm}
    \setlength\figurewidth{30mm}
\input{experiment_GC_Community_adaptive_T_7.tex}
\end{subfigure}
\caption{Left: Community graph with $n=60$ vertices in $10$ communities; Center and right: Phase transition, uniform vs. adaptive sampling for the recovery of truncated random walk matrices of random community graph, $T=7$ time steps. Center: Uniform sampling, pink dotted line: $4.8 r n$. Right: Adaptive sampling, pink dotted line: $2.4 r n$. Red curved line: $\operatorname{dof}(\f{A})= r(2n-r)$ of $(n \times n)$-matrix $\f{A}$ of rank-$r$.}
\label{fig:Community:T7}
\end{figure*}

\subsubsection{Dependence on Graph Topology} \label{sec:dep:graphtopology}
We now elucidate how the recovery of transition operators $\f{A}$ by \Cref{algo:IRLS:graphcompletion} depends on the \emph{topology} of an underlying \emph{graph}.

\paragraph{Random walk matrix} We recall that a random walk matrix, cf. \Cref{def:graph:operators}, is suitable to reveal structural information of a graph: The multiplicity of the eigenvalue $1$ is equal to the number of connected components; the second largest eigenvalue $\lambda_2$ that describes the mixing rate of the random walks; the spectral gap $|\lambda_1-\lambda_2|$ represents how well the graph is connected. We refer the readers to \cite{chung1997spectral} for a detailed discussion.

In \Cref{table:randomwalk}, we report on experiments on the recovery of (full-rank) random walk matrices $\f{A}=\f{P} = \f{D}^{-1}\f{W}$ associated to two different graphs, both with $n=50$ nodes: A very regular \emph{path graph} , and a more irregular community graph (one community of size $9$, eight of size $5$, and a single node) with inter-cluster connection probability of $1/50$, cf. \Cref{fig:random:graphmodels:community:path}. As in \Cref{sec:uniform:adaptive:experiments}, we use both uniform and adaptive sampling. We denote the number of degrees of freedom by $\operatorname{dof}(\f{A})$ and, in case of adaptive sampling, the number of samples located at time $t=1$ (averaged across 10 realizations) as $\overline{\text{m}_{\text{1}}}$. 

We observe that for the path graph, recovery by \Cref{algo:IRLS:graphcompletion} is possible for uniform sampling at an oversampling factor of $\rho=3$, as the recovery error $\text{Rec}_{\f{A}}$ is of the order of magnitude of the algorithmic tolerance with $\text{Rec}_{\f{A}} \approx 10^{-10}$, while exact recovery fails for $\rho=2.8$ even if adaptive sampling is chosen; i.e., the performance is essentially the same for uniform and adaptive sampling. For the community graph, for which we now consider $T=10$ time steps instead of $T=5$, an oversampling factor of $\rho = 8$ is not sufficient for uniform sampling to recover the transition operator, however, a much smaller sample complexity corresponding to $\rho=3.5$ leads already to exact recovery for adaptive sampling. While these graphs are simple examples, they show illustrate that the difficulty of the setup is negatively affected by the irregularity of the underlying graph.

\begin{figure}[H] 
\begin{center}
\begin{subfigure}[c]{0.24\textwidth}
\includegraphics[width=\linewidth, keepaspectratio]{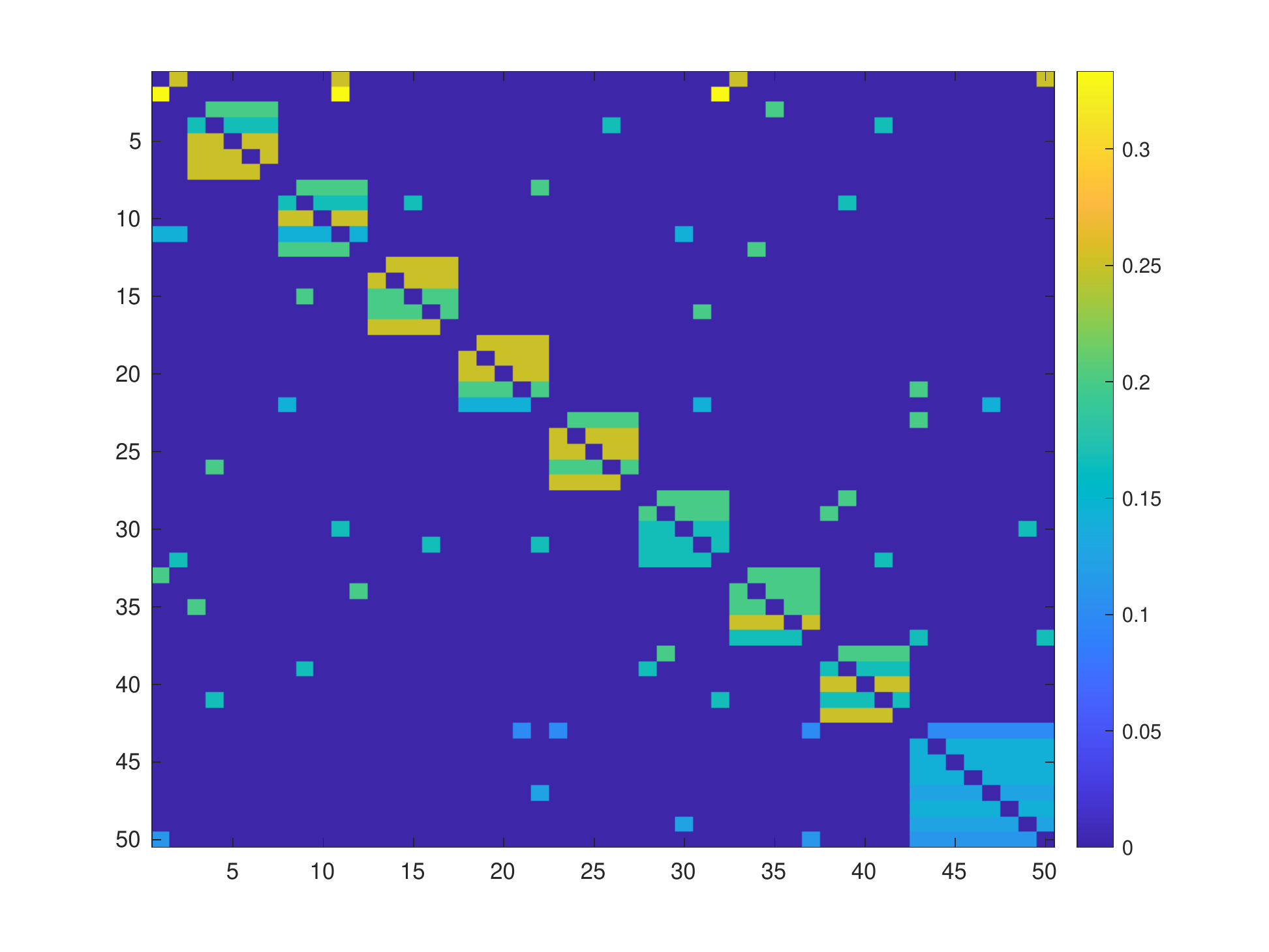}
\caption{Community graph}
\end{subfigure}
\begin{subfigure}[c]{0.24\textwidth}
\includegraphics[width=1\linewidth, keepaspectratio]{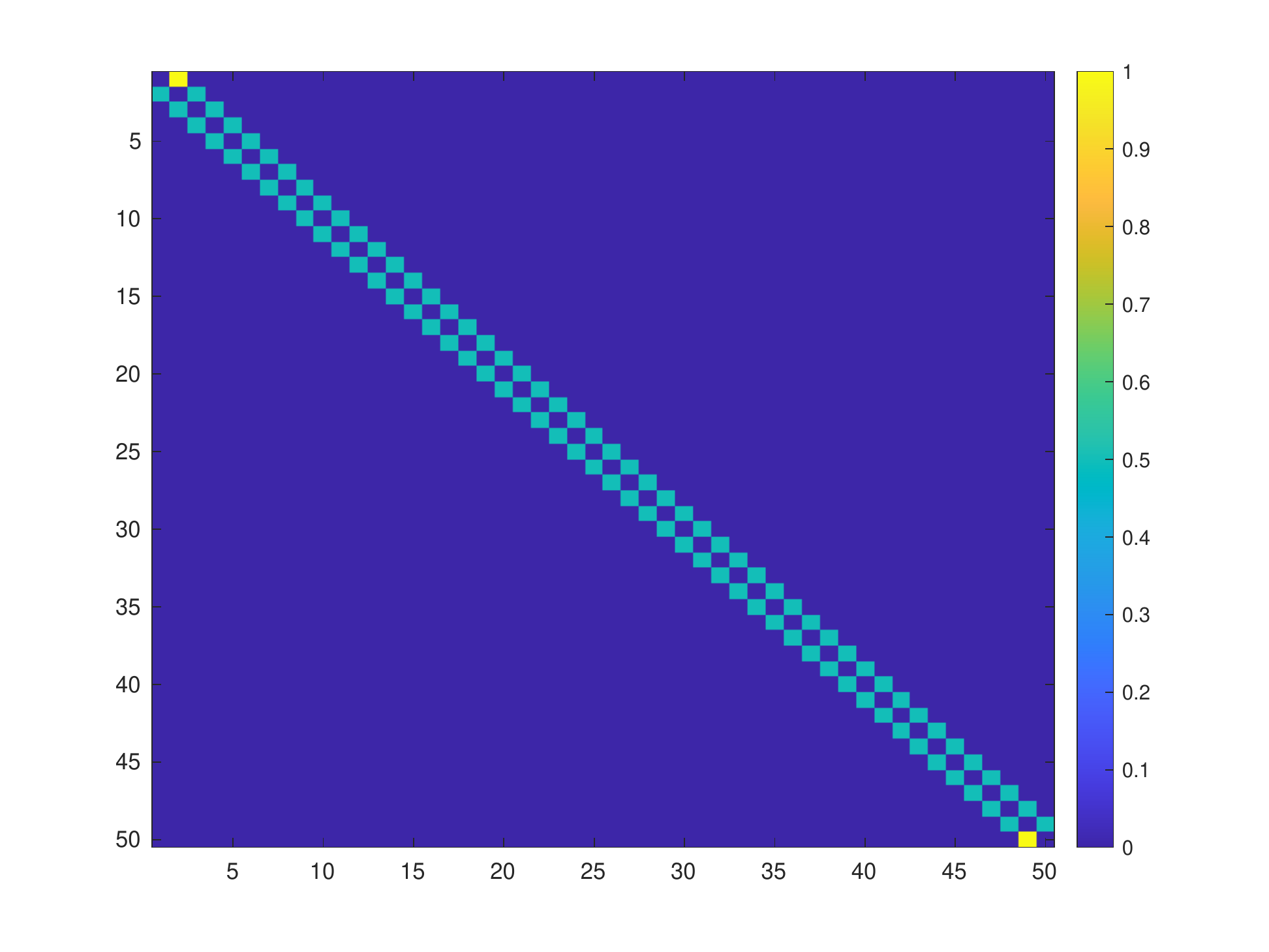}
\caption{Path graph}
\end{subfigure}
\end{center}
\caption{\footnotesize The plot of two random walk matrices used in the simulation. Both of matrices are sparse with significant nonzero entries. } \label{fig:random:graphmodels:community:path}
\vspace{-0.1in}
\end{figure}

\begin{table*}[t]
\begin{center}
\begin{tabular}{| l| l| l| l| l| l| l| l|}
\hline
Models & Sampling & $T$ &  $\text{m}_{\text{single}}$ & $\overline{\text{m}_{\text{1}}}$ &$\operatorname{dof}(\f{A})$& $\rho$ &$ \text{Rec}_{\f{A}}$  \\ \hline
Path graph    &     uniform &    5  &    1500 & &  2500      &          3    &     $9.6\cdot 10^{-11} \pm 9.2\cdot 10^{-20}$\noteC{How is it possible that the standard deviation is so small?}   \\ \hline
Path graph    &     adaptive &    5  &   & 2169 &  2500      &          2.8    &     $7.2\cdot 10^{-4} \pm 2.3\cdot 10^{-3}$   \\ \hline
Community        &     uniform &    10 &  2000& &  2500 & 8& $7.6\cdot 10^{-5} \pm \cdot  5.8\cdot10^{-8} $            \\ \hline
Community        &     adaptive &    10 &  &  2297&   2500 & 3.5& $4.4\cdot 10^{-13} \pm \cdot  5.8\cdot10^{-13} $            \\ 
  \hline
\end{tabular}\caption{Recovery errors $\text{Rec}_{\f{A}}$ of \Cref{algo:IRLS:graphcompletion} for random walk matrix of path/community graphs for different sampling sets.}\label{table:randomwalk}

\end{center}
\end{table*}

\paragraph{Heat diffusion operator} 
We now revisit the heat diffusion operators  $\f{A} = \mathcal{T}_{r}\left( \exp(-\tau \f{L})\right)$ from \Cref{sec:numerical:Tdependence}, focussing on how the number of time steps $T$,  the heat diffusion scale  $\tau$ and the structure of the underlying graph $\mathcal{G}$ determine the recoverability of $\f{A}$ from time-space samples that are sampled uniformly at random, for settings of slightly larger scale. To that end, we consider a Swiss roll graph with $n = |V| = 200$ nodes and graph representing the roads of the state of Minnesota \cite{UFSparseMatrix-2011,Kolodziej-2019SuiteSparse} with $n=|V|=2642$ nodes, using the default settings of the GSP toolbox \cite{perraudin2014}; see Figure \ref{fig:random:graphmodels} for a visualization.

 \begin{figure}[H] 
\centering
\includegraphics[width=.22\textwidth, keepaspectratio]{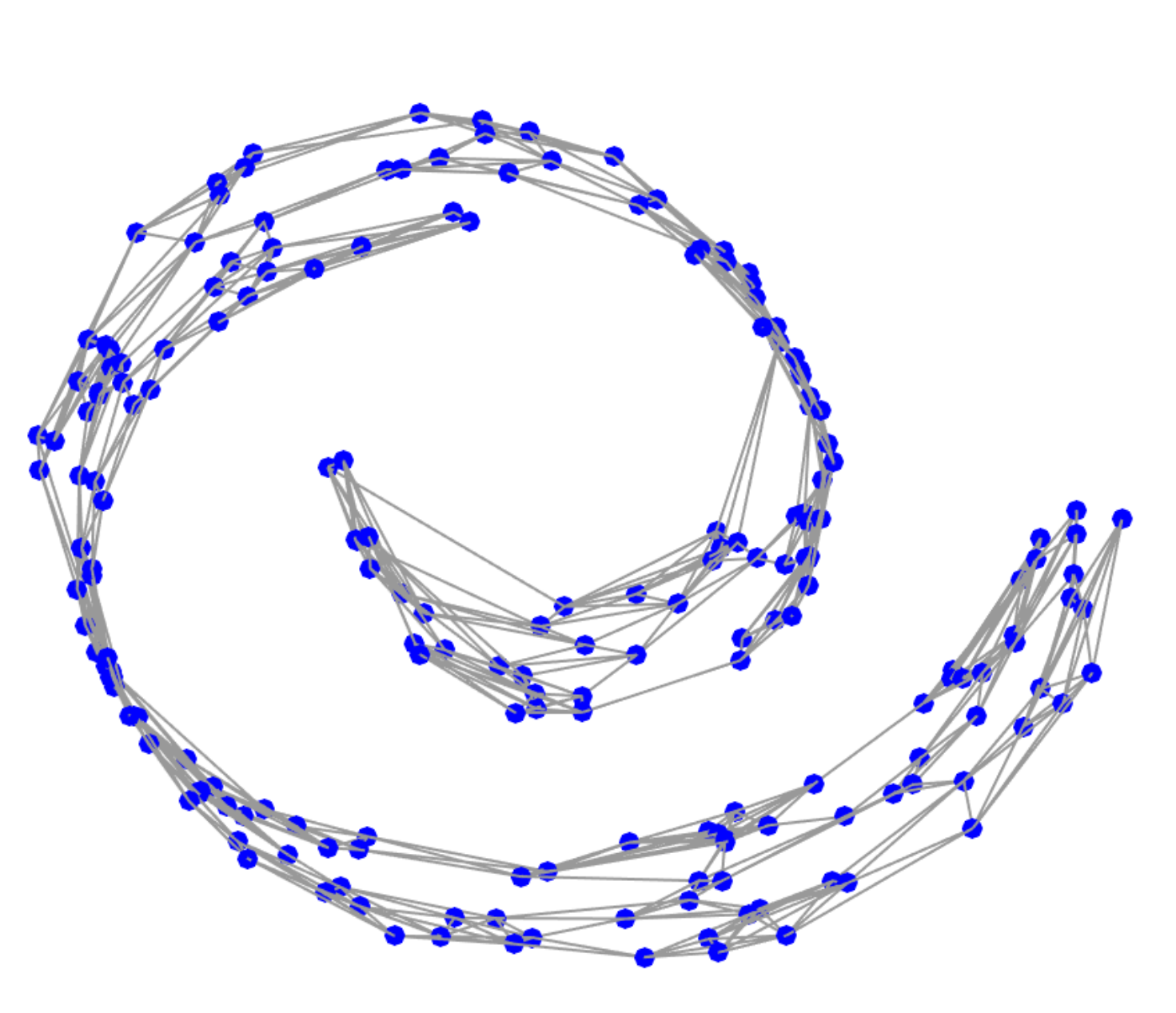}
\includegraphics[width=.22\textwidth, keepaspectratio]{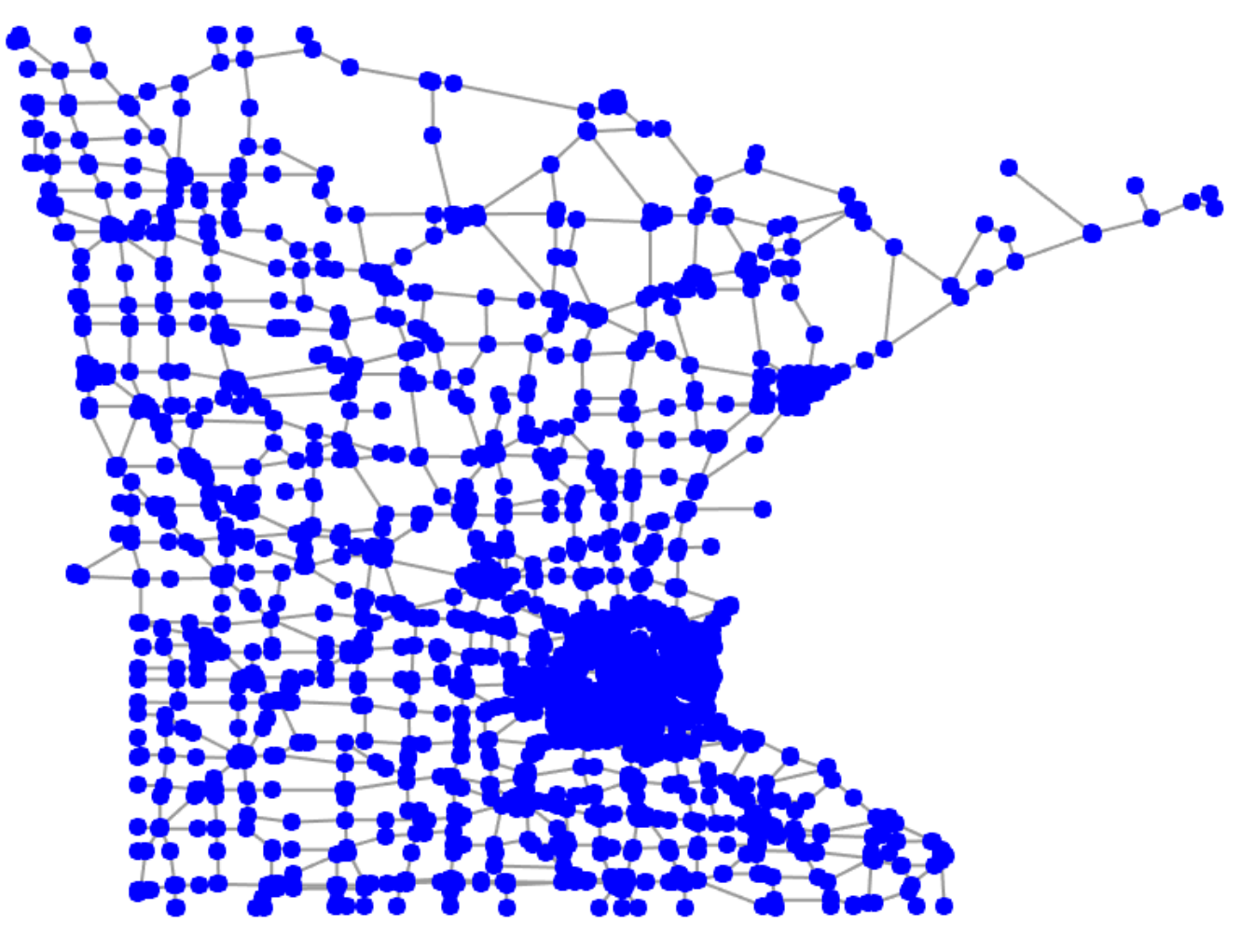}
\caption{\footnotesize Left: Swiss roll graph. Right: Minnesota road network graph.}\label{fig:random:graphmodels}
\vspace{-0.1in}
\end{figure}

Furthermore, we consider heat diffusion operators corresponding to \emph{slow} and \emph{fast} energy dissipation, corresponding to small and large choices of $\tau$. Since we choose $r=10$, we have that $\lambda_{min}(\f{A})=\exp(-\tau\lambda_{10}({\f{L})})$, which is larger for slow energy dissipation and smaller or faster energy dissipation. Here, $\lambda_{10}({\f{L}})$ corresponds to the $10$th largest eigenvalue of the Laplacian $\f{L}$.

\begin{table*}
\begin{center}
\begin{tabular}{|l|l|l|l|l|l|l|l|l|l}
\hline
Models  &  $\lambda_{min}(\f{A})$ &$T$ &  $\text{m}_{\text{single}}$ & $\operatorname{dof}(\f{A})$ &$\rho$ &$\text{Rec}_{\f{A}}$ \\ \hline
Swiss roll            &     0.21&   6 &1629    & 1955&  5&  $1.4\cdot 10^{-13} \pm 4.2\cdot 10^{-14} $   \\ \hline
Swiss roll            &     0.89&   6&  978  & 1955&  3&  $1.9\cdot 10^{-12} \pm 3.4\cdot 10^{-12} $   \\ \hline
Minnesota            &    0.64&   5 &  32082  & 26375&  6&  $8.3\cdot 10^{-13} \pm 5.2\cdot 10^{-13} $   \\ \hline
Minnesota            &     0.9&   5 &  21100  & 26375&  4&  $3.8\cdot 10^{-11} \pm 6.6\cdot 10^{-11} $   \\ \hline
\end{tabular}\caption{Parameter choices for accurate recovery of rank-$10$ heat diffusion operators using uniform sampling with replacement, sample complexities at phase transition.} \label{table:heatdiffusion} 
\end{center}
\end{table*}
In \Cref{table:heatdiffusion}, we report parameter choices and sample sizes and the smallest oversampling factor $\rho$ (stepsize $0.5$) that allows for accurate recovery, i.e., $\text{Rec}_{\f{A}}$ of the order of the stopping condition, of the transition operator via \Cref{algo:IRLS:graphcompletion}. We observe the transition operator is recoverable even if the expected number of samples $\text{m}_{\text{single}}$ for a single time step is below the number of degrees of freedom $\operatorname{dof}(\f{A})$, which would not be possible in the static setting of $T=1$. Furthermore, we see that in the case of faster energy dissipation ($\lambda_{min}(\f{A})$ small), the threshold oversampling factor $\rho$ is larger. This is consistent with the worse bound $\widetilde{\mu}_0$ in \cref{eq:mu:tilde:mobile:uniform} and the sufficient condition \cref{eq:mainthm:uniform:samplecomplexity} for the local convergence of \Cref{algo:IRLS:graphcompletion}.

\subsection{Comparison with Black-Box Nonlinear Optimization}  
We now compare the performance of \Cref{algo:IRLS:graphcompletion} for the problem studied in this paper to the one of a black-box nonlinear optimization solver applied to the nonlinear least squares objective \cref{eq:nonlinear:leastsquares} \cite{waltz2006interior}, as used by the wrapper function \texttt{fmincon} of MATLAB. For \texttt{fmincon}, we use the zero-padded observations $P_{\Omega}^* P_{\Omega}(\mathcal{Q}_T(\f{A}))$ as initialization.

Unlike \Cref{algo:IRLS:graphcompletion}, this method is not able to algorithmically utilize the low-rank structure of the problem in the case of rank-truncated transition operators $\f{A}$, which is why it is not suitable to handle large-scale problem instances with many unknowns such as, for example, those associated to the Minnesota road network graph considered in \Cref{table:heatdiffusion}--in fact, it is infeasible to run it on a personal computer already for transition operator sizes of $n > 200$, unlike \Cref{algo:IRLS:graphcompletion}.

Instead, we consider rank-truncated heat diffusion operator associated to the Swiss roll graph used in the first row of \Cref{table:heatdiffusion}, and the random walk matrices associated to a path graph and the community graph of \Cref{table:randomwalk}, each with uniform sampling.

Setting the maximal number of iterations equal to $200$, we report the observed recovery errors $\text{Rec}_{\f{A}}$ in \Cref{table:randomwalk:fmincon}. We see that for the oversampling factors for which \Cref{algo:IRLS:graphcompletion} essentially leads to exact recovery, \texttt{fmincon} exhibits recovery errors of the order $10^{-1}$ or $10^{-2}$ for the Swiss roll and community graph models, indicating that exact recovery does not happen. For the path graph, on the other hand, the recovery error is of order $10^{-7}$.
Taking the often larger computational cost of the generic nonlinear optimization solver into account, while it cannot be ruled out that an increase in the iteration number will eventually lead to smaller errors, we conclude that \Cref{algo:IRLS:graphcompletion} works significantly better than \texttt{fmincon} for irregular graphs.   

\begin{table*}[t]{\small 
\begin{center}
\begin{tabular}{|l|l|l|l|l|l|l|l|l|}
\hline
Models & Sampling & $n$ & Rank $r$& $T$ &  $\text{m}_{\text{single}}$ or $\overline{\text{m}_{\text{1}}}$ &$\operatorname{dof}(\f{A})$& $\rho$ &$ \text{Rec}_{\f{A}}$  \\ \hline
Swiss roll  & uniform & 200 & 10 & 6 & 1629 & 1955 & 5 & $5.5\cdot 10^{-1} \pm 3.3\cdot 10^{-2} $ \\ \hline
Path graph & adaptive & 50 & 50 & 5 & 2169 & 2500 & 3 & $6.8\cdot 10^{-8} \pm  2.4\cdot 10^{-8}$ \\ \hline
Community & adaptive  & 50 & 50 & 10 & 2297 & 2500 & 3.5 & $1.5\cdot 10^{-2} \pm 2.3\cdot 10^{-3} $ \\ 
  \hline
\end{tabular}
\caption{Recovery errors using interior-point solver of nonlinear least squares formulation \cref{eq:nonlinear:leastsquares} (\texttt{fmincon}).} \label{table:randomwalk:fmincon}
\end{center}}
\end{table*}

\subsection{Robustness to Noisy Observations} \label{sec:numerical:noisy:observations}
In all previous experiments in \Cref{sec:numerical:experiments}, we have assumed that the observations $\f{y} \in \R^{m}$ provided to the recovery method correspond to \emph{exact} space-time samples $\f{y} = P_{\Omega}(\y{Q}_T(\f{A}))$. While, taken literally, the local convergence statements of \Cref{thm:convergence:mobilesensors} only apply to this setting, for practical applicability it is important that the problem is also solvable in the presence of \emph{additive noise} such that $\f{y} = P_{\Omega}(\y{Q}_T(\f{A})) + \eta$, where the noise $\eta$ is unknown to the algorithm.  We investigate the noise robustness of the IRLS approach of  \Cref{algo:IRLS:graphcompletion} by reconsidering the Erd\H{o}s-R\'enyi and community graph transition operator models of \Cref{sec:numerical:Tdependence} and \Cref{sec:uniform:adaptive:experiments} for noisy observations with random spherical noise such that $\f{y} = P_{\Omega}(\y{Q}_T(\f{A}))  + \eta = P_{\Omega}(\y{Q}_T(\f{A})) + \frac{\|P_{\Omega}(\y{Q}_T(\f{A}))\|_2}{\sqrt{\text{SNR}}} \f{v}$, where $\f{v}$ is a vector drawn uniformly at random from the unit sphere and SNR correponds to the signal-to-noise ration $\text{SNR}  = \|P_{\Omega}(\y{Q}_T(\f{A}))\|_2^2/ \|\eta\|_2^2$. Despite the presence of noise, we apply  \Cref{algo:IRLS:graphcompletion} with regularization parameter $\lambda = 0$. 
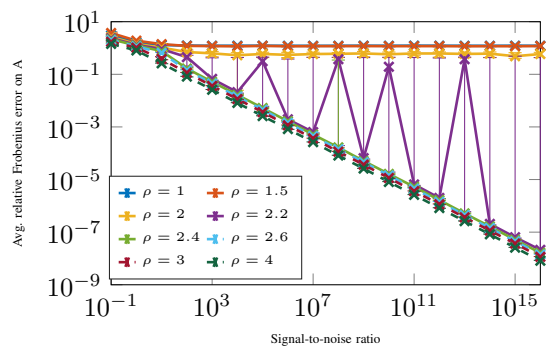
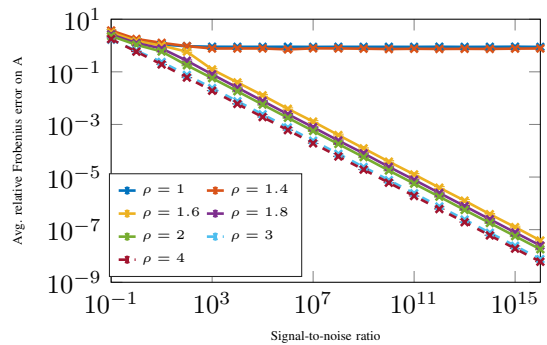
\begin{figure}[h]
\begin{subfigure}[b]{0.50\textwidth}
    \setlength\figureheight{35mm}
    \setlength\figurewidth{60mm}
    \input{experiment_GC_ER_T_7_noisy.tex}
    \caption{Heat diffusion operator, Erd\H{o}s-R\'enyi graph}
\end{subfigure}
\begin{subfigure}[b]{0.50\textwidth}
    \setlength\figureheight{35mm}
    \setlength\figurewidth{60mm}
        \input{experiment_GC_Community_T_7_noisy.tex}
        \caption{Random walk, Community graph}
\end{subfigure}
\caption{Median recovery error $\text{Rec}_{\f{A}}$ with $25\%$ and $75\%$ quantiles, vs. signal-to-noise ratio (SNR), $100$ realizations, different oversampling factors $\rho$. $T=7$ time steps, $n=50$ nodes, rank $r=20$.}
\label{fig:noisy:observations}
\end{figure}

We observe that for sample complexities below the phase transition thresholds in \Cref{sec:numerical:Tdependence} and \Cref{sec:uniform:adaptive:experiments}, the resulting recovery errors $\text{Rec}_{\f{A}}$ are consistently of the order $10^{-1}$ to $10^{1}$, as even in the noiseless case recovery is not possible. For $\rho$ chosen above the phase transition threshold, we observe a linear decrease in $\text{Rec}_{\f{A}}$ with respect to the SNR in the log-log plots of \Cref{fig:noisy:observations} with an approximate slope of $-1/2$, empirically supporting the relationship
\begin{equation} \label{eq:noise:dependence}
\|\f{\hat A} -\f{A}\|_F \asymp {1}/{\sqrt{\text{SNR}}} \asymp \|\eta\|_2
\end{equation}
whenever exact recovery occur in the noiseless case. For $\rho=2.2$ in \Cref{fig:noisy:observations}(a), we see that the accuracy of the outputs of \Cref{sec:numerical:experiments} has large variance: this is because the sample complexity is set to be right at the phase transition for the uniform sampling (see \Cref{fig:ER:T_dependence}). 

We recall that in view of \cref{eq:Hankel:rank:min:Glambda} and the majorization-minimization interpretation of IRLS \cite{Daubechies10,KMV21,KFMMV21}, it is possible to choose a regularization parameter $\lambda$ in  \Cref{algo:IRLS:graphcompletion} that is adapted to the noise level via cross validation, leading to a potential improvement in the dependency of $\text{Rec}_{\f{A}}$ with respect to $\|\eta\|_2$. This necessitates determining an additional free parameter in the method. In fact, \Cref{fig:noisy:observations} and \eqref{eq:noise:dependence} suggest that the improvement might be modest, of the order of a constant, and that the choice $\lambda = 0$ may be a valid option even in the case of noisy observations.

\section{Proof of the Low-Rank Property of Block Hankel Matrix}  \label{sec:proof:vandermonde}
Before proceeding with the proof of \Cref{thm:blkHankel:lowrank}, we show a corollary of \Cref{thm:blkHankel:lowrank} that generalizes the well-known Vandermonde decomposition for Hankel matrices. 
\begin{corollary}[Generalized Vandermonde Decomposition] \label{corollary:generalized:vandermonde}
Let $\f{H} \in \Rddn$ be a block Hankel matrix with $(n \times n)$ blocks, that has a positive semidefinite square extension $\overset{\square}{\f{H}} \in \y{M}_{Dn,Dn}$, $D=\max(d_1,d_2)$, such that
	\begin{itemize}
	\item its first block $\f{H}_1 \in \y{M}_{n}$ if of rank $r$, and
	\item at least one other block $\f{H}_{j} \in \y{M}_n$, $j>1$, is of rank $r$.
	\end{itemize}
Then there exists a triple $(\f{U},\f{N},\Sigma)$ with $\f{U} \in \y{M}_{n,r}$ with orthonormal columns, $\Sigma \in \y{M}_r$ positive definite diagonal, and $\f{N} \in \y{M}_r$ that is $\Sigma$-self-adjoint\footnote{which means that $\f{N} \Sigma = \Sigma \f{N}^*$.}, such that each block $\f{H}_k$ satisfies
	\[
	\f{H}_k = \f{U} \f{N}^{k-1} \Sigma \f{U}^*\,,
	\]	
	and $\f{H}$ has the generalized Vandermonde decomposition
	\begin{equation} \label{eq:gen:Vandermonde}
	\f{H} = \mathcal{V}_{d_1}(\f{U},\f{N}) \Sigma \mathcal{V}_{d_2}(\f{U},\f{N})^*,
	\end{equation}
	where
	\[
	\begin{split}
	\mathcal{V}_m(\f{N},\Lambda)&= \begin{pmatrix} (\f{U} \f{N}^0)^{*} &  (\f{U} \f{N}^1)^{*} & \hdots & (\f{U} \f{N}^{m-1})^{*} \end{pmatrix}^* \\
	&\in \R^{n m \times r}
	\end{split}
	\] 
	is a \emph{generalized Vandermonde matrix} with $m \in \N$.
\end{corollary}

\begin{proof}[{Proof of \Cref{corollary:generalized:vandermonde}}]
If $\f{Y} \in \y{M}_{n,r}$ and $\f{M} \in \y{M}_{r}$ with $\rank(\f{H}) = \rank(\f{Y}) = \rank(\f{M}) = r$ are the matrices from \Cref{thm:blkHankel:lowrank}.2, we can write the square extension $\overset{\square}{\f{H}}\in \y{M}_{Dn,Dn}$ of $\f{H}$ as $\overset{\square}{\f{H}}  = \f{L}\f{L}^*$, where $ \f{L} = \begin{bmatrix} \f{L}_{1}^*, &\ldots, &\f{L}_D^*  \end{bmatrix}^* \in \y{M}_{Dn,r}$ is a block matrix with $\f{L}_{1} = \f{Y}$ and 
\[
\f{L}_{j+1} = \f{Y} \f{M}^{j}
\]
for $1 \leq j \leq D-1$. Let $\f{Y}  = \f{U} \widetilde{\Sigma} \f{V}^*$ be a singular value decomposition, where $\widetilde{\Sigma} \in \y{M}_{r}$ contains the non-zero singular values of $\f{Y}$, $\f{U} \in \y{M}_{n,r}$ the corresponding left singular vectors and $\f{V} \in \y{M}_{r}$ the right singular vectors in their columns. Defining $\Sigma = \widetilde{\Sigma}^2$ and $\f{N}= \widetilde{\Sigma} \f{V}^*\f{M} \f{V} \widetilde{\Sigma}^{-1}$, we can write each block $\f{H}_k$ of the block Hankel matrix $\f{H}$ as
\begin{align*}
\f{H}_k &= \f{Y} \f{M}^{k-1} \f{Y}^* = \f{U} \widetilde{\Sigma} \f{V}^*  \f{M}^{k-1} \f{V} \widetilde{\Sigma}^* \f{U}^* \\&= \f{U} \widetilde{\Sigma} \f{V}^*  \f{M}^{k-1} \f{V} \widetilde{\Sigma}^{-1}  \widetilde{\Sigma}^2 \f{U}^* = \f{U} \f{N}^{k-1} \Sigma \f{U}^*,
\end{align*}
since
\begin{align*}
\f{N}^k &=  (\widetilde{\Sigma} \f{V}^*\f{M} \f{V} \widetilde{\Sigma}^{-1})^{k} = \widetilde{\Sigma} \f{V}^*\f{M} \f{V} \widetilde{\Sigma}^{-1} (\widetilde{\Sigma} \f{V}^*\f{M} \f{V} \widetilde{\Sigma}^{-1})^{k-1} \\&=  \widetilde{\Sigma} \f{V}^*\f{M}^k \f{V} \widetilde{\Sigma}^{-1}
\end{align*}
for each $k \in \N$, using the fact that $\f{V}$ has orthogonal columns. Furthermore, we can verify that $\f{N}$ is $\Sigma$-adjoint, since
\[
\Sigma \f{N}^* = \widetilde{\Sigma}^2 \widetilde{\Sigma}^{-1} \f{V}^* \f{M} \f{V} \widetilde{\Sigma} = \widetilde{\Sigma} \f{V}^* \f{M} \f{V} \widetilde{\Sigma}^{-1}  \widetilde{\Sigma}^2 = \f{N} \Sigma,
\]
using also that $\f{M}$ is symmetric and that $\Sigma = \widetilde{\Sigma}^2$. The $\Sigma$-adjointness of $\f{N}$ allows us to write the occurrence of $\f{H}_k$ in the $i$-th row block and the $j$-th column block as
\[
\f{H}_k   = \f{U} \f{N}^{i-1} \Sigma (\f{N}^*)^{j-1} \f{U}^*
\]
for $i$ and $j$ satisfying $k= i + j -1$. From this, we see that $\f{H}$ attains the generalized Vandermonde decomposition \cref{eq:gen:Vandermonde} since
\[
\f{H}_k =  \left(\mathcal{V}_{d_1}(\f{U},\f{N}) \Sigma \mathcal{V}_{d_2}(\f{U},\f{N})^*\right)_{i,j} = \f{U} \f{N}^{i-1} \Sigma (\f{N}^*)^{j-1} \f{U}^*
\]
for $i,j$ with  $k= i + j -1$, using the definition $\mathcal{V}_m(\f{N},\Lambda)= \begin{pmatrix} (\f{U} \f{N}^0)^{*} &  (\f{U} \f{N}^1)^{*} & \hdots & (\f{U} \f{N}^{m-1})^{*} \end{pmatrix}^*$ for $m=d_1$ and $m=d_2$, respectively.
\end{proof}
We now continue with the proof of \Cref{thm:blkHankel:lowrank}. The proof has similarities to the proof of a similar result for block Toeplitz matrices \cite[Lemma 2]{Yang16}.

\begin{proof}[Proof of \Cref{thm:blkHankel:lowrank}]
For the first statement of \Cref{thm:blkHankel:lowrank}, let $\f{A}$ be a rank-$r$ matrix and denote by $\f{A}=\f{U} \f{J} \f{U}^{-1}$ its Jordan decomposition, i.e., $\f{U} \in \C^{n \times n}$ is invertible and $\f{J} \in \C^{n \times n}$ is an upper triangular matrix with $\rank(\f{J}) = r$. This decomposition shows that $\mathcal{H}\big(\mathcal{Q}_T(\f{A})) \in \Rddn$ is similar to the matrix 
\begin{align*}
\f{S} &:=\begin{pmatrix}
\f{J}&\f{J}^2&\cdots&\f{J}^{d_1}\\
\f{J}^2&\f{J}^3&\cdots&\f{J}^{d_1+1}\\
\vdots&\vdots&\cdots&\vdots\\
\f{J}^{d_2}&\f{J}^{d_2+1}&\cdots&\f{J}^{T}
\end{pmatrix}\\&=\begin{pmatrix}\f{J}\\ \f{J}^2\\ \vdots \\ \f{J}^{d_2}\end{pmatrix}\begin{pmatrix} \Id &\f{J}& \cdots& \f{J}^{d_1-1}\end{pmatrix} =:\f{S}_1 \f{S}_2,
\end{align*}
where $\Id \in \Rnn$ is the identity matrix. We note that $\rank(\f{S}_1)\geq \rank(\f{J})=  r$, as the linear independence of each $r$ of its columns is implied by the linear independence of any $r$ columns of $\f{J}$. To show the reverse inequality, we observe that the rows of the lower blocks of $\f{S}_1$ are all, in fact, in the row space of $\f{J}$ as these blocks are simply subsequent powers of $\f{J}$, which implies that $\rank(\f{S}_1) \leq \rank(\f{J}) = r$. Furthermore, it holds that $\rank(\f{S}_2)=n$ due to the full rank of $\Id \in \Rnn$. The decomposition above shows that $\rank(\f{S})\leq \min\{r,n\}=r$. On the other hand, Let $\f{E}_1=\begin{pmatrix} \Id &0&0&\cdots&0 \end{pmatrix} \in \mathbb{R}^{n\times nd_2}$. Thus, $\f{S}_1=\f{S} \f{E}_1^{\top}$ and  $\rank(\f{S}\f{E}_1^\top)=r \leq \min\{\rank(\f{S}),\rank(\f{E}_1)\}$, which finally implies $\rank(\f{S})=r$.

We continue with the proof of the second statement of \Cref{thm:blkHankel:lowrank}. Since the block Hankel matrix $\f{H} \in \Rddn$ has a square extension $\overset{\square}{\f{H}} \in \y{M}_{Dn,Dn}$ that is positive semidefinite and of rank $r$, there exists a block matrix $\f{L} \in \y{M}_{D n \times r}$ such that 
\begin{align*}
\overset{\square}{\f{H}} &=   \begin{bmatrix}
 	\f{H}_1 & \f{H}_2 & \reflectbox{$\ddots$} & \f{H}_{D} \\
 	\f{H}_2 & \reflectbox{$\ddots$}& \reflectbox{$\ddots$} & \reflectbox{$\ddots$}\\
	\reflectbox{$\ddots$} &  \reflectbox{$\ddots$} & \reflectbox{$\ddots$} &  \f{H}_{2D-2} \\
	\f{H}_{D} &  \reflectbox{$\ddots$} & \f{H}_{2D-2} &  \f{H}_{2D-1}
 \end{bmatrix} \\&= \f{L} \f{L}^* = \begin{bmatrix} \f{L}_{1} \\ \f{L}_{2} \\ \vdots \\ \f{L}_{D}  \end{bmatrix} \begin{bmatrix} \f{L}_{1}^* & \f{L}_{2}^* \ldots &\f{L}_{D}^{*} \end{bmatrix},
\end{align*}
where the columns of $\f{L}$ are linear independent. In the last equality, we wrote $ \f{L} = \begin{bmatrix} \f{L}_{1}^*, &\ldots, &\f{L}_D^*  \end{bmatrix}^*$ with block row matrix with blocks $\f{L}_{j} \in \y{M}_{n,r}$ for each $j=1,\ldots,D$. Now we denote the lower submatrix of $\f{L}$ as $\f{L}_{L} =  \begin{bmatrix} \f{L}_{1}^*, &\ldots, &\f{L}_{D-1}^{*}  \end{bmatrix}^*$ and the upper submatrix $\f{L}$ as $\f{L}_{U} =  \begin{bmatrix} \f{L}_{2}^*, &\ldots, &\f{L}_{D}^{*}  \end{bmatrix}^*$. Due to the Hankel structure of $\overset{\square}{\f{H}}$, the $D-1$ left lower and $D-1$ right upper blocks coincide, and therefore $\f{L}_{U} \f{L}_{L}^* = \f{L}_{L} \f{L}_{U}^*$. Since by assumption at least one other matrix block of $\overset{\square}{\f{H}}$, besides the first one, is of rank $r = \rank(\overset{\square}{\f{H}})$, it follows that both $\f{L}_{L}$ and $\f{L}_{U}$ have full column rank. Furthermore, since this matrix is symmetric and its columns are both in the span of the columns of $\f{L}_{L}$ and $\f{L}_{U}$, then the column spans of $\f{L}_{U}$ and $\f{L}_{L}$ coincide. Thus, there exists a unique invertible matrix $\f{M} \in \y{M}_{r}$ such that
\begin{equation} \label{eq:Mexistunique}
\f{L}_{U} = \f{L}_{L} \f{M}.
\end{equation}
Comparing the $D-1$ blocks of the latter matrix, we note that $\f{L}_{j+1} = \f{L}_{j} \f{M}$ for each $j=1,\ldots,D-1$, and therefore  
\begin{equation} \label{eq:Yjp1}
\f{L}_{j+1} = \f{L}_{1} \f{M}^{j}
\end{equation} 
for each $j=1,\ldots,D-1$. Inserting this into $\overset{\square}{\f{H}} =  \f{L} \f{L}^*$, we observe that $\f{H}_2 = \f{L}_1 \f{M} \f{L}_1^* = \f{L}_1 \f{M}^* \f{L}_1^*$, implying that $\f{M}$ is symmetric. The representation \cref{eq:Hpsd:repr:Thm} of $\f{H}$ follows with $\f{Y} = \f{L}_1$ from inserting the definition of the right hand side and multiplying the resulting block matrices.
\end{proof}

\section{Proof Outline for \Cref{thm:convergence:mobilesensors}} \label{section:proof:main:theorem}

The proof of \Cref{thm:convergence:mobilesensors} consists of multiple steps. First, we formulate a local restricted isometry property, \Cref{assumption:1}, and show that it holds with high probability for uniform and adaptive sampling if enough samples are provided or if the local sampling probabilities are large enough, respectively (\Cref{lemma:main:sampling}). Using perturbation arguments, we then show with \Cref{eq:MatrixIRLS:tangentspace:localRIP:perturbation,lemma:MatrixIRLS:localRIP} that this regularity also extends to the neighborhood of the ground truth. 

In order to have a chance of establishing recovery guarantees, it is necessary to understand when a coodinatewise sampling operator provided related to the sampling set $\Omega \subset I$ is invertible restricted to a subspace associated to low-rank matrices. 

Let $\y{P}_{\f{T}_{\f{Z}}}: \MnOplus \to \MnOplus$ (or, in short, $\y{P}_\f{T}$) be the orthogonal projection onto $\f{T}_{\f{Z}}$. If $\y{H}$ is the block Hankel operator \cref{eq:BlockHankel:operatordef} and $E_{i,j,t}$ is element $(i,j,t)$ of the standard basis of $\MnOplus$, s.t.
\[
\langle E_{i,j,t}, \f{\widetilde{X}} \rangle_F   = \langle E_{i,j}, \f{X}_t \rangle_F = (\f{X}_t)_{i,j}
\]
for any $\f{\widetilde{X}} = \f{X}_1 \oplus \f{X}_2 \oplus \f{X}_3 \oplus \ldots \oplus \f{X}_T \in \MnOplus$, then we define \emph{normalized block Hankel operator} $\y{G}:\MnOplus \to \Rddn$ as
 \begin{equation} \label{def:G:operator}
\y{G}(\widetilde{\f{X}}) := \sum_{t=1}^{T} \sum_{i,j = 1}^{n} \langle E_{i,j,t}, \f{\widetilde{X}} \rangle_F \frac{\mathcal{H}(E_{i,j,t})}{\|\mathcal{H}(E_{i,j,t})\|_F}.
\end{equation}
With this definition, we formulate the following property.

\begin{property}[Local restricted isometry property] \label{assumption:1}
Let $\f{Z} \in \Mn$ be of rank $r$, let  $\f{T}_{\f{Z}}  \subset \Rddn$ be the associated tangent space \cref{eq:tangent:space:def} to the manifold of rank-$r$ matrices, and let $\alpha >0$. Let $\y{R}_{\Omega}: \MnOplus  \to \MnOplus$ be a self-adjoint, normalized sampling operator relative to a sampling set $\Omega \subset I$. We say that $\y{R}_{\Omega}$ satisfies the \emph{local restricted isometry property} with respect to $\f{T}_{\f{Z}}$ and constant $\alpha$ if 
\begin{equation} \label{eq:localRIP:condition}
\left\|\mathcal{P}_{\f{T}_{\f{Z}}}\mathcal{G}\mathcal{R}_{\Omega}\mathcal{G}^{*}\mathcal{P}_{\f{T}_{\f{Z}}}- \mathcal{P}_{\f{T}_{\f{Z}}}\mathcal{G}\mathcal{G}^{*}\mathcal{P}_{\f{T}_{\f{Z}}}\right\| \leq \alpha,
\end{equation}
where $\mathcal{P}_{\f{T}_{\f{Z}}}: \MnOplus  \to \MnOplus$ is the orthogonal projection onto the linear subspace $\f{T}_{\f{Z}}$. 
\end{property}

Condition \cref{eq:localRIP:condition} is referred as a \emph{local} restricted isometry property since it is not a restricted isometry property with respect to the entire manifold of low-rank matrices \cite{Recht10,Davenport16}, but rather one that holds with respect to a particular (tangent) subspace associated to the low-rank matrix manifold around a point. Similar conditions have been used for structured low-rank matrix completion \cite[Lemma 1]{chen_chi14}, \cite[Lemma 20]{ye_kim_jin_lee}.

With \Cref{lemma:main:sampling}, we establish \Cref{assumption:1} with high probability for uniform and adaptive sampling with respect to the block Hankel matrix $\f{H}_{\f{A}}$ associated to a transition operator $\f{A}$.

\begin{lemma}[Local RIP for sampling operators] \label{lemma:main:sampling}
Let $\f{A} \in \Mn$ be of rank-$r$, let $\f{T} := \f{T}_{\f{H}_{\f{A}}}$ be the tangent space to the rank-$r$ matrix manifold at the block Hankel matrix $\f{H}_{\f{A}} = \y{H}(\mathcal{Q}_T(\f{A}))$ associated to $\mathcal{Q}_T(\f{A})$. Let $0 < \alpha < 1$ and $\y{G}:\MnOplus \to \Rddn$ be the normalized block Hankel operator of \cref{def:G:operator}. 
There exists a constant $C > 0$ such that the following holds:
\begin{enumerate}
\item \emph{[Uniform sampling model]} Suppose that $\Omega$ is a random subset of cardinality $m$ uniformly drawn without replacement among the set of space-time samples $I = [n] \times [n] \times [T]$ . Let $\y{R}_{\Omega}: \MnOplus  \to \MnOplus$ be the normalized sampling operator
\begin{equation} \label{def:ROmega}
\f{L} \to \y{R}_{\Omega}(\f{L}) := \sum_{(i,j,t) \in \Omega} \frac{n^2 T}{ m} \langle E_{i,j,t}, \f{L}\rangle_F E_{i,j,t}.
\end{equation}
Then $\y{R}_{\Omega}$ satisfies the local restricted isometry property with respect to $\f{H}_{\f{A}}$ with constant $\alpha$ (\Cref{assumption:1}), with probability at least $1-n^{-2}$, provided that
\begin{equation} \label{eq:sample:complexity:condition1}
m \geq \frac{C  c_s}{\alpha^2} \mu_0 r n \log(n T)\,,
\end{equation} 
if $\f{H}_{\f{A}}$ is $\mu_0$-incoherent as per \Cref{def:incoherence}.

\item  \emph{[Adaptive sampling]} Suppose that $\Omega$ consists of random index triplets $(i,j,t) \in I$ that are independently observed according to Bernoulli distributions with probabilities $(p_{i,j,t})_{(i,j,t) \in I}$. Let $\y{R}_{\Omega}: \MnOplus  \to \MnOplus$ be the normalized sampling operator
\begin{equation} \label{def:ROmega:adaptive:mobile}
\f{L} \to \y{R}_{\Omega}(\f{L}) := \sum_{(i,j,t) \in \Omega} \frac{1}{p_{i,j,t}} \langle E_{i,j,t}, \f{L}\rangle_F E_{i,j,t}.
\end{equation}
Then $\y{R}_{\Omega}$ satisfies the local restricted isometry property with respect to $\f{H}_{\f{A}}$ with constant $\alpha$ (\Cref{assumption:1}), with probability at least $1 - n^{-2}$,  provided that, for each $(i,j,t) \in I$, we have
\begin{equation}  \label{eq:sample:complexity:condition:localcoherencemain}
p_{i,j,t} \geq \min\left( \frac{C c_s}{\alpha^2} \mu_{i,j,t}  \frac{r}{n T} \log(n T) ,1\right)\,,
\end{equation}
if $(\mu_{i,j,t})_{(i,j,t) \in I}$ are the local incoherences \cref{eq:mu:ijk} of $\f{H}_{\f{A}}$ as in \Cref{def:incoherence}.
\end{enumerate}
\end{lemma}
We note that the incoherence parameters of the matrix $\f{H}_{\f{A}}$ play an important role in quantifying the number of space-time samples that are sufficient to establish \Cref{assumption:1} for sampling operators. The proof can be found in \Cref{sec:proof:samplingoperator:assumption}. 

In \Cref{eq:MatrixIRLS:tangentspace:localRIP:perturbation}, we extend \Cref{assumption:1} to a neighborhood of $\f{H}_{\f{A}}$. 

\begin{lemma} \label{eq:MatrixIRLS:tangentspace:localRIP:perturbation}
Assume that the local restricted isometry property \Cref{assumption:1} holds true for a normalized sampling operator $\y{R}_{\Omega}: \MnOplus  \to \MnOplus$ with respect to $\f{H}_{\f{A}}= \HQTA$, where $\f{A}$ has a rank-$r$, and constant $\alpha > 0$. If $\f{H} \in \Rddn$ is of rank $r$ and
\begin{equation} \label{eq:HXtilde:closeness:lemma}
\f{H} \in  \y{B}_{\HA}\left( \frac\alpha8 \left(\sqrt{\left\| \y{R}_{\Omega} \right\| \left(1 + \alpha \right)} + 1\right)^{-1} \sigma_r(\HA) \right),
\end{equation}
then
\begin{equation*}
\left\|\PTX \mathcal{G}\mathcal{R}_{\Omega}\mathcal{G}^{*}\PTX- \PTX\mathcal{G}\mathcal{G}^{*}\PTX\right\| \leq 2 \alpha,
\end{equation*}
where $\f{T}_{\f{H}}$ is the tangent space to the rank-$r$ manifold at $\f{H}$.
\end{lemma}
The proof is postponed to \Cref{sec:RIP:neighborhood}.

As the next step, we establish a null space-type property that shows that not too much mass can be concentrated on the tangent space $\f{T}_{\f{H}}$ among block Hankel matrices in the null space of the sampling operator. 
\begin{lemma} \label{lemma:MatrixIRLS:localRIP}
 Let $\y{R}_{\Omega}: \MnOplus \to \MnOplus$ be a normalized sampling operator as in \cref{def:ROmega} or \cref{def:ROmega:adaptive:mobile}. If $\f{H} \in \Rddn$ is of rank $r$ and $T_{\f{H}} \subset \Rddn$ is the tangent space \cref{eq:tangent:space:def} to the rank-$r$ manifold at $\f{H}$, then
\begin{equation} \label{eq:PTGROmGPT:25}
\left\|\PTX\mathcal{G}\mathcal{R}_{\Omega}\mathcal{G}^{*}\PTX- \PTX\mathcal{G}\mathcal{G}^{*}\PTX\right\| \leq \frac{2}{5},
\end{equation}
implies
\[
\|\y{H}(\eta)\|_{F}^2 \leq \frac{5}{3}  \left( \left\|\y{R}_{\Omega}\right\|  + 8/5 \right)  \left\| \y{P}_{\f{T}_{\f{H}}^\perp} \y{H}(\eta)\right\|_F^2
\]
for each $\eta \in \ker \y{R}_{\Omega}$.
\end{lemma}
We refer to \Cref{sec:proof:localRIP:NSP} for the proof.

Finally, \Cref{thm:convergence:template} relates \Cref{assumption:1} with respect to the block Hankel matrix of a transition operator to the local quadratic convergence of \texttt{TOIRLS}.
\begin{proposition}[Local Convergence with Quadratic Rate] \label{thm:convergence:template}
There exists an absolute constant $c_0$ such that the following holds. Assume that $\f{A} \in \Mn$ is of rank $r$, and that \Cref{assumption:1} holds for the normalized sampling operator $\y{R}_{\Omega}: \MnOplus  \to \MnOplus$ with respect to $\f{H}_{\f{A}} := \HQTA$ and constant $\alpha = 1/5$. Let $\widetilde{\f{X}}\hk$ is the $k$-th iterate of \texttt{TOIRLS}  \Cref{algo:IRLS:graphcompletion} with inputs: $\Omega$, $\f{y} = P_{\Omega}(\y{Q}_T(\f{A}))$, and $\widetilde{r} = r$. If we assume that the smoothing parameter fulfills $\varepsilon_k = \sigma_{r+1}(\y{H}(\widetilde{\f{X}}\hk))$ and if $\y{H}(\widetilde{\f{X}}\hk) \in \y{B}_{\HA}\left( c_0 \|\y{R}_{\Omega}\|^{-3/2} r^{-1} \kappa^{-1}  (dn-r)^{-1/2} \sigma_r(\HA) \right)$, where $\kappa := \sigma_1(\f{H}_{\f{A}}) /  \sigma_r(\f{H}_{\f{A}})$ is the condition number of $\f{H}_{\f{A}}$,
then there exists $\nu$ such that for all $\ell\ge0$
\begin{equation*}
\begin{split}
&\|\y{H}(\widetilde{\f{X}}^{(k+\ell+1)}) - \f{H}_{\f{A}}\|   \\
&\leq \min(\nu \|\y{H}(\widetilde{\f{X}}^{(k+\ell)})-\f{H}_{\f{A}}\|^{2} , \|\y{H}(\widetilde{\f{X}}^{(k+\ell)})-\f{H}_{\f{A}}\|).
\end{split}
\end{equation*}
In other words, $\y{H}(\widetilde{\f{X}}^{(k+\ell)}) \xrightarrow{\ell \to \infty} \f{H}_{\f{A}}$ with quadratic convergence rate.
\end{proposition}
The proof (see \Cref{sec:weightoperator:structure}) crucially relies on estimates from \cite{KMV21} on the action of the weight operator $W_{\y{H}(\widetilde{\f{X}}\hk)}$ of \Cref{def:weightoperator:MatrixIRLS} where $\widetilde{\f{X}}\hk$ is an \texttt{TOIRLS} iterate, and combines them with \Cref{lemma:MatrixIRLS:localRIP}.

Putting the results of this section together amounts finally to the proof of \Cref{thm:convergence:mobilesensors}, which is detailed in \Cref{sec:thm:convergence:mobilesensors:proof}.

\section{Conclusion \& Outlook}\label{section:conclusion}
In this paper, we developed a framework for the learning of linear transition operators from random sparse observations of space-time samples, and provided a local convergence analysis for the non-convex optimization approach \texttt{TOIRLS} for solving that problem, quantifying the number and distribution of samples sufficient for convergence. The current work could be extended in several directions: the presented convergence analysis for \texttt{TOIRLS} is inherently local, i.e., requires an iterate that is already close to a low-rank ground truth matrix. The empirical results suggest that a global convergence of \texttt{TOIRLS} might be possible, despite being beyond the scope of this paper. Furthermore, not only entrywise, but general linear sampling operators could be considered, cf. \cref{eq:def:Somega}, as well as applications to a broader family of dynamical systems such as linear time-invariant systems with input terms \cref{LIS2}. Finally, it would be of interest to combine the setup considered in this paper with additional prior knowledge on the transition operator $\f{A}$, such as sparsity, which is common for example in the context of graph transition operators.

\appendix

\subsection{Incoherence Estimates} \label{sec:coherence:estimates}
In this subsection, we provide estimates for the incoherence parameters $\mu_{i,j,t}$ and $\mu_0$ of the block Hankel matrix $\mathcal{H}\big(\mathcal{Q}_T(\f{A})\big)$ associated with a linear operator $\f{A} \in \Mn$, which have been defined in \Cref{def:incoherence}. Also for use in other proofs, we state the following result that elucidates the action of the block Hankel operator $\y{H}$.
\begin{lemma} \label{lemma:Hankel:action}
Recalling the normalized block Hankel operator $\y{G}:\MnOplus \to \Rddn$ of \cref{def:G:operator}, let  $\left\{E_{i,j,t} \right\}_{(i,j,t) \in I}$ be the standard basis of $\MnOplus$ and $\left\{ \f{B}_{i,j,t}\right\}_{(i,j,t) \in I}$ the standard basis of the space of block Hankel matrices $\mathcal{H}(\MnOplus)$. Then we have that the diagonal operator $\y{D}: \MnOplus \to \MnOplus$
\begin{equation} \label{eq:DEijt:def}
\begin{split}
\y{D}(E_{i,j,t}) :&= \|\mathcal{H}(E_{i,j,t})\|_F E_{i,j,t} \\
&= \sqrt{\min(t,T+1-t,d_1,d_2) } E_{i,j,t},
\end{split}
\end{equation}
for each $(i,j,t) \in I$, satisfies $\y{H} = \y{G} \y{D}$, which is equivalent to
\[
\f{B}_{i,j,t}= \mathcal{H}(\y{D}^{-1} E_{i,j,t}) = \y{G}(E_{i,j,t}) \,.
\]
Furthermore, it holds that $\mathcal{H}^* \mathcal{H}: \MnOplus \to \MnOplus$ satisfies, for all $(i,j,t) \in I$,
\[
\mathcal{H}^* \mathcal{H} ( E_{i,j,t} ) = \min(t,T+1-t,d_1,d_2) E_{i,j,t}.
\]
\end{lemma}
\begin{proof}
The first statement follows by combining the definition \cref{eq:DEijt:def} with \cref{def:G:operator}, and by counting the number of occurrences of each block in \cref{eq:BlockHankel:operatordef}. The last statement follows from
\begin{align*}
\langle E_{i,j,t}, &\mathcal{H}^* \mathcal{H} ( E_{i,j,t} ) \rangle = \big\|\mathcal{H} ( E_{i,j,t} )\big\|_F^2 = \big\|\mathcal{H} ( E_{i,j,t} )\big\|_F^2 \\&= \big\|\y{D}(E_{i,j,t}) \big\|_F^2 = \min(t,T+1-t,d_1,d_2)
\end{align*}
for all  $(i,j,t) \in I$ due to \cref{eq:DEijt:def}.
\end{proof}
With the following lemma, we bound the local incoherence parameter as defined in \Cref{def:incoherence} by the incoherence parameter based on the related incoherence parameter in \cref{eq:incoherences:Ud1d2} in the spirit of \cite[(27)]{chen_chi14} and \cite{cai_wang_wei19}. 

\begin{lemma} \label{lemma:incoherence:bound:traditional}
Let $\f{Z} \in \Rddn$ be of rank $r$ with leading left and right singular vector matrices $\f{U} \in \R^{n d_1 \times r}$ and $\f{V} \in \R^{n d_2 \times r}$, respectively, and let $\f{T}_{\f{Z}}$ be the associated tangent space \cref{eq:tangent:space:def}. Suppose there exists a positive constant $\mu_0>0$ such that 
 \begin{align} \label{eq:incoherences:Ud1d2}
 \max_{1 \leq i \leq j \leq n, 1 \leq t \leq T} \| \f{U}^* \f{B}_{i,j,t}\|_F &\leq  \sqrt{\mu_0 \frac{r}{nd_1}},\nonumber\\ \max_{1 \leq i \leq j \leq n, 1 \leq t \leq T} \| \f{B}_{i,j,t}\f{V} \|_F &\leq \sqrt{\mu_0 \frac{r}{nd_2}},
\end{align}
where $\left\{ \f{B}_{i,j,t}\right\}_{(i,j,t) \in I}$ is the standard basis of the space of block Hankel matrices $\mathcal{H}(\MnOplus)$. 

Then, for each $(i,j,t) \in [n]\times[n]\times [T]$, 
 \[ \| P_{\f{T}_{\f{Z}}}( \f{B}_{i,j,t})\|_{F}^2 \leq    \mu_0 c_s \frac{r}{nT}, 
\]
where $c_s = \frac{T(T+1)}{d_1 d_2}$. In particular, $\f{Z}$ is $\mu_0$-incoherent in the sense of \Cref{def:incoherence}.
 \end{lemma}
 \begin{proof}
 It is well-known \cite[Eq. (3)]{recht} that the action of the projection operator  $P_{\f{T}_{\f{Z}}}$ can be written such that
 \begin{align*}
  P_{\f{T}}( \f{M}) &=\f{U} \f{U}^*  \f{M}+ \f{M}\f{V}\f{V}^*  -\f{U}\f{U}^*  \f{M}\f{V}\f{V}^* \\
  &=\f{U}\f{U}^* \f{M}(\f{I}-\f{V}\f{V}^*)+ \f{M} \f{V}\f{V}^*
 \end{align*}
 for any matrix $\f{M}$. Therefore, we estimate that
 \begin{equation} \label{eq:incoherence:conversion}
 \begin{split}
  \| P_{\f{T}}( \f{M})\|_{F}^2&=\|\f{U}\f{U}^*  \f{M}(\f{I}-\f{V}\f{V}^*)\|_{F}^2+\| \f{M}\f{V}\f{V}^*\|_{F}^2 \\
  &\leq \|\f{U}\f{U}^*  \f{M}\|_{F}^2\|(\f{I}-\f{V}\f{V}^*)\|^2+\|\f{M}\f{V}\f{V}^*\|_{F}^2 \\
  & \leq \|\f{U}\f{U}^*  \f{M}\|_{F}^2+\|\f{M}\f{V}\f{V}^*\|_{F}^2 \\&\leq \|\f{U}^* \f{M} \|_{F}^2+\| \f{M} \f{V}\|_{F}^2.
  \end{split}
  \end{equation}
  Thus, \cref{eq:incoherences:Ud1d2} implies that
  \begin{equation*}
    \| P_{\f{T}}( \f{B}_{i,j,t})\|_{F}^2  \leq \mu_0 \frac{r}{nd_1} + \mu_0 \frac{r}{nd_2} =  \frac{\mu_0 r}{n}\frac{d_1+d_2}{d_1d_2} = \mu_0 c_s \frac{r}{nT},
 \end{equation*}
 setting $\f{M} = \f{B}_{i,j,t}$ for any $(i,j,t) \in I$. 
 \end{proof}
  
Before providing the proof of the incoherence estimates for the examples of \Cref{sec:incoherence:estimates}, we note that it follows from \cref{eq:sample:complexity:condition:localcoherence} in \Cref{thm:convergence:mobilesensors} that for adaptive sampling with probabilities satisfying $p_{i,j,t} \geq C c_s \mu_{i,j,t}  \frac{r}{n T} \log(n T)$ for all $(i,j,t) \in I$, where $C$ is the constant from \cref{eq:sample:complexity:condition:localcoherence}.2 and $c_s= \frac{T(T+1)}{d_1 d_2}$, a number of
\begin{equation} \label{eq:mexp:lowerbound}
\begin{split}
m_{\text{exp}} := \mathbb{E}[|\Omega|] &= \sum_{(i,j,t) \in I} p_{i,j,t} \\
&\geq C c_s \sum_{(i,j,t) \in I} \mu_{i,j,t} \frac{r}{n T} \log(n T)
\end{split}
\end{equation}
expected samples will enable local convergence of \texttt{TOIRLS} in the adaptive model for $\Omega$.

\subsubsection{Orthogonal Matrices }\label{coherence:unitary}
To understand the incoherences of block Hankel matrices $\f{H}_{\f{A}} = \y{H}(\mathcal{Q}_T(\f{A}))$ associated to orthonormal matrices $\f{A} \in \mathbb{O}^{n}$, we observe that a compact singular value decomposition of $\f{H}_{\f{A}}$ can be given by $\f{H}_{\f{A}}  = \f{U}_{\f{H}_{\f{A}}} \Sigma \f{V}_{\f{H}_{\f{A}}}^*$ with
\begin{align*}
\f{U}_{\f{H}_{\f{A}}} = \frac{1}{\sqrt{d_1}}\begin{pmatrix}\f{A}\\ \f{A}^2\\ \vdots \\ \f{A}^{d_1}
\end{pmatrix} \,,\quad
\f{V}_{\f{H}_{\f{A}}} =  \frac{1}{\sqrt{d_2}} \begin{pmatrix}\f{I} \\ \f{A}  \\\cdots \\ \f{A}^{d_2-1} \end{pmatrix} \,,
\end{align*}
$\f{U}_{\f{H}_{\f{A}}}\in \mathcal{M}_{n d_1, n}$, $\f{V}_{\f{H}_{\f{A}}}\in \mathcal{M}_{n d_2, n}$ and $\Sigma = \sqrt{d_1d_2} \Id \in \Mn$.  Using this, we obtain the following proposition.

\begin{proposition}  \label{prop:orthogonalmatrix:incoherences}
If the transition operator  $\f{A} \in \mathbb{O}^{n}$ is an orthogonal matrix, then 
\begin{enumerate}
 \item $\f{H}_{\f{A}} = \y{H}(\mathcal{Q}_T(\f{A}))$ is $\mu_0$-incoherent with $\mu_0 \leq 1$. 
 \item The local incoherences of $\f{H}_{\f{A}} $ satisfy $\sum_{(i,j,t) \in I} \mu_{i,j,t}  \leq nT^2$.
 \end{enumerate}
As a consequence, for both uniform and adaptive sampling, a sample complexity of order $\Theta(n^2\log(nT))$ is sufficient to satisfy the assumption of \Cref{thm:convergence:mobilesensors}.  
\end{proposition}

\begin{proof} 
First, it is straightforward to verify from a block-wise computation that  $\|\f{U}_{\f{H}_{\f{A}}} ^{*}\f{B}_{i,j,t}\|_{F}=\frac{1}{\sqrt{d_1}}$ for each $(i,j,t) \in I$, since $\|(\f{A}^{j})^{*} \f{M}\|_F = \|\f{M}\|_F$ for each $j$, for each block $\f{M}$ of $\f{B}_{i,j,t}$, due to the preservation of norms through multiplication with orthogonal matrices.  Similarly, $\|\f{B}_{i,j,t}\f{V}_{\f{H}_{\f{A}}} \|_{F}=\frac{1}{\sqrt{d_2}}$ for each $(i,j,t) \in I$.

This implies that \cref{eq:incoherences:Ud1d2} is satisfied with $\mu_0$ as $r = \rank(\f{H}_{\f{A}})= \rank(\f{A}) = n$. In view of \Cref{lemma:incoherence:bound:traditional}, it follows that $\f{H}_{\f{A}}$ is $\mu_0$-incoherent in the sense of \Cref{def:incoherence} with $\mu_0 \leq 1$ since
\begin{equation*}
\| \y{P}_{\f{T}}( \f{B}_{i,j,t})\|_{F}^2 \leq   c_s \frac{1}{T}
\end{equation*}
for each $(i,j,t) \in I$, where $\y{P}_{\f{T}}$ is the projection operator onto the subspace $\f{T}= \f{T}_{\f{H}_{\f{A}}}$. This shows the first statement of \Cref{prop:orthogonalmatrix:incoherences}

Estimating the sum of local incoherences $\mu_{i,j,t}$ of $\f{H}_{\f{A}}$, we obtain

\begin{align*}
&\sum_{(i,j,t) \in I} \mu_{i,j,t}  =  \sum_{1\leq i\leq j\leq n,t=1,\ldots,T}   \frac{T}{c_s}  \|\mathcal{P}_{\f{T}} (\f{B}_{i,j,t})\|_F^2  \\&=  \sum_{1\leq i\leq j\leq n,t=1,\ldots,T}  \frac{d_1 d_2}{(d_1 + d_2)} \|\mathcal{P}_{\f{T}} (\f{B}_{i,j,t})\|_F^2 \nonumber \\
&\leq  \sum_{\substack{1\leq i\leq j\leq n\\ t=1,\ldots,T}}     \frac{ d_1 d_2}{ (d_1 + d_2)} \big(\|\f{U}_{\f{H}_{\f{A}}}^* \f{B}_{i,j,t}\|_F^2+\|\f{B}_{i,j,t}\f{V}_{\f{H}_{\f{A}}} \|_F^2\big) \\
&\leq \sum_{1\leq i\leq j\leq n,t=1,\ldots,T}     \frac{ d_1 d_2}{ (d_1 + d_2)} \left(\frac{1}{d_1}+ \frac{1}{d_2}\right)  \\
&=  \sum_{1\leq i\leq j\leq n,t=1,\ldots,T}  1 = Tn^2,
\end{align*}
using \cref{eq:incoherence:conversion} in the inequality. 

Therefore, in view of \cref{eq:mexp:lowerbound}, a sufficient number of expected space-time samples $m_{\text{exp}}$ to enable the local convergence guarantee of \Cref{prop:orthogonalmatrix:incoherences} is $ m_{\text{exp}} = \Theta(n^2\log(nT))$.
\end{proof}

\subsubsection{Positive Semi-Definite Matrices} \label{coherence:psd}
We now justify the bounds of \Cref{sec:incoherence:estimates} for positive semidefinite transition operators $\f{A} \in \Mn$. To this end, we provide a closed formula for a singular value decomposition for the associated block Hankel matrix $\f{H}_{\f{A}} = \y{H}(\mathcal{Q}_T(\f{A}))$ in \Cref{theorem:factorization}.
\begin{theorem}\label{theorem:factorization}
Suppose $\f{A} = \sum_{\ell=1}^r \lambda_\ell \f{u}_i\f{u}_i^* =  \f{U} \Lambda \f{U}^*$ is a positive semidefinite matrix with $r$ positive eigenvalues $\lambda_1,\ldots,\lambda_r$  and $\f{u}_1,\f{u}_2,\ldots, \f{u}_r$ are the corresponding eigenvectors, so that $\Lambda = \diag(\lambda_1,\ldots,\lambda_r)$ and $\f{U} = \begin{pmatrix} \f{u}_1 & \ldots & \f{u}_r\end{pmatrix}$. For a vector $\f{u} \in \R^n$, an integer $m$ and a scalar $\lambda$, we define
$$\mathcal{V}_{m,\lambda}(\f{u})=\begin{pmatrix} \f{u} \\ \lambda \f{u} \\ \vdots\\ \lambda^{m-1} \f{u}\end{pmatrix} \in \R^{n m}$$
and the \emph{generalized Vandermonde matrix}, in $\y{M}_{n m,r}$,
\begin{align*}
\mathcal{V}_m(\f{U},\Lambda)\!=\!\begin{pmatrix}\vline &  &\vline \\ \mathcal{V}_{m,\lambda_1}(\f{u}_1)&\ldots & \mathcal{V}_{m,\lambda_r}(\f{u}_r) \\ \vline & & \vline \end{pmatrix} 
\!=\! \begin{pmatrix} \f{U} \Lambda^0 \\  \f{U} \Lambda^1 \\ \vdots \\ \f{U} \Lambda^{m-1} \end{pmatrix}.
\end{align*}
Let $\vv{\Lambda^m}:= \diag( \vv{\lambda_\ell^m})_{\ell=1}^r$ with $\vv{\lambda_\ell^m} = \sqrt{\sum_{s=0}^{m-1}\lambda_\ell^{2s}}$ for all $\ell \in [r]$. Then a compact singular value decomposition of $\f{H}_{\f{A}}$ can be written such that
\begin{equation} \label{eq:SVD:blockhankel}
\f{H}_{\f{A}} = \f{U}_{\f{H}_{\f{A}}}\f{D} \f{V}_{\f{H}_{\f{A}}}^*
\end{equation}
where 
\begin{align*}
\f{U}_{\f{H}_{\f{A}}} &= \mathcal{V}_{d_1}(\f{U},\Lambda) (\vv{\Lambda^{d_1}})^{-1}\\ \f{V}_{\f{H}_{\f{A}}} &= \mathcal{V}_{d_1}(\f{U},\Lambda) (\vv{\Lambda^{d_2}})^{-1}\\
\f{D} &= (\vv{\Lambda^{d_1}})\Lambda (\vv{\Lambda^{d_2}})\,.
\end{align*}
In particular, the nonzero singular values of $\f{H}_{\f{A}}$ are 
$$\lambda_\ell\sqrt{ \sum_{t=0}^{d_1-1} \lambda_{\ell}^{2 t}}\sqrt{ \sum_{t=0}^{d_2-1} \lambda_{\ell}^{2 t}}, $$ 
for $\ell=1,\ldots, r$, and the first $r$ right and left singular vectors are $\{(\vv{\lambda_\ell^{d_2}})^{-1} \mathcal{V}_{d_2,\lambda_1}(\f{u}_i) \}_{\ell=1}^{r}$ and  $\{ (\vv{\lambda_\ell^{d_1}})^{-1}\mathcal{V}_{d_1,\lambda_1}(\f{u}_i)\}_{\ell=1}^{r}$, respectively.
\end{theorem}
\begin{proof}
To prove the statements, the equality \cref{eq:SVD:blockhankel} can be verified expanding the right hand side, and furthermore, since $\f{U}^* \f{U} = \Id_{r}$, the orthogonality of the columns of the singular vector matrices can be verified, i.e. $\f{U}_{\f{H}_{\f{A}}}^* \f{U}_{\f{H}_{\f{A}}}  = \Id_{r}$ and $\f{V}_{\f{H}_{\f{A}}}^* \f{V}_{\f{H}_{\f{A}}}  = \Id_{r}$.
\end{proof}

Following the notation in \eqref{eq:SVD:blockhankel}, we compute estimates of $\|\f{U}_{\f{H}_{\f{A}}} (\f{B}_{i,j,t})\|_F$ and $\| (\f{B}_{i,j,t})\f{V}_{\f{H}_{\f{A}}} \|_F$. As a preparation of what follows, we recall from \cref{eq:DEijt:def} that 
\[
\|\y{H}(E_{i,j,t})\|_{F}=
\begin{cases}
 \sqrt{t} & \text{if } t\leq \mathrm{min}\{d_1,d_2\}, \\
 \sqrt{ \mathrm{min}\{d_1,d_2\}}& \text{if }  \begin{split} \mathrm{min}\{d_1,d_2\}  < t \\
 \leq  \mathrm{max}\{d_1,d_2\} \end{split},\\
\sqrt{T+1-t} & \text{if } t > \mathrm{max}\{d_1,d_2\}, \\
\end{cases}
\]
if $E_{i,j,t}$ is the standard basis matrix of index $(i,j,t) \in I$ of $\MnOplus$. 
The above observation yields the following lemma. 
\begin{lemma} \label{estimateslemma}
Let $\f{B}_{i,j,t}  = \mathcal{H}(E_{i,j,t})  / \|\mathcal{H}(E_{i,j,t})\|_F \in \Rddn$ be the standard basis matrix of index $(i,j,t) \in I$ of $\mathcal{H}(\MnOplus)$. For $2 \leq d_1\leq d_2$, we have the following identities if $\f{U}_{\f{H}_{\f{A}}}$ and $\f{V}_{\f{H}_{\f{A}}}$ are as in \Cref{theorem:factorization}:
 \[
 \begin{split}
&\|\f{U}_{\f{H}_{\f{A}}} ^* \f{B}_{i,j,t}\|_F= \\
&=
\begin{cases}
\sqrt{\sum_{\ell=1}^{r}\frac{ \sum_{s=0}^{t-1} \lambda_{\ell}^{2 s}  }{\sum_{s=0}^{d_1-1} \lambda_{\ell}^{2s} } \frac{\|\f{u}_{\ell}^*E_{ij}\|^2}{t}} , & \text{if } t< d_1, \\
\sqrt{\sum_{\ell=1}^{r}\frac{\|\f{u}_{\ell}^*E_{ij}\|^2}{d_1}},   & \text{if }  d_1 \leq t \leq  d_2, \\
\sqrt{\sum_{\ell=1}^{r}\frac{ \sum_{s=t-d_2}^{d_1-1} \lambda_{\ell}^{2s} }{\sum_{s=0}^{d_1-1} \lambda_{\ell}^{2 s}} \frac{\|\f{u}_{\ell}^*E_{ij}\|^2}{T+1-t}},  & \text{if } t> d_2,
\end{cases}
\end{split}
\]
and 
 \[
 \begin{split}
&\|\f{B}_{i,j,t}\f{V}_{\f{H}_{\f{A}}}\|_F = \\
&=
\begin{cases}
\sqrt{\sum_{\ell=1}^{r}\frac{\sum_{s=0}^{t-1}\lambda_{\ell}^{2s}}{\sum_{s=0}^{d_2 - 1}\lambda_{\ell}^{2s}} \frac{\|E_{ij}\f{u}_{\ell}\|^2}{t}}, & \text{if } t< d_1, \\
\sqrt{\sum_{\ell=1}^{r}\frac{\sum_{s=t-d_1}^{t-1}\lambda_{\ell}^{2s}}{\sum_{s=0}^{d_2-1}\lambda_{\ell}^{2s}} \frac{\|E_{ij}\f{u}_{\ell}\|^2}{d_1}},   & \text{if }  d_1 \leq t \leq  d_2,\\
\sqrt{\sum_{\ell=1}^{r}\frac{\sum_{s=t-d_1}^{d_2-1}\lambda_{\ell}^{2s}}{\sum_{s=0}^{d_2-1}\lambda_{\ell}^{2s}} \frac{\|E_{ij}\f{u}_{\ell}\|^2}{T+1-t}},  & \text{if } t> d_2. \\
\end{cases}
\end{split}
\] 
\end{lemma}

We restrict our attention to \emph{normalized} positive semidefinite transition operators $\f{A}$ whose eigenvalues are within the interval $[0,1]$. To simplify the analysis and avoid unnecessary technicalities, we restrict ourselves to the case of $d_1=d_2$.

\begin{proposition}\label{estimates:positivedefinite} Let $\f{A}$ be a positive semidefinite transition operator as in \ref{theorem:factorization} and assume that $0 < \lambda_r \leq \ldots \leq \lambda_1 \leq 1$. Assume also that $d_1= d_2$ and $T=d_1+d_2-1 \geq 3$. Then:
\begin{enumerate}
 \item  $\f{H}_{\f{A}} = \y{H}(\mathcal{Q}_T(\f{A}))$ is $\mu_0$-incoherent with
 $$\mu_0  \leq  \max_{1\leq i\leq n}  \sum_{\ell=1}^{r}\frac{nd_2 (\f{u}_{\ell})_i^2}{r(1+\lambda_{\ell}^2+\cdots+\lambda_{\ell}^{2(d_2-1)})}.$$
 \item The local incoherences of $\f{H}_{\f{A}}$ satisfy 
 \begin{align}
 &\sum_{(i,j,t) \in I } \mu_{i,j,t} \notag  \\
 & \leq \frac{n^2 d_2}{r} \left(\sum_{t=1}^{d_2-1}\sum_{\ell=1}^{r} \frac{(1+\lambda_{\ell}^{2(d_2-t)})\sum_{s=0}^{t-1}\lambda_{\ell}^{2s}}{t\sum_{s=0}^{d_2-1}\lambda_{\ell}^{2s}}+\frac{1}{d_2} \right) \label{eq:loc:psd:incoherence:1} \\
 &\leq 4.4 n^2 T \log(T).   \label{eq:loc:psd:incoherence:2}
 \end{align}
Consequently, for the adaptive sampling model, it is possible to satisfy the assumption of \Cref{thm:convergence:mobilesensors} with
  \[
  m_{\text{exp}} = \mathbb{E}[|\Omega|] = \Theta( rn\log(nT)\log(T)) 
  \]
  expected samples. 
 \end{enumerate}
\end{proposition}

\begin{proof}
1. Define the function $g: [0,1] \times \N \times \N \mapsto \R$ such that 
\begin{equation} \label{eq:g:definition}
g(\lambda,d, t)=\frac{ \sum_{s=0}^{t-1} \lambda^{s} }{t \sum_{s=0}^{d-1} \lambda^{s}}.
\end{equation}
We observe that $g$ has following properties:
\begin{itemize}
\item For a fixed $d \in \mathbb{N}$ and $ \lambda \in [0, 1]$, $g(\lambda,d, t)$ is a decreasing function with respect to $t \in \mathbb{N}$.
\item For fixed $d  \in \mathbb{N}$  and $t  \in \mathbb{N}$, $g(\lambda,d, t)$ is a decreasing function with respect to $\lambda$ on $[0,\infty)$.
\end{itemize}
Using Lemma \ref{estimateslemma} and the properties of the function $g$, we obtain that
\begin{align*}
\max_{1 \leq i \leq j \leq n, 1 \leq t \leq T} \|\f{U}_{\f{H}_{\f{A}}}^* (\f{B}_{i,j,t})\|_F &=\max_{1 \leq i \leq j \leq n} \|\f{U}_{\f{H}_{\f{A}}}^* (\f{B}_{i,j,1})\|_F\\
&= \max_{1\leq i\leq n}\sqrt{\sum_{\ell=1}^{r}\frac{ (\f{u}_{\ell})_i^2}{\sum_{s=0}^{d_1-1}\lambda_{\ell}^{2s}}}.
\end{align*}
By symmetry, we have
\begin{align*}
\max_{1 \leq i \leq j \leq n, 1 \leq t \leq T}\|\f{B}_{i,j,t}\f{V}_{\f{H}_{\f{A}}}\|_F&= \max_{1\leq i\leq n}\sqrt{\sum_{\ell=1}^{r}\frac{ (\f{u}_{\ell})_i^2}{\sum_{s=0}^{d_2-1}\lambda_{\ell}^{2s}}}.
\end{align*}
From this and from the fact that $d_1 \leq d_2$, we see that we can choose 
$$\mu_0 =\max_{1\leq i\leq n}{\sum_{\ell=1}^{r}\frac{nd_2 (\f{u}_{\ell})_{i}^2}{r(\sum_{s=0}^{d_1-1}\lambda_{\ell}^{2s})}}$$ 
to satisfy the inequalities of \eqref{eq:incoherences:Ud1d2}, from which it follows that $\f{H}_{\f{A}}$ is $\mu_0$-incoherent by \Cref{lemma:incoherence:bound:traditional}.

2. To obtain an upper bound for $\sum_{i,j,t}\mu_{i,j,t}$, we use that
\begin{equation*}
\begin{split}
&\sum_{1\leq i ,j\leq n}\left(\|\f{U}_{\f{H}_{\f{A}}}^* \f{B}_{i,j,t}\|_F^2+\|\f{B}_{i,j,t}\f{V}_{\f{H}_{\f{A}}}\|_F^2 \right) \\
&= 
\begin{cases}
\sum_{\ell=1}^{r} 2 g(\lambda_{\ell}^2,d_2,t) n, & \text{ for }t < d_2, \\
\frac{ 2n}{d_2}, & \text{ for } t = d_2, \\
\sum_{\ell=1}^{r} 2 \lambda_{\ell}^{2(t-d_2)}g(\lambda_{\ell}^2,d_2,T-t+1) n,  & \text{ for }t > d_2.
\end{cases}
\end{split}
\end{equation*}
Therefore
 \begin{align*}
&\sum_{1\leq i, j\leq n} \sum_{t=1}^{T}\big( \| \f{U}_{\f{H}_{\f{A}}}^* (\f{B}_{i,j,t})\|_F^2+\|\f{B}_{i,j,t}\f{V}_{\f{H}_{\f{A}}}\|_F^2\big) \\
&= \sum_{t=1}^{d_2-1}\sum_{\ell=1}^{r} 2(1+\lambda_{\ell}^{2(d_2-t)})g(\lambda_{\ell}^2,d_2,t) n+\frac{2n}{d_2} = 2 G(\Lambda,T) n,
\end{align*} where $G(\Lambda,T):= \sum_{t=1}^{d_2-1} \sum_{\ell=1}^{r} (1+\lambda_{\ell}^{2(d_2-t)})g(\lambda_{\ell}^2,d_2,t) +\frac{1}{d_2}$ is a constant that only depends on the eigenvalues $(\lambda_{\ell})_\ell$, $d_2$ and $T$.  Finally, it follows then from \cref{eq:incoherence:conversion} that
\begin{align*}
\sum_{(i,j,t) \in I}\!\!  \mu_{i,j,t} 
&\leq \!\!\! \sum_{(i,j,t) \in I}     \frac{n d_2^2}{2r d_2} (\|\f{U}_{\f{H}_{\f{A}}}^* \f{B}_{i,j,t}\|_F^2\! +\! \|\f{B}_{i,j,t}\f{V}_{\f{H}_{\f{A}}}\|_F^2) \\
&\leq   \frac{n^2 d_2}{r} G(\Lambda,T),
\end{align*} 
which amounts to the first desired bound \cref{eq:loc:psd:incoherence:1}. Moreover, using the properties of $g$ described previously, one can show that 
 \[
 g(\lambda_{\ell}^2,d_2,t) \leq  g(0, d_2,t)=\frac{1}{t}.
\] 
Since furthermore $\sum_{t=1}^{d_2-1}\frac{1}{t} < c_\gamma + \log (d_2)$, where $c_\gamma < 0.58$ is the Euler-Mascheroni constant, we have 
  \begin{equation*} 
  \begin{split}
 G(\Lambda,L) &\leq \sum_{t=1}^{d_2-1} \sum_{\ell=1}^{r} \frac{1+\lambda_{\ell}^{2(d_2-t)}}{t} +\frac{1}{d_2} \\
 &\leq  2r(c_\gamma + \log(d_2)) +\frac{1}{d_2}  \\
 &\leq (1.16+ 2 \log(d_2) + 1/(r d_2)) r  \leq 4.4r \log(T),
 \end{split}
 \end{equation*} 
 using that $\lambda_\ell \leq 1$ for all $\ell \in [r]$ in the second inequality and $2 \leq d_2 \leq T$ in the last inequality. 
 This yields the second desired bound \cref{eq:loc:psd:incoherence:2} and concludes the proof.
 
Therefore, analgously as to the argument in the proof of \Cref{prop:orthogonalmatrix:incoherences}, it follows that a sufficient number of expected space-time samples $m_{\text{exp}}$ to enable the local convergence guarantee of \Cref{prop:orthogonalmatrix:incoherences} is $ m_{\text{exp}} = \Theta(r n \log(nT)\log(T))$.
\end{proof}

We conclude this section by noting that if $\f{A}$ is a rank-$r$ projection, this amounts to a positive semidefinite transition operator with $\lambda_1 = \lambda_2 = \ldots = \lambda_r = 1$. In this case, the function $g$ of \cref{eq:g:definition} can be simplified to $g(1,d, t) = \frac{1}{d}$, which simplifies the expression for $G(\Lambda,L)$ to $G(\Lambda,L) = (2 d_2-1)/d_2 = 2 - \frac{1}{d_2}$. This means that in fact, for rank-$r$ projection matrices, $m_{\text{exp}} = \Theta( r n  \log(nT))$ expected samples in the adaptive regime are sufficient.

\subsection{Proofs} \label{sec:proofs:appendix}
In the next sections, we provide the proofs of the main local convergence result for \texttt{TOIRLS}, \Cref{thm:convergence:mobilesensors}, as well as the proofs of \Cref{lemma:main:sampling,eq:MatrixIRLS:tangentspace:localRIP:perturbation,lemma:MatrixIRLS:localRIP,thm:convergence:template} which are auxiliary results for proving \Cref{thm:convergence:mobilesensors}.

\subsection{Proof of \Cref{thm:convergence:mobilesensors}} \label{sec:thm:convergence:mobilesensors:proof}
In this section, we provide the proof of \Cref{thm:convergence:mobilesensors}, which is based on combining \Cref{thm:convergence:template} and \Cref{lemma:main:sampling}. 
As an additional ingredient, we bound the spectral norm $\| \y{R}_{\Omega}\|$ of the normalized sampling operators $\y{R}_{\Omega}$ of \cref{def:ROmega,def:ROmega:adaptive:mobile}.

\begin{lemma} \label{lem:repetition}
Let $\Omega$ be a random subset of the index set $I= [n] \times [n] \times [T]$ of size $m$ that is sampled uniformly i.i.d. with replacement, where $m < n^2 T$. Let $\beta > 1$. Then with probability at least $1-(n^2 T)^{1-\beta}$, the maximal number of repetitions of any entry in $\Omega$ is less than 
$\frac{8}{3}\beta\log(n T)$ for $n \sqrt{T} \geq 9$ and $\beta > 1$.

 Consequently, we have that with probability of at least $1-(n^2 T)^{1-\beta}$, the operator $\y{R}_{\Omega}: \Rdd \to \Rdd$ of \cref{def:ROmega}
fulfills
\[
\|\y{R}_\Omega\| \leq \frac{8}{3}\beta \frac{n^2 T}{m} \log(n T),
\]
where $\|\y{R}_\Omega\|$ is the spectral norm of $\y{R}_\Omega$.
\end{lemma}
The proof of \Cref{lem:repetition} is a simple adaptation of \cite[Proposition 5]{recht}.
We proceed to the proof of our main result, \Cref{thm:convergence:mobilesensors}.
\begin{proof}[{Proof of \Cref{thm:convergence:mobilesensors}.1}]
By choosing $\beta = 2$ in \Cref{lem:repetition}, it follows that with probability of at least $1-(n^2 T)^{-1}$,
\begin{equation} \label{eq:lem:repetition:implication:uniformmobile}
\|\y{R}_\Omega\|_{2} \leq \frac{16}{3} \frac{n^2 T}{m} \log(n T)
\end{equation}
Recall that $d= \min(d_1,d_2)$ was chosen to be the minimum of the pencil parameters $d_1$ and $d_2$, which satisfy $d_1 + d_2 -1 = T$. Let $c_0$ be the constant of \Cref{thm:convergence:template} and $C$ the constant of \Cref{lemma:main:sampling}.1.

Fix now $\alpha = 1/5$. From the statement of \Cref{lemma:main:sampling}.1, if follows that if
\begin{equation} \label{eq:m:bound1}
m \geq  25 C c_s \mu_0 r n \log(nT),
\end{equation}
with a probability at least $1-n^{-2}$ the normalized sampling operator $\y{R}_{\Omega}: \MnOplus  \to \MnOplus$ of \cref{def:ROmega} satisfies
\begin{equation*}
\left\|\mathcal{P}_{\f{T}}\mathcal{G}\mathcal{R}_{\Omega}\mathcal{G}^{*}\mathcal{P}_{\f{T}}- \mathcal{P}_{\f{T}}\mathcal{G}\mathcal{G}^{*}\mathcal{P}_{\f{T}}\right\| \leq \frac{1}{5},
\end{equation*}
i.e., \Cref{assumption:1} is satisfied with respect to $\f{T}= \f{T}_{\f{H}_{\f{A}}}$ and constant $\alpha = 1/5$.

Let now $\widetilde{\f{X}}\hk$ be such that $\y{H}(\widetilde{\f{X}}\hk)$ satisfies assumption \cref{eq:proximity:assumption:mainthm:1} of \Cref{thm:convergence:mobilesensors}.1. It follows from \Cref{thm:convergence:template} and \cref{eq:lem:repetition:implication:uniformmobile} that on an event $E$ of probability of at least $1-(n^2 T)^{-1} - n^{-2} \geq 1- 2n^{-2} $, if $c_0$ is the constant of \Cref{thm:convergence:template}, $C$ the constant of \cref{eq:sample:complexity:condition1} and $\widetilde{c}_0 := c_0 (75 C c_s/16 )^{3/2} $, it holds that
\[
\begin{split}
&\|\y{H}(\widetilde{\f{X}}\hk) - \f{H}_{\f{A}}\| \\&\leq \widetilde{c}_0 \mu_0^{3/2} (nT)^{-3/2} r^{1/2} \kappa^{-1}  (dn-r)^{-1/2} \sigma_r(\f{H}_{\f{A}})  \\
&= c_0 \frac{ (25 C c_s)^{3/2} \mu_0^{3/2} r^{1/2}}{(16/3)^{3/2} n^{3/2} T^{3/2} \kappa  (dn-r)^{1/2}} \sigma_r(\f{H}_{\f{A}})   \\
&= c_0 \frac{ (25 C c_s)^{3/2} \mu_0^{3/2} r^{3/2} n^{3/2} \log^{3/2}(nT) \sigma_r(\f{H}_{\f{A}})}{(16/3)^{3/2} n^{3} T^{3/2} \log^{3/2}(nT) r \kappa  (dn-r)^{1/2}} \\
&\leq c_0 \frac{m^{3/2}\sigma_r(\f{H}_{\f{A}}) }{(16/3)^{3/2} n^{3} T^{3/2} \log^{3/2}(nT) r \kappa  (dn-r)^{1/2}} \\
&\leq c_0 \frac{\sigma_r(\f{H}_{\f{A}})}{\|\y{R}_{\Omega}\|^{3/2} r \kappa  (dn-r)^{1/2}}, \\
\end{split}
\]
using also \cref{eq:m:bound1} in the second inequality.

Therefore, the conclusion of \Cref{thm:convergence:template} holds with constant (see the proof of \Cref{thm:convergence:template} in \Cref{sec:weightoperator:structure})
\begin{align*}
\nu &= \frac{20}{3 \sigma_r(\f{H}_{\f{A}})} (1 + 6 \kappa) \left(\|\y{R}_{\Omega}\| + 8/5 \right) r \\&\leq  \frac{20}{3 \sigma_r(\f{H}_{\f{A}})} (1 + 6 \kappa) \left(\frac{16}{3} \frac{n^2 T}{m} \log(n T)+ 8/5 \right) r,
\end{align*}
which means that $\|\y{H}(\widetilde{\f{X}}\hkk) - \f{H}_{\f{A}}\|  \leq \min(\nu \|\y{H}(\widetilde{\f{X}}\hk)-\f{H}_{\f{A}}\|^{2} , \|\y{H}(\widetilde{\f{X}}\hk)-\f{H}_{\f{A}}\|)$ and furthermore, $\y{H}(\widetilde{\f{X}}^{(k+\ell)}) \xrightarrow{\ell \to \infty} \f{H}_{\f{A}}$, on the event $E$ from above.

This finishes the proof of \Cref{thm:convergence:mobilesensors}.1.
\end{proof}
\begin{proof}[{Proof of \Cref{thm:convergence:mobilesensors}.2}]
Let $C >0$ be the constant of \cref{eq:sample:complexity:condition:localcoherence}. To show \Cref{thm:convergence:mobilesensors} in the case of adaptive sampling, we recall the definition 
\begin{equation*}
\f{L} \to \y{R}_{\Omega}(\f{L}) = \sum_{(i,j,t) \in \Omega} \frac{1}{p_{i,j,t}} \langle E_{i,j,t}, \f{L}\rangle_F E_{i,j,t}.
\end{equation*}
of the normalized sampling operator $\y{R}_{\Omega}: \MnOplus  \to \MnOplus$ in this case, cf. \cref{def:ROmega:adaptive:mobile}.

Fix $\alpha = 1/5$. It follows from the definition of $\y{R}_{\Omega}$ and the Bernoulli sampling model that
\begin{equation} \label{eq:lem:repetition:implication:adaptivemobile}
\left\| \y{R}_{\Omega}\right\| \leq  \min_{(i,j,t) \in I} \frac{1}{p_{i,j,t}} \leq \frac{n T}{25 C c_s r \log(n T) \min_{(i,j,t) \in I} \mu_{i,j,t} },
\end{equation}
using assumption \cref{eq:sample:complexity:condition:localcoherence} in the last inequality. Under the same assumption, it follows from \Cref{lemma:main:sampling}.2 that with probability at least $1- n^{-2}$, the local isometry property on $\f{T} = \f{T}_{\HA}$ with constant $1/5$ holds, i.e.,
\begin{equation*}
\left\|\mathcal{P}_{\f{T}}\mathcal{G}\mathcal{R}_{\Omega}\mathcal{G}^{*}\mathcal{P}_{\f{T}}- \mathcal{P}_{\f{T}}\mathcal{G}\mathcal{G}^{*}\mathcal{P}_{\f{T}}\right\| \leq \frac{1}{5},
\end{equation*}
which entails that \Cref{assumption:1} is satisfied for $\alpha = 1/5$. As above, it follows from \Cref{thm:convergence:template} and \cref{eq:lem:repetition:implication:adaptivemobile} that on an event of probability at least $1-n^{-2} \geq 1 - 2n^{-2}$, if $c_0$ is the constant of \Cref{thm:convergence:template}, $C$ the constant of \cref{eq:sample:complexity:condition:localcoherence} and $\widetilde{c}_0 = c_0 (25 C c_s )^{3/2}$,
\[
\begin{split}
&\|\y{H}(\widetilde{\f{X}}\hk) - \f{H}_{\f{A}}\| \\&\leq
\widetilde{c_0}  \frac{ \min_{(i,j,t) \in I} \mu_{i,j,t}^{3/2} r^{1/2} \log^{3/2}(nT)}{ (nT)^{3/2} \kappa  (dn-r)^{1/2}} \sigma_r(\f{H}_{\f{A}})   \\
&\leq c_0 \frac{(25 C c_s)^{3/2} r^{1/2} \log^{3/2}(n T) \min_{(i,j,t) \in I} \mu_{i,j,t}^{3/2} }{ (n T)^{3/2} \kappa  (dn-r)^{1/2}} \sigma_r(\f{H}_{\f{A}}) \\
&\leq c_0 \frac{1}{\|\y{R}_{\Omega}\|^{3/2} r \kappa  (dn-r)^{1/2}} \sigma_r(\f{H}_{\f{A}}), \\
\end{split}
\]
and therefore, the conclusion of \Cref{thm:convergence:template} holds with constant
\[
\begin{split}
\nu &= \frac{20}{3 \sigma_r(\f{H}_{\f{A}})} (1 + 6 \kappa) \left(\|\y{R}_{\Omega}\| + 8/5 \right) r  \\
&\leq  \frac{20}{3 \sigma_r(\f{H}_{\f{A}})} (1 + 6 \kappa) \left(\frac{n T (\min_{(i,j,t) \in I} \mu_{i,j,t} )^{-1}}{25 C c_s r \log(n T) }+ 8/5 \right) r,
\end{split}
\]
which means that $\|\y{H}(\widetilde{\f{X}}\hkk) - \f{H}_{\f{A}}\|  \leq \min(\nu \|\y{H}(\widetilde{\f{X}}\hk)-\f{H}_{\f{A}}\|^{2} , \|\y{H}(\widetilde{\f{X}}\hk)-\f{H}_{\f{A}}\|)$ and furthermore, $\y{H}(\widetilde{\f{X}}^{(k+\ell)}) \xrightarrow{\ell \to \infty} \f{H}_{\f{A}}$, and therefore concludes the proof of \Cref{thm:convergence:mobilesensors}.
\end{proof}

\subsection{Proof of \Cref{lemma:main:sampling}} \label{sec:proof:samplingoperator:assumption}
In this section, we prove \Cref{lemma:main:sampling}, our main result about the regularity of the normalized sampling operators  $\y{R}_{\Omega}: \MnOplus \to \MnOplus$ for the uniform and adaptive sampling models, see \cref{def:ROmega} and \cref{def:ROmega:adaptive:mobile}, respectively. The proof uses a non-commutative Bernstein inequality:

\begin{lemma}[{Noncommutative Bernstein inequality, cf. \cite[Theorem 4]{recht} or \cite[Theorem 5.4.1]{Ver18}}] \label{theorem:matrix:Bernstein}
Let $\y{Z}_1,\ldots,\y{Z}_{m}$ be independent, Hermitian zero-mean random operators of dimension $n^2 d_1 d_2 \times n^2 d_1 d_2$. Suppose that $\rho^2 =  \|\mathbb{E} \sum_{\ell=1} \y{Z}_\ell \y{Z}_{\ell}\|$ and $\|\y{Z}_\ell \| \leq M$ almost surely for all $\ell \in [m]$. Then for any $\alpha > 0$,
\[
\mathbb{P} \left( \left\| \sum_{\ell=1}^m \y{Z}_{\ell} \right\|  \right) \leq 2 n^2 d_1 d_2 \exp \left( \frac{-\alpha^2/2}{ \rho^2 + M \alpha / 3} \right).
\]
\end{lemma}

The proof of \Cref{lemma:main:sampling}.1 follows the proof idea of \cite[Lemma 3]{chen_chi14} and \cite[Lemma 23]{LeeLiJinYe18}.
\begin{proof}[{Proof of \Cref{lemma:main:sampling}.1}]
If $\left\{E_{i,j,t} \right\}_{(i,j,t) \in I}$ is the standard basis of $\MnOplus$ and $\left\{ \f{B}_{i,j,t}\right\}_{(i,j,t) \in I}$ the standard basis of the space of block Hankel matrices $\mathcal{H}(\MnOplus)$, we recall from \Cref{lemma:Hankel:action} that $\y{G}(E_{i,j,t}) = \f{B}_{i,j,t}$ for each $(i,j,t) \in I$. 

We first assume a slightly different sampling model than that considered in the statement of \Cref{lemma:main:sampling}: let $\Omega = \{(i_{\ell},j_{\ell},t_{\ell})\}_{\ell=1}^{m} \subset I$ be a set of $m$ indices sampled uniformly i.i.d. \emph{with} replacement. 
For $\ell \in [m]$, define the operators $\y{Z}_{\ell}$ and $\widetilde{\y{Z}}_{\ell}$ such that
\begin{align*}
\mathcal{Z}_{\ell} &:= \frac{n^2 T}{m} \widetilde{\y{Z}}_{\ell}  -   \frac{1}{m} \mathcal{P}_{\f{T}} \mathcal{G} \mathcal{G}^* \mathcal{P}_{\f{T}}\\&:=  \frac{n^2 T}{m} \mathcal{P}_{\f{T}} \mathcal{G}E_{i_{\ell},j_{\ell},t_{\ell}} E_{i_{\ell},j_{\ell},t_{\ell}}^*  \mathcal{G}^* \mathcal{P}_{\f{T}}  - \frac{1}{m} \mathcal{P}_{\f{T}} \mathcal{G} \mathcal{G}^* \mathcal{P}_{\f{T}}.
\end{align*}
Then the expectation $\mathbb{E}[\widetilde{\mathcal{Z}}_{\ell}]$ of $\widetilde{\mathcal{Z}}_{\ell}$ satisfies
\begin{equation} \label{eq:ExpZell}
\begin{split}
\mathbb{E}[\widetilde{\mathcal{Z}}_{\ell}] &=  \mathbb{E}\left[\mathcal{P}_{\f{T}} \mathcal{G}E_{i_{\ell},j_{\ell},t_{\ell}} E_{i_{\ell},j_{\ell},t_{\ell}}^*  \mathcal{G}^* \mathcal{P}_{\f{T}}\right]\\& = \frac{1}{n^2 T} \sum_{i,j=1,i\leq j}^n \sum_{t=1}^T  \mathcal{P}_{\f{T}} \mathcal{G}E_{i,j,t} E_{i,j,t}^*  \mathcal{G}^* \mathcal{P}_{\f{T}}     \\
&= \frac{1}{n^2 T} \mathcal{P}_{\f{T}} \mathcal{G} \mathcal{G}^* \mathcal{P}_{\f{T}}
\end{split}
\end{equation}
and furthermore,
\begin{equation*}
\mathbb{E}[ \mathcal{Z}_{\ell} ] = \frac{n^2 T}{m} \mathbb{E} [ \widetilde{\y{Z}}_{\ell}  ] - \frac{1}{m} \mathcal{P}_{\f{T}} \mathcal{G} \mathcal{G}^* \mathcal{P}_{\f{T}} = 0.
\end{equation*}
Since for any $\f{M} \in \R^{d_1 n \times d_2 n}$,
\[
\widetilde{\mathcal{Z}}_{\ell}(\f{M}) =\langle \mathcal{P}_{\f{T}}(\f{B}_{i_{\ell},j_{\ell},t_{\ell}}), \f{M} \rangle_F \mathcal{P}_{\f{T}} (\f{B}_{i_{\ell},j_{\ell},t_{\ell}}),
\]
we obtain
\begin{align*}
\| \widetilde{\mathcal{Z}}_{\ell}(\f{M})\|_F &\leq \left|\langle \mathcal{P}_{\f{T}} (\f{B}_{i_{\ell},j_{\ell},t_{\ell}}), \f{M} \rangle_F\right|  \left\|\mathcal{P}_{\f{T}} (\f{B}_{i_{\ell},j_{\ell},t_{\ell}})\right\|_F \\&\leq  \|\mathcal{P}_{\f{T}} (\f{B}_{i_{\ell},j_{\ell},t_{\ell}})\|_F^2 \|\f{M}\|_F
\end{align*}
by Cauchy-Schwarz, and thus obtain
\begin{equation} \label{eq:PTGeomegaGPT:bound}
\begin{split}
\left\| \widetilde{\mathcal{Z}}_{\ell}\right\| &\leq  \left\|\mathcal{P}_{\f{T}} (\f{B}_{i_{\ell},j_{\ell},t_{\ell}})\right\|_F^2 \\
&\leq \max_{1\leq i \leq j \in [n], t \in [T]}  \| \mathcal{P}_{\f{T}} (\f{B}_{i,j,t}) \|_F^2 \leq \frac{\mu_0 c_s r}{n T}\,,
\end{split}
\end{equation}
using the incoherence assumption on $\f{H}_{\f{A}}$ in the last inequality, as well as $d_1 +d_2 -1 =T$ and the definition of $c_s = T(T+1)/(d_1 d_2)$. Analogously, we estimate that
\begin{equation} 
\begin{aligned}
\left\|\frac{1}{m} \mathcal{P}_{\f{T}} \mathcal{G} \mathcal{G}^* \mathcal{P}_{\f{T}}\right\| &\leq \frac{1}{m} \sum_{i,j=1,i\leq j}^n  \left\| \sum_{t=1}^T \mathcal{P}_{\f{T}} \mathcal{G}E_{i,j,t} E_{i,j,t}^*  \mathcal{G}^* \mathcal{P}_{\f{T}} \right\| \\
&\leq \frac{n^2 T}{m} \frac{\mu_0 c_s r }{n T} = \frac{\mu_0 c_s r n}{m}.
\label{eq:PTGGPT:bound}
\end{aligned}
\end{equation}
We observe that if $\y{A}$ and $\y{B}$ are positive semidefinite operators, it holds that $\|\y{A} - \y{B}\| \leq \max(\|\y{A}\|, \|\y{B}\|)$. Therefore, it follows from \cref{eq:PTGeomegaGPT:bound} and \cref{eq:PTGGPT:bound} that 
\begin{align} \label{eq:Zell:spectralnorm}
\|\y{Z}_{\ell} \| \leq \max\left( \frac{n^2 T}{m} \frac{\mu_0  c_s r}{ nT } , \frac{\mu_0 c_s r n}{m} \right) = \frac{\mu_0 c_s r n}{ m}
\end{align}
almost surely for all $\ell \in [m]$, as the operators involved are positive semidefinite. Further we compute that
\begin{align*}
\begin{split}
&\mathbb{E}\left[ \y{Z}_{\ell}\y{Z}_{\ell} \right] \\&= \frac{(n^2 T)^2}{m^2}  \mathbb{E}\left[ \widetilde{\y{Z}}_{\ell} \widetilde{\y{Z}}_{\ell}\right] -  \frac{n^2 T}{m^2} \mathbb{E}\left[\widetilde{\y{Z}}_{\ell}\right]  \mathcal{P}_{\f{T}} \mathcal{G} \mathcal{G}^* \mathcal{P}_{\f{T}} \\
&-  \frac{n^2 T}{m^2} \mathcal{P}_{\f{T}} \mathcal{G} \mathcal{G}^* \mathcal{P}_{\f{T}}\mathbb{E}\left[ \widetilde{\y{Z}}_{\ell} \right] + \frac{1}{m^2}  \mathcal{P}_{\f{T}} \mathcal{G} \mathcal{G}^* \mathcal{P}_{\f{T}} \mathcal{P}_{\f{T}} \mathcal{G} \mathcal{G}^* \mathcal{P}_{\f{T}} \\
&=  \frac{(n^2 T)^2}{m^2}  \mathbb{E}\left[ \widetilde{\y{Z}}_{\ell} \widetilde{\y{Z}}_{\ell}\right] - \frac{1}{m^2}  \mathcal{P}_{\f{T}} \mathcal{G} \mathcal{G}^* \mathcal{P}_{\f{T}} \mathcal{P}_{\f{T}} \mathcal{G} \mathcal{G}^* \mathcal{P}_{\f{T}},
\end{split}
\end{align*}
using that $\mathbb{E}\left[ \widetilde{\y{Z}}_{\ell} \right]  = \frac{1}{n^2 T}  \mathcal{P}_{\f{T}} \mathcal{G} \mathcal{G}^* \mathcal{P}_{\f{T}}$, cf. \cref{eq:ExpZell}. In order to estimate the latter terms, we observe that for any $\f{M} \in \R^{d_1 n \times d_2 n}$, 
\[
\begin{split}
\widetilde{\y{Z}}_{\ell} \widetilde{\y{Z}}_{\ell}(\f{M}) &= \langle \y{P}_{\f{T}}(\f{B}_{i_{\ell},j_{\ell},t_{\ell}}), \widetilde{\y{Z}}_{\ell}(\f{M})\rangle \y{P}_{\f{T}}(\f{B}_{i_{\ell},j_{\ell},t_{\ell}}) \\&= \left\| \y{P}_{\f{T}}(\f{B}_{i_{\ell},j_{\ell},t_{\ell}}) \right\|_F^2  \langle \y{P}_{\f{T}}(\f{B}_{i_{\ell},j_{\ell},t_{\ell}}), \f{M} \rangle \y{P}_{\f{T}}(\f{B}_{i_{\ell},j_{\ell},t_{\ell}}) \\
&=  \left\| \y{P}_{\f{T}}(\f{B}_{i_{\ell},j_{\ell},t_{\ell}}) \right\|_F^2 \widetilde{\y{Z}}_{l}(\f{M})
\end{split}
\]
and therefore
\begin{equation} \label{eq:widetildeZellsquared:bound}
\begin{split}
\left\| \mathbb{E}\widetilde{\y{Z}}_{\ell}\widetilde{\y{Z}}_{\ell} \right\| &= \max_{\|M\|_F=1}\left\| \mathbb{E} \left\| \y{P}_{\f{T}}(\f{B}_{i_{\ell},j_{\ell},t_{\ell}}) \right\|_F^2 \widetilde{\y{Z}}_{\ell}(\f{M}) \right\|_F \\&\leq  \max_{i\leq j \in [n] t \in [T]} \left\| \y{P}_{\f{T}}(\f{B}_{i,j,t}) \right\|_F^2 \left\| \mathbb{E}  \widetilde{\y{Z}}_{\ell} \right\| \\
&\leq \frac{\mu_0  c_s r}{ n T}    \frac{1}{n^2 T}  \|\mathcal{P}_{\f{T}} \mathcal{G} \mathcal{G}^* \mathcal{P}_{\f{T}}\| \leq  \frac{ \mu_0 c_s r}{n^3 T^2},
\end{split}
\end{equation}
using \cref{eq:PTGeomegaGPT:bound} in the second inequality and the fact $\|\mathcal{P}_{\f{T}} \mathcal{G} \mathcal{G}^* \mathcal{P}_{\f{T}}\| \leq 1$ in the third in equality. 
For the expectation of the squares of $\y{Z}_{\ell}$, we obtain
\begin{equation} \label{eq:Zellsquared:spectralnorm}
\begin{split}
 &\sum_{\ell=1}^m \left\| \Ex \y{Z}_{\ell}\y{Z}_{\ell} \right\| \\
 &\leq \frac{1}{m^{2}} \sum_{\ell=1}^m \! \max\!\Big( n^2T^2  \left\| \Ex \widetilde{\y{Z}}_{\ell}\widetilde{\y{Z}}_{\ell} \right\|\!\!, \!\left\| \mathcal{P}_{\f{T}} \mathcal{G} \mathcal{G}^* \mathcal{P}_{\f{T}} \mathcal{P}_{\f{T}} \mathcal{G} \mathcal{G}^* \mathcal{P}_{\f{T}} \right\| \! \Big) \\
&\leq \sum_{\ell=1}^m  \max\left( \frac{(n^2 T)^2}{(m)^2}  \frac{ \mu_0 c_s r}{n^3 T^2} , \frac{1}{m^2} \right) \\
&=  \sum_{\ell=1}^m  \max\left( \frac{n}{ m^2}  \frac{\mu_0 c_s r}{1} , \frac{1}{m^2} \right) \\
&\leq \frac{\mu_0 c_s r n}{m},
\end{split}
\end{equation}
using again that $\|\y{A} -\y{B}\| \leq \max( \|\y{A}\|,\|\y{B}\|)$ for positive semidefinite operators in the second inequality, \cref{eq:widetildeZellsquared:bound} and the fact that $\| \mathcal{P}_{\f{T}} \mathcal{G} \mathcal{G}^* \mathcal{P}_{\f{T}}\| \leq 1$ in the third inequality.

Next, recalling the definition \cref{def:ROmega} of the normalized sampling operator $\y{R}_{\Omega}: \MnOplus \to \MnOplus$ of the statement of \Cref{lemma:main:sampling}, we observe that $\y{P}_{\f{T}}\y{G}\y{R}_{\Omega}\y{G}^{*}\y{P}_{\f{T}} = \frac{n^2 T}{m} \sum_{\ell=1}^m \widetilde{\y{Z}}_{\ell}$.

Since the $\y{Z}_{\ell}$'s are Hermitian, we can now apply the matrix Bernstein inequality \cite[Theorem 4]{recht} in form of \Cref{theorem:matrix:Bernstein} above to obtain, for $0 < \alpha <1$, the estimate
\begin{equation*}
\begin{split}
    &\bb{P}\left( \left\| \y{P}_{\f{T}}\y{G}\y{R}_{\Omega}\y{G}^{*}\y{P}_{\f{T}} - \mathcal{P}_{\f{T}}\mathcal{G}\mathcal{G}^{*}\mathcal{P}_{\f{T}} \right\| \geq \alpha \right)	\\&\leq 2n^2d_1 d_2 \exp\left(- \frac{m \alpha^2}{2\mu_0 c_s rn(1+ \alpha/3)}\right) \\
&\leq 2 \frac{(T+1)^2}{4}n^2 \exp\left(- \frac{m \alpha^2}{ 2\mu_0 c_s rn(4/3)}\right) \\&= \frac{(T+1)^2 n^2}{2} \exp\left(- \frac{ 3 m \alpha^2}{8\mu_0 c_s rn}\right),\\
\end{split}
\end{equation*}
using the norm estimates of \cref{eq:Zell:spectralnorm} and \cref{eq:Zellsquared:spectralnorm} to estimate the respective quantities in \Cref{theorem:matrix:Bernstein}. From this, we see that
\[
\bb{P}\left( \left\| \y{P}_{\f{T}}\y{G}\y{R}_{\Omega}\y{G}^{*}\y{P}_{\f{T}} - \mathcal{P}_{\f{T}}\mathcal{G}\mathcal{G}^{*}\mathcal{P}_{\f{T}} \right\| \geq \alpha \right) \leq n^{-2}
\]
if $\frac{1}{2}(T+1)^2 n^{4} \leq \exp\left(\frac{3 m \alpha^2}{8 \mu_0 c_s rn}\right)$, which is further implied by the condition 
\[
m \geq \frac{16 c_s}{3 \alpha^2} \mu_0 r n \log(n T).
\]

This shows that for the constant $C:= \frac{16}{3}$, if \cref{eq:sample:complexity:condition1} is fulfilled, then with probability at least $1- n^{-2}$, 
\[
\left\| \y{P}_{\f{T}}\y{G}\y{R}_{\Omega}\y{G}^{*}\y{P}_{\f{T}} - \mathcal{P}_{\f{T}}\mathcal{G}\mathcal{G}^{*}\mathcal{P}_{\f{T}} \right\| < \alpha
\]
if $m$ i.i.d. samples are uniformly sampled with replacement. With the argument of \cite[Proposition 3]{recht}, we conclude that the statement with the same probability bound holds true for the sampling model where $\Omega$ is a random subset of cardinality $m$, uniformly drawn without replacement, if $m$ satisfies \cref{eq:sample:complexity:condition1}, which finishes the proof. 

\end{proof}
\begin{proof}[{Proof of \Cref{lemma:main:sampling}.2}]
To show the second part of \Cref{lemma:main:sampling}, we consider for each $(i,j,t) \in I$ a random variable $\delta_{i,j,t}$  that is $1$ if $(i,j,t) \in \Omega$ and $0$ otherwise. With that notation, $\y{R}_{\Omega}$ of \cref{def:ROmega} can be written as
\begin{align*}
\mathcal{P}_{\f{T}}\mathcal{G}\mathcal{R}_{\Omega}\mathcal{G}^{*}\mathcal{P}_{\f{T}} &= \sum_{(i,j,t) \in I} \frac{\delta_{i,j,t}}{p_{i,j,t}} \mathcal{P}_{\f{T}}\mathcal{G} E_{i,j,t} E_{i,j,t}^* \mathcal{G}^{*}\mathcal{P}_{\f{T}} \\& =: \sum_{(i,j,t) \in I} \frac{\delta_{i,j,t}}{p_{i,j,t}} \widetilde{\y{Z}}_{i,j,t},
\end{align*}
defining operators $\widetilde{\y{Z}}_{i,j,t}: \MnOplus \to \MnOplus$ for each $(i,j,t) \in I$. With this, we obtain
\begin{align*}
\begin{split}
&\mathcal{P}_{\f{T}}\mathcal{G}\mathcal{R}_{\Omega}\mathcal{G}^{*}\mathcal{P}_{\f{T}} - \mathcal{P}_{\f{T}}\mathcal{G}\mathcal{G}^{*}\mathcal{P}_{\f{T}}  \\&= \sum_{(i,j,t) \in I} \frac{\delta_{i,j,t}}{p_{i,j,t}} \widetilde{\y{Z}}_{i,j,t} - \sum_{(i,j,t) \in I}  \widetilde{\y{Z}}_{i,j,t} \\
&= \sum_{(i,j,t) \in I} \left(  \frac{\delta_{i,j,t}}{p_{i,j,t}}  - 1\right)  \widetilde{\y{Z}}_{i,j,t} =: \sum_{(i,j,t) \in I} \y{Z}_{i,j,t},
\end{split}
\end{align*}
defining the random operators $\mathcal{Z}_{i,j,t}$. Based on the assumption on the sampling model, the $\mathcal{Z}_{i,j,t}$ are independent and as the $\delta_{i,j,t}$ are Bernoulli variables with success probabilities $p_{i,j,t}$, it follows that 
\[
\mathbb{E}[\mathcal{Z}_{i,j,t}] =  \left( \frac{\mathbb{E}[\delta_{i,j,t}]}{p_{i,j,t}} - 1\right) \mathcal{P}_{\f{T}} \mathcal{G}E_{i,j,t} E_{i,j,t}^*  \mathcal{G}^* \mathcal{P}_{\f{T}}  = 0.
\]
Let $\f{M} \in \R^{d_1 n \times d_2 n}$ be arbitrary. Since
\begin{equation} \label{eq:Zijk}
\begin{split}
\widetilde{\y{Z}}_{i,j,t}(\f{M})&= \langle \f{B}_{i,j,t}, \mathcal{P}_{\f{T}}(\f{M}) \rangle_F \mathcal{P}_{\f{T}}(\f{B}_{i,j,t})\\&= \langle \mathcal{P}_{\f{T}}( \f{B}_{i,j,t}), \f{M} \rangle_F \mathcal{P}_{\f{T}}(\f{B}_{i,j,t})
\end{split}
\end{equation}
we obtain
\begin{align*}
\begin{split}
\| \widetilde{\y{Z}}_{i,j,t}(\f{M})\|_F &\leq  \left|\langle \mathcal{P}_{\f{T}}(\f{B}_{i,j,t}), \f{M} \rangle_F\right|  \|\mathcal{P}_{\f{T}}(\f{B}_{i,j,t})\|_F \\&\leq  \|\mathcal{P}_{\f{T}} (\f{B}_{i,j,t})\|_F^2 \|\f{M}\|_F \\
&\leq  \frac{r (d_1 + d_2)}{n d_1 d_2} \mu_{i,j,t}  \|\f{M}\|_F  \leq  \frac{ c_s r }{n T} \mu_{i,j,t} \|\f{M}\|_F, 
\end{split}
\end{align*}
using the definition of the local incoherence factor $\mu_{i,j,t}$, cf. \Cref{def:incoherence}.
This implies that
\[
\| \widetilde{\y{Z}}_{i,j,t}\| \leq \frac{ c_s r }{n T} \mu_{i,j,t}
\]
almost surely for each $i \leq j \leq n$ and each $t \leq T$ and, since $\delta_{i,j,t}/p_{i,j,t} \widetilde{\y{Z}}_{i,j,t}$ and $\widetilde{\y{Z}}_{i,j,t}$ are both positive semidefinite operators, that
\begin{align} \label{eq:Zijk:normbound}
\| \y{Z}_{i,j,t}\| &\leq \max \left(\frac{\delta_{i,j,t}}{p_{i,j,t}}\|\widetilde{\y{Z}}_{i,j,t}\|, \|\widetilde{\y{Z}}_{i,j,t}\|   \right)\\ &\leq \frac{1}{p_{i,j,t}} \|\widetilde{\y{Z}}_{i,j,t}\| \leq \frac{ c_s r }{p_{i,j,t} n T} \mu_{i,j,t}
\end{align}
almost surely as well. Furthermore, for the expectation of the squares of $\y{Z}_{i,j,t}$  we obtain
 \begin{align*}
\begin{split}
&\Ex \y{Z}_{i,j,t}\y{Z}_{i,j,t} \\&=  \Ex \left[ \frac{\delta_{i,j,t}^2}{p_{i,j,t}^2}\widetilde{\y{Z}}_{i,j,t}\widetilde{\y{Z}}_{i,j,t}  \right]  - 2 \Ex\left[\frac{\delta_{i,j,t}}{p_{i,j,t}}\widetilde{\y{Z}}_{i,j,t}\widetilde{\y{Z}}_{i,j,t}  \right]  + \widetilde{\y{Z}}_{i,j,t}\widetilde{\y{Z}}_{i,j,t}\\
& =  \frac{p_{i,j,t}}{p_{i,j,t}^2} \widetilde{\y{Z}}_{i,j,t}\widetilde{\y{Z}}_{i,j,t} - \widetilde{\y{Z}}_{i,j,t}\widetilde{\y{Z}}_{i,j,t} \\
&= \left(\frac{1}{p_{i,j,t}} - 1  \right)\widetilde{\y{Z}}_{i,j,t}\widetilde{\y{Z}}_{i,j,t}.
\end{split}
\end{align*}
Now, using \cref{eq:Zijk} 
and observing that for any $\f{M} \in \R^{d_1 n \times d_2 n}$
\begin{align*}
\widetilde{\y{Z}}_{i,j,t}\widetilde{\y{Z}}_{i,j,t}(\f{M}) &= \|\mathcal{P}_{\f{T}}( \f{B}_{i,j,t})\|_F^2 \langle \mathcal{P}_{\f{T}}( \f{B}_{i,j,t}),\f{M} \rangle_F \mathcal{P}_{\f{T}}(\f{B}_{i,j,t})\\& = \|\mathcal{P}_{\f{T}}( \f{B}_{i,j,t})\|_F^2 \mathcal{P}_{\f{T}} \f{B}_{i,j,t} \f{B}_{i,j,t}^* \mathcal{P}_{\f{T}} (\f{M}), 
\end{align*}
we obtain the spectral norm bound
\begin{align*}
\left\|\widetilde{\y{Z}}_{i,j,t}\widetilde{\y{Z}}_{i,j,t}\right\| &\leq  \|\mathcal{P}_{\f{T}}( \f{B}_{i,j,t})\|_F^2 \left\| \mathcal{P}_{\f{T}} \f{B}_{i,j,t} \f{B}_{i,j,t}^* \mathcal{P}_{\f{T}}\right\|\\& \leq   \frac{r (d_1 + d_2)}{n d_1 d_2} \mu_{i,j,t} \leq  \frac{ c_s r }{n T} \mu_{i,j,t},
\end{align*}
using the definition of $\mu_{i,j,t}$ and the fact that $\left\| \mathcal{P}_{\f{T}} \f{B}_{i,j,t} \f{B}_{i,j,t}^* \mathcal{P}_{\f{T}}\right\| \leq 1$, as well as $d_1 +d_2 -1 =T$ and the definition of $c_s = T(T+1)/(d_1 d_2)$ in the last inequality. Due to a similar argument as made in \cref{eq:Zijk:normbound}, we obtain that

\begin{equation*}
\begin{split}
&\left\| \sum_{(i,j,t) \in I}   \Ex \y{Z}_{i,j,t}\y{Z}_{i,j,t}\right\| \!\! = \left\| \sum_{(i,j,t) \in I}\!  \left(\frac{1}{p_{i,j,t}} - 1  \right)\!\widetilde{\y{Z}}_{i,j,t}\widetilde{\y{Z}}_{i,j,t} \right\|  \\
&= \left\| \sum_{(i,j,t) \in I}  \left(\frac{1}{p_{i,j,t}} - 1  \right)\|\mathcal{P}_{\f{T}}( \f{B}_{i,j,t})\|_F^2 \mathcal{P}_{\f{T}} \f{B}_{i,j,t} \f{B}_{i,j,t}^* \mathcal{P}_{\f{T}} \right\| \\
\end{split}
\end{equation*}
\begin{equation} \label{eq:ZijkZijk:spectralnorm}
\begin{split}
&\leq \max_{(i',j',t') \in I} \left(\frac{1}{p_{i,j,t}} - 1  \right)\|\mathcal{P}_{\f{T}}( \f{B}_{i,j,t})\|_F^2  \\
 &\qquad\cdot \left\| \sum_{(i,j,t) \in I} \mathcal{P}_{\f{T}} \f{B}_{i,j,t} \f{B}_{i,j,t}^* \mathcal{P}_{\f{T}} \right\| \\
&\leq \max_{(i',j',t') \in I}  \frac{ c_s r }{n T} \frac{\mu_{i',j',t'}}{p_{i',j,'t'}} \left\|\mathcal{P}_{\f{T}} \y{G} \y{G}^* \mathcal{P}_{\f{T}} \right\| \\&\leq \max_{(i',j',t') \in I}  \frac{ c_s r }{n T} \frac{\mu_{i',j',t'}}{p_{i',j,'t'}}  =: \widetilde{c}
\end{split}
\end{equation}
using the formulas for $\Ex \y{Z}_{i,j,t}\y{Z}_{i,j,t} $ and $\widetilde{\y{Z}}_{i,j,t}\widetilde{\y{Z}}_{i,j,t}$ from above, the fact that the \\
$\mathcal{P}_{\f{T}} \f{B}_{i,j,t} \f{B}_{i,j,t}^* \mathcal{P}_{\f{T}}$ are all positive semidefinite and the assumption the $p_{i,j,t} \leq 1 $ for all $(i,j,t) \in I$. Furthermore, we used the definition of the local coherences $\mu_{i',j',t'}$ from \Cref{def:incoherence} in the first inequality, and the fact that $ \left\|\mathcal{P}_{\f{T}} \y{G} \y{G}^* \mathcal{P}_{\f{T}} \right\| \leq 1$ in the last inequality.

As the $\y{Z}_{i,j,t}$ are Hermitian, we can now use \cref{eq:ZijkZijk:spectralnorm} and \cref{eq:Zijk:normbound} to apply the matrix Bernstein inequality \Cref{theorem:matrix:Bernstein} to estimate that
\begin{equation*}
\begin{split}
&\bb{P}\left( \left\| \mathcal{P}_{\f{T}}\mathcal{G}\mathcal{R}_{\Omega}\mathcal{G}^{*}\mathcal{P}_{\f{T}}- \mathcal{P}_{\f{T}}\mathcal{G}\mathcal{G}^{*}\mathcal{P}_{\f{T}} \right\| \geq \alpha \right)\\&	\leq 2 n^2 d_1 d_2 \exp\left(- \frac{\alpha^2/2}{\widetilde{c} + \widetilde{c} \alpha/3}\right) \\
&\leq 2 \frac{(T+1)^2}{4} n^2 \exp\left(- \frac{\alpha^2/2}{\widetilde{c} + \widetilde{c} \alpha/3}\right)\\&  \leq \frac{1}{2}(T+1)^2 n^2 \exp\left(- \frac{3 \alpha^2}{8 \widetilde{c}}\right) \leq n^{-2},\\
\end{split}
\end{equation*}
where the last inequality holds if $\widetilde{c}^{-1} \geq \frac{8}{3 \alpha^2} \left( 4 \log(n) + \log(1/2) + 2 \log(T+1)\right)$,
which, in view of the definition of $\widetilde{c}$ from \cref{eq:ZijkZijk:spectralnorm} , is implied by the condition 
\[
p_{i,j,t} \geq \frac{32}{3 \alpha^2} \mu_{i,j,t} c_s \frac{r}{n T } \log((T+1)n) \] for all $ 1 \leq  i \leq j \leq n, 1 \leq t \leq T$.

This shows that there exists an absolute constant $C>1$ such that if \cref{eq:sample:complexity:condition:localcoherence} is fulfilled for each $(i,j,t) \in I$, with probability at least $1- n^2$, it holds that 
\[
 \left\| \mathcal{P}_{\f{T}}\mathcal{G}\mathcal{R}_{\Omega}\mathcal{G}^{*}\mathcal{P}_{\f{T}}- \mathcal{P}_{\f{T}}\mathcal{G}\mathcal{G}^{*}\mathcal{P}_{\f{T}} \right\| < \alpha,
\]
which finishes the proof of \Cref{lemma:main:sampling}. 
\end{proof}

\subsection{Proof of \Cref{eq:MatrixIRLS:tangentspace:localRIP:perturbation}} \label{sec:RIP:neighborhood}
To show the perturbation result of \Cref{eq:MatrixIRLS:tangentspace:localRIP:perturbation}, we use ideas from the proof of \cite[Lemma 8]{cai_wang_wei19}. As an auxiliary result, we also use the following lemma.

\begin{lemma}[{\cite[Lemma 4.2]{wei_cai_chan_leung}, \cite[Eq. (30)]{KMV21}}] \label{lemma:PT:perturbation}
If $\f{T} := \f{T}_{\f{H}_{\f{A}}}$ and $\f{T}_{\f{H}}$ are the tangent spaces of the rank-$r$ matrix manifold at $\f{H}_{\f{A}}$ and $\f{H}$, respectively, then 
\[
\left\|\y{P}_{\f{T}} - \PTX  \right\| \leq \frac{4 \left\| \f{H}_{\f{A}} - \f{H} \right\|}{\sigma_{r}(\f{H}_{\f{A}} )}.
\]
\end{lemma}

\begin{proof}[{Proof of \Cref{eq:MatrixIRLS:tangentspace:localRIP:perturbation}}]
Recall that $\f{T}= T_{\f{H}_{\f{A}}}   \subset \Rddn$ is the tangent space onto $\y{M}_r$ at $\f{H}_{\f{A}}$. For any $\f{Z} \in \Rddn$, we have 
\begin{align*}
&\left\|\y{R}_{\Omega}\y{G}^{*} \y{P}_{\f{T}}(\f{Z})\right\|_F^2 \\&= \left\langle \y{R}_{\Omega}\y{G}^{*}\y{P}_{\f{T}}(\f{Z}), \y{R}_{\Omega}\y{G}^{*}\y{P}_{\f{T}}(\f{Z})\right\rangle  =  \left\langle \y{G}^{*}\y{P}_{\f{T}}(\f{Z}), \y{R}_{\Omega}^2 \y{G}^{*}\y{P}_{\f{T}}(\f{Z})\right\rangle  \\
&\leq  \left\| \y{R}_{\Omega} \right\|  \left\langle \y{G}^{*}\y{P}_{\f{T}}(\f{Z}),\y{R}_{\Omega}  \y{G}^{*} \y{P}_{\f{T}}(\f{Z})\right\rangle \\
&=  \left\| \y{R}_{\Omega} \right\|  \left\langle \f{Z}, \y{P}_{\f{T}}\y{G}\y{R}_{\Omega}  \y{G}^{*} \y{P}_{\f{T}}(\f{Z})\right\rangle   \\
&=  \left\| \y{R}_{\Omega} \right\| \big( \left\langle \f{Z}, \left(\y{P}_{\f{T}}\y{G}\y{R}_{\Omega}  \y{G}^{*} \y{P}_{\f{T}} - \y{P}_{\f{T}}\y{G} \y{G}^{*} \y{P}_{\f{T}} \right)\f{Z}\right\rangle \\
&\quad\quad\quad\quad\;\;+ \left\langle \f{Z},\y{P}_{\f{T}}\y{G} \y{G}^{*} \y{P}_{\f{T}} \f{Z}\right\rangle  \big) \\
&\leq  \left\| \y{R}_{\Omega} \right\| \left( \alpha \|\f{Z}\|_F^2 + \|\f{Z}\|_F^2 \right)  =  \left\| \y{R}_{\Omega} \right\| \left(1 + \alpha \right) \|\f{Z}\|_F^2
\end{align*}
using the fact that $\y{R}_{\Omega}$ is self-adjoint in the second inequality and \Cref{assumption:1} in the last inequality. From this, it follows that 
\begin{equation} \label{eq:PTGROmega:bound}
\left\|  \y{P}_{\f{T}}\y{G} \y{R}_{\Omega} \right\|  =  \left\|\y{R}_{\Omega}\y{G}^{*} \y{P}_{\f{T}}\right\| \leq \sqrt{\left\| \y{R}_{\Omega} \right\| \left(1 + \alpha \right)}.
\end{equation}
With this preparation, we can now apply the triangle inequality multiple times to estimate that
\begin{equation*}
\begin{split}
&\left\|\PTX \mathcal{G}\mathcal{R}_{\Omega}\mathcal{G}^{*}\PTX- \PTX \mathcal{G}\mathcal{G}^{*}\PTX \right\| 
\\& \leq \left\|\left( \PTX - \y{P}_{\f{T}} \right)\mathcal{G}\mathcal{R}_{\Omega}\mathcal{G}^{*}\PTX \right\|  +  \left\| \y{P}_{\f{T}} \mathcal{G}\mathcal{R}_{\Omega}\mathcal{G}^{*}\left( \PTX - \y{P}_{\f{T}} \right) \right\|   \\
 &+ \left\|  \y{P}_{\f{T}} \mathcal{G}\mathcal{R}_{\Omega}\mathcal{G}^{*}\y{P}_{\f{T}} -  \y{P}_{\f{T}} \mathcal{G}\mathcal{G}^{*}\y{P}_{\f{T}}  \right\| \\&+ \left\|  \y{P}_{\f{T}} \mathcal{G}\mathcal{G}^{*}\left(\y{P}_{\f{T}} -\PTX \right)  \right\|  + \left\|  \left(\PTX - \y{P}_{\f{T}}\right) \mathcal{G}\mathcal{G}^{*}\PTX  \right\|   \\
&\leq \left\|  \PTX - \y{P}_{\f{T}} \right\|  \left\|\mathcal{R}_{\Omega}\mathcal{G}^{*}\PTX \right\| +  \left\| \y{P}_{\f{T}} \mathcal{G}\mathcal{R}_{\Omega}\right\| \left\| \PTX - \y{P}_{\f{T}} \right\| + \alpha \\
&+ \left\|  \y{P}_{\f{T}}\right\|\left\|\mathcal{G}\mathcal{G}^{*}\right\| \left\|\y{P}_{\f{T}} -\PTX\right\| + \left\|\PTX - \y{P}_{\f{T}}\right\| \left\|\mathcal{G}\mathcal{G}^{*}\right\| \left\|\PTX  \right\| \\
&\leq  \left\|  \PTX - \y{P}_{\f{T}} \right\|  \left( 2  \left\|\mathcal{R}_{\Omega}\mathcal{G}^{*}\PTX \right\| \right) + \alpha+ 2 \left\|\PTX - \y{P}_{\f{T}}\right\| \\& \leq \frac{8 \left\| \f{H} - \f{H}_{\f{A}} \right\|}{\sigma_{r}(\f{H}_{\f{A}})}   \left(\sqrt{\left\| \y{R}_{\Omega} \right\| \left(1 + \alpha \right)} + 1 \right) + \alpha \\
&\leq \alpha + \alpha = 2 \alpha,
\end{split}
\end{equation*}
using the sub-multiplicativity of the spectral norm multiple times, \Cref{assumption:1}, in the second inequality, and \cref{eq:PTGROmega:bound} and \Cref{lemma:PT:perturbation} in the penultimate inequality. Finally, we conclude the proof by using the closeness assumption \cref{eq:HXtilde:closeness:lemma} in the last inequality.
\end{proof}

\subsection{Proof of \Cref{lemma:MatrixIRLS:localRIP}} \label{sec:proof:localRIP:NSP}

We present the proof of \Cref{lemma:MatrixIRLS:localRIP}, which is inspired by the proofs of \cite[Lemma 20]{ye_kim_jin_lee} and \cite[Lemma 1]{chen_chi14}, but refines the respective arguments.

\begin{proof}[{Proof of \Cref{lemma:MatrixIRLS:localRIP}}]
Let $\eta \in \ker \y{R}_{\Omega}$. Due to the entrywise nature of the normalized sampling operator $\y{R}_{\Omega}$, it holds that $\eta \in \ker \y{R}_{\Omega}$ if and only if $\y{D} \eta  \in \ker \y{R}_{\Omega}$ due to the diagonality of the diagonal operator $\y{D}: \MnOplus \to \MnOplus$ from \cref{eq:DEijt:def}. Therefore, it holds that $\y{R}_{\Omega} \y{G}^* \y{H}(\eta) = \y{R}_{\Omega} \y{G}^* \y{G} \y{D} (\eta) = \y{R}_{\Omega} \y{D} (\eta) = 0$, as $\y{G}^* \y{G} = \Id$ is the identity operator and as $\y{H} = \y{G} \y{D}$ due to \Cref{lemma:Hankel:action}, which implies further that $\y{G} \y{R}_{\Omega} \y{G}^* \y{H}(\eta) = 0$

Furthermore, this also implies that
\[
\left(\Id - \y{G}\y{G}^*\right) \y{H}(\eta) =  \left(\y{G} \y{D} -  \y{G}\y{G}^*\y{G} \y{D}\right) \eta = \left(\y{G} \y{D} -  \y{G} \y{D}\right) \eta = 0.
\]
Therefore, taking the scalar product with $\PTX\y{H}(\eta)$, we note that
\begin{equation} \label{eq:localRIP:proof:1}
\begin{split}
0 &= \langle \PTX\y{H}(\eta), \left( \y{G} \y{R}_{\Omega} \y{G}^* + \Id - \y{G}\y{G}^* \right) \y{H}(\eta) \rangle \\
&= \langle \PTX\y{H}(\eta), \left( \y{G} \y{R}_{\Omega} \y{G}^* + \Id - \y{G}\y{G}^* \right) \PTX \y{H}(\eta) \rangle \\&+ \langle \PTX\y{H}(\eta), \left( \y{G} \y{R}_{\Omega} \y{G}^* + \Id - \y{G}\y{G}^* \right)\y{P}_{T_{\f{H}}^\perp} \y{H}(\eta) \rangle,
\end{split}
\end{equation}
and furthermore, taking the scalar product with  $\y{P}_{T_{\f{H}}^\perp}\y{H}(\eta)$, we also observe that
\begin{align*}
0 &= \langle \y{P}_{T_{\f{H}}^\perp}\y{H}(\eta), \left( \y{G} \y{R}_{\Omega} \y{G}^* + \Id - \y{G}\y{G}^* \right) \y{H}(\eta) \rangle \\
&= \langle \y{P}_{T_{\f{H}}^\perp}\y{H}(\eta), \left( \y{G} \y{R}_{\Omega} \y{G}^* + \Id - \y{G}\y{G}^* \right) \PTX \y{H}(\eta) \rangle \\&+ \langle \y{P}_{T_{\f{H}}^\perp}\y{H}(\eta), \left( \y{G} \y{R}_{\Omega} \y{G}^* + \Id - \y{G}\y{G}^* \right)\y{P}_{T_{\f{H}}^\perp} \y{H}(\eta) \rangle,
\end{align*}
which is equivalent to 
\[
\begin{split}
&\langle \y{P}_{T_{\f{H}}^\perp}\y{H}(\eta), \left( \y{G} \y{R}_{\Omega} \y{G}^* + \Id - \y{G}\y{G}^* \right) \PTX \y{H}(\eta) \rangle\\& = \langle \PTX\y{H}(\eta), \left( \y{G} \y{R}_{\Omega} \y{G}^* + \Id - \y{G}\y{G}^* \right) \y{P}_{T_{\f{H}}^\perp} \y{H}(\eta) \rangle \\
&= - \langle \y{P}_{T_{\f{H}}^\perp}\y{H}(\eta), \left( \y{G} \y{R}_{\Omega} \y{G}^* + \Id - \y{G}\y{G}^* \right)\y{P}_{T_{\f{H}}^\perp} \y{H}(\eta) \rangle,
\end{split}
\]
where we used in the first equality the fact that $\y{G} \y{R}_{\Omega} \y{G}^* + \Id - \y{G}\y{G}^* $ is self-adjoint as a sum of self-adjoint operators. Inserting this into \cref{eq:localRIP:proof:1}, we obtain
\begin{equation} \label{eq:localRIP:proof:11}
\begin{split}
&\langle \PTX\y{H}(\eta), \left( \y{G} \y{R}_{\Omega} \y{G}^* + \Id - \y{G}\y{G}^* \right) \PTX \y{H}(\eta) \rangle \\&=  \langle \y{P}_{T_{\f{H}}^\perp}\y{H}(\eta), \left( \y{G} \y{R}_{\Omega} \y{G}^* + \Id - \y{G}\y{G}^* \right)\y{P}_{T_{\f{H}}^\perp} \y{H}(\eta) \rangle.
\end{split}
\end{equation}
We now bound the left and right hand side of the latter equality separately. On the one hand, we obtain a lower bound
\[
\begin{split}
& \left|\langle \PTX\y{H}(\eta), \left( \y{G} \y{R}_{\Omega} \y{G}^* + \Id - \y{G}\y{G}^* \right) \PTX \y{H}(\eta) \rangle \right| \\
&\geq \left|\langle \PTX\y{H}(\eta), \PTX \y{H}(\eta) \rangle\right| \\&- \left|\langle \PTX\y{H}(\eta), \left( \y{G} \y{R}_{\Omega} \y{G}^* - \y{G}\y{G}^* \right) \PTX \y{H}(\eta) \rangle\right|  \\
&= \left\| \PTX\y{H}(\eta)\right\|_F^2-\\&  \left|\langle \PTX\y{H}(\eta), \left( \PTX \y{G} \y{R}_{\Omega} \y{G}^*\PTX - \PTX\y{G}\y{G}^*\PTX \right) \PTX\y{H}(\eta) \rangle\right| \\
&\geq \left\| \PTX\y{H}(\eta)\right\|_F^2 -\\ & \left\|  \PTX \y{G} \y{R}_{\Omega} \y{G}^*\PTX - \PTX\y{G}\y{G}^*\PTX \right\| \left\| \PTX\y{H}(\eta)\right\|_F^2 \\
&\geq \left\| \PTX\y{H}(\eta)\right\|_F^2 -\frac{2}{5} \left\| \PTX\y{H}(\eta)\right\|_F^2,
\end{split}
\]
using the projection property $ \PTX^2 =  \PTX$ in the equality and \cref{eq:PTGROmGPT:25} in the last inequality, which implies that
\begin{align} \label{eq:localRIP:proof:2}
&\left\| \PTX\y{H}(\eta)\right\|_F^2 \nonumber\\&\leq \frac{5}{3} \left|\langle \PTX\y{H}(\eta), \left( \y{G} \y{R}_{\Omega} \y{G}^* + \Id - \y{G}\y{G}^* \right) \PTX \y{H}(\eta) \rangle \right| 
\end{align}

On the other hand, we have the upper bounds
\begin{equation*}
\begin{split}
& \left|\langle \y{P}_{T_{\f{H}}^\perp}\y{H}(\eta), \left( \y{G} \y{R}_{\Omega} \y{G}^*+ \Id - \y{G}\y{G}^* \right)\y{P}_{T_{\f{H}}^\perp} \y{H}(\eta) \rangle\right| \\
&\leq \left\| \y{P}_{T_{\f{H}}^\perp}\y{H}(\eta)\right\|_F \left\|  \y{G} \y{R}_{\Omega} \y{G}^*\right\| \left\|\y{P}_{T_{\f{H}}^\perp} \y{H}(\eta)\right\|_F \\&+ \left\| \y{P}_{T_{\f{H}}^\perp}\y{H}(\eta)\right\|_F \left\|\Id - \y{G}\y{G}^*\right\|   \left\|\y{P}_{T_{\f{H}}^\perp} \y{H}(\eta)\right\|_F \\
&\leq \left\|  \y{G}\right\| \left\|\y{R}_{\Omega}\right\|  \left\|\y{G}^*\right\| \left\|\y{P}_{T_{\f{H}}^\perp} \y{H}(\eta)\right\|_F^2   +  \left\|\y{P}_{T_{\f{H}}^\perp} \y{H}(\eta)\right\|_F^2 \\
&\leq \left( \left\|\y{R}_{\Omega}\right\|  + 1 \right) \left\|\y{P}_{T_{\f{H}}^\perp} \y{H}(\eta)\right\|_F^2,
\end{split}
\end{equation*}
using the sub-mulitiplicativity of the spectral norm and the fact that $\Id - \y{G}\y{G}^*$ is a projection in the third inequality, and observing that $\|\y{G}\| \leq 1$  and$\|\y{G}^*\| \leq 1$ in the last inequality. Combining this with \cref{eq:localRIP:proof:11} and \cref{eq:localRIP:proof:2}, this implies
\[
\left\| \PTX\y{H}(\eta)\right\|_F^2  \leq \frac{5}{3}  \left( \left\|\y{R}_{\Omega}\right\|  + 1 \right)  \left\|\y{P}_{T_{\f{H}}^\perp} \y{H}(\eta)\right\|_F^2
\]
and therefore
\[
\begin{split}
\|\y{H}(\eta)\|_F^2 &= \left\|\PTX\y{H}(\eta)\right\|_F^2 + \big\|\y{P}_{T_{\f{H}}^\perp}\y{H}(\eta)\big\|_F^2 \\&\leq  \frac{5}{3}  \left( \left\|\y{R}_{\Omega}\right\|  + 8/5 \right)  \left\|\y{P}_{T_{\f{H}}^\perp} \y{H}(\eta)\right\|_F^2,
\end{split}
\]
which finishes the proof.
\end{proof}
Next, we provide an auxiliary result of similar flavor as \Cref{lemma:MatrixIRLS:localRIP} to be used in the convergence analysis of \texttt{TOIRLS}.

\begin{lemma}\label{lemma:etaksigmarp1Xk}
Assume that \Cref{assumption:1} holds true for a normalized sampling operator $\y{R}_{\Omega}: \MnOplus  \to \MnOplus$ with respect to a rank-$r$ matrix $\f{H}_{\f{A}} \in \Rddn$ and constant $\alpha = 1/5$. If $\widetilde{\f{X}}\hk \in \MnOplus$ is such that the 
best rank-$r$ approximation of a matrix $\f{H}_k := \y{H}(\widetilde{\f{X}}\hk)$, i.e., 
\begin{equation} \label{eq:best:rankr:approximation}
\mathcal{T}_{r}(\f{H}_k) =  \argmin_{\f{Z} \in \Rddn : \rank(\f{Z}) \leq r} \|\f{Z} - \f{H}_k\|
\end{equation}
satisfies 
\[
\mathcal{T}_{r}(\f{H}_k)  \in \y{B}_{\HA}\left(\frac{ \sigma_r(\HA)}{ 32\sqrt{r}\kappa \left(\sqrt{ 6\|\y{R}_{\Omega}\|/5}+ 1 \right)} \right),
\]
then it holds that
\[
\|\y{H}(\eta\hk)\| < \sqrt{\frac{20}{3}  \left( \left\|\y{R}_{\Omega}\right\|  + 8/5 \right)}  \sqrt{dn-r} \sigma_{r+1}(\f{H}_k)
\]
where $\eta\hk = \widetilde{\f{X}}\hk - \y{Q}_T(\f{A})$.
\end{lemma}
\begin{proof}
If $\f{T}_k := T_{\mathcal{T}_{r}(\f{H}_k) }$  is tangent space onto the manifold of rank-$r$ matrices at $\mathcal{T}_{r}(\f{H}_k)$ and if $\f{U}_{\perp}\hk \in \R^{d_1 n \times (d_1 n - r)}$ and $\f{V}_{\perp}\hk \in \R^{d_2 n \times (d_2 n - r)}$ are the matrices with the last  $d_1 n - r$ and last $d_2 n - r$ left and right singular vectors of $\f{H}_k$, respectively, we can write the action of the projection $\y{P}_{\f{T}_k^\perp}: \Rddn \to \Rddn$ onto the orthogonal complement $\f{T}_k^\perp$ of $\f{T}_k$ as
\[
\y{P}_{\f{T}_k^\perp}(\f{Z}) =  \f{U}_{\perp}^{(k)}\f{U}_{\perp}^{(k)*} \f{Z} \f{V}_{\perp}^{(k)}\f{V}_{\perp}^{(k)*},
\]
cf., e.g., \cite{recht}. Let $d= \min(d_1, d_2)$. If $\Sigma_k^{\perp} \in \R^{(d_1 n - r) \times (d_2 n - r)}$ is diagonal with the last $d n - r$ singular values of $\f{H}_k$ ordered in an non-increasing way, we observe that
\[
\y{P}_{\f{T}_k^\perp}(\f{H}_k) =   \f{U}_{\perp}^{(k)} \Sigma_k^{\perp} \f{V}_{\perp}^{(k)*}.
\]
Now, if $\f{H}_{\f{A}} = \f{U}_0 \Sigma_0 \f{V}_0^T$ is a compact singular value decomposition of $\f{H}_{\f{A}}$, we estimate 
\[
\begin{split}
&\|\y{P}_{\f{T}_k^{\perp}}(\y{H}(\f{\eta}\hk))\|_{F} 
\leq  \|\y{P}_{\f{T}_k^{\perp}}(\f{H}_k)\|_{F} +   \|\y{P}_{\f{T}_k^{\perp}}(\f{H}_{\f{A}})\|_{F} \\
&\leq \sqrt{\sum_{i=r+1}^{d n} \sigma_{i}^2(\f{H}_k)} +  \left\| \f{U}_{\perp}^{(k)}\f{U}_{\perp}^{(k)*} \f{H}_{\f{A}} \f{V}_{\perp}^{(k)}\f{V}_{\perp}^{(k)*}\right\|_{F} \\
&\leq  \sqrt{\sum_{i=r+1}^{d n} \sigma_{i}^2(\f{H}_k)} + \|  \f{U}_{\perp}^{(k)}\| \left\|\f{U}_{\perp}^{(k)*} \f{H}_{\f{A}} \f{V}_{\perp}^{(k)}\right\|_{F} \|\f{V}_{\perp}^{(k)*}\|    \\ 
&\leq  \sqrt{dn-r} \sigma_{r+1}(\f{H}_k) + \|\f{U}_{\perp}^{(k)*} \f{U}_0\| \|\f{\Sigma}_0\|_F \|\f{V}_0^{*} \f{V}_{\perp}^{(k)}\| \\
\end{split}
\]
using the definition of $\eta\hk$, the triangle inequality, the fact that $\|\f{A} \f{B}\|_F \leq \|\f{A}\| \|\f{B}\|_F$ for all matrices $\f{A}$ and $\f{B}$, and that $ \|\f{V}_{\perp}^{(k)*}\|  = \|  \f{U}_{\perp}^{(k)}\| = 1$. 

By the classical perturbation bound due to Wedin \cite{Wedin72,Stewart06}, cf. also \cite[Lemma B.6]{KMV21},
\begin{align*}
\max\{\|\f{U}_{\perp}^{(k)*} \f{U}_0\|,\|\f{V}_0^{*} \f{V}_{\perp}^{(k)}\|\}  &\leq \sqrt{2} \frac{1}{\zeta} \| \y{H}(\eta\hk)\|\,,
\end{align*}
if $\f{H}_k \in \y{B}_{\HA}(\zeta)$ with $0 < \zeta < 1$. By assumption $\zeta < 1/2$, so
\[
\begin{split}
&\|\y{P}_{\f{T}_k^{\perp}}(\y{H}(\f{\eta}\hk))\|_{F} \\ &\leq  \sqrt{dn-r} \sigma_{r+1}(\f{H}_k) + 8 \| \y{H}(\eta\hk)\|^2 \sqrt{r} \sigma_1(\f{H}_{\f{A}}) \\
&\leq \sqrt{dn-r} \sigma_{r+1}(\f{H}_k) + 8 \| \y{H}(\eta\hk)\|^2  \sqrt{r} \kappa \sigma_r(\f{H}_{\f{A}}).
\end{split}
\]

Due to our assumptions, we can apply \Cref{eq:MatrixIRLS:tangentspace:localRIP:perturbation} for $\alpha = 1/5$ and further \Cref{lemma:MatrixIRLS:localRIP} for $\eta = \eta\hk$ to estimate that
\[
\begin{split}
&\|\y{H}(\eta\hk)\| \leq \|\y{H}(\eta\hk)\|_{F} \\&\leq \sqrt{\frac{5}{3}  \left( \left\|\y{R}_{\Omega}\right\|  + 8/5 \right)}  \|\y{P}_{\f{T}_k^{\perp}}(\y{H}(\eta\hk))\|_{F} \\
&\leq  \sqrt{\frac{5}{3}  \left( \left\|\y{R}_{\Omega}\right\|  + 8/5 \right)}\\& \left( \sqrt{dn-r} \cdot \sigma_{r+1}(\f{H}_k) + 8 \| \y{H}(\eta\hk)\|^2  \sqrt{r} \kappa \sigma_r(\f{H}_{\f{A}})\right)  \\
&\leq  \sqrt{\frac{5}{3}  \left( \left\|\y{R}_{\Omega}\right\|  + 8/5 \right)}  \sqrt{dn-r} \cdot \sigma_{r+1}(\f{H}_k)\\&\quad\quad\quad\quad+    \frac{8  \sqrt{r} \kappa \sqrt{\frac{5}{3}  \left( \left\|\y{R}_{\Omega}\right\|  + 8/5 \right)} }{\left(32\sqrt{r}\kappa\right) \left(\sqrt{ 6\|\y{R}_{\Omega}\|/5}+ 1 \right)}  \|\y{H}(\eta\hk)\|\\
&<   \sqrt{\frac{5}{3}  \left( \left\|\y{R}_{\Omega}\right\|  + 8/5 \right)}  \sqrt{dn-r} \cdot \sigma_{r+1}(\f{H}_k) + \frac{1}{2} \|\y{H}(\eta\hk)\|.
\end{split}
\]
Rearranging this estimate, we obtain
\[
\|\y{H}(\eta\hk)\|  < \sqrt{\frac{20}{3}  \left( \left\|\y{R}_{\Omega}\right\|  + 8/5 \right)}  \sqrt{dn-r} \sigma_{r+1}(\f{H}_k).
\]
\end{proof}

\subsection{Proof of \Cref{thm:convergence:template}} \label{sec:weightoperator:structure}
In this section, we provide the proof of \Cref{thm:convergence:template}. For this purpose, we use key results of \cite{KMV21}, adapted to our notation.
\begin{proposition}[{\cite[Lemma B.8 and Lemma B.9]{KMV21}}] \label{prop:weightoperator}
Let $\f{H}_{\f{A}} = \y{H}(\y{Q}_T(\f{A})) \in \Ran \y{H}$ be a matrix of rank $r$, let $\widetilde{\f{X}}\hk$ be the $k$-th iterate of \Cref{algo:IRLS:graphcompletion} for input parameters $\Omega$, $\f{y}=P_{\Omega}(\y{Q}_T(\f{A}))$, $\lambda = 0$ and $\widetilde{r} =r $. For $\f{H}_{k}:=\y{H}(\widetilde{\f{X}}\hk)$, assume that $\varepsilon_k=\sigma_{r+1}(\f{H}_{k})$ and that
\begin{equation*}
\| \y{H}(\eta) \|_F \leq C \|\y{P}_{\f{T}_k^{\perp}} \y{H}(\eta) \|_F \quad \quad \text{ for all } \eta \in \ker \y{R}_{\Omega}
\end{equation*}
for some constant $C$, where $\f{T}_k = \f{T}_{\mathcal{T}_{r}(\f{H}_{k})}$  is the tangent space onto the manifold of rank-$r$ matrices at $\mathcal{T}_{r}(\f{H}_{k})$. Then
\begin{equation*}
	\|\y{H}(\widetilde{\f{X}}\hkk) - \f{H}_{\f{A}}\| \leq C^2 \varepsilon_k^{2} \|W_{\f{H}_{k}}(\f{H}_{\f{A}})\|_{S_1},
\end{equation*}
where $W_{\f{H}_{k}} : \Rddn \to \Rddn$ is the optimal weight operator of $\f{H}_k$ as in \cref{eq:def:W}, and $\| \cdot\|_{S_1}$ describes the Schatten-$1$ norm.

Furthermore, if additionally $\y{H}(\widetilde{\f{X}}\hk) \in \y{B}_{\HA} (\zeta)$ for some $0 < \zeta < 1$, then
\begin{align*}
&\|\y{H}(\widetilde{\f{X}}\hkk) - \f{H}_{\f{A}}\|\leq C^2  r (1-\zeta)^{-2} \sigma_{r}(\f{H}_{\f{A}})^{-1} \\&\left(\varepsilon_k^{2} + 4 \varepsilon_k \|\y{H}(\widetilde{\f{X}}\hk) - \f{H}_{\f{A}}\| \kappa + 2 \|\y{H}(\widetilde{\f{X}}\hkk) - \f{H}_{\f{A}}\|^2 \kappa  \right),
\end{align*}
where $\kappa =  \sigma_{1}(\f{H}_{\f{A}})  /  \sigma_{r}(\f{H}_{\f{A}})$ is the condition number of $\f{H}_{\f{A}}$. 
\end{proposition}

We can now put \Cref{lemma:MatrixIRLS:localRIP}, \Cref{eq:MatrixIRLS:tangentspace:localRIP:perturbation} and \Cref{prop:weightoperator} together to prove \Cref{thm:convergence:template}, showing that we attain locally quadratic convergence under the stated assumptions.
\begin{proof}[Proof of \Cref{thm:convergence:template}]
	
Let $k \in \N$ and  $\widetilde{\f{X}}\hk$ be the $k$-th iterate of \Cref{algo:IRLS:graphcompletion} with inputs $\Omega$, $\f{y} = P_{\Omega}(\y{Q}_T(\f{A}))$, $\lambda = 0$ and $\widetilde{r} = r$.  First, we observe that 
\[
\begin{split}
\zeta_3 &=  \min\left( \zeta_1, \zeta_2 , \zeta_3,\frac{1}{2} \right) := \min\bigg( \frac{ \left(32\sqrt{r}\kappa\right)^{-1}}{\sqrt{ 6\|\y{R}_{\Omega}\|/5}+ 1}, \\
&\frac{3}{20r}\frac{(1 + 6 \kappa)^{-1}\sigma_r(\f{H}_{\f{A}})}{\|\y{R}_{\Omega}\| + 8/5}, \\
&\frac{3^{\frac{3}{2}} \left(1+ 6 \kappa \right)^{-1} \sigma_r(\f{H}_{\f{A}})} {20^{\frac{3}{2}}r \left(\left\|\y{R}_{\Omega}\right\|  + 8/5 \right)^{\frac{3}{2}} (dn-r)^{\frac{1}{2}} }  ,\frac{1}{2}  \bigg).
\end{split}
\]

We note that \Cref{assumption:1} is satisfied with respect to $\f{H}_{\f{A}}$ and constant $\alpha =  1/5$. Since  $\y{H}(\widetilde{\f{X}}\hk) \in \y{B}_{\HA}\left( \zeta_3 \right)$ (as $\zeta_3 < \zeta_2$),  we also have that a best rank-$r$ approximation $\mathcal{T}_{r}(\y{H}(\widetilde{\f{X}}\hk))$ of $\y{H}(\widetilde{\f{X}}\hk)$ satisfies 
\[
\begin{split}
&\big\| \mathcal{T}_{r}(\y{H}(\widetilde{\f{X}}\hk)) - \f{H}_{\f{A}}  \big\|\\ &\leq \big\| \y{H}(\widetilde{\f{X}}\hk) - \f{H}_{\f{A}}  \big\|  + 
 \big\| \y{H}(\widetilde{\f{X}}\hk) - \mathcal{T}_{r}(\y{H}(\widetilde{\f{X}}\hk)) \big\| \\
 &\leq 2 \big\| \y{H}(\widetilde{\f{X}}\hk) - \f{H}_{\f{A}}  \big\| \leq 2 \zeta_3 \leq \frac{1}{40} \left(\sqrt{ 6\|\y{R}_{\Omega}\|/5}+ 1 \right)^{-1},
 \end{split}
\]
from which it follows due to \Cref{eq:MatrixIRLS:tangentspace:localRIP:perturbation} that \cref{eq:PTGROmGPT:25} holds true for $\f{H} :=  \mathcal{T}_{r}(\y{H}(\widetilde{\f{X}}\hk))$. Thus, \Cref{lemma:MatrixIRLS:localRIP} implies that
\[
\|\y{H}(\eta\hk)\|_{F}^2 \leq \frac{5}{3}  \left( \left\|\y{R}_{\Omega}\right\|  + 8/5 \right)  \left\|\y{P}_{\f{T}_k^\perp} \y{H}(\eta\hk)\right\|_F^2,
\]
where $\eta\hk := \widetilde{\f{X}}\hk - \y{Q}_T(\f{A})$ and $\f{T}_k = \f{T}_{\f{H}}=  \f{T}_{\mathcal{T}_{r}(\y{H}(\widetilde{\f{X}}\hk))}$ is tangent space onto the manifold of rank-$r$ matrices at $\f{H}$.

Next, since $\varepsilon_k = \sigma_{r+1}(\y{H}(\widetilde{\f{X}}\hk))$, it follows from the Eckardt-Young-Mirsky theorem \cite{Mirsky60} that
\begin{align*}
\varepsilon_k &= \sigma_{r+1}(\y{H}(\widetilde{\f{X}}\hk)) \\&=\|\y{H}(\widetilde{\f{X}}\hk)- \f{H}\| \leq \|\y{H}(\widetilde{\f{X}}\hk) - \f{H}_{\f{A}}\| = \|\y{H}(\eta\hk)\|,
\end{align*}
where we used that $\f{H}_{\f{A}}$ is of rank $r$ in the inequality. Since also $\y{H}(\widetilde{\f{X}}\hk) \in \y{B}_{\HA}\left(1/2 \right)$, it follows therefore from \Cref{prop:weightoperator} that
\begin{equation} \label{eq:Hetakp1:normest}
\begin{split}
&\|\y{H}(\eta\hkk) \| \\&\leq \frac{20}{3}  \left( \left\|\y{R}_{\Omega}\right\|  + 8/5 \right)   r \sigma_{r}(\f{H}_{\f{A}})^{-1} \\ 
&\quad\quad\quad \left(\varepsilon_k^{2} + 4 \varepsilon_k \|\y{H}(\eta\hk)\| \kappa + 2 \|\y{H}(\eta\hk)\|^2 \kappa  \right) \\
&\leq \frac{20}{3}  \left( \left\|\y{R}_{\Omega}\right\|  + 8/5 \right)   r \sigma_{r}(\f{H}_{\f{A}})^{-1}  \left(1+ 6 \kappa \right) \|\y{H}(\eta\hk)\|^2,
\end{split}
\end{equation}
where $\kappa = \sigma_1(\f{H}_{\f{A}}) /  \sigma_r(\f{H}_{\f{A}})$ is the condition number of $\f{H}_{\f{A}}$ and $\eta\hkk = \widetilde{\f{X}}\hkk - \y{Q}_T(\f{A})$.

If, additionally, we assume that $\y{H}(\widetilde{\f{X}}\hk) \in \y{B}_{\HA}\left( \zeta_2 \right)$, then we can further bound the right hand side of \cref{eq:Hetakp1:normest} to obtain
\begin{equation} \label{eq:Heta:normdecay}
\|\y{H}(\eta\hkk)\|  < \|\y{H}(\eta\hk)\|\,,
\end{equation}
and also obtain a quadratic decay in the spectral error
\[
\|\y{H}(\eta\hkk)\|  \leq \nu \|\y{H}(\eta\hk)\|^2\,,
\]
with $\nu = \frac{20}{3 \sigma_r(\f{H}_{\f{A}})} (1 + 6 \kappa) \left(\|\y{R}_{\Omega}\| + 8/5 \right) r$. 

What remains to be shown is that the $(r+1)$-st singular value $\y{H}\left(\widetilde{\f{X}}\hk\right)$ is strictly decreasing from one iterate to the next.
If $\y{H}(\widetilde{\f{X}}\hk) \in \y{B}_{\HA}\left( \zeta_1\right)$, 
it follows that
\begin{equation} \label{eq:sigma:decreasing}
\begin{split}
&\sigma_{r+1}(\y{H}(\widetilde{\f{X}}\hkk)) \leq \|\y{H}(\eta\hkk)\| \\&\leq  \frac{20}{3}  \left( \left\|\y{R}_{\Omega}\right\|  + 8/5 \right)   r \sigma_{r}(\f{H}_{\f{A}})^{-1}  \left(1+ 6 \kappa \right) \|\y{H}(\eta\hk)\|^2 \\
&< \left(\frac{20}{3}  \left( \left\|\y{R}_{\Omega}\right\|  + 8/5 \right)\right)^{3/2}   r \\& \quad\quad \cdot \left(1+ 6 \kappa \right) \sqrt{dn-r}   \sigma_{r+1}(\y{H}(\widetilde{\f{X}}\hk)) \frac{\|\y{H}(\eta\hk)\|}{\sigma_{r}(\f{H}_{\f{A}})} \\
&\leq  \sigma_{r+1}(\y{H}(\widetilde{\f{X}}\hk)),
\end{split}
\end{equation}
using \cref{eq:Hetakp1:normest} in the second inequality, \Cref{lemma:etaksigmarp1Xk} in the third inequality, and $\y{H}(\widetilde{\f{X}}\hk) \in \y{B}_{\HA}\left( \zeta_3\right)$ in the last inequality.

Finally, we recall the update rule \eqref{eq:MatrixIRLS:epsdef}, which gives that \cref{eq:sigma:decreasing} implies that $\varepsilon_{k+1} = \sigma_{r+1}(\y{H}(\widetilde{\f{X}}\hkk))$, so that \cref{eq:Heta:normdecay} ensures that the assumptions of \Cref{thm:convergence:template} are not only fulfilled for iteration $k$, but also for iteration $k+1$. By induction, this implies that $\y{H}(\widetilde{\f{X}}^{(k+\ell)}) \xrightarrow{\ell \to \infty} \f{H}_{\f{A}}$, concluding the proof of \Cref{thm:convergence:template}.
\end{proof}

\subsection{Computational Details} \label{sec:computational:details}
In this section, we detail some aspects of anx efficient implementations of \texttt{TOIRLS}, cf. \Cref{algo:IRLS:graphcompletion}.

\subsubsection{Explicit Expression for Weighted Least Squares Solution}
First, we justify the explicit formula of \Cref{sec:computational:considerations} provided for the $k+1$-st iterate $\widetilde{\f{X}}\hkk$ of  \texttt{TOIRLS}.
\begin{lemma} \label{lemma:explicit:weighted:leastsquares}
For any $\lambda \geq 0$, it holds that the solution $\widetilde{\f{X}}\hkk$ of the weighted least squares problem \cref{eq:MatrixIRLS:Xdef} satisfies 
\begin{equation} \label{eq:Xtildek:implementation}
\widetilde{\f{X}}\hkk  = \widetilde{W}_{\f{H}_{k}}^{-1}  P_{\Omega}^* \left( \lambda\Id +  P_{\Omega} \widetilde{W}_{\f{H}_{k}}^{-1} P_{\Omega}^* \right)^{-1}(\f{y}),
\end{equation}
where $\widetilde{W}_{\f{H}_{k}}^{-1}: \Rddn \to \Rddn$ is the inverse of the effective weight operator $\widetilde{W}_{\f{H}_{k}}: \Rddn \to \Rddn$ of \Cref{def:weightoperator:MatrixIRLS}.
\end{lemma}
\begin{proof}
Using the substitution $\widetilde{\f{X}'} = \widetilde{W}_{\f{H}_{k}}^{1/2}(\widetilde{\f{X}})$ in \cref{eq:MatrixIRLS:Xdef}, we obtain that 
\[
\widetilde{\f{X}}\hkk  = \widetilde{W}_{\f{H}_{k}}^{-1/2}\bigg( \widetilde{\f{X}'}\hk\bigg) 
\]
where 
\[
\begin{split}
\widetilde{\f{X}'}\hkk &= \argmin\limits_{\widetilde{\f{X}'} \in \MnOplus} \left\{\langle \widetilde{\f{X}}', \widetilde{\f{X}}' \rangle + \frac{1}{\lambda} \left\| P_{\Omega}(\widetilde{W}_{\f{H}_{k}}^{-1/2}(\f{\widetilde{X}}')) - \f{y}\right\|_2^2 ) \right\} \\
&= \argmin\limits_{\widetilde{\f{X}'} \in \MnOplus} \left\{ \lambda \|\widetilde{\f{X}}'\|_F^2 + \left\| \y{F}(\f{\widetilde{X}}') - \f{y}\right\|_2^2 \right\}   \\
&= \y{F}^* \left( \lambda \Id +  \y{F} \y{F}^*\right)^{-1} \f{y} \\&= \widetilde{W}_{\f{H}_{k}}^{-1/2}  P_{\Omega}^* \left( \lambda\Id +  P_{\Omega} \widetilde{W}_{\f{H}_{k}}^{-1} P_{\Omega}^* \right)^{-1}(\f{y}),
\end{split}
\]
defining $\y{F}:= P_{\Omega}\circ \widetilde{W}_{\f{H}_{k}}^{-1/2}: \Rddn \to \R^{m}$ to interpret the problem as a ridge regression/$\ell_2$-penalized least squares problem in the second equality, and using the inner product implementation of ridge regression in the third equality (see, e.g., \cite[Theorem 2.4]{Fan-StatisticalFoundations2020}). This shows \cref{eq:Xtildek:implementation}.

For $\lambda = 0$, we note that
\[
\widetilde{\f{X}}\hkk = \argmin\limits_{\widetilde{\f{X}} \in \MnOplus: P_{\Omega}( \widetilde{\f{X}}) = \f{y}} \langle \widetilde{\f{X}}, \widetilde{W}_{\f{H}_{k}}(\widetilde{\f{X}}) \rangle = \widetilde{W}_{\f{H}_{k}}^{-1/2}\bigg( \widetilde{\f{X}'}\hk\bigg) 
\]
with 
\begin{align*}
\widetilde{\f{X}'}\hkk &= \argmin\limits_{\widetilde{\f{X}'} \in \MnOplus:  \y{F}(\widetilde{\f{X}'}) = \f{y}} \| \widetilde{\f{X}'}\|_F^2= \y{F}^* \left( \y{F} \y{F}^*\right)^{-1} \f{y}\\& = \widetilde{W}_{\f{H}_{k}}^{-1/2}  P_{\Omega}^* \left( P_{\Omega} \widetilde{W}_{\f{H}_{k}}^{-1} P_{\Omega}^* \right)^{-1}(\f{y})\,,
\end{align*}
using an analogous substitution as above and the fact that $\y{F}^{\dagger} = \y{F}^* \left( \y{F} \y{F}^*\right)^{-1}$ is the Moore-Penrose inverse of $\y{F}$ as defined above.
\end{proof}

\subsubsection{Efficient Implementation of \texttt{TOIRLS}} 
In this section we outline the main computational steps of an efficient implementation of \texttt{TOIRLS}. In particular, we provide an algorithm, \Cref{algo:TOIRLS:implementation}, for computing the weighted least squares solution \cref{eq:Xtildek:implementation} essentially via a conjugate gradient method applied to a $(r_k (n d_1 + n d_2 + r_k) \times r_k (n d_1 + n d_2 + r_k)) = O(r n  T) \times O(r n T)$ linear system.

In the following, we let $S_{k}:=\R^{r_{k}\left(n d_{1}+n d_{2}+r_{k}\right)}$, and we recall from the proof of \Cref{lemma:etaksigmarp1Xk} that if $\widetilde{\mathbf{X}}^{(k)} \in \Rddn$ is the $k$-th iterate of \texttt{TOIRLS} and $\mathbf{H}_{k}:=\y{H}(\widetilde{\f{X}}\hk)$, $\mathbf{T}_{k}:=\f{T}_{\y{T}_{r_{k}}\left(\y{H}\left(\widetilde{\mathbf{X}}^{(k)}\right)\right)}$ denotes tangent space onto the manifold of rank-$r_{k}$ matrices at $\y{T}_{r_{k}}\left(\y{H}(\widetilde{\f{X}}\hk)\right)$ (here, with $r_{k}$ instead of $r$), where $\y{T}_{r_{k}}\left(\mathbf{H}_{k}\right)$ is the best rank-r $r_{k}$ approximation of $\mathbf{H}_{k}$, cf. \cref{eq:best:rankr:approximation}. Given the subspace $\mathbf{T}_{k} \subset \Rddn$, we let $P_{\f{T}_{k}}: S_{k} \rightarrow \mathbf{T}_{k}$ be the parametrization operator defined, for $\gamma \in S_{k}$, as
$$
P_{\f{T}_{k}}(\gamma):=\mathbf{U}^{(k)} \Gamma_{1} \mathbf{V}^{(k) *}+\mathbf{U}^{(k)} \Gamma_{2}+\Gamma_{3} \mathbf{V}^{(k) *},
$$
where $\Gamma_{1} \in \mathbb{R}^{r_{k} \times r_{k}}, \Gamma_{2} \in \mathbb{R}^{r_{k} \times n d_{2}}$ and $\Gamma_{d} \in \mathbb{R}^{n d_{1} \times r_{k}}$ are matricizations of the first $r_{k}^{2}$, central $r_{k} n d_{2}$ and final $r_{k} n d_{2}$ coordinates of $\gamma \in S_{k}$, respectively. We note that the projection operator $\y{P}_{\f{T}_{\mathbf{k}}}: \Rddn \rightarrow \Rddn$ can be implemented via $\y{P}_{\mathbf{T}_{k}}=P_{\f{T}_{k}} P_{\f{T}_{k}}^{*} \cdot$ Recall that $\y{G}: \MnOplus \to \Rddn$ is the normalized block Hankel operator \cref{def:G:operator} and $\y{D}: \MnOplus \to \MnOplus$ the
diagonal operator of \cref{eq:DEijt:def}. Finally, $\mathbf{I}_{m}$ and $\mathbf{I}_{S_{k}}$ are identity matrices on $\mathbb{R}^{m}$ and $S_{k}$, respectively, and $\mathbf{D}_{S_{k}}: S_{k} \rightarrow S_{k}$ is a diagonal matrix that contains coordinates $\left(\sum_{\varepsilon_{k}, d_{1} n}\right)_{i i}^{-1}\left(\sum_{\varepsilon_{k}, d_{2} n}\right)_{j j}=$ $\max \left(\sigma_{i}^{(k)}, \varepsilon_{k}\right)^{-1} \max \left(\sigma_{j}^{(k)}, \varepsilon_{k}\right)^{-1}$ for $(i, j) \in\left\{(i, j) \in\left[n d_{1}\right] \times\left[n d_{2}\right]: i \leq r_{k}\right.$ or $\left.j \leq r_{k}\right\}$ on its diagonal, which is related to the weight operator $W_{\mathbf{H}_{k}}: \Rddn \rightarrow \Rddn$ of \cref{eq:def:W} by
\begin{equation} \label{eq:WHk:Dsk:def}
W_{\mathbf{H}_{k}}(\mathbf{Z})=\left(P_{T_{k}} \mathbf{D}_{S_{k}} P_{T_{k}}^{*}+\epsilon_{k}^{-2}\left(\mathrm{Id}-P_{T_{k}} P_{T_{k}}^{*}\right)\right) \mathbf{Z},
\end{equation}
cf. \cite[Appendix A, Eq. (12)]{KMV21}. With these notational conventions, we can formulate \Cref{algo:TOIRLS:implementation}.

\begin{algorithm}[h]
\caption{Implementation of $k+1$-st weighted least squares step of \texttt{TOIRLS}} \label{algo:TOIRLS:implementation}
\begin{algorithmic}
\STATE{\bfseries Input:}  Set $\Omega$, observations $\f{y} \in \R^m$, $\lambda \geq 0$, matrices $\f{U}\hk \in \R^{d_1 \times r_k}$, $\f{V}\hk \in \R^{d_2 \times r_k}$ of singular vectors and $r_k$ leading singular values $\sigma_1\hk,\ldots,\sigma_{r_k}\hk$  of $\f{H}_{k}$ , smoothing parameter $\varepsilon_k$, initialization $\gamma_{k+1}^{(0)} = P_{T_k}^*P_{T_{k-1}}(\gamma_{k}) \in \R^{r_k (n d_1 + n d_2 + r_k)}$ where $\gamma_{k} \in \R^{r_{k-1} (n d_1 + n d_2 + r_{k-1})}$ is respective parameter \cref{eq:gamma:system} of the $k$-th iteration.
\end{algorithmic}
Let $\f{K} := \lambda \varepsilon_k^{-2}\f{I} + P_{\Omega} \y{D}^{-2} P_{\Omega}^*$ and 
\begin{equation}
\begin{aligned}
\f{M} :=  P_{\f{T}_k}^* \y{G}\y{D}^{-1}  &P_{\Omega}^* \f{K}^{-1} P_{\Omega} \y{D}^{-1} \y{G}^* P_{\f{T}_k} \\&+ \frac{\f{D}_{S_k}^{-1}}{\f{D}_{S_k}^{-1}- \varepsilon_k^{2} \f{I}_{S_k}} -  P_{\f{T}_k}^*\y{G}\y{G}^* P_{\f{T}_k}.
\end{aligned}
\label{eq:Moperator:def}
\end{equation}
\begin{algorithmic}[1]
\STATE  Compute $\f{h}_k^0:= P_{\f{T}_k}^* \y{G}\y{D}^{-1}  P_{\Omega}^* \f{K}^{-1} \f{y} -\f{M} \gamma_{k+1}^{(0)}  \in \R^{r_k (n d_1 + n d_2 + r_k)}$.
\STATE Solve linear system
\begin{equation} \label{eq:gamma:system}
	\f{M} \Delta\f{\gamma}_{k+1} =  \f{h}_k^0	
\end{equation}
	 for $\Delta\gamma_{k+1} \in S_k$ by the \emph{conjugate gradient} method \cite{HestenesStiefel52,Meurant}.
\STATE Compute $\gamma_{k+1} = \gamma_{k+1}^{(0)} + \Delta\gamma_{k+1}$.
\STATE Compute residual $\f{p}_{k+1}:= \f{K}^{-1}(\f{y}- P_{\Omega}\y{D}^{-1} \y{G}^* P_{\f{T}_k} (\f{\gamma}_{k+1})) \in \R^m$ where 
\begin{equation} \label{def:Kmat}
\mathbf{K}:=\lambda \varepsilon_{k}^{-2} \mathbf{I}+P_{\Omega} \y{D}^{-2} P_{\Omega}^{*}.
\end{equation}
\end{algorithmic}
\begin{algorithmic}
\STATE{\bfseries Output:} $\f{p}_{k+1} \in \R^m$ and $\gamma_{k+1} \in  \R^{r_{k} (n d_1 + n d_2 +r_{k})}$.
\end{algorithmic}
\end{algorithm}

\Cref{lemma:Algo:implementation} shows that \Cref{algo:TOIRLS:implementation} indeed computes the weighted least squares solution $\widetilde{\mathbf{X}}^{(k)}$.

\begin{lemma} \label{lemma:Algo:implementation}
Let $\widetilde{\mathbf{X}}^{(k)} \in \Rddn$ be the $k$-th iterate of \texttt{TOIRLS} (\Cref{algo:IRLS:graphcompletion}) for an observation vector $\mathbf{y} \in \mathbb{R}^{m}$ with $m=|\Omega|, \lambda \geq 0$ and smoothing parameter $\varepsilon_{k}>0$, let $\mathbf{H}_{k}=\y{H}\left(\widetilde{\mathbf{X}}^{(k)}\right)$ and $r_{k}=\left|\left\{i \in[d n]: \sigma_{i}\left(\mathbf{H}_{k}\right) > \varepsilon_{k}\right\}\right|$. Then the following statements hold.
\begin{enumerate}
\item The $(k+1)$-st iterate $\widetilde{\mathbf{X}}^{(k+1)}$ of \Cref{algo:IRLS:graphcompletion} satisfies
$$
\widetilde{\mathbf{X}}^{(k+1)}=\y{D}^{-1} \widetilde{P}_{\Omega}^{*}\left(\mathbf{p}_{k+1}\right)+\y{D}^{-1} \y{G}^{*} P_{\f{T}_{k}}\left(\gamma_{k+1}\right)
$$
where $\mathbf{p}_{k+1} \in \mathbb{R}^{m}$ and $\gamma_{k+1} \in \mathbb{R}^{r_{k}\left(n d_{1}+n d_{2}+r_{k}\right)}$ is the output of \Cref{algo:TOIRLS:implementation} if the linear system of \cref{eq:gamma:system} is solved exactly.
\item The vector $\gamma_{k+1} \in \mathbb{R}^{r_{k}\left(n d_{1}+n d_{2}+r_{k}\right)}$ corresponding to an iterative solution of \cref{eq:gamma:system} using $N_{\text{CG\_inner}}$ iterations of a conjugate gradient method\footnote{or of related iterative solvers based on matrix-vector multiplication} can be computed in $O(N_{\text{CG\_inner}}  r_k T( m + n \log T + n r_k T))$ time. Thus, a parametrization of an approximation of the $(k+1)$-st iterate $\widetilde{\mathbf{X}}^{(k+1)}$ of \Cref{algo:IRLS:graphcompletion} 
can be computed in $O(N_{\text{CG\_inner}}  r_k T( m + n \log T + n r_k T))$ time. 
\end{enumerate}
\end{lemma}

The statement of \Cref{lemma:Algo:implementation}.2 enables \Cref{algo:TOIRLS:implementation} to compute an \emph{accurate} approximation of $\widetilde{\mathbf{X}}^{(k+1)}$ in $O(r_k T( m + n \log T + n r_k T))$ time in many situations, in particular, if $\widetilde{\mathbf{X}}^{(k)}$ is close to an image $\mathcal{Q}_T(\f{A})$ of a transition operator $\f{A}$ satisfying $P_{\Omega}( \mathcal{Q}_T(\f{A})) =\f{y}$ and if the normalized sampling operator $\y{R}_{\Omega}$ associated to the sampling set $\Omega$ satisfies a local restricted isometry property \Cref{eq:localRIP:condition}, as in this case, it can be shown that the condition number of the matrix $\f{M}$ of linear system \cref{eq:gamma:system} is a \emph{small constant}. We do not provide the full proof for that statement as it amounts to a variant of \cite[Theorem 4.2]{KMV21} and its proof.

\begin{proof}[{Proof of \Cref{lemma:Algo:implementation}.1}] We recall from \Cref{lemma:explicit:weighted:leastsquares} that
\begin{equation} \label{eq:Xkp1:WGk}
\begin{split}
&\widetilde{\mathbf{X}}^{(k+1)} \\&=\widetilde{W}_{\mathbf{H}_{k}}^{-1} P_{\Omega}^{*}\left(\lambda \operatorname{Id}+P_{\Omega} \widetilde{W}_{\mathbf{H}_{k}}^{-1} P_{\Omega}^{*}\right)^{-1}(\mathbf{y}) \\
&=\y{D}^{-1}\!\left(\widetilde{W}_{\y{G}}^{(k)}\right)^{-1}\! \widetilde{P}_{\Omega}^{*}\left(\!\lambda \operatorname{Id}+\widetilde{P}_{\Omega}\left(\widetilde{W}_{\y{G}}^{(k)}\right)^{-1} \! \widetilde{P}_{\Omega}^{*}\right)^{-1}\!\!(\mathbf{y})
\end{split}
\end{equation}
using the notation $\widetilde{W}_{\y{G}}^{(k)}:=\y{G}^{*} W_{\mathrm{H}_{k}} \y{G}$ and $\widetilde{P}_{\Omega}:=P_{\Omega} \y{D}^{-1} .$ The last inequality holds due to
\begin{align*}
\widetilde{W}_{\mathbf{H}_{k}}^{-1}&=\left(\y{H}^{*} W_{\mathbf{H}_{k}} \y{H}\right)^{-1}=\left(\y{D} \y{G}^{*} W_{\mathbf{H}_{k}}  \y{G} \y{D}\right)^{-1}\\&=\y{D}^{-1}\left(\widetilde{W}_{\y{G}}^{(k)}\right)^{-1} \y{D}^{-1}
\end{align*}
Using \cref{eq:WHk:Dsk:def} and $\y{G}^{*} \y{G}=\mathrm{Id}$, we can write
$$
\begin{aligned}
\left(\widetilde{W}_{\y{G}}^{(k)}\right)^{-1} &=\left(\y{G}^{*} W_{\mathbf{H}_{k}} \y{G}\right)^{-1}\\&=\left(\y{G}^{*}\left(P_{\f{T}_{k}}\left(\mathbf{D}_{S_{k}}-\varepsilon_{k}^{-2} \mathbf{I}_{S_{k}}\right) P_{\f{T}_{k}}^{*}+\varepsilon_{k}^{-2} \mathrm{Id}\right) \y{G}\right)^{-1} \\
&=\left(\y{G}^{*} P_{\f{T}_{k}}\left(\mathbf{D}_{S_{k}}-\varepsilon_{k}^{-2} \mathbf{I}_{S_{k}}\right) P_{\f{T}_{k}}^{*} \y{\y {G}}+\varepsilon_{k}^{-2} \mathrm{Id}\right)^{-1} .
\end{aligned}
$$
For the next step, we recall the Sherman-Morrison-Woodbury formula \cite{Woodbury50,Fornasier11}, \cite[(0.7.4.1)]{HornJohnson12}, which states that for any invertible matrices $\mathbf{B}, \mathbf{C}$ and matrices $\mathbf{E}, \mathbf{F}$ of appropriate dimensions,
\begin{equation} \label{eq:SMW:formula}
\left(\mathbf{B}+\f{E}\f{C}\f{F}^{*}\right)^{-1}\!=\f{B}^{-1}-\f{B}^{-1} \f{E}\left(\f{C}^{-1}\!\!+\f{F}^{*} \f{B}^{-1} \f{E}\right)^{-1} \f{F}^{*} \f{B}^{-1}.
\end{equation}
Applying \cref{eq:SMW:formula} for $\mathbf{B}=\varepsilon_{k}^{-2} \mathrm{Id}, \mathbf{E}=\y{G}^{*} P_{\mathbf{T}_{k}}, \mathbf{F}^{*}=P_{\mathbf{T}_{k}}^{*} \y{G}$ and $\mathbf{C}=\mathbf{D}_{S_{k}}-\varepsilon_{k}^{-2} \mathbf{I}_{S_{k}}$ yields then
\begin{equation} \label{eq:Wgkmin1}
\begin{split}
&\left(\widetilde{W}_{\y{G}}^{(k)}\right)^{-1}=\varepsilon_{k}^{2} \Big[\mathrm{Id} \\ &-\y{G}^{*} P_{\f{T}_{k}}\left(\varepsilon_{k}^{-2}\left(\mathbf{D}_{S_{k}}-\varepsilon_{k}^{-2} \mathbf{I}_{S_{k}}\right)^{-1}\!\!+ \!P_{\mathbf{T}_{k}}^{*} \y{G}\y{G}^* P_{\f{T}_{k}}\right)^{-1} \!\! P_{\f{T}_{k}}^{*} \y{G}\Big].
\end{split}
\end{equation}
Thus,
\begin{equation*}
\begin{aligned}
&\left(\lambda \operatorname{Id}+\widetilde{P}_{\Omega} \!\left(\widetilde{W}_{\y{G}}^{(k)}\right)^{-1} \!\!\widetilde{P}_{\Omega}^{*}\right)^{-1} \!\!\!\!=\varepsilon_{k}^{-2}\left(\lambda \varepsilon_{k}^{-2} \mathbf{I}+\widetilde{P}_{\Omega} \widetilde{P}_{\Omega}^{*}-\Xi\right)^{-1}
\end{aligned}
\end{equation*}
where
\begin{equation*}
\begin{aligned} 
\Xi:=\! \widetilde{P}_{\Omega} \y{G}^{*} P_{\f{T}_{k}}\left(\frac{\varepsilon_{k}^{-2} \mathbf{I}_{S_{k}}}{\mathbf{D}_{S_{k}}-\varepsilon_{k}^{-2} \mathbf{I}_{S_{k}}}+P_{\f{T}_{k}}^{*} \y{G} \y{G}^{*} P_{\f{T}_{k}}\right)^{-1}\!\!\! P_{\f{T}_{k}}^{*} \y{G} \widetilde{P}_{\Omega}^{*}.
\end{aligned}
\end{equation*}
Here, we can again use Sherman-Morrison-Woodbury \cref{eq:SMW:formula} with $\mathbf{C}=\mathbf{N}^{-1}, \mathbf{E}=\mathbf{F}=\widetilde{\mathbf{E}}$ and $\mathbf{B}=\f{K}$ with $\f{K}$ from \cref{def:Kmat}, $\widetilde{\mathbf{E}}=\widetilde{P}_{\Omega} \y{G}^{*} P_{\f{T}_{k}}$ and
\begin{equation}
\begin{aligned} \label{eq:def:N:appendix}
\mathbf{N}&:=\frac{\mathbf{D}_{S_{k}}^{-1}}{\mathbf{D}_{S_{k}}^{-1}-\varepsilon_{k}^{2} \mathbf{I}_{S_{k}}}-P_{\mathbf{T}_{k}}^{*} \y{G} \y{G}^{*} P_{\mathbf{T}_{k}}\\&=-\left(\frac{\varepsilon_{k}^{-2} \mathbf{I}_{S_{k}}}{\mathbf{D}_{S_{k}}-\varepsilon_{k}^{-2} \mathbf{I}_{S_{k}}}+P_{\mathbf{T}_{k}}^{*} \y{G} \y{G}^{*}P_{\f{T}_{k}}\right)
\end{aligned}
\end{equation}
to obtain
\begin{equation} \label{eq:widetilde:y} 
\begin{split}
\widetilde{\mathbf{y}} &:=\left(\lambda \operatorname{Id}+\widetilde{P}_{\Omega}\left(\widetilde{W}_{\y{G}}^{(k)}\right)^{-1} \widetilde{P}_{\Omega}^{*}\right)^{-1}(\mathbf{y}) \\
&=\varepsilon_{k}^{-2} \mathbf{K}^{-1}(\mathbf{y})-\varepsilon_{k}^{-2} \mathbf{K}^{-1} \widetilde{\mathbf{E}}\left(\widetilde{\mathbf{E}}^{*} \mathbf{K}^{-1} \widetilde{\mathbf{E}}+\mathbf{N}\right)^{-1} \widetilde{\mathbf{E}}^{*} \mathbf{K}^{-1}(\mathbf{y}) \\
&=\varepsilon_{k}^{-2} \mathbf{K}^{-1}\left(\mathbf{y}-\widetilde{\mathbf{E}} \gamma_{k+1}\right)\,,
\end{split}
\end{equation}
using the notation that $\gamma_{k+1} \in S_{k}$ is solution to the invertible linear system
$$
\left(\widetilde{\mathbf{E}}^{*} \mathbf{K}^{-1} \widetilde{\mathbf{E}}+\mathbf{N}\right) \gamma_{k+1}=\widetilde{\mathbf{E}}^{*} \mathbf{K}^{-1} \mathbf{y}
$$
Furthermore, $\mathbf{y}, \widetilde{\mathbf{y}}$ and $\gamma_{k+1}$ are related by
\begin{equation} \label{eq:epskEtilde}
\begin{split}
&\varepsilon_{k}^{2} \widetilde{\mathbf{E}}^{*}(\widetilde{\mathbf{y}}) \\&=\widetilde{\mathbf{E}}^{*} \mathbf{K}^{-1}\left(\mathbf{y}-\widetilde{\mathbf{E}} \gamma_{k+1}\right)\\&=\widetilde{\mathbf{E}}^{*} \mathbf{K}^{-1} \!\! \left(\mathbf{y}-\widetilde{\mathbf{E}}\left(\widetilde{\mathbf{E}}^{*} \mathbf{K}^{-1} \widetilde{\mathbf{E}}+\mathbf{N}\right)^{-1} \widetilde{\mathbf{E}}^{*} \mathbf{K}^{-1} \mathbf{y} \!\right)\!\! =\widetilde{\mathbf{E}}^{*} \mathbf{K}^{-1} \mathbf{y} \\
&\quad -\left(\widetilde{\mathbf{E}}^{*} \mathbf{K}^{-1} \widetilde{\mathbf{E}} + \mathbf{N} - \mathbf{N}\right)\left(\widetilde{\mathbf{E}}^{*} \mathbf{K}^{-1} \widetilde{\mathbf{E}}+\mathbf{N}\right)^{-1} \widetilde{\mathbf{E}}^{*} \mathbf{K}^{-1} \mathbf{y} \\
&=\mathbf{N}\left(\widetilde{\mathbf{E}}^{*} \mathbf{K}^{-1} \widetilde{\mathbf{E}}+\mathbf{N}\right)^{-1} \widetilde{\mathbf{E}}^{*} \mathbf{K}^{-1} \mathbf{y}.
\end{split}
\end{equation}
Inserting these equalities back into the expression \cref{eq:Xkp1:WGk} for $\widetilde{\mathbf{X}}^{(k+1)}$, we observe that
\begin{equation*}
\begin{split}
\widetilde{\f{X}}\hkk &= \y{D}^{-1} \big(\widetilde{W}_{\y{G}}\hk\big)^{-1} \widetilde{P}_{\Omega}^* \big( \lambda \Id + \widetilde{P}_{\Omega} \big(\widetilde{W}_{\y{G}}\hk\big)^{-1}  \widetilde{P}_{\Omega}^*\big)^{-1} (\f{y}) \\
&= \y{D}^{-1}  \big(\widetilde{W}_{\y{G}}\hk\big)^{-1} \widetilde{P}_{\Omega}^*(\widetilde{\f{y}}) \\
&= \varepsilon_k^2 \y{D}^{-1} \big[ \Id + \y{G}^* P_{\f{T}_k} \f{N}^{-1} P_{\f{T}_k}^* \y{G} \big] \widetilde{P}_{\Omega}^*(\widetilde{\f{y}}) \\
&= \varepsilon_k^2 \y{D}^{-1} \widetilde{P}_{\Omega}^*(\widetilde{\f{y}}) + \varepsilon_k^2 \y{D}^{-1}\y{G}^* P_{\f{T}_k} \f{N}^{-1}  \widetilde{\f{E}}^* (\widetilde{\f{y}}) \\
&= \varepsilon_k^2 \y{D}^{-1} \widetilde{P}_{\Omega}^*(\widetilde{\f{y}}) \\
&+ \varepsilon_k^2 \y{D}^{-1}\y{G}^* P_{\f{T}_k}   \left(  \widetilde{\f{E}}^* \f{K}^{-1}  \widetilde{\f{E}}+ \f{N} \right)^{-1}  \widetilde{\f{E}}^* \f{K}^{-1} (\widetilde{\f{y}}) \\
&=  \varepsilon_k^2 \y{D}^{-1} \widetilde{P}_{\Omega}^*(\widetilde{\f{y}}) + \varepsilon_k^2 \y{D}^{-1}\y{G}^* P_{\f{T}_k} ( \gamma_{k+1}) \\
&= \y{D}^{-1} \widetilde{P}_{\Omega}^* \f{K}^{-1} \left(\f{y} - \widetilde{E} \gamma_{k+1}\right) + \y{D}^{-1} \y{G}^* P_{\f{T}_k}(\gamma_{k+1}) \\
&= \y{D}^{-1} \widetilde{P}_{\Omega}^* (\f{p}_{k+1}) + \y{D}^{-1} \y{G}^* P_{\f{T}_k}(\gamma_{k+1}).
\end{split}
\end{equation*}

In the second equality, we used the definition of $\widetilde{\mathbf{y}}$; in the third equality, we used that $\left(\widetilde{W}_{\y{G}}^{(k)}\right)^{-1}=\varepsilon_{k}^{2}\left[\Id +\y{G}^{ * } P_{\f{T}_{k}} \mathbf{N}^{-1} P_{\f{T}_{k}}^{*} \y{G}\right]$, which follows from \cref{eq:Wgkmin1} and \cref{eq:def:N:appendix}. In the fifth equality, we used \cref{eq:epskEtilde}; in the sixth equality the definition of $\gamma_{k+1}$. In the seventh equality, we used \cref{eq:widetilde:y} and in the last equality, we use the definition $\mathbf{p}_{k+1}=\mathbf{K}^{-1}\left(\mathbf{y}-P_{\Omega} \y{D}^{-1} \y{G}^{*} P_{\mathbf{T}_{k}}\left(\gamma_{k+1}\right)\right)=\mathbf{K}^{-1}\left(\mathbf{y}-\widetilde{\mathbf{E}}\left(\gamma_{k+1}\right)\right)$ of the residual $\mathbf{p}_{k+1}$. This concludes the proof of \Cref{lemma:Algo:implementation}.
\end{proof} 

\begin{proof}[{Proof of \Cref{lemma:Algo:implementation}.2}]
In line $3$ of \Cref{algo:TOIRLS:implementation}, $\gamma_{k+1}$ is computed by adding $\gamma_{k+1}^{(0)}$ (which is an input) and $\Delta\gamma_{k+1}$, the solution of the positive definite linear system of \cref{eq:gamma:system}, which requires $O(r_k (n d_1 +n d_2 +r_k))$ computations.

To solve \cref{eq:gamma:system}, a conjugate gradient method can be applied, whose cost crucially depends on the cost of executing matrix-vector multiplications with the matrix $\f{M}$ from \cref{eq:Moperator:def}. We address the matrix-vector multiplication cost of the three summands of $\f{M}$ separately, as we can obtain $\f{M} \gamma$ for $\gamma \in \R^{r_k (n d_1 +n d_2 +r_k)}$ by adding/substracting the three resulting vectors in additional $3 r_k (n d_1 +n d_2 +r_k)$ time.
\begin{itemize}
\item \emph{Matrix-vector multiplication} with $\frac{\f{D}_{S_k}^{-1}}{\f{D}_{S_k}^{-1}- \varepsilon_k^{2} \f{I}_{S_k}}$: this matrix is diagonaly, since $\mathbf{D}_{S_{k}}$ is, resulting in a time complexity of $r_k (n d_1 +n d_2 +r_k)$.
\item \emph{Matrix-vector} multiplication with $P_{\f{T}_k}^* \y{G}\y{D}^{-1}  P_{\Omega}^* \f{K}^{-1} P_{\Omega} \y{D}^{-1} \y{G}^* P_{\f{T}_k}$. This can me implemented by the successive application of the three operators $P_{\f{T}_k}^* \y{G}\y{D}^{-1}  P_{\Omega}^*$,  $\f{K}^{-1}$ and  $P_{\Omega} \y{D}^{-1} \y{G}^* P_{\f{T}_k}$. Applying $P_{\Omega} \y{D}^{-1} \y{G}^* P_{\f{T}_k}$ can be done in $O( m T r_k +r_k^2 n T)+ mT = O( m T r_k +r_k^2 n T)$ time by evaluating the tangent space matrix returned by $P_{\f{T}_k}$ at $m T$ locations (which correspond to the support set of $\y{H}(P_{\Omega}^*(\f{y}))$) via Algorithm 4 of the paper \cite{KMV21} and averaging the entries across the Hankel blocks. $\f{K}^{-1}$ can be applied by in $r_k (n d_1 +n d_2 +r_k)$ time as $P_{\Omega} \y{D}^{-2} P_{\Omega}^{*}$ is a diagonal $(m \times m)$ matrix whose $i$-th diagonal entry is the inverse number of occurrences of the Hankel block which corresponds to the $i$-th observation in $\Omega$. Finally, the application of $P_{\f{T}_k}^* \y{G}\y{D}^{-1}  P_{\Omega}^*$ can be implemented via Algorithm 3 of \cite{KMV21} as the image of $ \y{G}\y{D}^{-1}  P_{\Omega}^*$ is a sparse $(n d_1 \times n d_2)$ matrix with a support set of size $m T$, giving a time complexity of $O( m T r_k +r_k^2 n T)$. Thus, the total time complexity of the entire matrix-vector multiplication is $O( m T r_k +r_k^2 n T)$.
\item \emph{Matrix-vector multiplication} with $P_{\f{T}_k}^*\y{G}\y{G}^* P_{\f{T}_k}$. We observe that block Hankel matrices of size ($n d_1 \times n d_2$) with $(n \times n)$ blocks can be embedded into a $(n T \times n T)$ block circulant matrix (up to reordering of columns), cf., e.g., \cite[Section 8.3.1]{Kailath1999Fast}, and such block circulant matrices can be diagonalized by a "block" discrete Fourier transform. Using the fast Fourier transform across blocks, it is possible to compute the image of $P_{\f{T}_k}^*\y{G}\y{G}^* P_{\f{T}_k}$ in $O(r_k n T \log T+ r_k^2 T^2 n)$ time, see also \cite[Section 3.4]{K19} for a related algorithm for Hankel matrices.
\end{itemize}
We refer to our MATLAB implementation for further details on the above.
Overall, the matrix-vector multiplication of $M$ with a vector $\gamma \in \R^{r_k (n d_1 +n d_2 +r_k)}$ can be performed in $O( r_k T( m + n \log T + n r_k T))$ time. Since the computations necessary to obtain $\f{h}_k^0$ and $\f{p}_{k+1}$ in line 1 and 4, respectively, involve only operations whose order we quantified above, this concludes the proof of \Cref{lemma:Algo:implementation}.2.
\end{proof}

  \bibliographystyle{ieeetr}
\bibliography{LaplacianCompletion}

\end{document}

%% file: experiment_Orthogonal_PT_uniform.tex
%
%
\begin{tikzpicture}

\begin{axis}[%
width=0.876\figurewidth,
height=\figureheight,
at={(0\figurewidth,0\figureheight)},
scale only axis,
point meta min=0,
point meta max=1,
axis on top,
xmin=0.990076335877863,
xmax=3.60992366412214,
xlabel style={font=\color{white!15!black}},
xlabel={Oversampling $\rho$},
ymin=2.5,
ymax=40.5,
ylabel style={font=\color{white!15!black}},
ylabel={Rank $r$},
axis background/.style={fill=white},
xlabel style={font=\tiny},ylabel style={font=\tiny},
xticklabel style = {font=\tiny}, yticklabel style = {font=\tiny},
colormap={mymap}{[1pt] rgb(0pt)=(0.2422,0.1504,0.6603); rgb(1pt)=(0.2444,0.1534,0.6728); rgb(2pt)=(0.2464,0.1569,0.6847); rgb(3pt)=(0.2484,0.1607,0.6961); rgb(4pt)=(0.2503,0.1648,0.7071); rgb(5pt)=(0.2522,0.1689,0.7179); rgb(6pt)=(0.254,0.1732,0.7286); rgb(7pt)=(0.2558,0.1773,0.7393); rgb(8pt)=(0.2576,0.1814,0.7501); rgb(9pt)=(0.2594,0.1854,0.761); rgb(11pt)=(0.2628,0.1932,0.7828); rgb(12pt)=(0.2645,0.1972,0.7937); rgb(13pt)=(0.2661,0.2011,0.8043); rgb(14pt)=(0.2676,0.2052,0.8148); rgb(15pt)=(0.2691,0.2094,0.8249); rgb(16pt)=(0.2704,0.2138,0.8346); rgb(17pt)=(0.2717,0.2184,0.8439); rgb(18pt)=(0.2729,0.2231,0.8528); rgb(19pt)=(0.274,0.228,0.8612); rgb(20pt)=(0.2749,0.233,0.8692); rgb(21pt)=(0.2758,0.2382,0.8767); rgb(22pt)=(0.2766,0.2435,0.884); rgb(23pt)=(0.2774,0.2489,0.8908); rgb(24pt)=(0.2781,0.2543,0.8973); rgb(25pt)=(0.2788,0.2598,0.9035); rgb(26pt)=(0.2794,0.2653,0.9094); rgb(27pt)=(0.2798,0.2708,0.915); rgb(28pt)=(0.2802,0.2764,0.9204); rgb(29pt)=(0.2806,0.2819,0.9255); rgb(30pt)=(0.2809,0.2875,0.9305); rgb(31pt)=(0.2811,0.293,0.9352); rgb(32pt)=(0.2813,0.2985,0.9397); rgb(33pt)=(0.2814,0.304,0.9441); rgb(34pt)=(0.2814,0.3095,0.9483); rgb(35pt)=(0.2813,0.315,0.9524); rgb(36pt)=(0.2811,0.3204,0.9563); rgb(37pt)=(0.2809,0.3259,0.96); rgb(38pt)=(0.2807,0.3313,0.9636); rgb(39pt)=(0.2803,0.3367,0.967); rgb(40pt)=(0.2798,0.3421,0.9702); rgb(41pt)=(0.2791,0.3475,0.9733); rgb(42pt)=(0.2784,0.3529,0.9763); rgb(43pt)=(0.2776,0.3583,0.9791); rgb(44pt)=(0.2766,0.3638,0.9817); rgb(45pt)=(0.2754,0.3693,0.984); rgb(46pt)=(0.2741,0.3748,0.9862); rgb(47pt)=(0.2726,0.3804,0.9881); rgb(48pt)=(0.271,0.386,0.9898); rgb(49pt)=(0.2691,0.3916,0.9912); rgb(50pt)=(0.267,0.3973,0.9924); rgb(51pt)=(0.2647,0.403,0.9935); rgb(52pt)=(0.2621,0.4088,0.9946); rgb(53pt)=(0.2591,0.4145,0.9955); rgb(54pt)=(0.2556,0.4203,0.9965); rgb(55pt)=(0.2517,0.4261,0.9974); rgb(56pt)=(0.2473,0.4319,0.9983); rgb(57pt)=(0.2424,0.4378,0.9991); rgb(58pt)=(0.2369,0.4437,0.9996); rgb(59pt)=(0.2311,0.4497,0.9995); rgb(60pt)=(0.225,0.4559,0.9985); rgb(61pt)=(0.2189,0.462,0.9968); rgb(62pt)=(0.2128,0.4682,0.9948); rgb(63pt)=(0.2066,0.4743,0.9926); rgb(64pt)=(0.2006,0.4803,0.9906); rgb(65pt)=(0.195,0.4861,0.9887); rgb(66pt)=(0.1903,0.4919,0.9867); rgb(67pt)=(0.1869,0.4975,0.9844); rgb(68pt)=(0.1847,0.503,0.9819); rgb(69pt)=(0.1831,0.5084,0.9793); rgb(70pt)=(0.1818,0.5138,0.9766); rgb(71pt)=(0.1806,0.5191,0.9738); rgb(72pt)=(0.1795,0.5244,0.9709); rgb(73pt)=(0.1785,0.5296,0.9677); rgb(74pt)=(0.1778,0.5349,0.9641); rgb(75pt)=(0.1773,0.5401,0.9602); rgb(76pt)=(0.1768,0.5452,0.956); rgb(77pt)=(0.1764,0.5504,0.9516); rgb(78pt)=(0.1755,0.5554,0.9473); rgb(79pt)=(0.174,0.5605,0.9432); rgb(80pt)=(0.1716,0.5655,0.9393); rgb(81pt)=(0.1686,0.5705,0.9357); rgb(82pt)=(0.1649,0.5755,0.9323); rgb(83pt)=(0.161,0.5805,0.9289); rgb(84pt)=(0.1573,0.5854,0.9254); rgb(85pt)=(0.154,0.5902,0.9218); rgb(86pt)=(0.1513,0.595,0.9182); rgb(87pt)=(0.1492,0.5997,0.9147); rgb(88pt)=(0.1475,0.6043,0.9113); rgb(89pt)=(0.1461,0.6089,0.908); rgb(90pt)=(0.1446,0.6135,0.905); rgb(91pt)=(0.1429,0.618,0.9022); rgb(92pt)=(0.1408,0.6226,0.8998); rgb(93pt)=(0.1383,0.6272,0.8975); rgb(94pt)=(0.1354,0.6317,0.8953); rgb(95pt)=(0.1321,0.6363,0.8932); rgb(96pt)=(0.1288,0.6408,0.891); rgb(97pt)=(0.1253,0.6453,0.8887); rgb(98pt)=(0.1219,0.6497,0.8862); rgb(99pt)=(0.1185,0.6541,0.8834); rgb(100pt)=(0.1152,0.6584,0.8804); rgb(101pt)=(0.1119,0.6627,0.877); rgb(102pt)=(0.1085,0.6669,0.8734); rgb(103pt)=(0.1048,0.671,0.8695); rgb(104pt)=(0.1009,0.675,0.8653); rgb(105pt)=(0.0964,0.6789,0.8609); rgb(106pt)=(0.0914,0.6828,0.8562); rgb(107pt)=(0.0855,0.6865,0.8513); rgb(108pt)=(0.0789,0.6902,0.8462); rgb(109pt)=(0.0713,0.6938,0.8409); rgb(110pt)=(0.0628,0.6972,0.8355); rgb(111pt)=(0.0535,0.7006,0.8299); rgb(112pt)=(0.0433,0.7039,0.8242); rgb(113pt)=(0.0328,0.7071,0.8183); rgb(114pt)=(0.0234,0.7103,0.8124); rgb(115pt)=(0.0155,0.7133,0.8064); rgb(116pt)=(0.0091,0.7163,0.8003); rgb(117pt)=(0.0046,0.7192,0.7941); rgb(118pt)=(0.0019,0.722,0.7878); rgb(119pt)=(0.0009,0.7248,0.7815); rgb(120pt)=(0.0018,0.7275,0.7752); rgb(121pt)=(0.0046,0.7301,0.7688); rgb(122pt)=(0.0094,0.7327,0.7623); rgb(123pt)=(0.0162,0.7352,0.7558); rgb(124pt)=(0.0253,0.7376,0.7492); rgb(125pt)=(0.0369,0.74,0.7426); rgb(126pt)=(0.0504,0.7423,0.7359); rgb(127pt)=(0.0638,0.7446,0.7292); rgb(128pt)=(0.077,0.7468,0.7224); rgb(129pt)=(0.0899,0.7489,0.7156); rgb(130pt)=(0.1023,0.751,0.7088); rgb(131pt)=(0.1141,0.7531,0.7019); rgb(132pt)=(0.1252,0.7552,0.695); rgb(133pt)=(0.1354,0.7572,0.6881); rgb(134pt)=(0.1448,0.7593,0.6812); rgb(135pt)=(0.1532,0.7614,0.6741); rgb(136pt)=(0.1609,0.7635,0.6671); rgb(137pt)=(0.1678,0.7656,0.6599); rgb(138pt)=(0.1741,0.7678,0.6527); rgb(139pt)=(0.1799,0.7699,0.6454); rgb(140pt)=(0.1853,0.7721,0.6379); rgb(141pt)=(0.1905,0.7743,0.6303); rgb(142pt)=(0.1954,0.7765,0.6225); rgb(143pt)=(0.2003,0.7787,0.6146); rgb(144pt)=(0.2061,0.7808,0.6065); rgb(145pt)=(0.2118,0.7828,0.5983); rgb(146pt)=(0.2178,0.7849,0.5899); rgb(147pt)=(0.2244,0.7869,0.5813); rgb(148pt)=(0.2318,0.7887,0.5725); rgb(149pt)=(0.2401,0.7905,0.5636); rgb(150pt)=(0.2491,0.7922,0.5546); rgb(151pt)=(0.2589,0.7937,0.5454); rgb(152pt)=(0.2695,0.7951,0.536); rgb(153pt)=(0.2809,0.7964,0.5266); rgb(154pt)=(0.2929,0.7975,0.517); rgb(155pt)=(0.3052,0.7985,0.5074); rgb(156pt)=(0.3176,0.7994,0.4975); rgb(157pt)=(0.3301,0.8002,0.4876); rgb(158pt)=(0.3424,0.8009,0.4774); rgb(159pt)=(0.3548,0.8016,0.4669); rgb(160pt)=(0.3671,0.8021,0.4563); rgb(161pt)=(0.3795,0.8026,0.4454); rgb(162pt)=(0.3921,0.8029,0.4344); rgb(163pt)=(0.405,0.8031,0.4233); rgb(164pt)=(0.4184,0.803,0.4122); rgb(165pt)=(0.4322,0.8028,0.4013); rgb(166pt)=(0.4463,0.8024,0.3904); rgb(167pt)=(0.4608,0.8018,0.3797); rgb(168pt)=(0.4753,0.8011,0.3691); rgb(169pt)=(0.4899,0.8002,0.3586); rgb(170pt)=(0.5044,0.7993,0.348); rgb(171pt)=(0.5187,0.7982,0.3374); rgb(172pt)=(0.5329,0.797,0.3267); rgb(173pt)=(0.547,0.7957,0.3159); rgb(175pt)=(0.5748,0.7929,0.2941); rgb(176pt)=(0.5886,0.7913,0.2833); rgb(177pt)=(0.6024,0.7896,0.2726); rgb(178pt)=(0.6161,0.7878,0.2622); rgb(179pt)=(0.6297,0.7859,0.2521); rgb(180pt)=(0.6433,0.7839,0.2423); rgb(181pt)=(0.6567,0.7818,0.2329); rgb(182pt)=(0.6701,0.7796,0.2239); rgb(183pt)=(0.6833,0.7773,0.2155); rgb(184pt)=(0.6963,0.775,0.2075); rgb(185pt)=(0.7091,0.7727,0.1998); rgb(186pt)=(0.7218,0.7703,0.1924); rgb(187pt)=(0.7344,0.7679,0.1852); rgb(188pt)=(0.7468,0.7654,0.1782); rgb(189pt)=(0.759,0.7629,0.1717); rgb(190pt)=(0.771,0.7604,0.1658); rgb(191pt)=(0.7829,0.7579,0.1608); rgb(192pt)=(0.7945,0.7554,0.157); rgb(193pt)=(0.806,0.7529,0.1546); rgb(194pt)=(0.8172,0.7505,0.1535); rgb(195pt)=(0.8281,0.7481,0.1536); rgb(196pt)=(0.8389,0.7457,0.1546); rgb(197pt)=(0.8495,0.7435,0.1564); rgb(198pt)=(0.86,0.7413,0.1587); rgb(199pt)=(0.8703,0.7392,0.1615); rgb(200pt)=(0.8804,0.7372,0.165); rgb(201pt)=(0.8903,0.7353,0.1695); rgb(202pt)=(0.9,0.7336,0.1749); rgb(203pt)=(0.9093,0.7321,0.1815); rgb(204pt)=(0.9184,0.7308,0.189); rgb(205pt)=(0.9272,0.7298,0.1973); rgb(206pt)=(0.9357,0.729,0.2061); rgb(207pt)=(0.944,0.7285,0.2151); rgb(208pt)=(0.9523,0.7284,0.2237); rgb(209pt)=(0.9606,0.7285,0.2312); rgb(210pt)=(0.9689,0.7292,0.2373); rgb(211pt)=(0.977,0.7304,0.2418); rgb(212pt)=(0.9842,0.733,0.2446); rgb(213pt)=(0.99,0.7365,0.2429); rgb(214pt)=(0.9946,0.7407,0.2394); rgb(215pt)=(0.9966,0.7458,0.2351); rgb(216pt)=(0.9971,0.7513,0.2309); rgb(217pt)=(0.9972,0.7569,0.2267); rgb(218pt)=(0.9971,0.7626,0.2224); rgb(219pt)=(0.9969,0.7683,0.2181); rgb(220pt)=(0.9966,0.774,0.2138); rgb(221pt)=(0.9962,0.7798,0.2095); rgb(222pt)=(0.9957,0.7856,0.2053); rgb(223pt)=(0.9949,0.7915,0.2012); rgb(224pt)=(0.9938,0.7974,0.1974); rgb(225pt)=(0.9923,0.8034,0.1939); rgb(226pt)=(0.9906,0.8095,0.1906); rgb(227pt)=(0.9885,0.8156,0.1875); rgb(228pt)=(0.9861,0.8218,0.1846); rgb(229pt)=(0.9835,0.828,0.1817); rgb(230pt)=(0.9807,0.8342,0.1787); rgb(231pt)=(0.9778,0.8404,0.1757); rgb(232pt)=(0.9748,0.8467,0.1726); rgb(233pt)=(0.972,0.8529,0.1695); rgb(234pt)=(0.9694,0.8591,0.1665); rgb(235pt)=(0.9671,0.8654,0.1636); rgb(236pt)=(0.9651,0.8716,0.1608); rgb(237pt)=(0.9634,0.8778,0.1582); rgb(238pt)=(0.9619,0.884,0.1557); rgb(239pt)=(0.9608,0.8902,0.1532); rgb(240pt)=(0.9601,0.8963,0.1507); rgb(241pt)=(0.9596,0.9023,0.148); rgb(242pt)=(0.9595,0.9084,0.145); rgb(243pt)=(0.9597,0.9143,0.1418); rgb(244pt)=(0.9601,0.9203,0.1382); rgb(245pt)=(0.9608,0.9262,0.1344); rgb(246pt)=(0.9618,0.932,0.1304); rgb(247pt)=(0.9629,0.9379,0.1261); rgb(248pt)=(0.9642,0.9437,0.1216); rgb(249pt)=(0.9657,0.9494,0.1168); rgb(250pt)=(0.9674,0.9552,0.1116); rgb(251pt)=(0.9692,0.9609,0.1061); rgb(252pt)=(0.9711,0.9667,0.1001); rgb(253pt)=(0.973,0.9724,0.0938); rgb(254pt)=(0.9749,0.9782,0.0872); rgb(255pt)=(0.9769,0.9839,0.0805)},
colorbar, colorbar style={ xticklabel style = {font=\tiny}, yticklabel style = {font=\tiny}} 
]
\addplot [forget plot] graphics [xmin=0.990076335877863, xmax=3.60992366412214, ymin=2.5, ymax=40.5] {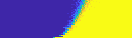};
\addplot [color=red, line width=1.5pt, forget plot]
  table[row sep=crcr]{%
2.59618951650383	40\\
2.57981759279519	37\\
2.56206054130931	34\\
2.54932937072785	32\\
2.53577628128896	30\\
2.5212877782767	28\\
2.50572510412442	26\\
2.48891613551298	24\\
2.47064374634516	22\\
2.45062860858625	20\\
2.43985701676486	19\\
2.42850290029811	18\\
2.41649963339172	17\\
2.40376846281026	16\\
2.39021537337138	15\\
2.37572687035912	14\\
2.36016419620683	13\\
2.34335522759539	12\\
2.32508283842757	11\\
2.30506770066866	10\\
2.28294199238051	9\\
2.25820755489268	8\\
2.23016596244153	7\\
2.1977943196778	6\\
2.15950679275107	5\\
2.11264664697509	4\\
2.05223341176021	3\\
};
\end{axis}
\end{tikzpicture}%

%% file: experiment_GC_ER_good_PT_uniform_T_1.tex
%
%
\definecolor{mycolor1}{rgb}{1.00000,0.00000,1.00000}%
\begin{tikzpicture}

\begin{axis}[%
width=0.2\figurewidth,
height=\figureheight,
at={(0\figurewidth,0\figureheight)},
scale only axis,
axis on top,
xmin=70.0142045454545,
xmax=1834.98579545455,
xlabel={Nr. of samples $m$},
ymin=0.5,
ymax=60.5,
ylabel style={font=\color{white!15!black}},
ylabel={Rank $r$},
ylabel near ticks,
xticklabels={,,,1000},
axis background/.style={fill=white},
xlabel style={font=\tiny},ylabel style={font=\tiny, at={(ticklabel cs:0.5)},anchor=near ticklabel},
xticklabel style = {font=\tiny}, yticklabel style = {font=\tiny},
colormap={mymap}{[1pt] rgb(0pt)=(0.2422,0.1504,0.6603); rgb(1pt)=(0.2444,0.1534,0.6728); rgb(2pt)=(0.2464,0.1569,0.6847); rgb(3pt)=(0.2484,0.1607,0.6961); rgb(4pt)=(0.2503,0.1648,0.7071); rgb(5pt)=(0.2522,0.1689,0.7179); rgb(6pt)=(0.254,0.1732,0.7286); rgb(7pt)=(0.2558,0.1773,0.7393); rgb(8pt)=(0.2576,0.1814,0.7501); rgb(9pt)=(0.2594,0.1854,0.761); rgb(11pt)=(0.2628,0.1932,0.7828); rgb(12pt)=(0.2645,0.1972,0.7937); rgb(13pt)=(0.2661,0.2011,0.8043); rgb(14pt)=(0.2676,0.2052,0.8148); rgb(15pt)=(0.2691,0.2094,0.8249); rgb(16pt)=(0.2704,0.2138,0.8346); rgb(17pt)=(0.2717,0.2184,0.8439); rgb(18pt)=(0.2729,0.2231,0.8528); rgb(19pt)=(0.274,0.228,0.8612); rgb(20pt)=(0.2749,0.233,0.8692); rgb(21pt)=(0.2758,0.2382,0.8767); rgb(22pt)=(0.2766,0.2435,0.884); rgb(23pt)=(0.2774,0.2489,0.8908); rgb(24pt)=(0.2781,0.2543,0.8973); rgb(25pt)=(0.2788,0.2598,0.9035); rgb(26pt)=(0.2794,0.2653,0.9094); rgb(27pt)=(0.2798,0.2708,0.915); rgb(28pt)=(0.2802,0.2764,0.9204); rgb(29pt)=(0.2806,0.2819,0.9255); rgb(30pt)=(0.2809,0.2875,0.9305); rgb(31pt)=(0.2811,0.293,0.9352); rgb(32pt)=(0.2813,0.2985,0.9397); rgb(33pt)=(0.2814,0.304,0.9441); rgb(34pt)=(0.2814,0.3095,0.9483); rgb(35pt)=(0.2813,0.315,0.9524); rgb(36pt)=(0.2811,0.3204,0.9563); rgb(37pt)=(0.2809,0.3259,0.96); rgb(38pt)=(0.2807,0.3313,0.9636); rgb(39pt)=(0.2803,0.3367,0.967); rgb(40pt)=(0.2798,0.3421,0.9702); rgb(41pt)=(0.2791,0.3475,0.9733); rgb(42pt)=(0.2784,0.3529,0.9763); rgb(43pt)=(0.2776,0.3583,0.9791); rgb(44pt)=(0.2766,0.3638,0.9817); rgb(45pt)=(0.2754,0.3693,0.984); rgb(46pt)=(0.2741,0.3748,0.9862); rgb(47pt)=(0.2726,0.3804,0.9881); rgb(48pt)=(0.271,0.386,0.9898); rgb(49pt)=(0.2691,0.3916,0.9912); rgb(50pt)=(0.267,0.3973,0.9924); rgb(51pt)=(0.2647,0.403,0.9935); rgb(52pt)=(0.2621,0.4088,0.9946); rgb(53pt)=(0.2591,0.4145,0.9955); rgb(54pt)=(0.2556,0.4203,0.9965); rgb(55pt)=(0.2517,0.4261,0.9974); rgb(56pt)=(0.2473,0.4319,0.9983); rgb(57pt)=(0.2424,0.4378,0.9991); rgb(58pt)=(0.2369,0.4437,0.9996); rgb(59pt)=(0.2311,0.4497,0.9995); rgb(60pt)=(0.225,0.4559,0.9985); rgb(61pt)=(0.2189,0.462,0.9968); rgb(62pt)=(0.2128,0.4682,0.9948); rgb(63pt)=(0.2066,0.4743,0.9926); rgb(64pt)=(0.2006,0.4803,0.9906); rgb(65pt)=(0.195,0.4861,0.9887); rgb(66pt)=(0.1903,0.4919,0.9867); rgb(67pt)=(0.1869,0.4975,0.9844); rgb(68pt)=(0.1847,0.503,0.9819); rgb(69pt)=(0.1831,0.5084,0.9793); rgb(70pt)=(0.1818,0.5138,0.9766); rgb(71pt)=(0.1806,0.5191,0.9738); rgb(72pt)=(0.1795,0.5244,0.9709); rgb(73pt)=(0.1785,0.5296,0.9677); rgb(74pt)=(0.1778,0.5349,0.9641); rgb(75pt)=(0.1773,0.5401,0.9602); rgb(76pt)=(0.1768,0.5452,0.956); rgb(77pt)=(0.1764,0.5504,0.9516); rgb(78pt)=(0.1755,0.5554,0.9473); rgb(79pt)=(0.174,0.5605,0.9432); rgb(80pt)=(0.1716,0.5655,0.9393); rgb(81pt)=(0.1686,0.5705,0.9357); rgb(82pt)=(0.1649,0.5755,0.9323); rgb(83pt)=(0.161,0.5805,0.9289); rgb(84pt)=(0.1573,0.5854,0.9254); rgb(85pt)=(0.154,0.5902,0.9218); rgb(86pt)=(0.1513,0.595,0.9182); rgb(87pt)=(0.1492,0.5997,0.9147); rgb(88pt)=(0.1475,0.6043,0.9113); rgb(89pt)=(0.1461,0.6089,0.908); rgb(90pt)=(0.1446,0.6135,0.905); rgb(91pt)=(0.1429,0.618,0.9022); rgb(92pt)=(0.1408,0.6226,0.8998); rgb(93pt)=(0.1383,0.6272,0.8975); rgb(94pt)=(0.1354,0.6317,0.8953); rgb(95pt)=(0.1321,0.6363,0.8932); rgb(96pt)=(0.1288,0.6408,0.891); rgb(97pt)=(0.1253,0.6453,0.8887); rgb(98pt)=(0.1219,0.6497,0.8862); rgb(99pt)=(0.1185,0.6541,0.8834); rgb(100pt)=(0.1152,0.6584,0.8804); rgb(101pt)=(0.1119,0.6627,0.877); rgb(102pt)=(0.1085,0.6669,0.8734); rgb(103pt)=(0.1048,0.671,0.8695); rgb(104pt)=(0.1009,0.675,0.8653); rgb(105pt)=(0.0964,0.6789,0.8609); rgb(106pt)=(0.0914,0.6828,0.8562); rgb(107pt)=(0.0855,0.6865,0.8513); rgb(108pt)=(0.0789,0.6902,0.8462); rgb(109pt)=(0.0713,0.6938,0.8409); rgb(110pt)=(0.0628,0.6972,0.8355); rgb(111pt)=(0.0535,0.7006,0.8299); rgb(112pt)=(0.0433,0.7039,0.8242); rgb(113pt)=(0.0328,0.7071,0.8183); rgb(114pt)=(0.0234,0.7103,0.8124); rgb(115pt)=(0.0155,0.7133,0.8064); rgb(116pt)=(0.0091,0.7163,0.8003); rgb(117pt)=(0.0046,0.7192,0.7941); rgb(118pt)=(0.0019,0.722,0.7878); rgb(119pt)=(0.0009,0.7248,0.7815); rgb(120pt)=(0.0018,0.7275,0.7752); rgb(121pt)=(0.0046,0.7301,0.7688); rgb(122pt)=(0.0094,0.7327,0.7623); rgb(123pt)=(0.0162,0.7352,0.7558); rgb(124pt)=(0.0253,0.7376,0.7492); rgb(125pt)=(0.0369,0.74,0.7426); rgb(126pt)=(0.0504,0.7423,0.7359); rgb(127pt)=(0.0638,0.7446,0.7292); rgb(128pt)=(0.077,0.7468,0.7224); rgb(129pt)=(0.0899,0.7489,0.7156); rgb(130pt)=(0.1023,0.751,0.7088); rgb(131pt)=(0.1141,0.7531,0.7019); rgb(132pt)=(0.1252,0.7552,0.695); rgb(133pt)=(0.1354,0.7572,0.6881); rgb(134pt)=(0.1448,0.7593,0.6812); rgb(135pt)=(0.1532,0.7614,0.6741); rgb(136pt)=(0.1609,0.7635,0.6671); rgb(137pt)=(0.1678,0.7656,0.6599); rgb(138pt)=(0.1741,0.7678,0.6527); rgb(139pt)=(0.1799,0.7699,0.6454); rgb(140pt)=(0.1853,0.7721,0.6379); rgb(141pt)=(0.1905,0.7743,0.6303); rgb(142pt)=(0.1954,0.7765,0.6225); rgb(143pt)=(0.2003,0.7787,0.6146); rgb(144pt)=(0.2061,0.7808,0.6065); rgb(145pt)=(0.2118,0.7828,0.5983); rgb(146pt)=(0.2178,0.7849,0.5899); rgb(147pt)=(0.2244,0.7869,0.5813); rgb(148pt)=(0.2318,0.7887,0.5725); rgb(149pt)=(0.2401,0.7905,0.5636); rgb(150pt)=(0.2491,0.7922,0.5546); rgb(151pt)=(0.2589,0.7937,0.5454); rgb(152pt)=(0.2695,0.7951,0.536); rgb(153pt)=(0.2809,0.7964,0.5266); rgb(154pt)=(0.2929,0.7975,0.517); rgb(155pt)=(0.3052,0.7985,0.5074); rgb(156pt)=(0.3176,0.7994,0.4975); rgb(157pt)=(0.3301,0.8002,0.4876); rgb(158pt)=(0.3424,0.8009,0.4774); rgb(159pt)=(0.3548,0.8016,0.4669); rgb(160pt)=(0.3671,0.8021,0.4563); rgb(161pt)=(0.3795,0.8026,0.4454); rgb(162pt)=(0.3921,0.8029,0.4344); rgb(163pt)=(0.405,0.8031,0.4233); rgb(164pt)=(0.4184,0.803,0.4122); rgb(165pt)=(0.4322,0.8028,0.4013); rgb(166pt)=(0.4463,0.8024,0.3904); rgb(167pt)=(0.4608,0.8018,0.3797); rgb(168pt)=(0.4753,0.8011,0.3691); rgb(169pt)=(0.4899,0.8002,0.3586); rgb(170pt)=(0.5044,0.7993,0.348); rgb(171pt)=(0.5187,0.7982,0.3374); rgb(172pt)=(0.5329,0.797,0.3267); rgb(173pt)=(0.547,0.7957,0.3159); rgb(175pt)=(0.5748,0.7929,0.2941); rgb(176pt)=(0.5886,0.7913,0.2833); rgb(177pt)=(0.6024,0.7896,0.2726); rgb(178pt)=(0.6161,0.7878,0.2622); rgb(179pt)=(0.6297,0.7859,0.2521); rgb(180pt)=(0.6433,0.7839,0.2423); rgb(181pt)=(0.6567,0.7818,0.2329); rgb(182pt)=(0.6701,0.7796,0.2239); rgb(183pt)=(0.6833,0.7773,0.2155); rgb(184pt)=(0.6963,0.775,0.2075); rgb(185pt)=(0.7091,0.7727,0.1998); rgb(186pt)=(0.7218,0.7703,0.1924); rgb(187pt)=(0.7344,0.7679,0.1852); rgb(188pt)=(0.7468,0.7654,0.1782); rgb(189pt)=(0.759,0.7629,0.1717); rgb(190pt)=(0.771,0.7604,0.1658); rgb(191pt)=(0.7829,0.7579,0.1608); rgb(192pt)=(0.7945,0.7554,0.157); rgb(193pt)=(0.806,0.7529,0.1546); rgb(194pt)=(0.8172,0.7505,0.1535); rgb(195pt)=(0.8281,0.7481,0.1536); rgb(196pt)=(0.8389,0.7457,0.1546); rgb(197pt)=(0.8495,0.7435,0.1564); rgb(198pt)=(0.86,0.7413,0.1587); rgb(199pt)=(0.8703,0.7392,0.1615); rgb(200pt)=(0.8804,0.7372,0.165); rgb(201pt)=(0.8903,0.7353,0.1695); rgb(202pt)=(0.9,0.7336,0.1749); rgb(203pt)=(0.9093,0.7321,0.1815); rgb(204pt)=(0.9184,0.7308,0.189); rgb(205pt)=(0.9272,0.7298,0.1973); rgb(206pt)=(0.9357,0.729,0.2061); rgb(207pt)=(0.944,0.7285,0.2151); rgb(208pt)=(0.9523,0.7284,0.2237); rgb(209pt)=(0.9606,0.7285,0.2312); rgb(210pt)=(0.9689,0.7292,0.2373); rgb(211pt)=(0.977,0.7304,0.2418); rgb(212pt)=(0.9842,0.733,0.2446); rgb(213pt)=(0.99,0.7365,0.2429); rgb(214pt)=(0.9946,0.7407,0.2394); rgb(215pt)=(0.9966,0.7458,0.2351); rgb(216pt)=(0.9971,0.7513,0.2309); rgb(217pt)=(0.9972,0.7569,0.2267); rgb(218pt)=(0.9971,0.7626,0.2224); rgb(219pt)=(0.9969,0.7683,0.2181); rgb(220pt)=(0.9966,0.774,0.2138); rgb(221pt)=(0.9962,0.7798,0.2095); rgb(222pt)=(0.9957,0.7856,0.2053); rgb(223pt)=(0.9949,0.7915,0.2012); rgb(224pt)=(0.9938,0.7974,0.1974); rgb(225pt)=(0.9923,0.8034,0.1939); rgb(226pt)=(0.9906,0.8095,0.1906); rgb(227pt)=(0.9885,0.8156,0.1875); rgb(228pt)=(0.9861,0.8218,0.1846); rgb(229pt)=(0.9835,0.828,0.1817); rgb(230pt)=(0.9807,0.8342,0.1787); rgb(231pt)=(0.9778,0.8404,0.1757); rgb(232pt)=(0.9748,0.8467,0.1726); rgb(233pt)=(0.972,0.8529,0.1695); rgb(234pt)=(0.9694,0.8591,0.1665); rgb(235pt)=(0.9671,0.8654,0.1636); rgb(236pt)=(0.9651,0.8716,0.1608); rgb(237pt)=(0.9634,0.8778,0.1582); rgb(238pt)=(0.9619,0.884,0.1557); rgb(239pt)=(0.9608,0.8902,0.1532); rgb(240pt)=(0.9601,0.8963,0.1507); rgb(241pt)=(0.9596,0.9023,0.148); rgb(242pt)=(0.9595,0.9084,0.145); rgb(243pt)=(0.9597,0.9143,0.1418); rgb(244pt)=(0.9601,0.9203,0.1382); rgb(245pt)=(0.9608,0.9262,0.1344); rgb(246pt)=(0.9618,0.932,0.1304); rgb(247pt)=(0.9629,0.9379,0.1261); rgb(248pt)=(0.9642,0.9437,0.1216); rgb(249pt)=(0.9657,0.9494,0.1168); rgb(250pt)=(0.9674,0.9552,0.1116); rgb(251pt)=(0.9692,0.9609,0.1061); rgb(252pt)=(0.9711,0.9667,0.1001); rgb(253pt)=(0.973,0.9724,0.0938); rgb(254pt)=(0.9749,0.9782,0.0872); rgb(255pt)=(0.9769,0.9839,0.0805)},
colorbar, colorbar style={at={(-3,1)},xticklabel style = {font=\tiny}, yticklabel style = {font=\tiny}}
]
\addplot [forget plot] graphics [xmin=70.0142045454545, xmax=1834.98579545455, ymin=0.5, ymax=60.5] {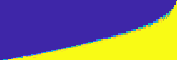};
\addplot [color=red, line width=1.0pt, forget plot]
  table[row sep=crcr]{%
1830	60\\
1827	58\\
1820	56\\
1809	54\\
1794	52\\
1775	50\\
1752	48\\
1725	46\\
1694	44\\
1659	42\\
1620	40\\
1577	38\\
1530	36\\
1479	34\\
1424	32\\
1365	30\\
1302	28\\
1235	26\\
1164	24\\
1089	22\\
1010	20\\
969	19\\
927	18\\
884	17\\
840	16\\
795	15\\
749	14\\
702	13\\
654	12\\
605	11\\
555	10\\
504	9\\
452	8\\
399	7\\
345	6\\
290	5\\
234	4\\
177	3\\
119	2\\
60	1\\
};
\addplot [color=mycolor1, dashdotted, line width=2.0pt, forget plot]
  table[row sep=crcr]{%
1980	22\\
90	1\\
};
\end{axis}
\end{tikzpicture}%

%% file: experiment_GC_ER_good_PT_uniform_T_4_pencil_3.tex
%
%
\definecolor{mycolor1}{rgb}{1.00000,0.00000,1.00000}%
\begin{tikzpicture}

\begin{axis}[%
width=0.858\figurewidth,
height=\figureheight,
at={(0\figurewidth,0\figureheight)},
scale only axis,
point meta min=0,
point meta max=1,
axis on top,
xmin=70.004215851602,
xmax=6004.9957841484,
xlabel style={font=\color{white!15!black}},
xlabel={$m$},
xticklabels={,,2000,4000,6000},
ymin=0.5,
ymax=60.5,
ylabel style={font=\color{white!15!black}},
axis background/.style={fill=white},
xlabel style={font=\tiny},ylabel style={font=\tiny},
xticklabel style = {font=\tiny}, yticklabel style = {font=\tiny},
colormap={mymap}{[1pt] rgb(0pt)=(0.2422,0.1504,0.6603); rgb(1pt)=(0.2444,0.1534,0.6728); rgb(2pt)=(0.2464,0.1569,0.6847); rgb(3pt)=(0.2484,0.1607,0.6961); rgb(4pt)=(0.2503,0.1648,0.7071); rgb(5pt)=(0.2522,0.1689,0.7179); rgb(6pt)=(0.254,0.1732,0.7286); rgb(7pt)=(0.2558,0.1773,0.7393); rgb(8pt)=(0.2576,0.1814,0.7501); rgb(9pt)=(0.2594,0.1854,0.761); rgb(11pt)=(0.2628,0.1932,0.7828); rgb(12pt)=(0.2645,0.1972,0.7937); rgb(13pt)=(0.2661,0.2011,0.8043); rgb(14pt)=(0.2676,0.2052,0.8148); rgb(15pt)=(0.2691,0.2094,0.8249); rgb(16pt)=(0.2704,0.2138,0.8346); rgb(17pt)=(0.2717,0.2184,0.8439); rgb(18pt)=(0.2729,0.2231,0.8528); rgb(19pt)=(0.274,0.228,0.8612); rgb(20pt)=(0.2749,0.233,0.8692); rgb(21pt)=(0.2758,0.2382,0.8767); rgb(22pt)=(0.2766,0.2435,0.884); rgb(23pt)=(0.2774,0.2489,0.8908); rgb(24pt)=(0.2781,0.2543,0.8973); rgb(25pt)=(0.2788,0.2598,0.9035); rgb(26pt)=(0.2794,0.2653,0.9094); rgb(27pt)=(0.2798,0.2708,0.915); rgb(28pt)=(0.2802,0.2764,0.9204); rgb(29pt)=(0.2806,0.2819,0.9255); rgb(30pt)=(0.2809,0.2875,0.9305); rgb(31pt)=(0.2811,0.293,0.9352); rgb(32pt)=(0.2813,0.2985,0.9397); rgb(33pt)=(0.2814,0.304,0.9441); rgb(34pt)=(0.2814,0.3095,0.9483); rgb(35pt)=(0.2813,0.315,0.9524); rgb(36pt)=(0.2811,0.3204,0.9563); rgb(37pt)=(0.2809,0.3259,0.96); rgb(38pt)=(0.2807,0.3313,0.9636); rgb(39pt)=(0.2803,0.3367,0.967); rgb(40pt)=(0.2798,0.3421,0.9702); rgb(41pt)=(0.2791,0.3475,0.9733); rgb(42pt)=(0.2784,0.3529,0.9763); rgb(43pt)=(0.2776,0.3583,0.9791); rgb(44pt)=(0.2766,0.3638,0.9817); rgb(45pt)=(0.2754,0.3693,0.984); rgb(46pt)=(0.2741,0.3748,0.9862); rgb(47pt)=(0.2726,0.3804,0.9881); rgb(48pt)=(0.271,0.386,0.9898); rgb(49pt)=(0.2691,0.3916,0.9912); rgb(50pt)=(0.267,0.3973,0.9924); rgb(51pt)=(0.2647,0.403,0.9935); rgb(52pt)=(0.2621,0.4088,0.9946); rgb(53pt)=(0.2591,0.4145,0.9955); rgb(54pt)=(0.2556,0.4203,0.9965); rgb(55pt)=(0.2517,0.4261,0.9974); rgb(56pt)=(0.2473,0.4319,0.9983); rgb(57pt)=(0.2424,0.4378,0.9991); rgb(58pt)=(0.2369,0.4437,0.9996); rgb(59pt)=(0.2311,0.4497,0.9995); rgb(60pt)=(0.225,0.4559,0.9985); rgb(61pt)=(0.2189,0.462,0.9968); rgb(62pt)=(0.2128,0.4682,0.9948); rgb(63pt)=(0.2066,0.4743,0.9926); rgb(64pt)=(0.2006,0.4803,0.9906); rgb(65pt)=(0.195,0.4861,0.9887); rgb(66pt)=(0.1903,0.4919,0.9867); rgb(67pt)=(0.1869,0.4975,0.9844); rgb(68pt)=(0.1847,0.503,0.9819); rgb(69pt)=(0.1831,0.5084,0.9793); rgb(70pt)=(0.1818,0.5138,0.9766); rgb(71pt)=(0.1806,0.5191,0.9738); rgb(72pt)=(0.1795,0.5244,0.9709); rgb(73pt)=(0.1785,0.5296,0.9677); rgb(74pt)=(0.1778,0.5349,0.9641); rgb(75pt)=(0.1773,0.5401,0.9602); rgb(76pt)=(0.1768,0.5452,0.956); rgb(77pt)=(0.1764,0.5504,0.9516); rgb(78pt)=(0.1755,0.5554,0.9473); rgb(79pt)=(0.174,0.5605,0.9432); rgb(80pt)=(0.1716,0.5655,0.9393); rgb(81pt)=(0.1686,0.5705,0.9357); rgb(82pt)=(0.1649,0.5755,0.9323); rgb(83pt)=(0.161,0.5805,0.9289); rgb(84pt)=(0.1573,0.5854,0.9254); rgb(85pt)=(0.154,0.5902,0.9218); rgb(86pt)=(0.1513,0.595,0.9182); rgb(87pt)=(0.1492,0.5997,0.9147); rgb(88pt)=(0.1475,0.6043,0.9113); rgb(89pt)=(0.1461,0.6089,0.908); rgb(90pt)=(0.1446,0.6135,0.905); rgb(91pt)=(0.1429,0.618,0.9022); rgb(92pt)=(0.1408,0.6226,0.8998); rgb(93pt)=(0.1383,0.6272,0.8975); rgb(94pt)=(0.1354,0.6317,0.8953); rgb(95pt)=(0.1321,0.6363,0.8932); rgb(96pt)=(0.1288,0.6408,0.891); rgb(97pt)=(0.1253,0.6453,0.8887); rgb(98pt)=(0.1219,0.6497,0.8862); rgb(99pt)=(0.1185,0.6541,0.8834); rgb(100pt)=(0.1152,0.6584,0.8804); rgb(101pt)=(0.1119,0.6627,0.877); rgb(102pt)=(0.1085,0.6669,0.8734); rgb(103pt)=(0.1048,0.671,0.8695); rgb(104pt)=(0.1009,0.675,0.8653); rgb(105pt)=(0.0964,0.6789,0.8609); rgb(106pt)=(0.0914,0.6828,0.8562); rgb(107pt)=(0.0855,0.6865,0.8513); rgb(108pt)=(0.0789,0.6902,0.8462); rgb(109pt)=(0.0713,0.6938,0.8409); rgb(110pt)=(0.0628,0.6972,0.8355); rgb(111pt)=(0.0535,0.7006,0.8299); rgb(112pt)=(0.0433,0.7039,0.8242); rgb(113pt)=(0.0328,0.7071,0.8183); rgb(114pt)=(0.0234,0.7103,0.8124); rgb(115pt)=(0.0155,0.7133,0.8064); rgb(116pt)=(0.0091,0.7163,0.8003); rgb(117pt)=(0.0046,0.7192,0.7941); rgb(118pt)=(0.0019,0.722,0.7878); rgb(119pt)=(0.0009,0.7248,0.7815); rgb(120pt)=(0.0018,0.7275,0.7752); rgb(121pt)=(0.0046,0.7301,0.7688); rgb(122pt)=(0.0094,0.7327,0.7623); rgb(123pt)=(0.0162,0.7352,0.7558); rgb(124pt)=(0.0253,0.7376,0.7492); rgb(125pt)=(0.0369,0.74,0.7426); rgb(126pt)=(0.0504,0.7423,0.7359); rgb(127pt)=(0.0638,0.7446,0.7292); rgb(128pt)=(0.077,0.7468,0.7224); rgb(129pt)=(0.0899,0.7489,0.7156); rgb(130pt)=(0.1023,0.751,0.7088); rgb(131pt)=(0.1141,0.7531,0.7019); rgb(132pt)=(0.1252,0.7552,0.695); rgb(133pt)=(0.1354,0.7572,0.6881); rgb(134pt)=(0.1448,0.7593,0.6812); rgb(135pt)=(0.1532,0.7614,0.6741); rgb(136pt)=(0.1609,0.7635,0.6671); rgb(137pt)=(0.1678,0.7656,0.6599); rgb(138pt)=(0.1741,0.7678,0.6527); rgb(139pt)=(0.1799,0.7699,0.6454); rgb(140pt)=(0.1853,0.7721,0.6379); rgb(141pt)=(0.1905,0.7743,0.6303); rgb(142pt)=(0.1954,0.7765,0.6225); rgb(143pt)=(0.2003,0.7787,0.6146); rgb(144pt)=(0.2061,0.7808,0.6065); rgb(145pt)=(0.2118,0.7828,0.5983); rgb(146pt)=(0.2178,0.7849,0.5899); rgb(147pt)=(0.2244,0.7869,0.5813); rgb(148pt)=(0.2318,0.7887,0.5725); rgb(149pt)=(0.2401,0.7905,0.5636); rgb(150pt)=(0.2491,0.7922,0.5546); rgb(151pt)=(0.2589,0.7937,0.5454); rgb(152pt)=(0.2695,0.7951,0.536); rgb(153pt)=(0.2809,0.7964,0.5266); rgb(154pt)=(0.2929,0.7975,0.517); rgb(155pt)=(0.3052,0.7985,0.5074); rgb(156pt)=(0.3176,0.7994,0.4975); rgb(157pt)=(0.3301,0.8002,0.4876); rgb(158pt)=(0.3424,0.8009,0.4774); rgb(159pt)=(0.3548,0.8016,0.4669); rgb(160pt)=(0.3671,0.8021,0.4563); rgb(161pt)=(0.3795,0.8026,0.4454); rgb(162pt)=(0.3921,0.8029,0.4344); rgb(163pt)=(0.405,0.8031,0.4233); rgb(164pt)=(0.4184,0.803,0.4122); rgb(165pt)=(0.4322,0.8028,0.4013); rgb(166pt)=(0.4463,0.8024,0.3904); rgb(167pt)=(0.4608,0.8018,0.3797); rgb(168pt)=(0.4753,0.8011,0.3691); rgb(169pt)=(0.4899,0.8002,0.3586); rgb(170pt)=(0.5044,0.7993,0.348); rgb(171pt)=(0.5187,0.7982,0.3374); rgb(172pt)=(0.5329,0.797,0.3267); rgb(173pt)=(0.547,0.7957,0.3159); rgb(175pt)=(0.5748,0.7929,0.2941); rgb(176pt)=(0.5886,0.7913,0.2833); rgb(177pt)=(0.6024,0.7896,0.2726); rgb(178pt)=(0.6161,0.7878,0.2622); rgb(179pt)=(0.6297,0.7859,0.2521); rgb(180pt)=(0.6433,0.7839,0.2423); rgb(181pt)=(0.6567,0.7818,0.2329); rgb(182pt)=(0.6701,0.7796,0.2239); rgb(183pt)=(0.6833,0.7773,0.2155); rgb(184pt)=(0.6963,0.775,0.2075); rgb(185pt)=(0.7091,0.7727,0.1998); rgb(186pt)=(0.7218,0.7703,0.1924); rgb(187pt)=(0.7344,0.7679,0.1852); rgb(188pt)=(0.7468,0.7654,0.1782); rgb(189pt)=(0.759,0.7629,0.1717); rgb(190pt)=(0.771,0.7604,0.1658); rgb(191pt)=(0.7829,0.7579,0.1608); rgb(192pt)=(0.7945,0.7554,0.157); rgb(193pt)=(0.806,0.7529,0.1546); rgb(194pt)=(0.8172,0.7505,0.1535); rgb(195pt)=(0.8281,0.7481,0.1536); rgb(196pt)=(0.8389,0.7457,0.1546); rgb(197pt)=(0.8495,0.7435,0.1564); rgb(198pt)=(0.86,0.7413,0.1587); rgb(199pt)=(0.8703,0.7392,0.1615); rgb(200pt)=(0.8804,0.7372,0.165); rgb(201pt)=(0.8903,0.7353,0.1695); rgb(202pt)=(0.9,0.7336,0.1749); rgb(203pt)=(0.9093,0.7321,0.1815); rgb(204pt)=(0.9184,0.7308,0.189); rgb(205pt)=(0.9272,0.7298,0.1973); rgb(206pt)=(0.9357,0.729,0.2061); rgb(207pt)=(0.944,0.7285,0.2151); rgb(208pt)=(0.9523,0.7284,0.2237); rgb(209pt)=(0.9606,0.7285,0.2312); rgb(210pt)=(0.9689,0.7292,0.2373); rgb(211pt)=(0.977,0.7304,0.2418); rgb(212pt)=(0.9842,0.733,0.2446); rgb(213pt)=(0.99,0.7365,0.2429); rgb(214pt)=(0.9946,0.7407,0.2394); rgb(215pt)=(0.9966,0.7458,0.2351); rgb(216pt)=(0.9971,0.7513,0.2309); rgb(217pt)=(0.9972,0.7569,0.2267); rgb(218pt)=(0.9971,0.7626,0.2224); rgb(219pt)=(0.9969,0.7683,0.2181); rgb(220pt)=(0.9966,0.774,0.2138); rgb(221pt)=(0.9962,0.7798,0.2095); rgb(222pt)=(0.9957,0.7856,0.2053); rgb(223pt)=(0.9949,0.7915,0.2012); rgb(224pt)=(0.9938,0.7974,0.1974); rgb(225pt)=(0.9923,0.8034,0.1939); rgb(226pt)=(0.9906,0.8095,0.1906); rgb(227pt)=(0.9885,0.8156,0.1875); rgb(228pt)=(0.9861,0.8218,0.1846); rgb(229pt)=(0.9835,0.828,0.1817); rgb(230pt)=(0.9807,0.8342,0.1787); rgb(231pt)=(0.9778,0.8404,0.1757); rgb(232pt)=(0.9748,0.8467,0.1726); rgb(233pt)=(0.972,0.8529,0.1695); rgb(234pt)=(0.9694,0.8591,0.1665); rgb(235pt)=(0.9671,0.8654,0.1636); rgb(236pt)=(0.9651,0.8716,0.1608); rgb(237pt)=(0.9634,0.8778,0.1582); rgb(238pt)=(0.9619,0.884,0.1557); rgb(239pt)=(0.9608,0.8902,0.1532); rgb(240pt)=(0.9601,0.8963,0.1507); rgb(241pt)=(0.9596,0.9023,0.148); rgb(242pt)=(0.9595,0.9084,0.145); rgb(243pt)=(0.9597,0.9143,0.1418); rgb(244pt)=(0.9601,0.9203,0.1382); rgb(245pt)=(0.9608,0.9262,0.1344); rgb(246pt)=(0.9618,0.932,0.1304); rgb(247pt)=(0.9629,0.9379,0.1261); rgb(248pt)=(0.9642,0.9437,0.1216); rgb(249pt)=(0.9657,0.9494,0.1168); rgb(250pt)=(0.9674,0.9552,0.1116); rgb(251pt)=(0.9692,0.9609,0.1061); rgb(252pt)=(0.9711,0.9667,0.1001); rgb(253pt)=(0.973,0.9724,0.0938); rgb(254pt)=(0.9749,0.9782,0.0872); rgb(255pt)=(0.9769,0.9839,0.0805)}
]
\addplot [forget plot] graphics [xmin=70.004215851602, xmax=6004.9957841484, ymin=0.5, ymax=60.5] {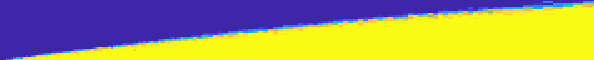};
\addplot [color=red, line width=1.0pt, forget plot]
  table[row sep=crcr]{%
1830	60\\
1827	58\\
1820	56\\
1809	54\\
1794	52\\
1775	50\\
1752	48\\
1725	46\\
1694	44\\
1659	42\\
1620	40\\
1577	38\\
1530	36\\
1479	34\\
1424	32\\
1365	30\\
1302	28\\
1235	26\\
1164	24\\
1089	22\\
1010	20\\
927	18\\
840	16\\
749	14\\
654	12\\
555	10\\
452	8\\
345	6\\
234	4\\
119	2\\
60	1\\
};
\addplot [color=mycolor1, dashdotted, line width=2.0pt, forget plot]
  table[row sep=crcr]{%
5400	60\\
90	1\\
};
\end{axis}
\end{tikzpicture}%

%% file: experiment_GC_ER_good_PT_uniform_T_7.tex
%
%
\definecolor{mycolor1}{rgb}{1.00000,0.00000,1.00000}%
\begin{tikzpicture}

\begin{axis}[%
width=0.859\figurewidth,
height=\figureheight,
at={(0\figurewidth,0\figureheight)},
scale only axis,
point meta min=0,
point meta max=1,
axis on top,
xmin=62.5,
xmax=6012.5,
xlabel={$m$},
xticklabels={,,2000,4000,6000},
ymin=0.5,
ymax=60.5,
axis background/.style={fill=white},
xlabel style={font=\tiny},ylabel style={font=\tiny},
xticklabel style = {font=\tiny}, yticklabel style = {font=\tiny},
colormap={mymap}{[1pt] rgb(0pt)=(0.2422,0.1504,0.6603); rgb(1pt)=(0.2444,0.1534,0.6728); rgb(2pt)=(0.2464,0.1569,0.6847); rgb(3pt)=(0.2484,0.1607,0.6961); rgb(4pt)=(0.2503,0.1648,0.7071); rgb(5pt)=(0.2522,0.1689,0.7179); rgb(6pt)=(0.254,0.1732,0.7286); rgb(7pt)=(0.2558,0.1773,0.7393); rgb(8pt)=(0.2576,0.1814,0.7501); rgb(9pt)=(0.2594,0.1854,0.761); rgb(11pt)=(0.2628,0.1932,0.7828); rgb(12pt)=(0.2645,0.1972,0.7937); rgb(13pt)=(0.2661,0.2011,0.8043); rgb(14pt)=(0.2676,0.2052,0.8148); rgb(15pt)=(0.2691,0.2094,0.8249); rgb(16pt)=(0.2704,0.2138,0.8346); rgb(17pt)=(0.2717,0.2184,0.8439); rgb(18pt)=(0.2729,0.2231,0.8528); rgb(19pt)=(0.274,0.228,0.8612); rgb(20pt)=(0.2749,0.233,0.8692); rgb(21pt)=(0.2758,0.2382,0.8767); rgb(22pt)=(0.2766,0.2435,0.884); rgb(23pt)=(0.2774,0.2489,0.8908); rgb(24pt)=(0.2781,0.2543,0.8973); rgb(25pt)=(0.2788,0.2598,0.9035); rgb(26pt)=(0.2794,0.2653,0.9094); rgb(27pt)=(0.2798,0.2708,0.915); rgb(28pt)=(0.2802,0.2764,0.9204); rgb(29pt)=(0.2806,0.2819,0.9255); rgb(30pt)=(0.2809,0.2875,0.9305); rgb(31pt)=(0.2811,0.293,0.9352); rgb(32pt)=(0.2813,0.2985,0.9397); rgb(33pt)=(0.2814,0.304,0.9441); rgb(34pt)=(0.2814,0.3095,0.9483); rgb(35pt)=(0.2813,0.315,0.9524); rgb(36pt)=(0.2811,0.3204,0.9563); rgb(37pt)=(0.2809,0.3259,0.96); rgb(38pt)=(0.2807,0.3313,0.9636); rgb(39pt)=(0.2803,0.3367,0.967); rgb(40pt)=(0.2798,0.3421,0.9702); rgb(41pt)=(0.2791,0.3475,0.9733); rgb(42pt)=(0.2784,0.3529,0.9763); rgb(43pt)=(0.2776,0.3583,0.9791); rgb(44pt)=(0.2766,0.3638,0.9817); rgb(45pt)=(0.2754,0.3693,0.984); rgb(46pt)=(0.2741,0.3748,0.9862); rgb(47pt)=(0.2726,0.3804,0.9881); rgb(48pt)=(0.271,0.386,0.9898); rgb(49pt)=(0.2691,0.3916,0.9912); rgb(50pt)=(0.267,0.3973,0.9924); rgb(51pt)=(0.2647,0.403,0.9935); rgb(52pt)=(0.2621,0.4088,0.9946); rgb(53pt)=(0.2591,0.4145,0.9955); rgb(54pt)=(0.2556,0.4203,0.9965); rgb(55pt)=(0.2517,0.4261,0.9974); rgb(56pt)=(0.2473,0.4319,0.9983); rgb(57pt)=(0.2424,0.4378,0.9991); rgb(58pt)=(0.2369,0.4437,0.9996); rgb(59pt)=(0.2311,0.4497,0.9995); rgb(60pt)=(0.225,0.4559,0.9985); rgb(61pt)=(0.2189,0.462,0.9968); rgb(62pt)=(0.2128,0.4682,0.9948); rgb(63pt)=(0.2066,0.4743,0.9926); rgb(64pt)=(0.2006,0.4803,0.9906); rgb(65pt)=(0.195,0.4861,0.9887); rgb(66pt)=(0.1903,0.4919,0.9867); rgb(67pt)=(0.1869,0.4975,0.9844); rgb(68pt)=(0.1847,0.503,0.9819); rgb(69pt)=(0.1831,0.5084,0.9793); rgb(70pt)=(0.1818,0.5138,0.9766); rgb(71pt)=(0.1806,0.5191,0.9738); rgb(72pt)=(0.1795,0.5244,0.9709); rgb(73pt)=(0.1785,0.5296,0.9677); rgb(74pt)=(0.1778,0.5349,0.9641); rgb(75pt)=(0.1773,0.5401,0.9602); rgb(76pt)=(0.1768,0.5452,0.956); rgb(77pt)=(0.1764,0.5504,0.9516); rgb(78pt)=(0.1755,0.5554,0.9473); rgb(79pt)=(0.174,0.5605,0.9432); rgb(80pt)=(0.1716,0.5655,0.9393); rgb(81pt)=(0.1686,0.5705,0.9357); rgb(82pt)=(0.1649,0.5755,0.9323); rgb(83pt)=(0.161,0.5805,0.9289); rgb(84pt)=(0.1573,0.5854,0.9254); rgb(85pt)=(0.154,0.5902,0.9218); rgb(86pt)=(0.1513,0.595,0.9182); rgb(87pt)=(0.1492,0.5997,0.9147); rgb(88pt)=(0.1475,0.6043,0.9113); rgb(89pt)=(0.1461,0.6089,0.908); rgb(90pt)=(0.1446,0.6135,0.905); rgb(91pt)=(0.1429,0.618,0.9022); rgb(92pt)=(0.1408,0.6226,0.8998); rgb(93pt)=(0.1383,0.6272,0.8975); rgb(94pt)=(0.1354,0.6317,0.8953); rgb(95pt)=(0.1321,0.6363,0.8932); rgb(96pt)=(0.1288,0.6408,0.891); rgb(97pt)=(0.1253,0.6453,0.8887); rgb(98pt)=(0.1219,0.6497,0.8862); rgb(99pt)=(0.1185,0.6541,0.8834); rgb(100pt)=(0.1152,0.6584,0.8804); rgb(101pt)=(0.1119,0.6627,0.877); rgb(102pt)=(0.1085,0.6669,0.8734); rgb(103pt)=(0.1048,0.671,0.8695); rgb(104pt)=(0.1009,0.675,0.8653); rgb(105pt)=(0.0964,0.6789,0.8609); rgb(106pt)=(0.0914,0.6828,0.8562); rgb(107pt)=(0.0855,0.6865,0.8513); rgb(108pt)=(0.0789,0.6902,0.8462); rgb(109pt)=(0.0713,0.6938,0.8409); rgb(110pt)=(0.0628,0.6972,0.8355); rgb(111pt)=(0.0535,0.7006,0.8299); rgb(112pt)=(0.0433,0.7039,0.8242); rgb(113pt)=(0.0328,0.7071,0.8183); rgb(114pt)=(0.0234,0.7103,0.8124); rgb(115pt)=(0.0155,0.7133,0.8064); rgb(116pt)=(0.0091,0.7163,0.8003); rgb(117pt)=(0.0046,0.7192,0.7941); rgb(118pt)=(0.0019,0.722,0.7878); rgb(119pt)=(0.0009,0.7248,0.7815); rgb(120pt)=(0.0018,0.7275,0.7752); rgb(121pt)=(0.0046,0.7301,0.7688); rgb(122pt)=(0.0094,0.7327,0.7623); rgb(123pt)=(0.0162,0.7352,0.7558); rgb(124pt)=(0.0253,0.7376,0.7492); rgb(125pt)=(0.0369,0.74,0.7426); rgb(126pt)=(0.0504,0.7423,0.7359); rgb(127pt)=(0.0638,0.7446,0.7292); rgb(128pt)=(0.077,0.7468,0.7224); rgb(129pt)=(0.0899,0.7489,0.7156); rgb(130pt)=(0.1023,0.751,0.7088); rgb(131pt)=(0.1141,0.7531,0.7019); rgb(132pt)=(0.1252,0.7552,0.695); rgb(133pt)=(0.1354,0.7572,0.6881); rgb(134pt)=(0.1448,0.7593,0.6812); rgb(135pt)=(0.1532,0.7614,0.6741); rgb(136pt)=(0.1609,0.7635,0.6671); rgb(137pt)=(0.1678,0.7656,0.6599); rgb(138pt)=(0.1741,0.7678,0.6527); rgb(139pt)=(0.1799,0.7699,0.6454); rgb(140pt)=(0.1853,0.7721,0.6379); rgb(141pt)=(0.1905,0.7743,0.6303); rgb(142pt)=(0.1954,0.7765,0.6225); rgb(143pt)=(0.2003,0.7787,0.6146); rgb(144pt)=(0.2061,0.7808,0.6065); rgb(145pt)=(0.2118,0.7828,0.5983); rgb(146pt)=(0.2178,0.7849,0.5899); rgb(147pt)=(0.2244,0.7869,0.5813); rgb(148pt)=(0.2318,0.7887,0.5725); rgb(149pt)=(0.2401,0.7905,0.5636); rgb(150pt)=(0.2491,0.7922,0.5546); rgb(151pt)=(0.2589,0.7937,0.5454); rgb(152pt)=(0.2695,0.7951,0.536); rgb(153pt)=(0.2809,0.7964,0.5266); rgb(154pt)=(0.2929,0.7975,0.517); rgb(155pt)=(0.3052,0.7985,0.5074); rgb(156pt)=(0.3176,0.7994,0.4975); rgb(157pt)=(0.3301,0.8002,0.4876); rgb(158pt)=(0.3424,0.8009,0.4774); rgb(159pt)=(0.3548,0.8016,0.4669); rgb(160pt)=(0.3671,0.8021,0.4563); rgb(161pt)=(0.3795,0.8026,0.4454); rgb(162pt)=(0.3921,0.8029,0.4344); rgb(163pt)=(0.405,0.8031,0.4233); rgb(164pt)=(0.4184,0.803,0.4122); rgb(165pt)=(0.4322,0.8028,0.4013); rgb(166pt)=(0.4463,0.8024,0.3904); rgb(167pt)=(0.4608,0.8018,0.3797); rgb(168pt)=(0.4753,0.8011,0.3691); rgb(169pt)=(0.4899,0.8002,0.3586); rgb(170pt)=(0.5044,0.7993,0.348); rgb(171pt)=(0.5187,0.7982,0.3374); rgb(172pt)=(0.5329,0.797,0.3267); rgb(173pt)=(0.547,0.7957,0.3159); rgb(175pt)=(0.5748,0.7929,0.2941); rgb(176pt)=(0.5886,0.7913,0.2833); rgb(177pt)=(0.6024,0.7896,0.2726); rgb(178pt)=(0.6161,0.7878,0.2622); rgb(179pt)=(0.6297,0.7859,0.2521); rgb(180pt)=(0.6433,0.7839,0.2423); rgb(181pt)=(0.6567,0.7818,0.2329); rgb(182pt)=(0.6701,0.7796,0.2239); rgb(183pt)=(0.6833,0.7773,0.2155); rgb(184pt)=(0.6963,0.775,0.2075); rgb(185pt)=(0.7091,0.7727,0.1998); rgb(186pt)=(0.7218,0.7703,0.1924); rgb(187pt)=(0.7344,0.7679,0.1852); rgb(188pt)=(0.7468,0.7654,0.1782); rgb(189pt)=(0.759,0.7629,0.1717); rgb(190pt)=(0.771,0.7604,0.1658); rgb(191pt)=(0.7829,0.7579,0.1608); rgb(192pt)=(0.7945,0.7554,0.157); rgb(193pt)=(0.806,0.7529,0.1546); rgb(194pt)=(0.8172,0.7505,0.1535); rgb(195pt)=(0.8281,0.7481,0.1536); rgb(196pt)=(0.8389,0.7457,0.1546); rgb(197pt)=(0.8495,0.7435,0.1564); rgb(198pt)=(0.86,0.7413,0.1587); rgb(199pt)=(0.8703,0.7392,0.1615); rgb(200pt)=(0.8804,0.7372,0.165); rgb(201pt)=(0.8903,0.7353,0.1695); rgb(202pt)=(0.9,0.7336,0.1749); rgb(203pt)=(0.9093,0.7321,0.1815); rgb(204pt)=(0.9184,0.7308,0.189); rgb(205pt)=(0.9272,0.7298,0.1973); rgb(206pt)=(0.9357,0.729,0.2061); rgb(207pt)=(0.944,0.7285,0.2151); rgb(208pt)=(0.9523,0.7284,0.2237); rgb(209pt)=(0.9606,0.7285,0.2312); rgb(210pt)=(0.9689,0.7292,0.2373); rgb(211pt)=(0.977,0.7304,0.2418); rgb(212pt)=(0.9842,0.733,0.2446); rgb(213pt)=(0.99,0.7365,0.2429); rgb(214pt)=(0.9946,0.7407,0.2394); rgb(215pt)=(0.9966,0.7458,0.2351); rgb(216pt)=(0.9971,0.7513,0.2309); rgb(217pt)=(0.9972,0.7569,0.2267); rgb(218pt)=(0.9971,0.7626,0.2224); rgb(219pt)=(0.9969,0.7683,0.2181); rgb(220pt)=(0.9966,0.774,0.2138); rgb(221pt)=(0.9962,0.7798,0.2095); rgb(222pt)=(0.9957,0.7856,0.2053); rgb(223pt)=(0.9949,0.7915,0.2012); rgb(224pt)=(0.9938,0.7974,0.1974); rgb(225pt)=(0.9923,0.8034,0.1939); rgb(226pt)=(0.9906,0.8095,0.1906); rgb(227pt)=(0.9885,0.8156,0.1875); rgb(228pt)=(0.9861,0.8218,0.1846); rgb(229pt)=(0.9835,0.828,0.1817); rgb(230pt)=(0.9807,0.8342,0.1787); rgb(231pt)=(0.9778,0.8404,0.1757); rgb(232pt)=(0.9748,0.8467,0.1726); rgb(233pt)=(0.972,0.8529,0.1695); rgb(234pt)=(0.9694,0.8591,0.1665); rgb(235pt)=(0.9671,0.8654,0.1636); rgb(236pt)=(0.9651,0.8716,0.1608); rgb(237pt)=(0.9634,0.8778,0.1582); rgb(238pt)=(0.9619,0.884,0.1557); rgb(239pt)=(0.9608,0.8902,0.1532); rgb(240pt)=(0.9601,0.8963,0.1507); rgb(241pt)=(0.9596,0.9023,0.148); rgb(242pt)=(0.9595,0.9084,0.145); rgb(243pt)=(0.9597,0.9143,0.1418); rgb(244pt)=(0.9601,0.9203,0.1382); rgb(245pt)=(0.9608,0.9262,0.1344); rgb(246pt)=(0.9618,0.932,0.1304); rgb(247pt)=(0.9629,0.9379,0.1261); rgb(248pt)=(0.9642,0.9437,0.1216); rgb(249pt)=(0.9657,0.9494,0.1168); rgb(250pt)=(0.9674,0.9552,0.1116); rgb(251pt)=(0.9692,0.9609,0.1061); rgb(252pt)=(0.9711,0.9667,0.1001); rgb(253pt)=(0.973,0.9724,0.0938); rgb(254pt)=(0.9749,0.9782,0.0872); rgb(255pt)=(0.9769,0.9839,0.0805)}
]
\addplot [forget plot] graphics [xmin=62.5, xmax=6012.5, ymin=0.5, ymax=60.5] {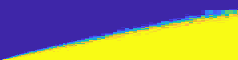};
\addplot [color=red, line width=1.0pt, forget plot]
  table[row sep=crcr]{%
1830	60\\
1829	59\\
1827	58\\
1824	57\\
1820	56\\
1815	55\\
1809	54\\
1802	53\\
1794	52\\
1785	51\\
1739	47\\
1725	46\\
1710	45\\
1694	44\\
1677	43\\
1659	42\\
1640	41\\
1620	40\\
1599	39\\
1577	38\\
1554	37\\
1530	36\\
1505	35\\
1479	34\\
1452	33\\
1424	32\\
1395	31\\
1365	30\\
1334	29\\
1302	28\\
1269	27\\
1235	26\\
1200	25\\
1164	24\\
1127	23\\
1089	22\\
1050	21\\
1010	20\\
969	19\\
927	18\\
884	17\\
840	16\\
795	15\\
749	14\\
702	13\\
654	12\\
605	11\\
555	10\\
504	9\\
452	8\\
399	7\\
345	6\\
290	5\\
234	4\\
177	3\\
119	2\\
60	1\\
};
\addplot [color=mycolor1, dashdotted, line width=2.0pt, forget plot]
  table[row sep=crcr]{%
5400	60\\
90	1\\
};
\end{axis}
\end{tikzpicture}%

%% file: experiment_GC_ER_good_PT_uniform_T_7_pencil_1.tex
%
%
\definecolor{mycolor1}{rgb}{1.00000,0.00000,1.00000}%
\begin{tikzpicture}

\begin{axis}[%
width=0.859\figurewidth,
height=0.955\figureheight,
at={(0\figurewidth,0\figureheight)},
scale only axis,
point meta min=0,
point meta max=1,
axis on top,
xmin=62.5,
xmax=6012.5,
xlabel style={font=\color{white!15!black}},
xlabel={Nr. of samples $m$},
xticklabels={,,2000,4000,6000},
ymin=0.5,
ymax=60.5,
ylabel style={font=\color{white!15!black}},
ylabel={Rank $r$},
ylabel near ticks,
axis background/.style={fill=white},
legend style={at={(0.5,0.97)}, anchor=north, legend cell align=left, align=left, draw=white!15!black},
xlabel style={font=\tiny},ylabel style={font=\tiny, anchor=near ticklabel},
xticklabel style = {font=\tiny}, yticklabel style = {font=\tiny},
colormap={mymap}{[1pt] rgb(0pt)=(0.2422,0.1504,0.6603); rgb(1pt)=(0.2444,0.1534,0.6728); rgb(2pt)=(0.2464,0.1569,0.6847); rgb(3pt)=(0.2484,0.1607,0.6961); rgb(4pt)=(0.2503,0.1648,0.7071); rgb(5pt)=(0.2522,0.1689,0.7179); rgb(6pt)=(0.254,0.1732,0.7286); rgb(7pt)=(0.2558,0.1773,0.7393); rgb(8pt)=(0.2576,0.1814,0.7501); rgb(9pt)=(0.2594,0.1854,0.761); rgb(11pt)=(0.2628,0.1932,0.7828); rgb(12pt)=(0.2645,0.1972,0.7937); rgb(13pt)=(0.2661,0.2011,0.8043); rgb(14pt)=(0.2676,0.2052,0.8148); rgb(15pt)=(0.2691,0.2094,0.8249); rgb(16pt)=(0.2704,0.2138,0.8346); rgb(17pt)=(0.2717,0.2184,0.8439); rgb(18pt)=(0.2729,0.2231,0.8528); rgb(19pt)=(0.274,0.228,0.8612); rgb(20pt)=(0.2749,0.233,0.8692); rgb(21pt)=(0.2758,0.2382,0.8767); rgb(22pt)=(0.2766,0.2435,0.884); rgb(23pt)=(0.2774,0.2489,0.8908); rgb(24pt)=(0.2781,0.2543,0.8973); rgb(25pt)=(0.2788,0.2598,0.9035); rgb(26pt)=(0.2794,0.2653,0.9094); rgb(27pt)=(0.2798,0.2708,0.915); rgb(28pt)=(0.2802,0.2764,0.9204); rgb(29pt)=(0.2806,0.2819,0.9255); rgb(30pt)=(0.2809,0.2875,0.9305); rgb(31pt)=(0.2811,0.293,0.9352); rgb(32pt)=(0.2813,0.2985,0.9397); rgb(33pt)=(0.2814,0.304,0.9441); rgb(34pt)=(0.2814,0.3095,0.9483); rgb(35pt)=(0.2813,0.315,0.9524); rgb(36pt)=(0.2811,0.3204,0.9563); rgb(37pt)=(0.2809,0.3259,0.96); rgb(38pt)=(0.2807,0.3313,0.9636); rgb(39pt)=(0.2803,0.3367,0.967); rgb(40pt)=(0.2798,0.3421,0.9702); rgb(41pt)=(0.2791,0.3475,0.9733); rgb(42pt)=(0.2784,0.3529,0.9763); rgb(43pt)=(0.2776,0.3583,0.9791); rgb(44pt)=(0.2766,0.3638,0.9817); rgb(45pt)=(0.2754,0.3693,0.984); rgb(46pt)=(0.2741,0.3748,0.9862); rgb(47pt)=(0.2726,0.3804,0.9881); rgb(48pt)=(0.271,0.386,0.9898); rgb(49pt)=(0.2691,0.3916,0.9912); rgb(50pt)=(0.267,0.3973,0.9924); rgb(51pt)=(0.2647,0.403,0.9935); rgb(52pt)=(0.2621,0.4088,0.9946); rgb(53pt)=(0.2591,0.4145,0.9955); rgb(54pt)=(0.2556,0.4203,0.9965); rgb(55pt)=(0.2517,0.4261,0.9974); rgb(56pt)=(0.2473,0.4319,0.9983); rgb(57pt)=(0.2424,0.4378,0.9991); rgb(58pt)=(0.2369,0.4437,0.9996); rgb(59pt)=(0.2311,0.4497,0.9995); rgb(60pt)=(0.225,0.4559,0.9985); rgb(61pt)=(0.2189,0.462,0.9968); rgb(62pt)=(0.2128,0.4682,0.9948); rgb(63pt)=(0.2066,0.4743,0.9926); rgb(64pt)=(0.2006,0.4803,0.9906); rgb(65pt)=(0.195,0.4861,0.9887); rgb(66pt)=(0.1903,0.4919,0.9867); rgb(67pt)=(0.1869,0.4975,0.9844); rgb(68pt)=(0.1847,0.503,0.9819); rgb(69pt)=(0.1831,0.5084,0.9793); rgb(70pt)=(0.1818,0.5138,0.9766); rgb(71pt)=(0.1806,0.5191,0.9738); rgb(72pt)=(0.1795,0.5244,0.9709); rgb(73pt)=(0.1785,0.5296,0.9677); rgb(74pt)=(0.1778,0.5349,0.9641); rgb(75pt)=(0.1773,0.5401,0.9602); rgb(76pt)=(0.1768,0.5452,0.956); rgb(77pt)=(0.1764,0.5504,0.9516); rgb(78pt)=(0.1755,0.5554,0.9473); rgb(79pt)=(0.174,0.5605,0.9432); rgb(80pt)=(0.1716,0.5655,0.9393); rgb(81pt)=(0.1686,0.5705,0.9357); rgb(82pt)=(0.1649,0.5755,0.9323); rgb(83pt)=(0.161,0.5805,0.9289); rgb(84pt)=(0.1573,0.5854,0.9254); rgb(85pt)=(0.154,0.5902,0.9218); rgb(86pt)=(0.1513,0.595,0.9182); rgb(87pt)=(0.1492,0.5997,0.9147); rgb(88pt)=(0.1475,0.6043,0.9113); rgb(89pt)=(0.1461,0.6089,0.908); rgb(90pt)=(0.1446,0.6135,0.905); rgb(91pt)=(0.1429,0.618,0.9022); rgb(92pt)=(0.1408,0.6226,0.8998); rgb(93pt)=(0.1383,0.6272,0.8975); rgb(94pt)=(0.1354,0.6317,0.8953); rgb(95pt)=(0.1321,0.6363,0.8932); rgb(96pt)=(0.1288,0.6408,0.891); rgb(97pt)=(0.1253,0.6453,0.8887); rgb(98pt)=(0.1219,0.6497,0.8862); rgb(99pt)=(0.1185,0.6541,0.8834); rgb(100pt)=(0.1152,0.6584,0.8804); rgb(101pt)=(0.1119,0.6627,0.877); rgb(102pt)=(0.1085,0.6669,0.8734); rgb(103pt)=(0.1048,0.671,0.8695); rgb(104pt)=(0.1009,0.675,0.8653); rgb(105pt)=(0.0964,0.6789,0.8609); rgb(106pt)=(0.0914,0.6828,0.8562); rgb(107pt)=(0.0855,0.6865,0.8513); rgb(108pt)=(0.0789,0.6902,0.8462); rgb(109pt)=(0.0713,0.6938,0.8409); rgb(110pt)=(0.0628,0.6972,0.8355); rgb(111pt)=(0.0535,0.7006,0.8299); rgb(112pt)=(0.0433,0.7039,0.8242); rgb(113pt)=(0.0328,0.7071,0.8183); rgb(114pt)=(0.0234,0.7103,0.8124); rgb(115pt)=(0.0155,0.7133,0.8064); rgb(116pt)=(0.0091,0.7163,0.8003); rgb(117pt)=(0.0046,0.7192,0.7941); rgb(118pt)=(0.0019,0.722,0.7878); rgb(119pt)=(0.0009,0.7248,0.7815); rgb(120pt)=(0.0018,0.7275,0.7752); rgb(121pt)=(0.0046,0.7301,0.7688); rgb(122pt)=(0.0094,0.7327,0.7623); rgb(123pt)=(0.0162,0.7352,0.7558); rgb(124pt)=(0.0253,0.7376,0.7492); rgb(125pt)=(0.0369,0.74,0.7426); rgb(126pt)=(0.0504,0.7423,0.7359); rgb(127pt)=(0.0638,0.7446,0.7292); rgb(128pt)=(0.077,0.7468,0.7224); rgb(129pt)=(0.0899,0.7489,0.7156); rgb(130pt)=(0.1023,0.751,0.7088); rgb(131pt)=(0.1141,0.7531,0.7019); rgb(132pt)=(0.1252,0.7552,0.695); rgb(133pt)=(0.1354,0.7572,0.6881); rgb(134pt)=(0.1448,0.7593,0.6812); rgb(135pt)=(0.1532,0.7614,0.6741); rgb(136pt)=(0.1609,0.7635,0.6671); rgb(137pt)=(0.1678,0.7656,0.6599); rgb(138pt)=(0.1741,0.7678,0.6527); rgb(139pt)=(0.1799,0.7699,0.6454); rgb(140pt)=(0.1853,0.7721,0.6379); rgb(141pt)=(0.1905,0.7743,0.6303); rgb(142pt)=(0.1954,0.7765,0.6225); rgb(143pt)=(0.2003,0.7787,0.6146); rgb(144pt)=(0.2061,0.7808,0.6065); rgb(145pt)=(0.2118,0.7828,0.5983); rgb(146pt)=(0.2178,0.7849,0.5899); rgb(147pt)=(0.2244,0.7869,0.5813); rgb(148pt)=(0.2318,0.7887,0.5725); rgb(149pt)=(0.2401,0.7905,0.5636); rgb(150pt)=(0.2491,0.7922,0.5546); rgb(151pt)=(0.2589,0.7937,0.5454); rgb(152pt)=(0.2695,0.7951,0.536); rgb(153pt)=(0.2809,0.7964,0.5266); rgb(154pt)=(0.2929,0.7975,0.517); rgb(155pt)=(0.3052,0.7985,0.5074); rgb(156pt)=(0.3176,0.7994,0.4975); rgb(157pt)=(0.3301,0.8002,0.4876); rgb(158pt)=(0.3424,0.8009,0.4774); rgb(159pt)=(0.3548,0.8016,0.4669); rgb(160pt)=(0.3671,0.8021,0.4563); rgb(161pt)=(0.3795,0.8026,0.4454); rgb(162pt)=(0.3921,0.8029,0.4344); rgb(163pt)=(0.405,0.8031,0.4233); rgb(164pt)=(0.4184,0.803,0.4122); rgb(165pt)=(0.4322,0.8028,0.4013); rgb(166pt)=(0.4463,0.8024,0.3904); rgb(167pt)=(0.4608,0.8018,0.3797); rgb(168pt)=(0.4753,0.8011,0.3691); rgb(169pt)=(0.4899,0.8002,0.3586); rgb(170pt)=(0.5044,0.7993,0.348); rgb(171pt)=(0.5187,0.7982,0.3374); rgb(172pt)=(0.5329,0.797,0.3267); rgb(173pt)=(0.547,0.7957,0.3159); rgb(175pt)=(0.5748,0.7929,0.2941); rgb(176pt)=(0.5886,0.7913,0.2833); rgb(177pt)=(0.6024,0.7896,0.2726); rgb(178pt)=(0.6161,0.7878,0.2622); rgb(179pt)=(0.6297,0.7859,0.2521); rgb(180pt)=(0.6433,0.7839,0.2423); rgb(181pt)=(0.6567,0.7818,0.2329); rgb(182pt)=(0.6701,0.7796,0.2239); rgb(183pt)=(0.6833,0.7773,0.2155); rgb(184pt)=(0.6963,0.775,0.2075); rgb(185pt)=(0.7091,0.7727,0.1998); rgb(186pt)=(0.7218,0.7703,0.1924); rgb(187pt)=(0.7344,0.7679,0.1852); rgb(188pt)=(0.7468,0.7654,0.1782); rgb(189pt)=(0.759,0.7629,0.1717); rgb(190pt)=(0.771,0.7604,0.1658); rgb(191pt)=(0.7829,0.7579,0.1608); rgb(192pt)=(0.7945,0.7554,0.157); rgb(193pt)=(0.806,0.7529,0.1546); rgb(194pt)=(0.8172,0.7505,0.1535); rgb(195pt)=(0.8281,0.7481,0.1536); rgb(196pt)=(0.8389,0.7457,0.1546); rgb(197pt)=(0.8495,0.7435,0.1564); rgb(198pt)=(0.86,0.7413,0.1587); rgb(199pt)=(0.8703,0.7392,0.1615); rgb(200pt)=(0.8804,0.7372,0.165); rgb(201pt)=(0.8903,0.7353,0.1695); rgb(202pt)=(0.9,0.7336,0.1749); rgb(203pt)=(0.9093,0.7321,0.1815); rgb(204pt)=(0.9184,0.7308,0.189); rgb(205pt)=(0.9272,0.7298,0.1973); rgb(206pt)=(0.9357,0.729,0.2061); rgb(207pt)=(0.944,0.7285,0.2151); rgb(208pt)=(0.9523,0.7284,0.2237); rgb(209pt)=(0.9606,0.7285,0.2312); rgb(210pt)=(0.9689,0.7292,0.2373); rgb(211pt)=(0.977,0.7304,0.2418); rgb(212pt)=(0.9842,0.733,0.2446); rgb(213pt)=(0.99,0.7365,0.2429); rgb(214pt)=(0.9946,0.7407,0.2394); rgb(215pt)=(0.9966,0.7458,0.2351); rgb(216pt)=(0.9971,0.7513,0.2309); rgb(217pt)=(0.9972,0.7569,0.2267); rgb(218pt)=(0.9971,0.7626,0.2224); rgb(219pt)=(0.9969,0.7683,0.2181); rgb(220pt)=(0.9966,0.774,0.2138); rgb(221pt)=(0.9962,0.7798,0.2095); rgb(222pt)=(0.9957,0.7856,0.2053); rgb(223pt)=(0.9949,0.7915,0.2012); rgb(224pt)=(0.9938,0.7974,0.1974); rgb(225pt)=(0.9923,0.8034,0.1939); rgb(226pt)=(0.9906,0.8095,0.1906); rgb(227pt)=(0.9885,0.8156,0.1875); rgb(228pt)=(0.9861,0.8218,0.1846); rgb(229pt)=(0.9835,0.828,0.1817); rgb(230pt)=(0.9807,0.8342,0.1787); rgb(231pt)=(0.9778,0.8404,0.1757); rgb(232pt)=(0.9748,0.8467,0.1726); rgb(233pt)=(0.972,0.8529,0.1695); rgb(234pt)=(0.9694,0.8591,0.1665); rgb(235pt)=(0.9671,0.8654,0.1636); rgb(236pt)=(0.9651,0.8716,0.1608); rgb(237pt)=(0.9634,0.8778,0.1582); rgb(238pt)=(0.9619,0.884,0.1557); rgb(239pt)=(0.9608,0.8902,0.1532); rgb(240pt)=(0.9601,0.8963,0.1507); rgb(241pt)=(0.9596,0.9023,0.148); rgb(242pt)=(0.9595,0.9084,0.145); rgb(243pt)=(0.9597,0.9143,0.1418); rgb(244pt)=(0.9601,0.9203,0.1382); rgb(245pt)=(0.9608,0.9262,0.1344); rgb(246pt)=(0.9618,0.932,0.1304); rgb(247pt)=(0.9629,0.9379,0.1261); rgb(248pt)=(0.9642,0.9437,0.1216); rgb(249pt)=(0.9657,0.9494,0.1168); rgb(250pt)=(0.9674,0.9552,0.1116); rgb(251pt)=(0.9692,0.9609,0.1061); rgb(252pt)=(0.9711,0.9667,0.1001); rgb(253pt)=(0.973,0.9724,0.0938); rgb(254pt)=(0.9749,0.9782,0.0872); rgb(255pt)=(0.9769,0.9839,0.0805)}
]
\addplot [forget plot] graphics [xmin=62.5, xmax=6012.5, ymin=0.5, ymax=60.5] {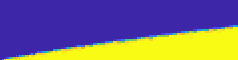};
\addplot [color=red, line width=1.0pt]
  table[row sep=crcr]{%
1830	60\\
1827	58\\
1820	56\\
1809	54\\
1794	52\\
1775	50\\
1752	48\\
1725	46\\
1694	44\\
1659	42\\
1620	40\\
1577	38\\
1530	36\\
1479	34\\
1424	32\\
1365	30\\
1302	28\\
1235	26\\
1164	24\\
1089	22\\
1010	20\\
927	18\\
840	16\\
749	14\\
654	12\\
555	10\\
452	8\\
345	6\\
234	4\\
119	2\\
60	1\\
};

\addplot [color=mycolor1, dashdotted, line width=2.0pt]
  table[row sep=crcr]{%
5400	60\\
90	1\\
};

\end{axis}

\begin{axis}[%
width=1.227\figurewidth,
height=1.227\figureheight,
at={(-0.144\figurewidth,-0.135\figureheight)},
scale only axis,
point meta min=0,
point meta max=1,
xmin=0,
xmax=1,
ymin=0,
ymax=1,
axis line style={draw=none},
ticks=none,
axis x line*=bottom,
axis y line*=left,
legend style={legend cell align=left, align=left, draw=white!15!black},
xlabel style={font=\tiny},ylabel style={font=\tiny},
]
\end{axis}
\end{tikzpicture}%

%% file: experiment_GC_ER_good_PT_uniform_T_7_pencil_2.tex
%
%
\definecolor{mycolor1}{rgb}{1.00000,0.00000,1.00000}%
\begin{tikzpicture}

\begin{axis}[%
width=0.859\figurewidth,
height=\figureheight,
at={(0.5\figurewidth,0.5\figureheight)},
scale only axis,
point meta min=0,
point meta max=1,
axis on top,
xmin=62.5,
xmax=6012.5,
xlabel style={font=\color{white!15!black}},
xlabel={$m$},
xticklabels={,,2000,4000,6000},
yticklabels={},
ymin=0.5,
ymax=60.5,
axis background/.style={fill=white},
xlabel style={font=\tiny},ylabel style={font=\tiny},
xticklabel style = {font=\tiny}, yticklabel style = {font=\tiny},
colormap={mymap}{[1pt] rgb(0pt)=(0.2422,0.1504,0.6603); rgb(1pt)=(0.2444,0.1534,0.6728); rgb(2pt)=(0.2464,0.1569,0.6847); rgb(3pt)=(0.2484,0.1607,0.6961); rgb(4pt)=(0.2503,0.1648,0.7071); rgb(5pt)=(0.2522,0.1689,0.7179); rgb(6pt)=(0.254,0.1732,0.7286); rgb(7pt)=(0.2558,0.1773,0.7393); rgb(8pt)=(0.2576,0.1814,0.7501); rgb(9pt)=(0.2594,0.1854,0.761); rgb(11pt)=(0.2628,0.1932,0.7828); rgb(12pt)=(0.2645,0.1972,0.7937); rgb(13pt)=(0.2661,0.2011,0.8043); rgb(14pt)=(0.2676,0.2052,0.8148); rgb(15pt)=(0.2691,0.2094,0.8249); rgb(16pt)=(0.2704,0.2138,0.8346); rgb(17pt)=(0.2717,0.2184,0.8439); rgb(18pt)=(0.2729,0.2231,0.8528); rgb(19pt)=(0.274,0.228,0.8612); rgb(20pt)=(0.2749,0.233,0.8692); rgb(21pt)=(0.2758,0.2382,0.8767); rgb(22pt)=(0.2766,0.2435,0.884); rgb(23pt)=(0.2774,0.2489,0.8908); rgb(24pt)=(0.2781,0.2543,0.8973); rgb(25pt)=(0.2788,0.2598,0.9035); rgb(26pt)=(0.2794,0.2653,0.9094); rgb(27pt)=(0.2798,0.2708,0.915); rgb(28pt)=(0.2802,0.2764,0.9204); rgb(29pt)=(0.2806,0.2819,0.9255); rgb(30pt)=(0.2809,0.2875,0.9305); rgb(31pt)=(0.2811,0.293,0.9352); rgb(32pt)=(0.2813,0.2985,0.9397); rgb(33pt)=(0.2814,0.304,0.9441); rgb(34pt)=(0.2814,0.3095,0.9483); rgb(35pt)=(0.2813,0.315,0.9524); rgb(36pt)=(0.2811,0.3204,0.9563); rgb(37pt)=(0.2809,0.3259,0.96); rgb(38pt)=(0.2807,0.3313,0.9636); rgb(39pt)=(0.2803,0.3367,0.967); rgb(40pt)=(0.2798,0.3421,0.9702); rgb(41pt)=(0.2791,0.3475,0.9733); rgb(42pt)=(0.2784,0.3529,0.9763); rgb(43pt)=(0.2776,0.3583,0.9791); rgb(44pt)=(0.2766,0.3638,0.9817); rgb(45pt)=(0.2754,0.3693,0.984); rgb(46pt)=(0.2741,0.3748,0.9862); rgb(47pt)=(0.2726,0.3804,0.9881); rgb(48pt)=(0.271,0.386,0.9898); rgb(49pt)=(0.2691,0.3916,0.9912); rgb(50pt)=(0.267,0.3973,0.9924); rgb(51pt)=(0.2647,0.403,0.9935); rgb(52pt)=(0.2621,0.4088,0.9946); rgb(53pt)=(0.2591,0.4145,0.9955); rgb(54pt)=(0.2556,0.4203,0.9965); rgb(55pt)=(0.2517,0.4261,0.9974); rgb(56pt)=(0.2473,0.4319,0.9983); rgb(57pt)=(0.2424,0.4378,0.9991); rgb(58pt)=(0.2369,0.4437,0.9996); rgb(59pt)=(0.2311,0.4497,0.9995); rgb(60pt)=(0.225,0.4559,0.9985); rgb(61pt)=(0.2189,0.462,0.9968); rgb(62pt)=(0.2128,0.4682,0.9948); rgb(63pt)=(0.2066,0.4743,0.9926); rgb(64pt)=(0.2006,0.4803,0.9906); rgb(65pt)=(0.195,0.4861,0.9887); rgb(66pt)=(0.1903,0.4919,0.9867); rgb(67pt)=(0.1869,0.4975,0.9844); rgb(68pt)=(0.1847,0.503,0.9819); rgb(69pt)=(0.1831,0.5084,0.9793); rgb(70pt)=(0.1818,0.5138,0.9766); rgb(71pt)=(0.1806,0.5191,0.9738); rgb(72pt)=(0.1795,0.5244,0.9709); rgb(73pt)=(0.1785,0.5296,0.9677); rgb(74pt)=(0.1778,0.5349,0.9641); rgb(75pt)=(0.1773,0.5401,0.9602); rgb(76pt)=(0.1768,0.5452,0.956); rgb(77pt)=(0.1764,0.5504,0.9516); rgb(78pt)=(0.1755,0.5554,0.9473); rgb(79pt)=(0.174,0.5605,0.9432); rgb(80pt)=(0.1716,0.5655,0.9393); rgb(81pt)=(0.1686,0.5705,0.9357); rgb(82pt)=(0.1649,0.5755,0.9323); rgb(83pt)=(0.161,0.5805,0.9289); rgb(84pt)=(0.1573,0.5854,0.9254); rgb(85pt)=(0.154,0.5902,0.9218); rgb(86pt)=(0.1513,0.595,0.9182); rgb(87pt)=(0.1492,0.5997,0.9147); rgb(88pt)=(0.1475,0.6043,0.9113); rgb(89pt)=(0.1461,0.6089,0.908); rgb(90pt)=(0.1446,0.6135,0.905); rgb(91pt)=(0.1429,0.618,0.9022); rgb(92pt)=(0.1408,0.6226,0.8998); rgb(93pt)=(0.1383,0.6272,0.8975); rgb(94pt)=(0.1354,0.6317,0.8953); rgb(95pt)=(0.1321,0.6363,0.8932); rgb(96pt)=(0.1288,0.6408,0.891); rgb(97pt)=(0.1253,0.6453,0.8887); rgb(98pt)=(0.1219,0.6497,0.8862); rgb(99pt)=(0.1185,0.6541,0.8834); rgb(100pt)=(0.1152,0.6584,0.8804); rgb(101pt)=(0.1119,0.6627,0.877); rgb(102pt)=(0.1085,0.6669,0.8734); rgb(103pt)=(0.1048,0.671,0.8695); rgb(104pt)=(0.1009,0.675,0.8653); rgb(105pt)=(0.0964,0.6789,0.8609); rgb(106pt)=(0.0914,0.6828,0.8562); rgb(107pt)=(0.0855,0.6865,0.8513); rgb(108pt)=(0.0789,0.6902,0.8462); rgb(109pt)=(0.0713,0.6938,0.8409); rgb(110pt)=(0.0628,0.6972,0.8355); rgb(111pt)=(0.0535,0.7006,0.8299); rgb(112pt)=(0.0433,0.7039,0.8242); rgb(113pt)=(0.0328,0.7071,0.8183); rgb(114pt)=(0.0234,0.7103,0.8124); rgb(115pt)=(0.0155,0.7133,0.8064); rgb(116pt)=(0.0091,0.7163,0.8003); rgb(117pt)=(0.0046,0.7192,0.7941); rgb(118pt)=(0.0019,0.722,0.7878); rgb(119pt)=(0.0009,0.7248,0.7815); rgb(120pt)=(0.0018,0.7275,0.7752); rgb(121pt)=(0.0046,0.7301,0.7688); rgb(122pt)=(0.0094,0.7327,0.7623); rgb(123pt)=(0.0162,0.7352,0.7558); rgb(124pt)=(0.0253,0.7376,0.7492); rgb(125pt)=(0.0369,0.74,0.7426); rgb(126pt)=(0.0504,0.7423,0.7359); rgb(127pt)=(0.0638,0.7446,0.7292); rgb(128pt)=(0.077,0.7468,0.7224); rgb(129pt)=(0.0899,0.7489,0.7156); rgb(130pt)=(0.1023,0.751,0.7088); rgb(131pt)=(0.1141,0.7531,0.7019); rgb(132pt)=(0.1252,0.7552,0.695); rgb(133pt)=(0.1354,0.7572,0.6881); rgb(134pt)=(0.1448,0.7593,0.6812); rgb(135pt)=(0.1532,0.7614,0.6741); rgb(136pt)=(0.1609,0.7635,0.6671); rgb(137pt)=(0.1678,0.7656,0.6599); rgb(138pt)=(0.1741,0.7678,0.6527); rgb(139pt)=(0.1799,0.7699,0.6454); rgb(140pt)=(0.1853,0.7721,0.6379); rgb(141pt)=(0.1905,0.7743,0.6303); rgb(142pt)=(0.1954,0.7765,0.6225); rgb(143pt)=(0.2003,0.7787,0.6146); rgb(144pt)=(0.2061,0.7808,0.6065); rgb(145pt)=(0.2118,0.7828,0.5983); rgb(146pt)=(0.2178,0.7849,0.5899); rgb(147pt)=(0.2244,0.7869,0.5813); rgb(148pt)=(0.2318,0.7887,0.5725); rgb(149pt)=(0.2401,0.7905,0.5636); rgb(150pt)=(0.2491,0.7922,0.5546); rgb(151pt)=(0.2589,0.7937,0.5454); rgb(152pt)=(0.2695,0.7951,0.536); rgb(153pt)=(0.2809,0.7964,0.5266); rgb(154pt)=(0.2929,0.7975,0.517); rgb(155pt)=(0.3052,0.7985,0.5074); rgb(156pt)=(0.3176,0.7994,0.4975); rgb(157pt)=(0.3301,0.8002,0.4876); rgb(158pt)=(0.3424,0.8009,0.4774); rgb(159pt)=(0.3548,0.8016,0.4669); rgb(160pt)=(0.3671,0.8021,0.4563); rgb(161pt)=(0.3795,0.8026,0.4454); rgb(162pt)=(0.3921,0.8029,0.4344); rgb(163pt)=(0.405,0.8031,0.4233); rgb(164pt)=(0.4184,0.803,0.4122); rgb(165pt)=(0.4322,0.8028,0.4013); rgb(166pt)=(0.4463,0.8024,0.3904); rgb(167pt)=(0.4608,0.8018,0.3797); rgb(168pt)=(0.4753,0.8011,0.3691); rgb(169pt)=(0.4899,0.8002,0.3586); rgb(170pt)=(0.5044,0.7993,0.348); rgb(171pt)=(0.5187,0.7982,0.3374); rgb(172pt)=(0.5329,0.797,0.3267); rgb(173pt)=(0.547,0.7957,0.3159); rgb(175pt)=(0.5748,0.7929,0.2941); rgb(176pt)=(0.5886,0.7913,0.2833); rgb(177pt)=(0.6024,0.7896,0.2726); rgb(178pt)=(0.6161,0.7878,0.2622); rgb(179pt)=(0.6297,0.7859,0.2521); rgb(180pt)=(0.6433,0.7839,0.2423); rgb(181pt)=(0.6567,0.7818,0.2329); rgb(182pt)=(0.6701,0.7796,0.2239); rgb(183pt)=(0.6833,0.7773,0.2155); rgb(184pt)=(0.6963,0.775,0.2075); rgb(185pt)=(0.7091,0.7727,0.1998); rgb(186pt)=(0.7218,0.7703,0.1924); rgb(187pt)=(0.7344,0.7679,0.1852); rgb(188pt)=(0.7468,0.7654,0.1782); rgb(189pt)=(0.759,0.7629,0.1717); rgb(190pt)=(0.771,0.7604,0.1658); rgb(191pt)=(0.7829,0.7579,0.1608); rgb(192pt)=(0.7945,0.7554,0.157); rgb(193pt)=(0.806,0.7529,0.1546); rgb(194pt)=(0.8172,0.7505,0.1535); rgb(195pt)=(0.8281,0.7481,0.1536); rgb(196pt)=(0.8389,0.7457,0.1546); rgb(197pt)=(0.8495,0.7435,0.1564); rgb(198pt)=(0.86,0.7413,0.1587); rgb(199pt)=(0.8703,0.7392,0.1615); rgb(200pt)=(0.8804,0.7372,0.165); rgb(201pt)=(0.8903,0.7353,0.1695); rgb(202pt)=(0.9,0.7336,0.1749); rgb(203pt)=(0.9093,0.7321,0.1815); rgb(204pt)=(0.9184,0.7308,0.189); rgb(205pt)=(0.9272,0.7298,0.1973); rgb(206pt)=(0.9357,0.729,0.2061); rgb(207pt)=(0.944,0.7285,0.2151); rgb(208pt)=(0.9523,0.7284,0.2237); rgb(209pt)=(0.9606,0.7285,0.2312); rgb(210pt)=(0.9689,0.7292,0.2373); rgb(211pt)=(0.977,0.7304,0.2418); rgb(212pt)=(0.9842,0.733,0.2446); rgb(213pt)=(0.99,0.7365,0.2429); rgb(214pt)=(0.9946,0.7407,0.2394); rgb(215pt)=(0.9966,0.7458,0.2351); rgb(216pt)=(0.9971,0.7513,0.2309); rgb(217pt)=(0.9972,0.7569,0.2267); rgb(218pt)=(0.9971,0.7626,0.2224); rgb(219pt)=(0.9969,0.7683,0.2181); rgb(220pt)=(0.9966,0.774,0.2138); rgb(221pt)=(0.9962,0.7798,0.2095); rgb(222pt)=(0.9957,0.7856,0.2053); rgb(223pt)=(0.9949,0.7915,0.2012); rgb(224pt)=(0.9938,0.7974,0.1974); rgb(225pt)=(0.9923,0.8034,0.1939); rgb(226pt)=(0.9906,0.8095,0.1906); rgb(227pt)=(0.9885,0.8156,0.1875); rgb(228pt)=(0.9861,0.8218,0.1846); rgb(229pt)=(0.9835,0.828,0.1817); rgb(230pt)=(0.9807,0.8342,0.1787); rgb(231pt)=(0.9778,0.8404,0.1757); rgb(232pt)=(0.9748,0.8467,0.1726); rgb(233pt)=(0.972,0.8529,0.1695); rgb(234pt)=(0.9694,0.8591,0.1665); rgb(235pt)=(0.9671,0.8654,0.1636); rgb(236pt)=(0.9651,0.8716,0.1608); rgb(237pt)=(0.9634,0.8778,0.1582); rgb(238pt)=(0.9619,0.884,0.1557); rgb(239pt)=(0.9608,0.8902,0.1532); rgb(240pt)=(0.9601,0.8963,0.1507); rgb(241pt)=(0.9596,0.9023,0.148); rgb(242pt)=(0.9595,0.9084,0.145); rgb(243pt)=(0.9597,0.9143,0.1418); rgb(244pt)=(0.9601,0.9203,0.1382); rgb(245pt)=(0.9608,0.9262,0.1344); rgb(246pt)=(0.9618,0.932,0.1304); rgb(247pt)=(0.9629,0.9379,0.1261); rgb(248pt)=(0.9642,0.9437,0.1216); rgb(249pt)=(0.9657,0.9494,0.1168); rgb(250pt)=(0.9674,0.9552,0.1116); rgb(251pt)=(0.9692,0.9609,0.1061); rgb(252pt)=(0.9711,0.9667,0.1001); rgb(253pt)=(0.973,0.9724,0.0938); rgb(254pt)=(0.9749,0.9782,0.0872); rgb(255pt)=(0.9769,0.9839,0.0805)},
]
\addplot [forget plot] graphics [xmin=62.5, xmax=6012.5, ymin=0.5, ymax=60.5] {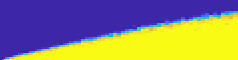};
\addplot [color=red, line width=1.0pt]
  table[row sep=crcr]{%
1830	60\\
1827	58\\
1820	56\\
1809	54\\
1794	52\\
1775	50\\
1752	48\\
1725	46\\
1694	44\\
1659	42\\
1620	40\\
1577	38\\
1530	36\\
1479	34\\
1424	32\\
1365	30\\
1302	28\\
1235	26\\
1164	24\\
1089	22\\
1010	20\\
927	18\\
840	16\\
749	14\\
654	12\\
555	10\\
452	8\\
345	6\\
234	4\\
119	2\\
60	1\\
};

\addplot [color=mycolor1, dashdotted, line width=2.0pt]
  table[row sep=crcr]{%
5400	60\\
90	1\\
};

\end{axis}
\end{tikzpicture}%

%% file: experiment_GC_ER_good_PT_uniform_T_7_pencil_3.tex
%
%
\definecolor{mycolor1}{rgb}{1.00000,0.00000,1.00000}%
\begin{tikzpicture}

\begin{axis}[%
width=0.859\figurewidth,
height=\figureheight,
at={(0\figurewidth,0\figureheight)},
scale only axis,
point meta min=0,
point meta max=1,
axis on top,
xmin=62.5,
xmax=6012.5,
xlabel style={font=\color{white!15!black}},
xlabel={$m$},
xticklabels={,,2000,4000,6000},
yticklabels={},
ymin=0.5,
ymax=60.5,
ylabel style={font=\color{white!15!black}},
axis background/.style={fill=white},
xlabel style={font=\tiny},ylabel style={font=\tiny},
xticklabel style = {font=\tiny}, yticklabel style = {font=\tiny},
colormap={mymap}{[1pt] rgb(0pt)=(0.2422,0.1504,0.6603); rgb(1pt)=(0.2444,0.1534,0.6728); rgb(2pt)=(0.2464,0.1569,0.6847); rgb(3pt)=(0.2484,0.1607,0.6961); rgb(4pt)=(0.2503,0.1648,0.7071); rgb(5pt)=(0.2522,0.1689,0.7179); rgb(6pt)=(0.254,0.1732,0.7286); rgb(7pt)=(0.2558,0.1773,0.7393); rgb(8pt)=(0.2576,0.1814,0.7501); rgb(9pt)=(0.2594,0.1854,0.761); rgb(11pt)=(0.2628,0.1932,0.7828); rgb(12pt)=(0.2645,0.1972,0.7937); rgb(13pt)=(0.2661,0.2011,0.8043); rgb(14pt)=(0.2676,0.2052,0.8148); rgb(15pt)=(0.2691,0.2094,0.8249); rgb(16pt)=(0.2704,0.2138,0.8346); rgb(17pt)=(0.2717,0.2184,0.8439); rgb(18pt)=(0.2729,0.2231,0.8528); rgb(19pt)=(0.274,0.228,0.8612); rgb(20pt)=(0.2749,0.233,0.8692); rgb(21pt)=(0.2758,0.2382,0.8767); rgb(22pt)=(0.2766,0.2435,0.884); rgb(23pt)=(0.2774,0.2489,0.8908); rgb(24pt)=(0.2781,0.2543,0.8973); rgb(25pt)=(0.2788,0.2598,0.9035); rgb(26pt)=(0.2794,0.2653,0.9094); rgb(27pt)=(0.2798,0.2708,0.915); rgb(28pt)=(0.2802,0.2764,0.9204); rgb(29pt)=(0.2806,0.2819,0.9255); rgb(30pt)=(0.2809,0.2875,0.9305); rgb(31pt)=(0.2811,0.293,0.9352); rgb(32pt)=(0.2813,0.2985,0.9397); rgb(33pt)=(0.2814,0.304,0.9441); rgb(34pt)=(0.2814,0.3095,0.9483); rgb(35pt)=(0.2813,0.315,0.9524); rgb(36pt)=(0.2811,0.3204,0.9563); rgb(37pt)=(0.2809,0.3259,0.96); rgb(38pt)=(0.2807,0.3313,0.9636); rgb(39pt)=(0.2803,0.3367,0.967); rgb(40pt)=(0.2798,0.3421,0.9702); rgb(41pt)=(0.2791,0.3475,0.9733); rgb(42pt)=(0.2784,0.3529,0.9763); rgb(43pt)=(0.2776,0.3583,0.9791); rgb(44pt)=(0.2766,0.3638,0.9817); rgb(45pt)=(0.2754,0.3693,0.984); rgb(46pt)=(0.2741,0.3748,0.9862); rgb(47pt)=(0.2726,0.3804,0.9881); rgb(48pt)=(0.271,0.386,0.9898); rgb(49pt)=(0.2691,0.3916,0.9912); rgb(50pt)=(0.267,0.3973,0.9924); rgb(51pt)=(0.2647,0.403,0.9935); rgb(52pt)=(0.2621,0.4088,0.9946); rgb(53pt)=(0.2591,0.4145,0.9955); rgb(54pt)=(0.2556,0.4203,0.9965); rgb(55pt)=(0.2517,0.4261,0.9974); rgb(56pt)=(0.2473,0.4319,0.9983); rgb(57pt)=(0.2424,0.4378,0.9991); rgb(58pt)=(0.2369,0.4437,0.9996); rgb(59pt)=(0.2311,0.4497,0.9995); rgb(60pt)=(0.225,0.4559,0.9985); rgb(61pt)=(0.2189,0.462,0.9968); rgb(62pt)=(0.2128,0.4682,0.9948); rgb(63pt)=(0.2066,0.4743,0.9926); rgb(64pt)=(0.2006,0.4803,0.9906); rgb(65pt)=(0.195,0.4861,0.9887); rgb(66pt)=(0.1903,0.4919,0.9867); rgb(67pt)=(0.1869,0.4975,0.9844); rgb(68pt)=(0.1847,0.503,0.9819); rgb(69pt)=(0.1831,0.5084,0.9793); rgb(70pt)=(0.1818,0.5138,0.9766); rgb(71pt)=(0.1806,0.5191,0.9738); rgb(72pt)=(0.1795,0.5244,0.9709); rgb(73pt)=(0.1785,0.5296,0.9677); rgb(74pt)=(0.1778,0.5349,0.9641); rgb(75pt)=(0.1773,0.5401,0.9602); rgb(76pt)=(0.1768,0.5452,0.956); rgb(77pt)=(0.1764,0.5504,0.9516); rgb(78pt)=(0.1755,0.5554,0.9473); rgb(79pt)=(0.174,0.5605,0.9432); rgb(80pt)=(0.1716,0.5655,0.9393); rgb(81pt)=(0.1686,0.5705,0.9357); rgb(82pt)=(0.1649,0.5755,0.9323); rgb(83pt)=(0.161,0.5805,0.9289); rgb(84pt)=(0.1573,0.5854,0.9254); rgb(85pt)=(0.154,0.5902,0.9218); rgb(86pt)=(0.1513,0.595,0.9182); rgb(87pt)=(0.1492,0.5997,0.9147); rgb(88pt)=(0.1475,0.6043,0.9113); rgb(89pt)=(0.1461,0.6089,0.908); rgb(90pt)=(0.1446,0.6135,0.905); rgb(91pt)=(0.1429,0.618,0.9022); rgb(92pt)=(0.1408,0.6226,0.8998); rgb(93pt)=(0.1383,0.6272,0.8975); rgb(94pt)=(0.1354,0.6317,0.8953); rgb(95pt)=(0.1321,0.6363,0.8932); rgb(96pt)=(0.1288,0.6408,0.891); rgb(97pt)=(0.1253,0.6453,0.8887); rgb(98pt)=(0.1219,0.6497,0.8862); rgb(99pt)=(0.1185,0.6541,0.8834); rgb(100pt)=(0.1152,0.6584,0.8804); rgb(101pt)=(0.1119,0.6627,0.877); rgb(102pt)=(0.1085,0.6669,0.8734); rgb(103pt)=(0.1048,0.671,0.8695); rgb(104pt)=(0.1009,0.675,0.8653); rgb(105pt)=(0.0964,0.6789,0.8609); rgb(106pt)=(0.0914,0.6828,0.8562); rgb(107pt)=(0.0855,0.6865,0.8513); rgb(108pt)=(0.0789,0.6902,0.8462); rgb(109pt)=(0.0713,0.6938,0.8409); rgb(110pt)=(0.0628,0.6972,0.8355); rgb(111pt)=(0.0535,0.7006,0.8299); rgb(112pt)=(0.0433,0.7039,0.8242); rgb(113pt)=(0.0328,0.7071,0.8183); rgb(114pt)=(0.0234,0.7103,0.8124); rgb(115pt)=(0.0155,0.7133,0.8064); rgb(116pt)=(0.0091,0.7163,0.8003); rgb(117pt)=(0.0046,0.7192,0.7941); rgb(118pt)=(0.0019,0.722,0.7878); rgb(119pt)=(0.0009,0.7248,0.7815); rgb(120pt)=(0.0018,0.7275,0.7752); rgb(121pt)=(0.0046,0.7301,0.7688); rgb(122pt)=(0.0094,0.7327,0.7623); rgb(123pt)=(0.0162,0.7352,0.7558); rgb(124pt)=(0.0253,0.7376,0.7492); rgb(125pt)=(0.0369,0.74,0.7426); rgb(126pt)=(0.0504,0.7423,0.7359); rgb(127pt)=(0.0638,0.7446,0.7292); rgb(128pt)=(0.077,0.7468,0.7224); rgb(129pt)=(0.0899,0.7489,0.7156); rgb(130pt)=(0.1023,0.751,0.7088); rgb(131pt)=(0.1141,0.7531,0.7019); rgb(132pt)=(0.1252,0.7552,0.695); rgb(133pt)=(0.1354,0.7572,0.6881); rgb(134pt)=(0.1448,0.7593,0.6812); rgb(135pt)=(0.1532,0.7614,0.6741); rgb(136pt)=(0.1609,0.7635,0.6671); rgb(137pt)=(0.1678,0.7656,0.6599); rgb(138pt)=(0.1741,0.7678,0.6527); rgb(139pt)=(0.1799,0.7699,0.6454); rgb(140pt)=(0.1853,0.7721,0.6379); rgb(141pt)=(0.1905,0.7743,0.6303); rgb(142pt)=(0.1954,0.7765,0.6225); rgb(143pt)=(0.2003,0.7787,0.6146); rgb(144pt)=(0.2061,0.7808,0.6065); rgb(145pt)=(0.2118,0.7828,0.5983); rgb(146pt)=(0.2178,0.7849,0.5899); rgb(147pt)=(0.2244,0.7869,0.5813); rgb(148pt)=(0.2318,0.7887,0.5725); rgb(149pt)=(0.2401,0.7905,0.5636); rgb(150pt)=(0.2491,0.7922,0.5546); rgb(151pt)=(0.2589,0.7937,0.5454); rgb(152pt)=(0.2695,0.7951,0.536); rgb(153pt)=(0.2809,0.7964,0.5266); rgb(154pt)=(0.2929,0.7975,0.517); rgb(155pt)=(0.3052,0.7985,0.5074); rgb(156pt)=(0.3176,0.7994,0.4975); rgb(157pt)=(0.3301,0.8002,0.4876); rgb(158pt)=(0.3424,0.8009,0.4774); rgb(159pt)=(0.3548,0.8016,0.4669); rgb(160pt)=(0.3671,0.8021,0.4563); rgb(161pt)=(0.3795,0.8026,0.4454); rgb(162pt)=(0.3921,0.8029,0.4344); rgb(163pt)=(0.405,0.8031,0.4233); rgb(164pt)=(0.4184,0.803,0.4122); rgb(165pt)=(0.4322,0.8028,0.4013); rgb(166pt)=(0.4463,0.8024,0.3904); rgb(167pt)=(0.4608,0.8018,0.3797); rgb(168pt)=(0.4753,0.8011,0.3691); rgb(169pt)=(0.4899,0.8002,0.3586); rgb(170pt)=(0.5044,0.7993,0.348); rgb(171pt)=(0.5187,0.7982,0.3374); rgb(172pt)=(0.5329,0.797,0.3267); rgb(173pt)=(0.547,0.7957,0.3159); rgb(175pt)=(0.5748,0.7929,0.2941); rgb(176pt)=(0.5886,0.7913,0.2833); rgb(177pt)=(0.6024,0.7896,0.2726); rgb(178pt)=(0.6161,0.7878,0.2622); rgb(179pt)=(0.6297,0.7859,0.2521); rgb(180pt)=(0.6433,0.7839,0.2423); rgb(181pt)=(0.6567,0.7818,0.2329); rgb(182pt)=(0.6701,0.7796,0.2239); rgb(183pt)=(0.6833,0.7773,0.2155); rgb(184pt)=(0.6963,0.775,0.2075); rgb(185pt)=(0.7091,0.7727,0.1998); rgb(186pt)=(0.7218,0.7703,0.1924); rgb(187pt)=(0.7344,0.7679,0.1852); rgb(188pt)=(0.7468,0.7654,0.1782); rgb(189pt)=(0.759,0.7629,0.1717); rgb(190pt)=(0.771,0.7604,0.1658); rgb(191pt)=(0.7829,0.7579,0.1608); rgb(192pt)=(0.7945,0.7554,0.157); rgb(193pt)=(0.806,0.7529,0.1546); rgb(194pt)=(0.8172,0.7505,0.1535); rgb(195pt)=(0.8281,0.7481,0.1536); rgb(196pt)=(0.8389,0.7457,0.1546); rgb(197pt)=(0.8495,0.7435,0.1564); rgb(198pt)=(0.86,0.7413,0.1587); rgb(199pt)=(0.8703,0.7392,0.1615); rgb(200pt)=(0.8804,0.7372,0.165); rgb(201pt)=(0.8903,0.7353,0.1695); rgb(202pt)=(0.9,0.7336,0.1749); rgb(203pt)=(0.9093,0.7321,0.1815); rgb(204pt)=(0.9184,0.7308,0.189); rgb(205pt)=(0.9272,0.7298,0.1973); rgb(206pt)=(0.9357,0.729,0.2061); rgb(207pt)=(0.944,0.7285,0.2151); rgb(208pt)=(0.9523,0.7284,0.2237); rgb(209pt)=(0.9606,0.7285,0.2312); rgb(210pt)=(0.9689,0.7292,0.2373); rgb(211pt)=(0.977,0.7304,0.2418); rgb(212pt)=(0.9842,0.733,0.2446); rgb(213pt)=(0.99,0.7365,0.2429); rgb(214pt)=(0.9946,0.7407,0.2394); rgb(215pt)=(0.9966,0.7458,0.2351); rgb(216pt)=(0.9971,0.7513,0.2309); rgb(217pt)=(0.9972,0.7569,0.2267); rgb(218pt)=(0.9971,0.7626,0.2224); rgb(219pt)=(0.9969,0.7683,0.2181); rgb(220pt)=(0.9966,0.774,0.2138); rgb(221pt)=(0.9962,0.7798,0.2095); rgb(222pt)=(0.9957,0.7856,0.2053); rgb(223pt)=(0.9949,0.7915,0.2012); rgb(224pt)=(0.9938,0.7974,0.1974); rgb(225pt)=(0.9923,0.8034,0.1939); rgb(226pt)=(0.9906,0.8095,0.1906); rgb(227pt)=(0.9885,0.8156,0.1875); rgb(228pt)=(0.9861,0.8218,0.1846); rgb(229pt)=(0.9835,0.828,0.1817); rgb(230pt)=(0.9807,0.8342,0.1787); rgb(231pt)=(0.9778,0.8404,0.1757); rgb(232pt)=(0.9748,0.8467,0.1726); rgb(233pt)=(0.972,0.8529,0.1695); rgb(234pt)=(0.9694,0.8591,0.1665); rgb(235pt)=(0.9671,0.8654,0.1636); rgb(236pt)=(0.9651,0.8716,0.1608); rgb(237pt)=(0.9634,0.8778,0.1582); rgb(238pt)=(0.9619,0.884,0.1557); rgb(239pt)=(0.9608,0.8902,0.1532); rgb(240pt)=(0.9601,0.8963,0.1507); rgb(241pt)=(0.9596,0.9023,0.148); rgb(242pt)=(0.9595,0.9084,0.145); rgb(243pt)=(0.9597,0.9143,0.1418); rgb(244pt)=(0.9601,0.9203,0.1382); rgb(245pt)=(0.9608,0.9262,0.1344); rgb(246pt)=(0.9618,0.932,0.1304); rgb(247pt)=(0.9629,0.9379,0.1261); rgb(248pt)=(0.9642,0.9437,0.1216); rgb(249pt)=(0.9657,0.9494,0.1168); rgb(250pt)=(0.9674,0.9552,0.1116); rgb(251pt)=(0.9692,0.9609,0.1061); rgb(252pt)=(0.9711,0.9667,0.1001); rgb(253pt)=(0.973,0.9724,0.0938); rgb(254pt)=(0.9749,0.9782,0.0872); rgb(255pt)=(0.9769,0.9839,0.0805)}
]
\addplot [forget plot] graphics [xmin=62.5, xmax=6012.5, ymin=0.5, ymax=60.5] {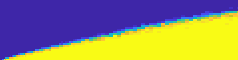};
\addplot [color=red, line width=1.0pt, forget plot]
  table[row sep=crcr]{%
1830	60\\
1827	58\\
1820	56\\
1809	54\\
1794	52\\
1775	50\\
1752	48\\
1725	46\\
1694	44\\
1659	42\\
1620	40\\
1577	38\\
1530	36\\
1479	34\\
1424	32\\
1365	30\\
1302	28\\
1235	26\\
1164	24\\
1089	22\\
1010	20\\
927	18\\
840	16\\
749	14\\
654	12\\
555	10\\
452	8\\
345	6\\
234	4\\
119	2\\
60	1\\
};
\addplot [color=mycolor1, dashdotted, line width=2.0pt, forget plot]
  table[row sep=crcr]{%
5400	60\\
90	1\\
};
\end{axis}
\end{tikzpicture}%

%% file: experiment_GC_ER_good_PT_uniform_T_7_pencil_4.tex
%
%
\definecolor{mycolor1}{rgb}{1.00000,0.00000,1.00000}%
\begin{tikzpicture}

\begin{axis}[%
width=0.859\figurewidth,
height=\figureheight,
at={(0\figurewidth,0\figureheight)},
scale only axis,
point meta min=0,
point meta max=1,
axis on top,
xmin=62.5,
xmax=6012.5,
xlabel style={font=\color{white!15!black}},
xlabel={$m$},
xticklabels={,,2000,4000,6000},
yticklabels={},
ymin=0.5,
ymax=60.5,
ylabel style={font=\color{white!15!black}},
axis background/.style={fill=white},
xlabel style={font=\tiny},ylabel style={font=\tiny},
xticklabel style = {font=\tiny}, yticklabel style = {font=\tiny},
colormap={mymap}{[1pt] rgb(0pt)=(0.2422,0.1504,0.6603); rgb(1pt)=(0.2444,0.1534,0.6728); rgb(2pt)=(0.2464,0.1569,0.6847); rgb(3pt)=(0.2484,0.1607,0.6961); rgb(4pt)=(0.2503,0.1648,0.7071); rgb(5pt)=(0.2522,0.1689,0.7179); rgb(6pt)=(0.254,0.1732,0.7286); rgb(7pt)=(0.2558,0.1773,0.7393); rgb(8pt)=(0.2576,0.1814,0.7501); rgb(9pt)=(0.2594,0.1854,0.761); rgb(11pt)=(0.2628,0.1932,0.7828); rgb(12pt)=(0.2645,0.1972,0.7937); rgb(13pt)=(0.2661,0.2011,0.8043); rgb(14pt)=(0.2676,0.2052,0.8148); rgb(15pt)=(0.2691,0.2094,0.8249); rgb(16pt)=(0.2704,0.2138,0.8346); rgb(17pt)=(0.2717,0.2184,0.8439); rgb(18pt)=(0.2729,0.2231,0.8528); rgb(19pt)=(0.274,0.228,0.8612); rgb(20pt)=(0.2749,0.233,0.8692); rgb(21pt)=(0.2758,0.2382,0.8767); rgb(22pt)=(0.2766,0.2435,0.884); rgb(23pt)=(0.2774,0.2489,0.8908); rgb(24pt)=(0.2781,0.2543,0.8973); rgb(25pt)=(0.2788,0.2598,0.9035); rgb(26pt)=(0.2794,0.2653,0.9094); rgb(27pt)=(0.2798,0.2708,0.915); rgb(28pt)=(0.2802,0.2764,0.9204); rgb(29pt)=(0.2806,0.2819,0.9255); rgb(30pt)=(0.2809,0.2875,0.9305); rgb(31pt)=(0.2811,0.293,0.9352); rgb(32pt)=(0.2813,0.2985,0.9397); rgb(33pt)=(0.2814,0.304,0.9441); rgb(34pt)=(0.2814,0.3095,0.9483); rgb(35pt)=(0.2813,0.315,0.9524); rgb(36pt)=(0.2811,0.3204,0.9563); rgb(37pt)=(0.2809,0.3259,0.96); rgb(38pt)=(0.2807,0.3313,0.9636); rgb(39pt)=(0.2803,0.3367,0.967); rgb(40pt)=(0.2798,0.3421,0.9702); rgb(41pt)=(0.2791,0.3475,0.9733); rgb(42pt)=(0.2784,0.3529,0.9763); rgb(43pt)=(0.2776,0.3583,0.9791); rgb(44pt)=(0.2766,0.3638,0.9817); rgb(45pt)=(0.2754,0.3693,0.984); rgb(46pt)=(0.2741,0.3748,0.9862); rgb(47pt)=(0.2726,0.3804,0.9881); rgb(48pt)=(0.271,0.386,0.9898); rgb(49pt)=(0.2691,0.3916,0.9912); rgb(50pt)=(0.267,0.3973,0.9924); rgb(51pt)=(0.2647,0.403,0.9935); rgb(52pt)=(0.2621,0.4088,0.9946); rgb(53pt)=(0.2591,0.4145,0.9955); rgb(54pt)=(0.2556,0.4203,0.9965); rgb(55pt)=(0.2517,0.4261,0.9974); rgb(56pt)=(0.2473,0.4319,0.9983); rgb(57pt)=(0.2424,0.4378,0.9991); rgb(58pt)=(0.2369,0.4437,0.9996); rgb(59pt)=(0.2311,0.4497,0.9995); rgb(60pt)=(0.225,0.4559,0.9985); rgb(61pt)=(0.2189,0.462,0.9968); rgb(62pt)=(0.2128,0.4682,0.9948); rgb(63pt)=(0.2066,0.4743,0.9926); rgb(64pt)=(0.2006,0.4803,0.9906); rgb(65pt)=(0.195,0.4861,0.9887); rgb(66pt)=(0.1903,0.4919,0.9867); rgb(67pt)=(0.1869,0.4975,0.9844); rgb(68pt)=(0.1847,0.503,0.9819); rgb(69pt)=(0.1831,0.5084,0.9793); rgb(70pt)=(0.1818,0.5138,0.9766); rgb(71pt)=(0.1806,0.5191,0.9738); rgb(72pt)=(0.1795,0.5244,0.9709); rgb(73pt)=(0.1785,0.5296,0.9677); rgb(74pt)=(0.1778,0.5349,0.9641); rgb(75pt)=(0.1773,0.5401,0.9602); rgb(76pt)=(0.1768,0.5452,0.956); rgb(77pt)=(0.1764,0.5504,0.9516); rgb(78pt)=(0.1755,0.5554,0.9473); rgb(79pt)=(0.174,0.5605,0.9432); rgb(80pt)=(0.1716,0.5655,0.9393); rgb(81pt)=(0.1686,0.5705,0.9357); rgb(82pt)=(0.1649,0.5755,0.9323); rgb(83pt)=(0.161,0.5805,0.9289); rgb(84pt)=(0.1573,0.5854,0.9254); rgb(85pt)=(0.154,0.5902,0.9218); rgb(86pt)=(0.1513,0.595,0.9182); rgb(87pt)=(0.1492,0.5997,0.9147); rgb(88pt)=(0.1475,0.6043,0.9113); rgb(89pt)=(0.1461,0.6089,0.908); rgb(90pt)=(0.1446,0.6135,0.905); rgb(91pt)=(0.1429,0.618,0.9022); rgb(92pt)=(0.1408,0.6226,0.8998); rgb(93pt)=(0.1383,0.6272,0.8975); rgb(94pt)=(0.1354,0.6317,0.8953); rgb(95pt)=(0.1321,0.6363,0.8932); rgb(96pt)=(0.1288,0.6408,0.891); rgb(97pt)=(0.1253,0.6453,0.8887); rgb(98pt)=(0.1219,0.6497,0.8862); rgb(99pt)=(0.1185,0.6541,0.8834); rgb(100pt)=(0.1152,0.6584,0.8804); rgb(101pt)=(0.1119,0.6627,0.877); rgb(102pt)=(0.1085,0.6669,0.8734); rgb(103pt)=(0.1048,0.671,0.8695); rgb(104pt)=(0.1009,0.675,0.8653); rgb(105pt)=(0.0964,0.6789,0.8609); rgb(106pt)=(0.0914,0.6828,0.8562); rgb(107pt)=(0.0855,0.6865,0.8513); rgb(108pt)=(0.0789,0.6902,0.8462); rgb(109pt)=(0.0713,0.6938,0.8409); rgb(110pt)=(0.0628,0.6972,0.8355); rgb(111pt)=(0.0535,0.7006,0.8299); rgb(112pt)=(0.0433,0.7039,0.8242); rgb(113pt)=(0.0328,0.7071,0.8183); rgb(114pt)=(0.0234,0.7103,0.8124); rgb(115pt)=(0.0155,0.7133,0.8064); rgb(116pt)=(0.0091,0.7163,0.8003); rgb(117pt)=(0.0046,0.7192,0.7941); rgb(118pt)=(0.0019,0.722,0.7878); rgb(119pt)=(0.0009,0.7248,0.7815); rgb(120pt)=(0.0018,0.7275,0.7752); rgb(121pt)=(0.0046,0.7301,0.7688); rgb(122pt)=(0.0094,0.7327,0.7623); rgb(123pt)=(0.0162,0.7352,0.7558); rgb(124pt)=(0.0253,0.7376,0.7492); rgb(125pt)=(0.0369,0.74,0.7426); rgb(126pt)=(0.0504,0.7423,0.7359); rgb(127pt)=(0.0638,0.7446,0.7292); rgb(128pt)=(0.077,0.7468,0.7224); rgb(129pt)=(0.0899,0.7489,0.7156); rgb(130pt)=(0.1023,0.751,0.7088); rgb(131pt)=(0.1141,0.7531,0.7019); rgb(132pt)=(0.1252,0.7552,0.695); rgb(133pt)=(0.1354,0.7572,0.6881); rgb(134pt)=(0.1448,0.7593,0.6812); rgb(135pt)=(0.1532,0.7614,0.6741); rgb(136pt)=(0.1609,0.7635,0.6671); rgb(137pt)=(0.1678,0.7656,0.6599); rgb(138pt)=(0.1741,0.7678,0.6527); rgb(139pt)=(0.1799,0.7699,0.6454); rgb(140pt)=(0.1853,0.7721,0.6379); rgb(141pt)=(0.1905,0.7743,0.6303); rgb(142pt)=(0.1954,0.7765,0.6225); rgb(143pt)=(0.2003,0.7787,0.6146); rgb(144pt)=(0.2061,0.7808,0.6065); rgb(145pt)=(0.2118,0.7828,0.5983); rgb(146pt)=(0.2178,0.7849,0.5899); rgb(147pt)=(0.2244,0.7869,0.5813); rgb(148pt)=(0.2318,0.7887,0.5725); rgb(149pt)=(0.2401,0.7905,0.5636); rgb(150pt)=(0.2491,0.7922,0.5546); rgb(151pt)=(0.2589,0.7937,0.5454); rgb(152pt)=(0.2695,0.7951,0.536); rgb(153pt)=(0.2809,0.7964,0.5266); rgb(154pt)=(0.2929,0.7975,0.517); rgb(155pt)=(0.3052,0.7985,0.5074); rgb(156pt)=(0.3176,0.7994,0.4975); rgb(157pt)=(0.3301,0.8002,0.4876); rgb(158pt)=(0.3424,0.8009,0.4774); rgb(159pt)=(0.3548,0.8016,0.4669); rgb(160pt)=(0.3671,0.8021,0.4563); rgb(161pt)=(0.3795,0.8026,0.4454); rgb(162pt)=(0.3921,0.8029,0.4344); rgb(163pt)=(0.405,0.8031,0.4233); rgb(164pt)=(0.4184,0.803,0.4122); rgb(165pt)=(0.4322,0.8028,0.4013); rgb(166pt)=(0.4463,0.8024,0.3904); rgb(167pt)=(0.4608,0.8018,0.3797); rgb(168pt)=(0.4753,0.8011,0.3691); rgb(169pt)=(0.4899,0.8002,0.3586); rgb(170pt)=(0.5044,0.7993,0.348); rgb(171pt)=(0.5187,0.7982,0.3374); rgb(172pt)=(0.5329,0.797,0.3267); rgb(173pt)=(0.547,0.7957,0.3159); rgb(175pt)=(0.5748,0.7929,0.2941); rgb(176pt)=(0.5886,0.7913,0.2833); rgb(177pt)=(0.6024,0.7896,0.2726); rgb(178pt)=(0.6161,0.7878,0.2622); rgb(179pt)=(0.6297,0.7859,0.2521); rgb(180pt)=(0.6433,0.7839,0.2423); rgb(181pt)=(0.6567,0.7818,0.2329); rgb(182pt)=(0.6701,0.7796,0.2239); rgb(183pt)=(0.6833,0.7773,0.2155); rgb(184pt)=(0.6963,0.775,0.2075); rgb(185pt)=(0.7091,0.7727,0.1998); rgb(186pt)=(0.7218,0.7703,0.1924); rgb(187pt)=(0.7344,0.7679,0.1852); rgb(188pt)=(0.7468,0.7654,0.1782); rgb(189pt)=(0.759,0.7629,0.1717); rgb(190pt)=(0.771,0.7604,0.1658); rgb(191pt)=(0.7829,0.7579,0.1608); rgb(192pt)=(0.7945,0.7554,0.157); rgb(193pt)=(0.806,0.7529,0.1546); rgb(194pt)=(0.8172,0.7505,0.1535); rgb(195pt)=(0.8281,0.7481,0.1536); rgb(196pt)=(0.8389,0.7457,0.1546); rgb(197pt)=(0.8495,0.7435,0.1564); rgb(198pt)=(0.86,0.7413,0.1587); rgb(199pt)=(0.8703,0.7392,0.1615); rgb(200pt)=(0.8804,0.7372,0.165); rgb(201pt)=(0.8903,0.7353,0.1695); rgb(202pt)=(0.9,0.7336,0.1749); rgb(203pt)=(0.9093,0.7321,0.1815); rgb(204pt)=(0.9184,0.7308,0.189); rgb(205pt)=(0.9272,0.7298,0.1973); rgb(206pt)=(0.9357,0.729,0.2061); rgb(207pt)=(0.944,0.7285,0.2151); rgb(208pt)=(0.9523,0.7284,0.2237); rgb(209pt)=(0.9606,0.7285,0.2312); rgb(210pt)=(0.9689,0.7292,0.2373); rgb(211pt)=(0.977,0.7304,0.2418); rgb(212pt)=(0.9842,0.733,0.2446); rgb(213pt)=(0.99,0.7365,0.2429); rgb(214pt)=(0.9946,0.7407,0.2394); rgb(215pt)=(0.9966,0.7458,0.2351); rgb(216pt)=(0.9971,0.7513,0.2309); rgb(217pt)=(0.9972,0.7569,0.2267); rgb(218pt)=(0.9971,0.7626,0.2224); rgb(219pt)=(0.9969,0.7683,0.2181); rgb(220pt)=(0.9966,0.774,0.2138); rgb(221pt)=(0.9962,0.7798,0.2095); rgb(222pt)=(0.9957,0.7856,0.2053); rgb(223pt)=(0.9949,0.7915,0.2012); rgb(224pt)=(0.9938,0.7974,0.1974); rgb(225pt)=(0.9923,0.8034,0.1939); rgb(226pt)=(0.9906,0.8095,0.1906); rgb(227pt)=(0.9885,0.8156,0.1875); rgb(228pt)=(0.9861,0.8218,0.1846); rgb(229pt)=(0.9835,0.828,0.1817); rgb(230pt)=(0.9807,0.8342,0.1787); rgb(231pt)=(0.9778,0.8404,0.1757); rgb(232pt)=(0.9748,0.8467,0.1726); rgb(233pt)=(0.972,0.8529,0.1695); rgb(234pt)=(0.9694,0.8591,0.1665); rgb(235pt)=(0.9671,0.8654,0.1636); rgb(236pt)=(0.9651,0.8716,0.1608); rgb(237pt)=(0.9634,0.8778,0.1582); rgb(238pt)=(0.9619,0.884,0.1557); rgb(239pt)=(0.9608,0.8902,0.1532); rgb(240pt)=(0.9601,0.8963,0.1507); rgb(241pt)=(0.9596,0.9023,0.148); rgb(242pt)=(0.9595,0.9084,0.145); rgb(243pt)=(0.9597,0.9143,0.1418); rgb(244pt)=(0.9601,0.9203,0.1382); rgb(245pt)=(0.9608,0.9262,0.1344); rgb(246pt)=(0.9618,0.932,0.1304); rgb(247pt)=(0.9629,0.9379,0.1261); rgb(248pt)=(0.9642,0.9437,0.1216); rgb(249pt)=(0.9657,0.9494,0.1168); rgb(250pt)=(0.9674,0.9552,0.1116); rgb(251pt)=(0.9692,0.9609,0.1061); rgb(252pt)=(0.9711,0.9667,0.1001); rgb(253pt)=(0.973,0.9724,0.0938); rgb(254pt)=(0.9749,0.9782,0.0872); rgb(255pt)=(0.9769,0.9839,0.0805)}
]
\addplot [forget plot] graphics [xmin=62.5, xmax=6012.5, ymin=0.5, ymax=60.5] {experiment_GC_ER_good_PT_uniform_T_7-1.png};
\addplot [color=red, line width=1.0pt, forget plot]
  table[row sep=crcr]{%
1830	60\\
1829	59\\
1827	58\\
1824	57\\
1820	56\\
1815	55\\
1809	54\\
1802	53\\
1794	52\\
1785	51\\
1739	47\\
1725	46\\
1710	45\\
1694	44\\
1677	43\\
1659	42\\
1640	41\\
1620	40\\
1599	39\\
1577	38\\
1554	37\\
1530	36\\
1505	35\\
1479	34\\
1452	33\\
1424	32\\
1395	31\\
1365	30\\
1334	29\\
1302	28\\
1269	27\\
1235	26\\
1200	25\\
1164	24\\
1127	23\\
1089	22\\
1050	21\\
1010	20\\
969	19\\
927	18\\
884	17\\
840	16\\
795	15\\
749	14\\
702	13\\
654	12\\
605	11\\
555	10\\
504	9\\
452	8\\
399	7\\
345	6\\
290	5\\
234	4\\
177	3\\
119	2\\
60	1\\
};
\addplot [color=mycolor1, dashdotted, line width=2.0pt, forget plot]
  table[row sep=crcr]{%
5400	60\\
90	1\\
};
\end{axis}
\end{tikzpicture}%

%% file: experiment_GC_ER_T_7_coherences.tex
%
%
\definecolor{mycolor1}{rgb}{0.00000,0.44700,0.74100}%
\definecolor{mycolor2}{rgb}{0.85000,0.32500,0.09800}%
\definecolor{mycolor3}{rgb}{0.92900,0.69400,0.12500}%
\definecolor{mycolor4}{rgb}{0.49400,0.18400,0.55600}%
\begin{tikzpicture}

\begin{axis}[%
width=0.951\figurewidth,
height=\figureheight,
at={(0\figurewidth,0\figureheight)},
scale only axis,
xmin=0,
xmax=60,
xlabel style={font=\color{white!15!black}},
xlabel={Rank $r$},
ymin=0,
ymax=33,
ylabel style={font=\color{white!15!black}},
ylabel={Incoherence quantity $c_s \mu_0$},
xticklabel style = {font=\tiny}, yticklabel style = {font=\tiny},
ylabel near ticks,
no markers,
axis background/.style={fill=white},
legend style={at={(0.43,0.58)}, anchor=south west, legend cell align=left, align=left, draw=white!15!black, font=\tiny, legend columns = 2},
legend image post style={xscale=0.3},
xlabel style={font=\tiny},ylabel style={font=\tiny},
]
\addplot [scatter, color=mycolor1, line width=1pt, mark size=5pt,  mark options={solid, mycolor1}, error bars/.cd,  y fixed, y dir = both, y explicit, error mark options={line width=0.5pt,mark size=1.5pt,rotate=90}]
 table[row sep=crcr, y error plus index=2, y error minus index=3]{%
1	11.8473504500876	0.705891213292352	0.705891213292352\\
2	30.9578066320885	0.929956702077105	0.929956702077105\\
3	24.8349442982264	0.422460098170348	0.422460098170348\\
4	21.0969599943414	0.285625084916117	0.285625084916117\\
5	19.0212001959305	0.382083125934626	0.382083125934626\\
6	18.8757607571293	1.20844562940902	1.20844562940902\\
7	14.1471899511825	0.127406805091763	0.127406805091763\\
8	12.7416017509262	0.163903720792578	0.163903720792578\\
9	13.966384803358	0.169742765120804	0.169742765120804\\
10	12.711141537716	0.130378629954596	0.130378629954596\\
11	14.363102260915	0.130894170667308	0.130894170667308\\
12	13.286203091783	0.122110274050043	0.122110274050043\\
13	12.4308767161341	0.104868417422578	0.104868417422578\\
14	12.229929646771	0.0995197342836622	0.0995197342836622\\
15	11.5893710572225	0.0880920260241718	0.0880920260241718\\
16	11.8537252260054	0.0938401468592483	0.0938401468592483\\
17	11.206861351562	0.0881430900424216	0.0881430900424216\\
18	11.8280816012644	0.0562327247635353	0.0562327247635353\\
19	11.3991674511952	0.0530369726982467	0.0530369726982467\\
20	19.8984748653129	0.942779835096579	0.942779835096579\\
21	19.0391965215503	0.649657150787614	0.649657150787614\\
22	18.2227403419235	0.539649751176042	0.539649751176042\\
23	17.4014053775154	0.629180340297779	0.629180340297779\\
24	16.6139797350642	0.636068100501059	0.636068100501059\\
25	15.9181357871301	0.681076771434088	0.681076771434088\\
26	15.4707030271569	0.390323821677768	0.390323821677768\\
27	14.9073618253209	0.428546824431524	0.428546824431524\\
28	14.3920273990275	0.392402094147047	0.392402094147047\\
29	13.8839574572237	0.42365412601765	0.42365412601765\\
30	13.4564570849618	0.366467257449254	0.366467257449254\\
31	12.9667439207181	0.388311941048636	0.388311941048636\\
32	12.5747710007422	0.330985826363055	0.330985826363055\\
33	12.2585194982625	0.289192333204249	0.289192333204249\\
34	11.907014596688	0.26056905563248	0.26056905563248\\
35	11.5627137446903	0.24266374659452	0.24266374659452\\
36	11.2318920254907	0.264876713556465	0.264876713556465\\
37	10.9487722877298	0.224036352425299	0.224036352425299\\
38	10.6148862627706	0.257971965274423	0.257971965274423\\
39	10.4154340911867	0.158512799075755	0.158512799075755\\
40	10.1478835144526	0.157317393373673	0.157317393373673\\
41	9.93438057301653	0.149388423798091	0.149388423798091\\
42	9.7154725587553	0.129254308358195	0.129254308358195\\
43	9.50309176208529	0.117527490031927	0.117527490031927\\
44	9.26894454172886	0.106106881112938	0.106106881112938\\
45	9.08882687533545	0.101348372687477	0.101348372687477\\
46	8.93378015392876	0.0866712801831448	0.0866712801831448\\
47	8.74974987734601	0.0811894636252288	0.0811894636252288\\
48	8.58879739889969	0.0668062049283659	0.0668062049283659\\
49	8.44176393000409	0.0531703989344479	0.0531703989344479\\
50	8.29643578069845	0.041125610680422	0.041125610680422\\
51	8.15416265396417	0.0381646984989893	0.0381646984989893\\
52	8.01667202080695	0.0331317598148832	0.0331317598148832\\
53	7.87548640554845	0.0260832061945867	0.0260832061945867\\
54	7.74627721420891	0.0195411861494292	0.0195411861494292\\
55	7.61941522924595	0.0105751039736486	0.0105751039736486\\
56	7.49186247914254	0.00671687196895389	0.00671687196895389\\
57	7.3655309565386	0.00264546369384783	0.00264546369384783\\
58	7.24072739722919	0.000736123250037726	0.000736123250037726\\
59	7.11862767526694	4.43921170250462e-05	4.43921170250462e-05\\
60	7	2.66154307931419e-15	2.66154307931419e-15\\
};
\addlegendentry{$d_1=1$}

\addplot [scatter, color=mycolor2, line width=1pt, mark size=5pt,  mark options={solid, mycolor2}, error bars/.cd,  y fixed, y dir = both, y explicit, error mark options={line width=0.5pt,mark size=1.5pt,rotate=90}]
 table[row sep=crcr, y error plus index=2, y error minus index=3]{%
1	6.91095442921777	0.411769874420537	0.411769874420537\\
2	18.0651466699218	0.41162965420526	0.41162965420526\\
3	13.1684957688826	0.246136161195927	0.246136161195927\\
4	10.9529340159433	0.208115189287749	0.208115189287749\\
5	10.1123383311142	0.149139523360424	0.149139523360424\\
6	9.25903961883248	0.138472668307103	0.138472668307103\\
7	8.22074122905973	0.0896914601636756	0.0896914601636756\\
8	7.93006156302222	0.093030060673832	0.093030060673832\\
9	8.02746444475939	0.100669779051075	0.100669779051075\\
10	7.59393929420028	0.108833967700215	0.108833967700215\\
11	7.78995415605256	0.0458657735364982	0.0458657735364982\\
12	7.65230694065007	0.0204483401169902	0.0204483401169902\\
13	7.96171959738887	0.0289933915528032	0.0289933915528032\\
14	7.47625112288233	0.0325123774868432	0.0325123774868432\\
15	7.2324912999004	0.0282536644167204	0.0282536644167204\\
16	7.09462515269267	0.0227575189743313	0.0227575189743313\\
17	6.89148370079155	0.026117900087224	0.026117900087224\\
18	6.54084928585464	0.0291793497354935	0.0291793497354935\\
19	6.24345235562219	0.0324065223492056	0.0324065223492056\\
20	6.34852459005543	0.059513348361871	0.059513348361871\\
21	6.42773711805813	0.049370825184429	0.049370825184429\\
22	6.42033587088632	0.038828517654212	0.038828517654212\\
23	6.22501540702111	0.0400654896889011	0.0400654896889011\\
24	6.38414386153769	0.0334365895244686	0.0334365895244686\\
25	6.27048621391351	0.0383706896914936	0.0383706896914936\\
26	6.31980147664379	0.0394774650959859	0.0394774650959859\\
27	6.12901044703817	0.0364010369126503	0.0364010369126503\\
28	6.14732044109271	0.0173766366199578	0.0173766366199578\\
29	5.97344539551769	0.0290566169199263	0.0290566169199263\\
30	5.78143023079393	0.0260495464210776	0.0260495464210776\\
31	5.70448752297694	0.0370558525526214	0.0370558525526214\\
32	5.76105230811258	0.0248838466772106	0.0248838466772106\\
33	5.89315051320019	0.013059494931698	0.013059494931698\\
34	5.94350768120196	0.0241219698705335	0.0241219698705335\\
35	6.07032004637343	0.0229180335329699	0.0229180335329699\\
36	5.95489264395748	0.0248790581024608	0.0248790581024608\\
37	5.94012518867752	0.0227009321136521	0.0227009321136521\\
38	5.81935390278054	0.0229178714712234	0.0229178714712234\\
39	5.71081437515798	0.0211002754150133	0.0211002754150133\\
40	10.2833170237285	0.154422757779771	0.154422757779771\\
41	10.0592531756302	0.10845066818752	0.10845066818752\\
42	9.8033206232278	0.179046949916905	0.179046949916905\\
43	9.58174517594207	0.119961008465489	0.119961008465489\\
44	9.37934526987111	0.0959878506270096	0.0959878506270096\\
45	9.15671484646677	0.0978690330820741	0.0978690330820741\\
46	8.95594395168039	0.126504373241224	0.126504373241224\\
47	8.74306355212751	0.129203923334034	0.129203923334034\\
48	8.58445708345538	0.115140948828515	0.115140948828515\\
49	8.40401473570023	0.109372882582983	0.109372882582983\\
50	8.26276252465325	0.0903157222076525	0.0903157222076525\\
51	8.08327966812319	0.105249895529421	0.105249895529421\\
52	7.91889603811267	0.118497992117789	0.118497992117789\\
53	7.78688846803918	0.0914023747979173	0.0914023747979173\\
54	7.63837631737539	0.0920315593399233	0.0920315593399233\\
55	7.50775257987239	0.088288161429737	0.088288161429737\\
56	7.36566271552623	0.120074995062982	0.120074995062982\\
57	7.26043432418254	0.0630986523665884	0.0630986523665884\\
58	7.11736422553466	0.10071413394741	0.10071413394741\\
59	6.98886299908796	0.0981911847651036	0.0981911847651036\\
60	6.88012995873018	0.0855839804798926	0.0855839804798926\\
};
\addlegendentry{$d_1=2$}

\addplot [scatter, color=mycolor3, line width=1pt, mark size=5pt,  mark options={solid, mycolor3}, error bars/.cd,  y fixed, y dir = both, y explicit, error mark options={line width=0.5pt,mark size=1.5pt,rotate=90}]
 table[row sep=crcr, y error plus index=2, y error minus index=3]{%
1	5.52876354337422	0.329415899536429	0.329415899536429\\
2	12.8208237200723	0.339713310823093	0.339713310823093\\
3	11.6377599312726	0.232801458263156	0.232801458263156\\
4	9.4824329590919	0.145684840385365	0.145684840385365\\
5	9.81525564909464	0.179991847737082	0.179991847737082\\
6	8.50464737321985	0.155972944952865	0.155972944952865\\
7	8.1531813824051	0.146254376246486	0.146254376246486\\
8	7.46262512588549	0.110794329639295	0.110794329639295\\
9	7.35730449708883	0.0885839220104072	0.0885839220104072\\
10	6.81567479528296	0.0874583886348271	0.0874583886348271\\
11	6.57342835727625	0.0686992709206543	0.0686992709206543\\
12	6.06571610280692	0.054870580812073	0.054870580812073\\
13	5.92008000243388	0.0716977472085463	0.0716977472085463\\
14	6.08176664429823	0.0562337275462222	0.0562337275462222\\
15	5.72156826689107	0.0390611474275505	0.0390611474275505\\
16	5.64099927048613	0.0464395672517806	0.0464395672517806\\
17	5.41842395172061	0.0454474357034074	0.0454474357034074\\
18	5.9785503044988	0.041411110353189	0.041411110353189\\
19	6.03893068994625	0.0348859506389201	0.0348859506389201\\
20	5.73322821407161	0.0371690739077229	0.0371690739077229\\
21	5.49604316312097	0.0350180904815664	0.0350180904815664\\
22	5.35047997111076	0.0301433697178384	0.0301433697178384\\
23	5.74082057817164	0.00950599373998975	0.00950599373998975\\
24	5.52378033833318	0.00800802088782268	0.00800802088782268\\
25	5.34646722703437	0.00617426537539178	0.00617426537539178\\
26	5.28169934770987	0.00294284727146671	0.00294284727146671\\
27	5.15729622398434	0.00435865645226942	0.00435865645226942\\
28	5.02531647699779	0.00848367524436756	0.00848367524436756\\
29	5.0173801998426	0.00435653040940721	0.00435653040940721\\
30	5.01903706988615	0.00976376001036932	0.00976376001036932\\
31	4.95334769367477	0.0122808314859702	0.0122808314859702\\
32	4.8916696463166	0.0102466400247665	0.0102466400247665\\
33	4.79133935603454	0.010320847441468	0.010320847441468\\
34	4.74879480473561	0.0108576720062663	0.0108576720062663\\
35	4.74695418601078	0.00730173726794244	0.00730173726794244\\
36	4.78332677186159	0.0049344688194676	0.0049344688194676\\
37	4.74842488729691	0.00540158897736733	0.00540158897736733\\
38	4.79329635098512	0.00334780867237357	0.00334780867237357\\
39	4.70701301930103	0.00316441049603207	0.00316441049603207\\
40	4.63818915456773	0.0030136695004022	0.0030136695004022\\
41	4.55419910839514	0.00173190303591711	0.00173190303591711\\
42	4.5797503347361	0.0018779312135161	0.0018779312135161\\
43	4.48652725146815	0.00206187162276605	0.00206187162276605\\
44	4.55943506336	0.0191925278443545	0.0191925278443545\\
45	4.52523424000257	0.0209027650422503	0.0209027650422503\\
46	4.52212057693567	0.019554866854752	0.019554866854752\\
47	4.44482416119637	0.021267056361101	0.021267056361101\\
48	4.36095159963231	0.0211338153218846	0.0211338153218846\\
49	4.27577970213048	0.0189043679155239	0.0189043679155239\\
50	4.27286440131709	0.00238065753540417	0.00238065753540417\\
51	4.34451896956332	0.00983263797932529	0.00983263797932529\\
52	4.37107748986406	0.00201842347851545	0.00201842347851545\\
53	4.2920970089041	0.00142668554731469	0.00142668554731469\\
54	4.22104321641924	0.00125269256988037	0.00125269256988037\\
55	4.17224802291978	0.00140917002569314	0.00140917002569314\\
56	4.20466986610065	0.0028162348549604	0.0028162348549604\\
57	4.27391333920438	0.00169747055450461	0.00169747055450461\\
58	4.22220087649209	0.00182779954475685	0.00182779954475685\\
59	4.19738115976965	0.00168800732868149	0.00168800732868149\\
60	6.83701573719391	0.111379224336908	0.111379224336908\\
};
\addlegendentry{$d_1=3$}

\addplot [scatter, color=mycolor4, line width=1pt, mark size=5pt,  mark options={solid, mycolor4}, error bars/.cd,  y fixed, y dir = both, y explicit, error mark options={line width=0.5pt,mark size=1.5pt,rotate=90}]
 table[row sep=crcr, y error plus index=2, y error minus index=3]{%
1	5.18321582191332	0.308827405815402	0.308827405815402\\
2	12.8335913020592	0.397039506811345	0.397039506811345\\
3	9.53814123532774	0.236043976310956	0.236043976310956\\
4	11.4300098500894	0.163314177217122	0.163314177217122\\
5	9.15407753116076	0.122327630322683	0.122327630322683\\
6	7.27969102917645	0.113676796167539	0.113676796167539\\
7	6.51095214484448	0.0836459552649071	0.0836459552649071\\
8	6.62773711295463	0.0716936288059831	0.0716936288059831\\
9	6.19615129748167	0.0765953791325314	0.0765953791325314\\
10	5.59458280321049	0.0663659056949855	0.0663659056949855\\
11	5.69141000638126	0.0579431482824703	0.0579431482824703\\
12	5.48155040502155	0.0520017836029846	0.0520017836029846\\
13	5.53182870248703	0.0531437408319369	0.0531437408319369\\
14	5.35184871769999	0.0483893961716069	0.0483893961716069\\
15	5.2464681744826	0.0429757391562001	0.0429757391562001\\
16	5.47368364878823	0.0383282034276961	0.0383282034276961\\
17	5.38633553866462	0.0313331091605439	0.0313331091605439\\
18	5.20196350170451	0.0231985228741091	0.0231985228741091\\
19	5.17557880638726	0.0297469886920496	0.0297469886920496\\
20	4.97071009452095	0.025396112654896	0.025396112654896\\
21	4.90250724620348	0.0246296404073122	0.0246296404073122\\
22	4.83635750019879	0.0225615588168409	0.0225615588168409\\
23	4.65936820317765	0.0233828553615623	0.0233828553615623\\
24	4.5711063365018	0.0202400085109438	0.0202400085109438\\
25	4.67097844280332	0.0226277952313327	0.0226277952313327\\
26	4.52122096103957	0.0233269581330501	0.0233269581330501\\
27	4.52218797113337	0.0192989468760273	0.0192989468760273\\
28	4.5405521368377	0.0185025332729582	0.0185025332729582\\
29	4.47305944314123	0.0195922128534385	0.0195922128534385\\
30	4.56527032812222	0.014821345091094	0.014821345091094\\
31	4.42560750616895	0.0147091791816079	0.0147091791816079\\
32	4.51838518935409	0.0143619412086879	0.0143619412086879\\
33	4.5904671469529	0.0131747277710401	0.0131747277710401\\
34	4.52813077469735	0.0158398395893441	0.0158398395893441\\
35	4.48979263692201	0.0100810213953268	0.0100810213953268\\
36	4.42536900846046	0.0127900440174563	0.0127900440174563\\
37	4.38542717002782	0.0119166680343419	0.0119166680343419\\
38	4.32189557741416	0.0102357656683265	0.0102357656683265\\
39	4.29106137164935	0.0121288302207496	0.0121288302207496\\
40	4.46106395832096	0.0114610361344738	0.0114610361344738\\
41	4.35434724835002	0.0101041899169277	0.0101041899169277\\
42	4.25162343887409	0.0102559798512898	0.0102559798512898\\
43	4.17752527780017	0.00930350412109058	0.00930350412109058\\
44	4.10907593599267	0.00967968972127163	0.00967968972127163\\
45	4.10058356943177	0.00815192013583802	0.00815192013583802\\
46	4.02551594138133	0.00884055645678919	0.00884055645678919\\
47	3.94671404442952	0.00776924127576643	0.00776924127576643\\
48	4.08084048147667	0.0181319370089737	0.0181319370089737\\
49	4.00072432118293	0.0171464963024172	0.0171464963024172\\
50	3.98683769435533	0.0151899143332162	0.0151899143332162\\
51	3.91857958009301	0.0153485384573505	0.0153485384573505\\
52	3.86830489862407	0.0149541383591627	0.0149541383591627\\
53	3.8708825043617	0.0166764494181665	0.0166764494181665\\
54	3.80967561368373	0.0152694155467867	0.0152694155467867\\
55	3.77124700266926	0.00669749212035316	0.00669749212035316\\
56	3.71666495484673	0.00610132189813111	0.00610132189813111\\
57	3.69053226266979	0.00556667309793411	0.00556667309793411\\
58	3.69346296922819	0.00614620079601238	0.00614620079601238\\
59	3.70793336698426	0.00649990697206079	0.00649990697206079\\
60	3.65387237927329	0.00513790896577273	0.00513790896577273\\
};
\addlegendentry{$d_1=4$}

\end{axis}
\end{tikzpicture}%

%% file: experiment_GC_ER_good_PT_adaptive_T_1.tex
%
%
\definecolor{mycolor1}{rgb}{1.00000,0.00000,1.00000}%
\begin{tikzpicture}

\begin{axis}[%
width=0.276\figurewidth,
height=\figureheight,
at={(0\figurewidth,0\figureheight)},
scale only axis,
axis on top,
xmin=70.0142045454545,
xmax=1834.98579545455,
xlabel style={font=\color{white!15!black}},
xlabel={Nr. of expected samples $m_{\text{exp}}$},
ymin=0.5,
ymax=60.5,
ylabel style={font=\color{white!15!black}},
ylabel={Rank $r$},
xticklabels={,,,1000},
axis background/.style={fill=white},
xlabel style={font=\tiny},ylabel near ticks, ylabel style={font=\tiny, anchor=near ticklabel},  
xticklabel style = {font=\tiny}, yticklabel style = {font=\tiny},
colormap={mymap}{[1pt] rgb(0pt)=(0.2422,0.1504,0.6603); rgb(1pt)=(0.2444,0.1534,0.6728); rgb(2pt)=(0.2464,0.1569,0.6847); rgb(3pt)=(0.2484,0.1607,0.6961); rgb(4pt)=(0.2503,0.1648,0.7071); rgb(5pt)=(0.2522,0.1689,0.7179); rgb(6pt)=(0.254,0.1732,0.7286); rgb(7pt)=(0.2558,0.1773,0.7393); rgb(8pt)=(0.2576,0.1814,0.7501); rgb(9pt)=(0.2594,0.1854,0.761); rgb(11pt)=(0.2628,0.1932,0.7828); rgb(12pt)=(0.2645,0.1972,0.7937); rgb(13pt)=(0.2661,0.2011,0.8043); rgb(14pt)=(0.2676,0.2052,0.8148); rgb(15pt)=(0.2691,0.2094,0.8249); rgb(16pt)=(0.2704,0.2138,0.8346); rgb(17pt)=(0.2717,0.2184,0.8439); rgb(18pt)=(0.2729,0.2231,0.8528); rgb(19pt)=(0.274,0.228,0.8612); rgb(20pt)=(0.2749,0.233,0.8692); rgb(21pt)=(0.2758,0.2382,0.8767); rgb(22pt)=(0.2766,0.2435,0.884); rgb(23pt)=(0.2774,0.2489,0.8908); rgb(24pt)=(0.2781,0.2543,0.8973); rgb(25pt)=(0.2788,0.2598,0.9035); rgb(26pt)=(0.2794,0.2653,0.9094); rgb(27pt)=(0.2798,0.2708,0.915); rgb(28pt)=(0.2802,0.2764,0.9204); rgb(29pt)=(0.2806,0.2819,0.9255); rgb(30pt)=(0.2809,0.2875,0.9305); rgb(31pt)=(0.2811,0.293,0.9352); rgb(32pt)=(0.2813,0.2985,0.9397); rgb(33pt)=(0.2814,0.304,0.9441); rgb(34pt)=(0.2814,0.3095,0.9483); rgb(35pt)=(0.2813,0.315,0.9524); rgb(36pt)=(0.2811,0.3204,0.9563); rgb(37pt)=(0.2809,0.3259,0.96); rgb(38pt)=(0.2807,0.3313,0.9636); rgb(39pt)=(0.2803,0.3367,0.967); rgb(40pt)=(0.2798,0.3421,0.9702); rgb(41pt)=(0.2791,0.3475,0.9733); rgb(42pt)=(0.2784,0.3529,0.9763); rgb(43pt)=(0.2776,0.3583,0.9791); rgb(44pt)=(0.2766,0.3638,0.9817); rgb(45pt)=(0.2754,0.3693,0.984); rgb(46pt)=(0.2741,0.3748,0.9862); rgb(47pt)=(0.2726,0.3804,0.9881); rgb(48pt)=(0.271,0.386,0.9898); rgb(49pt)=(0.2691,0.3916,0.9912); rgb(50pt)=(0.267,0.3973,0.9924); rgb(51pt)=(0.2647,0.403,0.9935); rgb(52pt)=(0.2621,0.4088,0.9946); rgb(53pt)=(0.2591,0.4145,0.9955); rgb(54pt)=(0.2556,0.4203,0.9965); rgb(55pt)=(0.2517,0.4261,0.9974); rgb(56pt)=(0.2473,0.4319,0.9983); rgb(57pt)=(0.2424,0.4378,0.9991); rgb(58pt)=(0.2369,0.4437,0.9996); rgb(59pt)=(0.2311,0.4497,0.9995); rgb(60pt)=(0.225,0.4559,0.9985); rgb(61pt)=(0.2189,0.462,0.9968); rgb(62pt)=(0.2128,0.4682,0.9948); rgb(63pt)=(0.2066,0.4743,0.9926); rgb(64pt)=(0.2006,0.4803,0.9906); rgb(65pt)=(0.195,0.4861,0.9887); rgb(66pt)=(0.1903,0.4919,0.9867); rgb(67pt)=(0.1869,0.4975,0.9844); rgb(68pt)=(0.1847,0.503,0.9819); rgb(69pt)=(0.1831,0.5084,0.9793); rgb(70pt)=(0.1818,0.5138,0.9766); rgb(71pt)=(0.1806,0.5191,0.9738); rgb(72pt)=(0.1795,0.5244,0.9709); rgb(73pt)=(0.1785,0.5296,0.9677); rgb(74pt)=(0.1778,0.5349,0.9641); rgb(75pt)=(0.1773,0.5401,0.9602); rgb(76pt)=(0.1768,0.5452,0.956); rgb(77pt)=(0.1764,0.5504,0.9516); rgb(78pt)=(0.1755,0.5554,0.9473); rgb(79pt)=(0.174,0.5605,0.9432); rgb(80pt)=(0.1716,0.5655,0.9393); rgb(81pt)=(0.1686,0.5705,0.9357); rgb(82pt)=(0.1649,0.5755,0.9323); rgb(83pt)=(0.161,0.5805,0.9289); rgb(84pt)=(0.1573,0.5854,0.9254); rgb(85pt)=(0.154,0.5902,0.9218); rgb(86pt)=(0.1513,0.595,0.9182); rgb(87pt)=(0.1492,0.5997,0.9147); rgb(88pt)=(0.1475,0.6043,0.9113); rgb(89pt)=(0.1461,0.6089,0.908); rgb(90pt)=(0.1446,0.6135,0.905); rgb(91pt)=(0.1429,0.618,0.9022); rgb(92pt)=(0.1408,0.6226,0.8998); rgb(93pt)=(0.1383,0.6272,0.8975); rgb(94pt)=(0.1354,0.6317,0.8953); rgb(95pt)=(0.1321,0.6363,0.8932); rgb(96pt)=(0.1288,0.6408,0.891); rgb(97pt)=(0.1253,0.6453,0.8887); rgb(98pt)=(0.1219,0.6497,0.8862); rgb(99pt)=(0.1185,0.6541,0.8834); rgb(100pt)=(0.1152,0.6584,0.8804); rgb(101pt)=(0.1119,0.6627,0.877); rgb(102pt)=(0.1085,0.6669,0.8734); rgb(103pt)=(0.1048,0.671,0.8695); rgb(104pt)=(0.1009,0.675,0.8653); rgb(105pt)=(0.0964,0.6789,0.8609); rgb(106pt)=(0.0914,0.6828,0.8562); rgb(107pt)=(0.0855,0.6865,0.8513); rgb(108pt)=(0.0789,0.6902,0.8462); rgb(109pt)=(0.0713,0.6938,0.8409); rgb(110pt)=(0.0628,0.6972,0.8355); rgb(111pt)=(0.0535,0.7006,0.8299); rgb(112pt)=(0.0433,0.7039,0.8242); rgb(113pt)=(0.0328,0.7071,0.8183); rgb(114pt)=(0.0234,0.7103,0.8124); rgb(115pt)=(0.0155,0.7133,0.8064); rgb(116pt)=(0.0091,0.7163,0.8003); rgb(117pt)=(0.0046,0.7192,0.7941); rgb(118pt)=(0.0019,0.722,0.7878); rgb(119pt)=(0.0009,0.7248,0.7815); rgb(120pt)=(0.0018,0.7275,0.7752); rgb(121pt)=(0.0046,0.7301,0.7688); rgb(122pt)=(0.0094,0.7327,0.7623); rgb(123pt)=(0.0162,0.7352,0.7558); rgb(124pt)=(0.0253,0.7376,0.7492); rgb(125pt)=(0.0369,0.74,0.7426); rgb(126pt)=(0.0504,0.7423,0.7359); rgb(127pt)=(0.0638,0.7446,0.7292); rgb(128pt)=(0.077,0.7468,0.7224); rgb(129pt)=(0.0899,0.7489,0.7156); rgb(130pt)=(0.1023,0.751,0.7088); rgb(131pt)=(0.1141,0.7531,0.7019); rgb(132pt)=(0.1252,0.7552,0.695); rgb(133pt)=(0.1354,0.7572,0.6881); rgb(134pt)=(0.1448,0.7593,0.6812); rgb(135pt)=(0.1532,0.7614,0.6741); rgb(136pt)=(0.1609,0.7635,0.6671); rgb(137pt)=(0.1678,0.7656,0.6599); rgb(138pt)=(0.1741,0.7678,0.6527); rgb(139pt)=(0.1799,0.7699,0.6454); rgb(140pt)=(0.1853,0.7721,0.6379); rgb(141pt)=(0.1905,0.7743,0.6303); rgb(142pt)=(0.1954,0.7765,0.6225); rgb(143pt)=(0.2003,0.7787,0.6146); rgb(144pt)=(0.2061,0.7808,0.6065); rgb(145pt)=(0.2118,0.7828,0.5983); rgb(146pt)=(0.2178,0.7849,0.5899); rgb(147pt)=(0.2244,0.7869,0.5813); rgb(148pt)=(0.2318,0.7887,0.5725); rgb(149pt)=(0.2401,0.7905,0.5636); rgb(150pt)=(0.2491,0.7922,0.5546); rgb(151pt)=(0.2589,0.7937,0.5454); rgb(152pt)=(0.2695,0.7951,0.536); rgb(153pt)=(0.2809,0.7964,0.5266); rgb(154pt)=(0.2929,0.7975,0.517); rgb(155pt)=(0.3052,0.7985,0.5074); rgb(156pt)=(0.3176,0.7994,0.4975); rgb(157pt)=(0.3301,0.8002,0.4876); rgb(158pt)=(0.3424,0.8009,0.4774); rgb(159pt)=(0.3548,0.8016,0.4669); rgb(160pt)=(0.3671,0.8021,0.4563); rgb(161pt)=(0.3795,0.8026,0.4454); rgb(162pt)=(0.3921,0.8029,0.4344); rgb(163pt)=(0.405,0.8031,0.4233); rgb(164pt)=(0.4184,0.803,0.4122); rgb(165pt)=(0.4322,0.8028,0.4013); rgb(166pt)=(0.4463,0.8024,0.3904); rgb(167pt)=(0.4608,0.8018,0.3797); rgb(168pt)=(0.4753,0.8011,0.3691); rgb(169pt)=(0.4899,0.8002,0.3586); rgb(170pt)=(0.5044,0.7993,0.348); rgb(171pt)=(0.5187,0.7982,0.3374); rgb(172pt)=(0.5329,0.797,0.3267); rgb(173pt)=(0.547,0.7957,0.3159); rgb(175pt)=(0.5748,0.7929,0.2941); rgb(176pt)=(0.5886,0.7913,0.2833); rgb(177pt)=(0.6024,0.7896,0.2726); rgb(178pt)=(0.6161,0.7878,0.2622); rgb(179pt)=(0.6297,0.7859,0.2521); rgb(180pt)=(0.6433,0.7839,0.2423); rgb(181pt)=(0.6567,0.7818,0.2329); rgb(182pt)=(0.6701,0.7796,0.2239); rgb(183pt)=(0.6833,0.7773,0.2155); rgb(184pt)=(0.6963,0.775,0.2075); rgb(185pt)=(0.7091,0.7727,0.1998); rgb(186pt)=(0.7218,0.7703,0.1924); rgb(187pt)=(0.7344,0.7679,0.1852); rgb(188pt)=(0.7468,0.7654,0.1782); rgb(189pt)=(0.759,0.7629,0.1717); rgb(190pt)=(0.771,0.7604,0.1658); rgb(191pt)=(0.7829,0.7579,0.1608); rgb(192pt)=(0.7945,0.7554,0.157); rgb(193pt)=(0.806,0.7529,0.1546); rgb(194pt)=(0.8172,0.7505,0.1535); rgb(195pt)=(0.8281,0.7481,0.1536); rgb(196pt)=(0.8389,0.7457,0.1546); rgb(197pt)=(0.8495,0.7435,0.1564); rgb(198pt)=(0.86,0.7413,0.1587); rgb(199pt)=(0.8703,0.7392,0.1615); rgb(200pt)=(0.8804,0.7372,0.165); rgb(201pt)=(0.8903,0.7353,0.1695); rgb(202pt)=(0.9,0.7336,0.1749); rgb(203pt)=(0.9093,0.7321,0.1815); rgb(204pt)=(0.9184,0.7308,0.189); rgb(205pt)=(0.9272,0.7298,0.1973); rgb(206pt)=(0.9357,0.729,0.2061); rgb(207pt)=(0.944,0.7285,0.2151); rgb(208pt)=(0.9523,0.7284,0.2237); rgb(209pt)=(0.9606,0.7285,0.2312); rgb(210pt)=(0.9689,0.7292,0.2373); rgb(211pt)=(0.977,0.7304,0.2418); rgb(212pt)=(0.9842,0.733,0.2446); rgb(213pt)=(0.99,0.7365,0.2429); rgb(214pt)=(0.9946,0.7407,0.2394); rgb(215pt)=(0.9966,0.7458,0.2351); rgb(216pt)=(0.9971,0.7513,0.2309); rgb(217pt)=(0.9972,0.7569,0.2267); rgb(218pt)=(0.9971,0.7626,0.2224); rgb(219pt)=(0.9969,0.7683,0.2181); rgb(220pt)=(0.9966,0.774,0.2138); rgb(221pt)=(0.9962,0.7798,0.2095); rgb(222pt)=(0.9957,0.7856,0.2053); rgb(223pt)=(0.9949,0.7915,0.2012); rgb(224pt)=(0.9938,0.7974,0.1974); rgb(225pt)=(0.9923,0.8034,0.1939); rgb(226pt)=(0.9906,0.8095,0.1906); rgb(227pt)=(0.9885,0.8156,0.1875); rgb(228pt)=(0.9861,0.8218,0.1846); rgb(229pt)=(0.9835,0.828,0.1817); rgb(230pt)=(0.9807,0.8342,0.1787); rgb(231pt)=(0.9778,0.8404,0.1757); rgb(232pt)=(0.9748,0.8467,0.1726); rgb(233pt)=(0.972,0.8529,0.1695); rgb(234pt)=(0.9694,0.8591,0.1665); rgb(235pt)=(0.9671,0.8654,0.1636); rgb(236pt)=(0.9651,0.8716,0.1608); rgb(237pt)=(0.9634,0.8778,0.1582); rgb(238pt)=(0.9619,0.884,0.1557); rgb(239pt)=(0.9608,0.8902,0.1532); rgb(240pt)=(0.9601,0.8963,0.1507); rgb(241pt)=(0.9596,0.9023,0.148); rgb(242pt)=(0.9595,0.9084,0.145); rgb(243pt)=(0.9597,0.9143,0.1418); rgb(244pt)=(0.9601,0.9203,0.1382); rgb(245pt)=(0.9608,0.9262,0.1344); rgb(246pt)=(0.9618,0.932,0.1304); rgb(247pt)=(0.9629,0.9379,0.1261); rgb(248pt)=(0.9642,0.9437,0.1216); rgb(249pt)=(0.9657,0.9494,0.1168); rgb(250pt)=(0.9674,0.9552,0.1116); rgb(251pt)=(0.9692,0.9609,0.1061); rgb(252pt)=(0.9711,0.9667,0.1001); rgb(253pt)=(0.973,0.9724,0.0938); rgb(254pt)=(0.9749,0.9782,0.0872); rgb(255pt)=(0.9769,0.9839,0.0805)},
colorbar, colorbar style={at={(-3.75,1)}, xticklabel style = {font=\tiny}, yticklabel style = {font=\tiny}}
]
\addplot [forget plot] graphics [xmin=70.0142045454545, xmax=1834.98579545455, ymin=0.5, ymax=60.5] {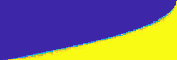};
\addplot [color=red, line width=1.0pt, forget plot]
  table[row sep=crcr]{%
1830	60\\
1829	59\\
1827	58\\
1824	57\\
1820	56\\
1815	55\\
1809	54\\
1802	53\\
1794	52\\
1785	51\\
1775	50\\
1764	49\\
1752	48\\
1739	47\\
1725	46\\
1710	45\\
1694	44\\
1677	43\\
1659	42\\
1640	41\\
1620	40\\
1599	39\\
1577	38\\
1554	37\\
1530	36\\
1505	35\\
1479	34\\
1452	33\\
1424	32\\
1395	31\\
1365	30\\
1334	29\\
1302	28\\
1269	27\\
1235	26\\
1200	25\\
1164	24\\
1127	23\\
1089	22\\
1050	21\\
1010	20\\
969	19\\
927	18\\
884	17\\
840	16\\
795	15\\
749	14\\
702	13\\
654	12\\
605	11\\
555	10\\
504	9\\
452	8\\
399	7\\
345	6\\
290	5\\
234	4\\
177	3\\
119	2\\
60	1\\
};
\addplot [color=mycolor1, dashdotted, line width=2.0pt, forget plot]
  table[row sep=crcr]{%
1890	21\\
90	1\\
};
\end{axis}
\end{tikzpicture}%

%% file: experiment_GC_ER_good_PT_adaptive_T_4_pencil_3.tex
%
%
\definecolor{mycolor1}{rgb}{1.00000,0.00000,1.00000}%
\begin{tikzpicture}

\begin{axis}[%
width=0.859\figurewidth,
height=\figureheight,
at={(0\figurewidth,0\figureheight)},
scale only axis,
point meta min=0,
point meta max=1,
axis on top,
xmin=62.5,
xmax=6012.5,
xlabel style={font=\color{white!15!black}},
xlabel={$m_{\text{exp}}$},
xticklabels={,,2000,4000,6000},
ymin=0.5,
ymax=60.5,
ylabel style={font=\color{white!15!black}},
axis background/.style={fill=white},
xlabel style={font=\tiny},ylabel style={font=\tiny},
xticklabel style = {font=\tiny}, yticklabel style = {font=\tiny},
colormap={mymap}{[1pt] rgb(0pt)=(0.2422,0.1504,0.6603); rgb(1pt)=(0.2444,0.1534,0.6728); rgb(2pt)=(0.2464,0.1569,0.6847); rgb(3pt)=(0.2484,0.1607,0.6961); rgb(4pt)=(0.2503,0.1648,0.7071); rgb(5pt)=(0.2522,0.1689,0.7179); rgb(6pt)=(0.254,0.1732,0.7286); rgb(7pt)=(0.2558,0.1773,0.7393); rgb(8pt)=(0.2576,0.1814,0.7501); rgb(9pt)=(0.2594,0.1854,0.761); rgb(11pt)=(0.2628,0.1932,0.7828); rgb(12pt)=(0.2645,0.1972,0.7937); rgb(13pt)=(0.2661,0.2011,0.8043); rgb(14pt)=(0.2676,0.2052,0.8148); rgb(15pt)=(0.2691,0.2094,0.8249); rgb(16pt)=(0.2704,0.2138,0.8346); rgb(17pt)=(0.2717,0.2184,0.8439); rgb(18pt)=(0.2729,0.2231,0.8528); rgb(19pt)=(0.274,0.228,0.8612); rgb(20pt)=(0.2749,0.233,0.8692); rgb(21pt)=(0.2758,0.2382,0.8767); rgb(22pt)=(0.2766,0.2435,0.884); rgb(23pt)=(0.2774,0.2489,0.8908); rgb(24pt)=(0.2781,0.2543,0.8973); rgb(25pt)=(0.2788,0.2598,0.9035); rgb(26pt)=(0.2794,0.2653,0.9094); rgb(27pt)=(0.2798,0.2708,0.915); rgb(28pt)=(0.2802,0.2764,0.9204); rgb(29pt)=(0.2806,0.2819,0.9255); rgb(30pt)=(0.2809,0.2875,0.9305); rgb(31pt)=(0.2811,0.293,0.9352); rgb(32pt)=(0.2813,0.2985,0.9397); rgb(33pt)=(0.2814,0.304,0.9441); rgb(34pt)=(0.2814,0.3095,0.9483); rgb(35pt)=(0.2813,0.315,0.9524); rgb(36pt)=(0.2811,0.3204,0.9563); rgb(37pt)=(0.2809,0.3259,0.96); rgb(38pt)=(0.2807,0.3313,0.9636); rgb(39pt)=(0.2803,0.3367,0.967); rgb(40pt)=(0.2798,0.3421,0.9702); rgb(41pt)=(0.2791,0.3475,0.9733); rgb(42pt)=(0.2784,0.3529,0.9763); rgb(43pt)=(0.2776,0.3583,0.9791); rgb(44pt)=(0.2766,0.3638,0.9817); rgb(45pt)=(0.2754,0.3693,0.984); rgb(46pt)=(0.2741,0.3748,0.9862); rgb(47pt)=(0.2726,0.3804,0.9881); rgb(48pt)=(0.271,0.386,0.9898); rgb(49pt)=(0.2691,0.3916,0.9912); rgb(50pt)=(0.267,0.3973,0.9924); rgb(51pt)=(0.2647,0.403,0.9935); rgb(52pt)=(0.2621,0.4088,0.9946); rgb(53pt)=(0.2591,0.4145,0.9955); rgb(54pt)=(0.2556,0.4203,0.9965); rgb(55pt)=(0.2517,0.4261,0.9974); rgb(56pt)=(0.2473,0.4319,0.9983); rgb(57pt)=(0.2424,0.4378,0.9991); rgb(58pt)=(0.2369,0.4437,0.9996); rgb(59pt)=(0.2311,0.4497,0.9995); rgb(60pt)=(0.225,0.4559,0.9985); rgb(61pt)=(0.2189,0.462,0.9968); rgb(62pt)=(0.2128,0.4682,0.9948); rgb(63pt)=(0.2066,0.4743,0.9926); rgb(64pt)=(0.2006,0.4803,0.9906); rgb(65pt)=(0.195,0.4861,0.9887); rgb(66pt)=(0.1903,0.4919,0.9867); rgb(67pt)=(0.1869,0.4975,0.9844); rgb(68pt)=(0.1847,0.503,0.9819); rgb(69pt)=(0.1831,0.5084,0.9793); rgb(70pt)=(0.1818,0.5138,0.9766); rgb(71pt)=(0.1806,0.5191,0.9738); rgb(72pt)=(0.1795,0.5244,0.9709); rgb(73pt)=(0.1785,0.5296,0.9677); rgb(74pt)=(0.1778,0.5349,0.9641); rgb(75pt)=(0.1773,0.5401,0.9602); rgb(76pt)=(0.1768,0.5452,0.956); rgb(77pt)=(0.1764,0.5504,0.9516); rgb(78pt)=(0.1755,0.5554,0.9473); rgb(79pt)=(0.174,0.5605,0.9432); rgb(80pt)=(0.1716,0.5655,0.9393); rgb(81pt)=(0.1686,0.5705,0.9357); rgb(82pt)=(0.1649,0.5755,0.9323); rgb(83pt)=(0.161,0.5805,0.9289); rgb(84pt)=(0.1573,0.5854,0.9254); rgb(85pt)=(0.154,0.5902,0.9218); rgb(86pt)=(0.1513,0.595,0.9182); rgb(87pt)=(0.1492,0.5997,0.9147); rgb(88pt)=(0.1475,0.6043,0.9113); rgb(89pt)=(0.1461,0.6089,0.908); rgb(90pt)=(0.1446,0.6135,0.905); rgb(91pt)=(0.1429,0.618,0.9022); rgb(92pt)=(0.1408,0.6226,0.8998); rgb(93pt)=(0.1383,0.6272,0.8975); rgb(94pt)=(0.1354,0.6317,0.8953); rgb(95pt)=(0.1321,0.6363,0.8932); rgb(96pt)=(0.1288,0.6408,0.891); rgb(97pt)=(0.1253,0.6453,0.8887); rgb(98pt)=(0.1219,0.6497,0.8862); rgb(99pt)=(0.1185,0.6541,0.8834); rgb(100pt)=(0.1152,0.6584,0.8804); rgb(101pt)=(0.1119,0.6627,0.877); rgb(102pt)=(0.1085,0.6669,0.8734); rgb(103pt)=(0.1048,0.671,0.8695); rgb(104pt)=(0.1009,0.675,0.8653); rgb(105pt)=(0.0964,0.6789,0.8609); rgb(106pt)=(0.0914,0.6828,0.8562); rgb(107pt)=(0.0855,0.6865,0.8513); rgb(108pt)=(0.0789,0.6902,0.8462); rgb(109pt)=(0.0713,0.6938,0.8409); rgb(110pt)=(0.0628,0.6972,0.8355); rgb(111pt)=(0.0535,0.7006,0.8299); rgb(112pt)=(0.0433,0.7039,0.8242); rgb(113pt)=(0.0328,0.7071,0.8183); rgb(114pt)=(0.0234,0.7103,0.8124); rgb(115pt)=(0.0155,0.7133,0.8064); rgb(116pt)=(0.0091,0.7163,0.8003); rgb(117pt)=(0.0046,0.7192,0.7941); rgb(118pt)=(0.0019,0.722,0.7878); rgb(119pt)=(0.0009,0.7248,0.7815); rgb(120pt)=(0.0018,0.7275,0.7752); rgb(121pt)=(0.0046,0.7301,0.7688); rgb(122pt)=(0.0094,0.7327,0.7623); rgb(123pt)=(0.0162,0.7352,0.7558); rgb(124pt)=(0.0253,0.7376,0.7492); rgb(125pt)=(0.0369,0.74,0.7426); rgb(126pt)=(0.0504,0.7423,0.7359); rgb(127pt)=(0.0638,0.7446,0.7292); rgb(128pt)=(0.077,0.7468,0.7224); rgb(129pt)=(0.0899,0.7489,0.7156); rgb(130pt)=(0.1023,0.751,0.7088); rgb(131pt)=(0.1141,0.7531,0.7019); rgb(132pt)=(0.1252,0.7552,0.695); rgb(133pt)=(0.1354,0.7572,0.6881); rgb(134pt)=(0.1448,0.7593,0.6812); rgb(135pt)=(0.1532,0.7614,0.6741); rgb(136pt)=(0.1609,0.7635,0.6671); rgb(137pt)=(0.1678,0.7656,0.6599); rgb(138pt)=(0.1741,0.7678,0.6527); rgb(139pt)=(0.1799,0.7699,0.6454); rgb(140pt)=(0.1853,0.7721,0.6379); rgb(141pt)=(0.1905,0.7743,0.6303); rgb(142pt)=(0.1954,0.7765,0.6225); rgb(143pt)=(0.2003,0.7787,0.6146); rgb(144pt)=(0.2061,0.7808,0.6065); rgb(145pt)=(0.2118,0.7828,0.5983); rgb(146pt)=(0.2178,0.7849,0.5899); rgb(147pt)=(0.2244,0.7869,0.5813); rgb(148pt)=(0.2318,0.7887,0.5725); rgb(149pt)=(0.2401,0.7905,0.5636); rgb(150pt)=(0.2491,0.7922,0.5546); rgb(151pt)=(0.2589,0.7937,0.5454); rgb(152pt)=(0.2695,0.7951,0.536); rgb(153pt)=(0.2809,0.7964,0.5266); rgb(154pt)=(0.2929,0.7975,0.517); rgb(155pt)=(0.3052,0.7985,0.5074); rgb(156pt)=(0.3176,0.7994,0.4975); rgb(157pt)=(0.3301,0.8002,0.4876); rgb(158pt)=(0.3424,0.8009,0.4774); rgb(159pt)=(0.3548,0.8016,0.4669); rgb(160pt)=(0.3671,0.8021,0.4563); rgb(161pt)=(0.3795,0.8026,0.4454); rgb(162pt)=(0.3921,0.8029,0.4344); rgb(163pt)=(0.405,0.8031,0.4233); rgb(164pt)=(0.4184,0.803,0.4122); rgb(165pt)=(0.4322,0.8028,0.4013); rgb(166pt)=(0.4463,0.8024,0.3904); rgb(167pt)=(0.4608,0.8018,0.3797); rgb(168pt)=(0.4753,0.8011,0.3691); rgb(169pt)=(0.4899,0.8002,0.3586); rgb(170pt)=(0.5044,0.7993,0.348); rgb(171pt)=(0.5187,0.7982,0.3374); rgb(172pt)=(0.5329,0.797,0.3267); rgb(173pt)=(0.547,0.7957,0.3159); rgb(175pt)=(0.5748,0.7929,0.2941); rgb(176pt)=(0.5886,0.7913,0.2833); rgb(177pt)=(0.6024,0.7896,0.2726); rgb(178pt)=(0.6161,0.7878,0.2622); rgb(179pt)=(0.6297,0.7859,0.2521); rgb(180pt)=(0.6433,0.7839,0.2423); rgb(181pt)=(0.6567,0.7818,0.2329); rgb(182pt)=(0.6701,0.7796,0.2239); rgb(183pt)=(0.6833,0.7773,0.2155); rgb(184pt)=(0.6963,0.775,0.2075); rgb(185pt)=(0.7091,0.7727,0.1998); rgb(186pt)=(0.7218,0.7703,0.1924); rgb(187pt)=(0.7344,0.7679,0.1852); rgb(188pt)=(0.7468,0.7654,0.1782); rgb(189pt)=(0.759,0.7629,0.1717); rgb(190pt)=(0.771,0.7604,0.1658); rgb(191pt)=(0.7829,0.7579,0.1608); rgb(192pt)=(0.7945,0.7554,0.157); rgb(193pt)=(0.806,0.7529,0.1546); rgb(194pt)=(0.8172,0.7505,0.1535); rgb(195pt)=(0.8281,0.7481,0.1536); rgb(196pt)=(0.8389,0.7457,0.1546); rgb(197pt)=(0.8495,0.7435,0.1564); rgb(198pt)=(0.86,0.7413,0.1587); rgb(199pt)=(0.8703,0.7392,0.1615); rgb(200pt)=(0.8804,0.7372,0.165); rgb(201pt)=(0.8903,0.7353,0.1695); rgb(202pt)=(0.9,0.7336,0.1749); rgb(203pt)=(0.9093,0.7321,0.1815); rgb(204pt)=(0.9184,0.7308,0.189); rgb(205pt)=(0.9272,0.7298,0.1973); rgb(206pt)=(0.9357,0.729,0.2061); rgb(207pt)=(0.944,0.7285,0.2151); rgb(208pt)=(0.9523,0.7284,0.2237); rgb(209pt)=(0.9606,0.7285,0.2312); rgb(210pt)=(0.9689,0.7292,0.2373); rgb(211pt)=(0.977,0.7304,0.2418); rgb(212pt)=(0.9842,0.733,0.2446); rgb(213pt)=(0.99,0.7365,0.2429); rgb(214pt)=(0.9946,0.7407,0.2394); rgb(215pt)=(0.9966,0.7458,0.2351); rgb(216pt)=(0.9971,0.7513,0.2309); rgb(217pt)=(0.9972,0.7569,0.2267); rgb(218pt)=(0.9971,0.7626,0.2224); rgb(219pt)=(0.9969,0.7683,0.2181); rgb(220pt)=(0.9966,0.774,0.2138); rgb(221pt)=(0.9962,0.7798,0.2095); rgb(222pt)=(0.9957,0.7856,0.2053); rgb(223pt)=(0.9949,0.7915,0.2012); rgb(224pt)=(0.9938,0.7974,0.1974); rgb(225pt)=(0.9923,0.8034,0.1939); rgb(226pt)=(0.9906,0.8095,0.1906); rgb(227pt)=(0.9885,0.8156,0.1875); rgb(228pt)=(0.9861,0.8218,0.1846); rgb(229pt)=(0.9835,0.828,0.1817); rgb(230pt)=(0.9807,0.8342,0.1787); rgb(231pt)=(0.9778,0.8404,0.1757); rgb(232pt)=(0.9748,0.8467,0.1726); rgb(233pt)=(0.972,0.8529,0.1695); rgb(234pt)=(0.9694,0.8591,0.1665); rgb(235pt)=(0.9671,0.8654,0.1636); rgb(236pt)=(0.9651,0.8716,0.1608); rgb(237pt)=(0.9634,0.8778,0.1582); rgb(238pt)=(0.9619,0.884,0.1557); rgb(239pt)=(0.9608,0.8902,0.1532); rgb(240pt)=(0.9601,0.8963,0.1507); rgb(241pt)=(0.9596,0.9023,0.148); rgb(242pt)=(0.9595,0.9084,0.145); rgb(243pt)=(0.9597,0.9143,0.1418); rgb(244pt)=(0.9601,0.9203,0.1382); rgb(245pt)=(0.9608,0.9262,0.1344); rgb(246pt)=(0.9618,0.932,0.1304); rgb(247pt)=(0.9629,0.9379,0.1261); rgb(248pt)=(0.9642,0.9437,0.1216); rgb(249pt)=(0.9657,0.9494,0.1168); rgb(250pt)=(0.9674,0.9552,0.1116); rgb(251pt)=(0.9692,0.9609,0.1061); rgb(252pt)=(0.9711,0.9667,0.1001); rgb(253pt)=(0.973,0.9724,0.0938); rgb(254pt)=(0.9749,0.9782,0.0872); rgb(255pt)=(0.9769,0.9839,0.0805)},
]
\addplot [forget plot] graphics [xmin=62.5, xmax=6012.5, ymin=0.5, ymax=60.5] {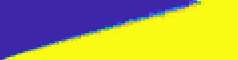};
\addplot [color=red, line width=1.0pt, forget plot]
  table[row sep=crcr]{%
1830	60\\
1827	58\\
1820	56\\
1809	54\\
1794	52\\
1775	50\\
1752	48\\
1725	46\\
1694	44\\
1659	42\\
1620	40\\
1577	38\\
1530	36\\
1479	34\\
1424	32\\
1365	30\\
1302	28\\
1235	26\\
1164	24\\
1089	22\\
1010	20\\
927	18\\
840	16\\
749	14\\
654	12\\
555	10\\
452	8\\
345	6\\
234	4\\
119	2\\
60	1\\
};
\addplot [color=mycolor1, dashdotted, line width=2.0pt, forget plot]
  table[row sep=crcr]{%
5400	60\\
90	1\\
};
\end{axis}
\end{tikzpicture}%

%% file: experiment_GC_ER_good_PT_adaptive_T_7.tex
%
%
\definecolor{mycolor1}{rgb}{1.00000,0.00000,1.00000}%
\begin{tikzpicture}

\begin{axis}[%
width=0.859\figurewidth,
height=\figureheight,
at={(0\figurewidth,0\figureheight)},
scale only axis,
point meta min=0,
point meta max=1,
axis on top,
xmin=62.5,
xmax=6012.5,
xlabel style={font=\color{white!15!black}},
xlabel={$m_{\text{exp}}$},
xticklabels={,,2000,4000,6000},
ymin=0.5,
ymax=60.5,
ylabel style={font=\color{white!15!black}},
axis background/.style={fill=white},
xlabel style={font=\tiny},ylabel style={font=\tiny},
xticklabel style = {font=\tiny}, yticklabel style = {font=\tiny},
colormap={mymap}{[1pt] rgb(0pt)=(0.2422,0.1504,0.6603); rgb(1pt)=(0.2444,0.1534,0.6728); rgb(2pt)=(0.2464,0.1569,0.6847); rgb(3pt)=(0.2484,0.1607,0.6961); rgb(4pt)=(0.2503,0.1648,0.7071); rgb(5pt)=(0.2522,0.1689,0.7179); rgb(6pt)=(0.254,0.1732,0.7286); rgb(7pt)=(0.2558,0.1773,0.7393); rgb(8pt)=(0.2576,0.1814,0.7501); rgb(9pt)=(0.2594,0.1854,0.761); rgb(11pt)=(0.2628,0.1932,0.7828); rgb(12pt)=(0.2645,0.1972,0.7937); rgb(13pt)=(0.2661,0.2011,0.8043); rgb(14pt)=(0.2676,0.2052,0.8148); rgb(15pt)=(0.2691,0.2094,0.8249); rgb(16pt)=(0.2704,0.2138,0.8346); rgb(17pt)=(0.2717,0.2184,0.8439); rgb(18pt)=(0.2729,0.2231,0.8528); rgb(19pt)=(0.274,0.228,0.8612); rgb(20pt)=(0.2749,0.233,0.8692); rgb(21pt)=(0.2758,0.2382,0.8767); rgb(22pt)=(0.2766,0.2435,0.884); rgb(23pt)=(0.2774,0.2489,0.8908); rgb(24pt)=(0.2781,0.2543,0.8973); rgb(25pt)=(0.2788,0.2598,0.9035); rgb(26pt)=(0.2794,0.2653,0.9094); rgb(27pt)=(0.2798,0.2708,0.915); rgb(28pt)=(0.2802,0.2764,0.9204); rgb(29pt)=(0.2806,0.2819,0.9255); rgb(30pt)=(0.2809,0.2875,0.9305); rgb(31pt)=(0.2811,0.293,0.9352); rgb(32pt)=(0.2813,0.2985,0.9397); rgb(33pt)=(0.2814,0.304,0.9441); rgb(34pt)=(0.2814,0.3095,0.9483); rgb(35pt)=(0.2813,0.315,0.9524); rgb(36pt)=(0.2811,0.3204,0.9563); rgb(37pt)=(0.2809,0.3259,0.96); rgb(38pt)=(0.2807,0.3313,0.9636); rgb(39pt)=(0.2803,0.3367,0.967); rgb(40pt)=(0.2798,0.3421,0.9702); rgb(41pt)=(0.2791,0.3475,0.9733); rgb(42pt)=(0.2784,0.3529,0.9763); rgb(43pt)=(0.2776,0.3583,0.9791); rgb(44pt)=(0.2766,0.3638,0.9817); rgb(45pt)=(0.2754,0.3693,0.984); rgb(46pt)=(0.2741,0.3748,0.9862); rgb(47pt)=(0.2726,0.3804,0.9881); rgb(48pt)=(0.271,0.386,0.9898); rgb(49pt)=(0.2691,0.3916,0.9912); rgb(50pt)=(0.267,0.3973,0.9924); rgb(51pt)=(0.2647,0.403,0.9935); rgb(52pt)=(0.2621,0.4088,0.9946); rgb(53pt)=(0.2591,0.4145,0.9955); rgb(54pt)=(0.2556,0.4203,0.9965); rgb(55pt)=(0.2517,0.4261,0.9974); rgb(56pt)=(0.2473,0.4319,0.9983); rgb(57pt)=(0.2424,0.4378,0.9991); rgb(58pt)=(0.2369,0.4437,0.9996); rgb(59pt)=(0.2311,0.4497,0.9995); rgb(60pt)=(0.225,0.4559,0.9985); rgb(61pt)=(0.2189,0.462,0.9968); rgb(62pt)=(0.2128,0.4682,0.9948); rgb(63pt)=(0.2066,0.4743,0.9926); rgb(64pt)=(0.2006,0.4803,0.9906); rgb(65pt)=(0.195,0.4861,0.9887); rgb(66pt)=(0.1903,0.4919,0.9867); rgb(67pt)=(0.1869,0.4975,0.9844); rgb(68pt)=(0.1847,0.503,0.9819); rgb(69pt)=(0.1831,0.5084,0.9793); rgb(70pt)=(0.1818,0.5138,0.9766); rgb(71pt)=(0.1806,0.5191,0.9738); rgb(72pt)=(0.1795,0.5244,0.9709); rgb(73pt)=(0.1785,0.5296,0.9677); rgb(74pt)=(0.1778,0.5349,0.9641); rgb(75pt)=(0.1773,0.5401,0.9602); rgb(76pt)=(0.1768,0.5452,0.956); rgb(77pt)=(0.1764,0.5504,0.9516); rgb(78pt)=(0.1755,0.5554,0.9473); rgb(79pt)=(0.174,0.5605,0.9432); rgb(80pt)=(0.1716,0.5655,0.9393); rgb(81pt)=(0.1686,0.5705,0.9357); rgb(82pt)=(0.1649,0.5755,0.9323); rgb(83pt)=(0.161,0.5805,0.9289); rgb(84pt)=(0.1573,0.5854,0.9254); rgb(85pt)=(0.154,0.5902,0.9218); rgb(86pt)=(0.1513,0.595,0.9182); rgb(87pt)=(0.1492,0.5997,0.9147); rgb(88pt)=(0.1475,0.6043,0.9113); rgb(89pt)=(0.1461,0.6089,0.908); rgb(90pt)=(0.1446,0.6135,0.905); rgb(91pt)=(0.1429,0.618,0.9022); rgb(92pt)=(0.1408,0.6226,0.8998); rgb(93pt)=(0.1383,0.6272,0.8975); rgb(94pt)=(0.1354,0.6317,0.8953); rgb(95pt)=(0.1321,0.6363,0.8932); rgb(96pt)=(0.1288,0.6408,0.891); rgb(97pt)=(0.1253,0.6453,0.8887); rgb(98pt)=(0.1219,0.6497,0.8862); rgb(99pt)=(0.1185,0.6541,0.8834); rgb(100pt)=(0.1152,0.6584,0.8804); rgb(101pt)=(0.1119,0.6627,0.877); rgb(102pt)=(0.1085,0.6669,0.8734); rgb(103pt)=(0.1048,0.671,0.8695); rgb(104pt)=(0.1009,0.675,0.8653); rgb(105pt)=(0.0964,0.6789,0.8609); rgb(106pt)=(0.0914,0.6828,0.8562); rgb(107pt)=(0.0855,0.6865,0.8513); rgb(108pt)=(0.0789,0.6902,0.8462); rgb(109pt)=(0.0713,0.6938,0.8409); rgb(110pt)=(0.0628,0.6972,0.8355); rgb(111pt)=(0.0535,0.7006,0.8299); rgb(112pt)=(0.0433,0.7039,0.8242); rgb(113pt)=(0.0328,0.7071,0.8183); rgb(114pt)=(0.0234,0.7103,0.8124); rgb(115pt)=(0.0155,0.7133,0.8064); rgb(116pt)=(0.0091,0.7163,0.8003); rgb(117pt)=(0.0046,0.7192,0.7941); rgb(118pt)=(0.0019,0.722,0.7878); rgb(119pt)=(0.0009,0.7248,0.7815); rgb(120pt)=(0.0018,0.7275,0.7752); rgb(121pt)=(0.0046,0.7301,0.7688); rgb(122pt)=(0.0094,0.7327,0.7623); rgb(123pt)=(0.0162,0.7352,0.7558); rgb(124pt)=(0.0253,0.7376,0.7492); rgb(125pt)=(0.0369,0.74,0.7426); rgb(126pt)=(0.0504,0.7423,0.7359); rgb(127pt)=(0.0638,0.7446,0.7292); rgb(128pt)=(0.077,0.7468,0.7224); rgb(129pt)=(0.0899,0.7489,0.7156); rgb(130pt)=(0.1023,0.751,0.7088); rgb(131pt)=(0.1141,0.7531,0.7019); rgb(132pt)=(0.1252,0.7552,0.695); rgb(133pt)=(0.1354,0.7572,0.6881); rgb(134pt)=(0.1448,0.7593,0.6812); rgb(135pt)=(0.1532,0.7614,0.6741); rgb(136pt)=(0.1609,0.7635,0.6671); rgb(137pt)=(0.1678,0.7656,0.6599); rgb(138pt)=(0.1741,0.7678,0.6527); rgb(139pt)=(0.1799,0.7699,0.6454); rgb(140pt)=(0.1853,0.7721,0.6379); rgb(141pt)=(0.1905,0.7743,0.6303); rgb(142pt)=(0.1954,0.7765,0.6225); rgb(143pt)=(0.2003,0.7787,0.6146); rgb(144pt)=(0.2061,0.7808,0.6065); rgb(145pt)=(0.2118,0.7828,0.5983); rgb(146pt)=(0.2178,0.7849,0.5899); rgb(147pt)=(0.2244,0.7869,0.5813); rgb(148pt)=(0.2318,0.7887,0.5725); rgb(149pt)=(0.2401,0.7905,0.5636); rgb(150pt)=(0.2491,0.7922,0.5546); rgb(151pt)=(0.2589,0.7937,0.5454); rgb(152pt)=(0.2695,0.7951,0.536); rgb(153pt)=(0.2809,0.7964,0.5266); rgb(154pt)=(0.2929,0.7975,0.517); rgb(155pt)=(0.3052,0.7985,0.5074); rgb(156pt)=(0.3176,0.7994,0.4975); rgb(157pt)=(0.3301,0.8002,0.4876); rgb(158pt)=(0.3424,0.8009,0.4774); rgb(159pt)=(0.3548,0.8016,0.4669); rgb(160pt)=(0.3671,0.8021,0.4563); rgb(161pt)=(0.3795,0.8026,0.4454); rgb(162pt)=(0.3921,0.8029,0.4344); rgb(163pt)=(0.405,0.8031,0.4233); rgb(164pt)=(0.4184,0.803,0.4122); rgb(165pt)=(0.4322,0.8028,0.4013); rgb(166pt)=(0.4463,0.8024,0.3904); rgb(167pt)=(0.4608,0.8018,0.3797); rgb(168pt)=(0.4753,0.8011,0.3691); rgb(169pt)=(0.4899,0.8002,0.3586); rgb(170pt)=(0.5044,0.7993,0.348); rgb(171pt)=(0.5187,0.7982,0.3374); rgb(172pt)=(0.5329,0.797,0.3267); rgb(173pt)=(0.547,0.7957,0.3159); rgb(175pt)=(0.5748,0.7929,0.2941); rgb(176pt)=(0.5886,0.7913,0.2833); rgb(177pt)=(0.6024,0.7896,0.2726); rgb(178pt)=(0.6161,0.7878,0.2622); rgb(179pt)=(0.6297,0.7859,0.2521); rgb(180pt)=(0.6433,0.7839,0.2423); rgb(181pt)=(0.6567,0.7818,0.2329); rgb(182pt)=(0.6701,0.7796,0.2239); rgb(183pt)=(0.6833,0.7773,0.2155); rgb(184pt)=(0.6963,0.775,0.2075); rgb(185pt)=(0.7091,0.7727,0.1998); rgb(186pt)=(0.7218,0.7703,0.1924); rgb(187pt)=(0.7344,0.7679,0.1852); rgb(188pt)=(0.7468,0.7654,0.1782); rgb(189pt)=(0.759,0.7629,0.1717); rgb(190pt)=(0.771,0.7604,0.1658); rgb(191pt)=(0.7829,0.7579,0.1608); rgb(192pt)=(0.7945,0.7554,0.157); rgb(193pt)=(0.806,0.7529,0.1546); rgb(194pt)=(0.8172,0.7505,0.1535); rgb(195pt)=(0.8281,0.7481,0.1536); rgb(196pt)=(0.8389,0.7457,0.1546); rgb(197pt)=(0.8495,0.7435,0.1564); rgb(198pt)=(0.86,0.7413,0.1587); rgb(199pt)=(0.8703,0.7392,0.1615); rgb(200pt)=(0.8804,0.7372,0.165); rgb(201pt)=(0.8903,0.7353,0.1695); rgb(202pt)=(0.9,0.7336,0.1749); rgb(203pt)=(0.9093,0.7321,0.1815); rgb(204pt)=(0.9184,0.7308,0.189); rgb(205pt)=(0.9272,0.7298,0.1973); rgb(206pt)=(0.9357,0.729,0.2061); rgb(207pt)=(0.944,0.7285,0.2151); rgb(208pt)=(0.9523,0.7284,0.2237); rgb(209pt)=(0.9606,0.7285,0.2312); rgb(210pt)=(0.9689,0.7292,0.2373); rgb(211pt)=(0.977,0.7304,0.2418); rgb(212pt)=(0.9842,0.733,0.2446); rgb(213pt)=(0.99,0.7365,0.2429); rgb(214pt)=(0.9946,0.7407,0.2394); rgb(215pt)=(0.9966,0.7458,0.2351); rgb(216pt)=(0.9971,0.7513,0.2309); rgb(217pt)=(0.9972,0.7569,0.2267); rgb(218pt)=(0.9971,0.7626,0.2224); rgb(219pt)=(0.9969,0.7683,0.2181); rgb(220pt)=(0.9966,0.774,0.2138); rgb(221pt)=(0.9962,0.7798,0.2095); rgb(222pt)=(0.9957,0.7856,0.2053); rgb(223pt)=(0.9949,0.7915,0.2012); rgb(224pt)=(0.9938,0.7974,0.1974); rgb(225pt)=(0.9923,0.8034,0.1939); rgb(226pt)=(0.9906,0.8095,0.1906); rgb(227pt)=(0.9885,0.8156,0.1875); rgb(228pt)=(0.9861,0.8218,0.1846); rgb(229pt)=(0.9835,0.828,0.1817); rgb(230pt)=(0.9807,0.8342,0.1787); rgb(231pt)=(0.9778,0.8404,0.1757); rgb(232pt)=(0.9748,0.8467,0.1726); rgb(233pt)=(0.972,0.8529,0.1695); rgb(234pt)=(0.9694,0.8591,0.1665); rgb(235pt)=(0.9671,0.8654,0.1636); rgb(236pt)=(0.9651,0.8716,0.1608); rgb(237pt)=(0.9634,0.8778,0.1582); rgb(238pt)=(0.9619,0.884,0.1557); rgb(239pt)=(0.9608,0.8902,0.1532); rgb(240pt)=(0.9601,0.8963,0.1507); rgb(241pt)=(0.9596,0.9023,0.148); rgb(242pt)=(0.9595,0.9084,0.145); rgb(243pt)=(0.9597,0.9143,0.1418); rgb(244pt)=(0.9601,0.9203,0.1382); rgb(245pt)=(0.9608,0.9262,0.1344); rgb(246pt)=(0.9618,0.932,0.1304); rgb(247pt)=(0.9629,0.9379,0.1261); rgb(248pt)=(0.9642,0.9437,0.1216); rgb(249pt)=(0.9657,0.9494,0.1168); rgb(250pt)=(0.9674,0.9552,0.1116); rgb(251pt)=(0.9692,0.9609,0.1061); rgb(252pt)=(0.9711,0.9667,0.1001); rgb(253pt)=(0.973,0.9724,0.0938); rgb(254pt)=(0.9749,0.9782,0.0872); rgb(255pt)=(0.9769,0.9839,0.0805)},
]
\addplot [forget plot] graphics [xmin=62.5, xmax=6012.5, ymin=0.5, ymax=60.5] {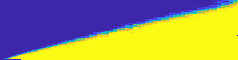};
\addplot [color=red, line width=1.0pt, forget plot]
  table[row sep=crcr]{%
1830	60\\
1829	59\\
1827	58\\
1824	57\\
1820	56\\
1815	55\\
1809	54\\
1802	53\\
1794	52\\
1785	51\\
1775	50\\
1764	49\\
1752	48\\
1739	47\\
1725	46\\
1710	45\\
1694	44\\
1677	43\\
1659	42\\
1640	41\\
1620	40\\
1599	39\\
1577	38\\
1554	37\\
1530	36\\
1505	35\\
1479	34\\
1452	33\\
1424	32\\
1395	31\\
1365	30\\
1334	29\\
1302	28\\
1269	27\\
1235	26\\
1200	25\\
1164	24\\
1127	23\\
1089	22\\
1050	21\\
1010	20\\
969	19\\
927	18\\
884	17\\
840	16\\
795	15\\
749	14\\
702	13\\
654	12\\
605	11\\
555	10\\
504	9\\
452	8\\
399	7\\
345	6\\
290	5\\
234	4\\
177	3\\
119	2\\
60	1\\
};
\addplot [color=mycolor1, dashdotted, line width=2.0pt, forget plot]
  table[row sep=crcr]{%
5400	60\\
90	1\\
};
\end{axis}
\end{tikzpicture}%

%% file: experiment_GC_Community_uniform_T_7.tex
%
%
\definecolor{mycolor1}{rgb}{1.00000,0.00000,1.00000}%
\begin{tikzpicture}

\begin{axis}[%
width=0.876\figurewidth,
height=\figureheight,
at={(0\figurewidth,0\figureheight)},
scale only axis,
point meta min=0,
point meta max=1,
axis on top,
xmin=0.0625,
xmax=15.0125,
xlabel style={font=\color{white!15!black}},
xlabel={Nr. of expected samples $m_{\text{exp}}$},
xtick={3,9,15},
minor x tick num=1,
xticklabels={3000,9000,15000},
ymin=0.5,
ymax=60.5,
ylabel style={font=\color{white!15!black}},
ylabel={Rank $r$},
axis background/.style={fill=white},
xlabel style={font=\tiny},ylabel style={font=\tiny},
xticklabel style = {font=\tiny}, yticklabel style = {font=\tiny}, 
colormap={mymap}{[1pt] rgb(0pt)=(0.2422,0.1504,0.6603); rgb(1pt)=(0.2444,0.1534,0.6728); rgb(2pt)=(0.2464,0.1569,0.6847); rgb(3pt)=(0.2484,0.1607,0.6961); rgb(4pt)=(0.2503,0.1648,0.7071); rgb(5pt)=(0.2522,0.1689,0.7179); rgb(6pt)=(0.254,0.1732,0.7286); rgb(7pt)=(0.2558,0.1773,0.7393); rgb(8pt)=(0.2576,0.1814,0.7501); rgb(9pt)=(0.2594,0.1854,0.761); rgb(11pt)=(0.2628,0.1932,0.7828); rgb(12pt)=(0.2645,0.1972,0.7937); rgb(13pt)=(0.2661,0.2011,0.8043); rgb(14pt)=(0.2676,0.2052,0.8148); rgb(15pt)=(0.2691,0.2094,0.8249); rgb(16pt)=(0.2704,0.2138,0.8346); rgb(17pt)=(0.2717,0.2184,0.8439); rgb(18pt)=(0.2729,0.2231,0.8528); rgb(19pt)=(0.274,0.228,0.8612); rgb(20pt)=(0.2749,0.233,0.8692); rgb(21pt)=(0.2758,0.2382,0.8767); rgb(22pt)=(0.2766,0.2435,0.884); rgb(23pt)=(0.2774,0.2489,0.8908); rgb(24pt)=(0.2781,0.2543,0.8973); rgb(25pt)=(0.2788,0.2598,0.9035); rgb(26pt)=(0.2794,0.2653,0.9094); rgb(27pt)=(0.2798,0.2708,0.915); rgb(28pt)=(0.2802,0.2764,0.9204); rgb(29pt)=(0.2806,0.2819,0.9255); rgb(30pt)=(0.2809,0.2875,0.9305); rgb(31pt)=(0.2811,0.293,0.9352); rgb(32pt)=(0.2813,0.2985,0.9397); rgb(33pt)=(0.2814,0.304,0.9441); rgb(34pt)=(0.2814,0.3095,0.9483); rgb(35pt)=(0.2813,0.315,0.9524); rgb(36pt)=(0.2811,0.3204,0.9563); rgb(37pt)=(0.2809,0.3259,0.96); rgb(38pt)=(0.2807,0.3313,0.9636); rgb(39pt)=(0.2803,0.3367,0.967); rgb(40pt)=(0.2798,0.3421,0.9702); rgb(41pt)=(0.2791,0.3475,0.9733); rgb(42pt)=(0.2784,0.3529,0.9763); rgb(43pt)=(0.2776,0.3583,0.9791); rgb(44pt)=(0.2766,0.3638,0.9817); rgb(45pt)=(0.2754,0.3693,0.984); rgb(46pt)=(0.2741,0.3748,0.9862); rgb(47pt)=(0.2726,0.3804,0.9881); rgb(48pt)=(0.271,0.386,0.9898); rgb(49pt)=(0.2691,0.3916,0.9912); rgb(50pt)=(0.267,0.3973,0.9924); rgb(51pt)=(0.2647,0.403,0.9935); rgb(52pt)=(0.2621,0.4088,0.9946); rgb(53pt)=(0.2591,0.4145,0.9955); rgb(54pt)=(0.2556,0.4203,0.9965); rgb(55pt)=(0.2517,0.4261,0.9974); rgb(56pt)=(0.2473,0.4319,0.9983); rgb(57pt)=(0.2424,0.4378,0.9991); rgb(58pt)=(0.2369,0.4437,0.9996); rgb(59pt)=(0.2311,0.4497,0.9995); rgb(60pt)=(0.225,0.4559,0.9985); rgb(61pt)=(0.2189,0.462,0.9968); rgb(62pt)=(0.2128,0.4682,0.9948); rgb(63pt)=(0.2066,0.4743,0.9926); rgb(64pt)=(0.2006,0.4803,0.9906); rgb(65pt)=(0.195,0.4861,0.9887); rgb(66pt)=(0.1903,0.4919,0.9867); rgb(67pt)=(0.1869,0.4975,0.9844); rgb(68pt)=(0.1847,0.503,0.9819); rgb(69pt)=(0.1831,0.5084,0.9793); rgb(70pt)=(0.1818,0.5138,0.9766); rgb(71pt)=(0.1806,0.5191,0.9738); rgb(72pt)=(0.1795,0.5244,0.9709); rgb(73pt)=(0.1785,0.5296,0.9677); rgb(74pt)=(0.1778,0.5349,0.9641); rgb(75pt)=(0.1773,0.5401,0.9602); rgb(76pt)=(0.1768,0.5452,0.956); rgb(77pt)=(0.1764,0.5504,0.9516); rgb(78pt)=(0.1755,0.5554,0.9473); rgb(79pt)=(0.174,0.5605,0.9432); rgb(80pt)=(0.1716,0.5655,0.9393); rgb(81pt)=(0.1686,0.5705,0.9357); rgb(82pt)=(0.1649,0.5755,0.9323); rgb(83pt)=(0.161,0.5805,0.9289); rgb(84pt)=(0.1573,0.5854,0.9254); rgb(85pt)=(0.154,0.5902,0.9218); rgb(86pt)=(0.1513,0.595,0.9182); rgb(87pt)=(0.1492,0.5997,0.9147); rgb(88pt)=(0.1475,0.6043,0.9113); rgb(89pt)=(0.1461,0.6089,0.908); rgb(90pt)=(0.1446,0.6135,0.905); rgb(91pt)=(0.1429,0.618,0.9022); rgb(92pt)=(0.1408,0.6226,0.8998); rgb(93pt)=(0.1383,0.6272,0.8975); rgb(94pt)=(0.1354,0.6317,0.8953); rgb(95pt)=(0.1321,0.6363,0.8932); rgb(96pt)=(0.1288,0.6408,0.891); rgb(97pt)=(0.1253,0.6453,0.8887); rgb(98pt)=(0.1219,0.6497,0.8862); rgb(99pt)=(0.1185,0.6541,0.8834); rgb(100pt)=(0.1152,0.6584,0.8804); rgb(101pt)=(0.1119,0.6627,0.877); rgb(102pt)=(0.1085,0.6669,0.8734); rgb(103pt)=(0.1048,0.671,0.8695); rgb(104pt)=(0.1009,0.675,0.8653); rgb(105pt)=(0.0964,0.6789,0.8609); rgb(106pt)=(0.0914,0.6828,0.8562); rgb(107pt)=(0.0855,0.6865,0.8513); rgb(108pt)=(0.0789,0.6902,0.8462); rgb(109pt)=(0.0713,0.6938,0.8409); rgb(110pt)=(0.0628,0.6972,0.8355); rgb(111pt)=(0.0535,0.7006,0.8299); rgb(112pt)=(0.0433,0.7039,0.8242); rgb(113pt)=(0.0328,0.7071,0.8183); rgb(114pt)=(0.0234,0.7103,0.8124); rgb(115pt)=(0.0155,0.7133,0.8064); rgb(116pt)=(0.0091,0.7163,0.8003); rgb(117pt)=(0.0046,0.7192,0.7941); rgb(118pt)=(0.0019,0.722,0.7878); rgb(119pt)=(0.0009,0.7248,0.7815); rgb(120pt)=(0.0018,0.7275,0.7752); rgb(121pt)=(0.0046,0.7301,0.7688); rgb(122pt)=(0.0094,0.7327,0.7623); rgb(123pt)=(0.0162,0.7352,0.7558); rgb(124pt)=(0.0253,0.7376,0.7492); rgb(125pt)=(0.0369,0.74,0.7426); rgb(126pt)=(0.0504,0.7423,0.7359); rgb(127pt)=(0.0638,0.7446,0.7292); rgb(128pt)=(0.077,0.7468,0.7224); rgb(129pt)=(0.0899,0.7489,0.7156); rgb(130pt)=(0.1023,0.751,0.7088); rgb(131pt)=(0.1141,0.7531,0.7019); rgb(132pt)=(0.1252,0.7552,0.695); rgb(133pt)=(0.1354,0.7572,0.6881); rgb(134pt)=(0.1448,0.7593,0.6812); rgb(135pt)=(0.1532,0.7614,0.6741); rgb(136pt)=(0.1609,0.7635,0.6671); rgb(137pt)=(0.1678,0.7656,0.6599); rgb(138pt)=(0.1741,0.7678,0.6527); rgb(139pt)=(0.1799,0.7699,0.6454); rgb(140pt)=(0.1853,0.7721,0.6379); rgb(141pt)=(0.1905,0.7743,0.6303); rgb(142pt)=(0.1954,0.7765,0.6225); rgb(143pt)=(0.2003,0.7787,0.6146); rgb(144pt)=(0.2061,0.7808,0.6065); rgb(145pt)=(0.2118,0.7828,0.5983); rgb(146pt)=(0.2178,0.7849,0.5899); rgb(147pt)=(0.2244,0.7869,0.5813); rgb(148pt)=(0.2318,0.7887,0.5725); rgb(149pt)=(0.2401,0.7905,0.5636); rgb(150pt)=(0.2491,0.7922,0.5546); rgb(151pt)=(0.2589,0.7937,0.5454); rgb(152pt)=(0.2695,0.7951,0.536); rgb(153pt)=(0.2809,0.7964,0.5266); rgb(154pt)=(0.2929,0.7975,0.517); rgb(155pt)=(0.3052,0.7985,0.5074); rgb(156pt)=(0.3176,0.7994,0.4975); rgb(157pt)=(0.3301,0.8002,0.4876); rgb(158pt)=(0.3424,0.8009,0.4774); rgb(159pt)=(0.3548,0.8016,0.4669); rgb(160pt)=(0.3671,0.8021,0.4563); rgb(161pt)=(0.3795,0.8026,0.4454); rgb(162pt)=(0.3921,0.8029,0.4344); rgb(163pt)=(0.405,0.8031,0.4233); rgb(164pt)=(0.4184,0.803,0.4122); rgb(165pt)=(0.4322,0.8028,0.4013); rgb(166pt)=(0.4463,0.8024,0.3904); rgb(167pt)=(0.4608,0.8018,0.3797); rgb(168pt)=(0.4753,0.8011,0.3691); rgb(169pt)=(0.4899,0.8002,0.3586); rgb(170pt)=(0.5044,0.7993,0.348); rgb(171pt)=(0.5187,0.7982,0.3374); rgb(172pt)=(0.5329,0.797,0.3267); rgb(173pt)=(0.547,0.7957,0.3159); rgb(175pt)=(0.5748,0.7929,0.2941); rgb(176pt)=(0.5886,0.7913,0.2833); rgb(177pt)=(0.6024,0.7896,0.2726); rgb(178pt)=(0.6161,0.7878,0.2622); rgb(179pt)=(0.6297,0.7859,0.2521); rgb(180pt)=(0.6433,0.7839,0.2423); rgb(181pt)=(0.6567,0.7818,0.2329); rgb(182pt)=(0.6701,0.7796,0.2239); rgb(183pt)=(0.6833,0.7773,0.2155); rgb(184pt)=(0.6963,0.775,0.2075); rgb(185pt)=(0.7091,0.7727,0.1998); rgb(186pt)=(0.7218,0.7703,0.1924); rgb(187pt)=(0.7344,0.7679,0.1852); rgb(188pt)=(0.7468,0.7654,0.1782); rgb(189pt)=(0.759,0.7629,0.1717); rgb(190pt)=(0.771,0.7604,0.1658); rgb(191pt)=(0.7829,0.7579,0.1608); rgb(192pt)=(0.7945,0.7554,0.157); rgb(193pt)=(0.806,0.7529,0.1546); rgb(194pt)=(0.8172,0.7505,0.1535); rgb(195pt)=(0.8281,0.7481,0.1536); rgb(196pt)=(0.8389,0.7457,0.1546); rgb(197pt)=(0.8495,0.7435,0.1564); rgb(198pt)=(0.86,0.7413,0.1587); rgb(199pt)=(0.8703,0.7392,0.1615); rgb(200pt)=(0.8804,0.7372,0.165); rgb(201pt)=(0.8903,0.7353,0.1695); rgb(202pt)=(0.9,0.7336,0.1749); rgb(203pt)=(0.9093,0.7321,0.1815); rgb(204pt)=(0.9184,0.7308,0.189); rgb(205pt)=(0.9272,0.7298,0.1973); rgb(206pt)=(0.9357,0.729,0.2061); rgb(207pt)=(0.944,0.7285,0.2151); rgb(208pt)=(0.9523,0.7284,0.2237); rgb(209pt)=(0.9606,0.7285,0.2312); rgb(210pt)=(0.9689,0.7292,0.2373); rgb(211pt)=(0.977,0.7304,0.2418); rgb(212pt)=(0.9842,0.733,0.2446); rgb(213pt)=(0.99,0.7365,0.2429); rgb(214pt)=(0.9946,0.7407,0.2394); rgb(215pt)=(0.9966,0.7458,0.2351); rgb(216pt)=(0.9971,0.7513,0.2309); rgb(217pt)=(0.9972,0.7569,0.2267); rgb(218pt)=(0.9971,0.7626,0.2224); rgb(219pt)=(0.9969,0.7683,0.2181); rgb(220pt)=(0.9966,0.774,0.2138); rgb(221pt)=(0.9962,0.7798,0.2095); rgb(222pt)=(0.9957,0.7856,0.2053); rgb(223pt)=(0.9949,0.7915,0.2012); rgb(224pt)=(0.9938,0.7974,0.1974); rgb(225pt)=(0.9923,0.8034,0.1939); rgb(226pt)=(0.9906,0.8095,0.1906); rgb(227pt)=(0.9885,0.8156,0.1875); rgb(228pt)=(0.9861,0.8218,0.1846); rgb(229pt)=(0.9835,0.828,0.1817); rgb(230pt)=(0.9807,0.8342,0.1787); rgb(231pt)=(0.9778,0.8404,0.1757); rgb(232pt)=(0.9748,0.8467,0.1726); rgb(233pt)=(0.972,0.8529,0.1695); rgb(234pt)=(0.9694,0.8591,0.1665); rgb(235pt)=(0.9671,0.8654,0.1636); rgb(236pt)=(0.9651,0.8716,0.1608); rgb(237pt)=(0.9634,0.8778,0.1582); rgb(238pt)=(0.9619,0.884,0.1557); rgb(239pt)=(0.9608,0.8902,0.1532); rgb(240pt)=(0.9601,0.8963,0.1507); rgb(241pt)=(0.9596,0.9023,0.148); rgb(242pt)=(0.9595,0.9084,0.145); rgb(243pt)=(0.9597,0.9143,0.1418); rgb(244pt)=(0.9601,0.9203,0.1382); rgb(245pt)=(0.9608,0.9262,0.1344); rgb(246pt)=(0.9618,0.932,0.1304); rgb(247pt)=(0.9629,0.9379,0.1261); rgb(248pt)=(0.9642,0.9437,0.1216); rgb(249pt)=(0.9657,0.9494,0.1168); rgb(250pt)=(0.9674,0.9552,0.1116); rgb(251pt)=(0.9692,0.9609,0.1061); rgb(252pt)=(0.9711,0.9667,0.1001); rgb(253pt)=(0.973,0.9724,0.0938); rgb(254pt)=(0.9749,0.9782,0.0872); rgb(255pt)=(0.9769,0.9839,0.0805)}
]
\addplot [forget plot] graphics [xmin=0.0625, xmax=15.0125, ymin=0.5, ymax=60.5] {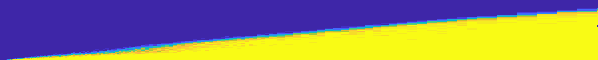};
\addplot [color=red, line width=1.0pt, forget plot]
  table[row sep=crcr]{%
3.600	60\\
3.599	59\\
3.596	58\\
3.591	57\\
3.584	56\\
3.575	55\\
3.564	54\\
3.551	53\\
3.536	52\\
3.519	51\\
3.500	50\\
3.479	49\\
3.456	48\\
3.431	47\\
3.404	46\\
3.375	45\\
3.344	44\\
3.311	43\\
3.276	42\\
3.239	41\\
3.200	40\\
3.159	39\\
3.116	38\\
3.071	37\\
3.024	36\\
2.975	35\\
2.924	34\\
2.871	33\\
2.816	32\\
2.759	31\\
2.700	30\\
2.639	29\\
2.576	28\\
2.511	27\\
2.444	26\\
2.375	25\\
2.304	24\\
2.231	23\\
2.156	22\\
2.079	21\\
2.000	20\\
1.919	19\\
1.836	18\\
1.751	17\\
1.664	16\\
1.575	15\\
1.484	14\\
1.391	13\\
1.296	12\\
1.199	11\\
1.100	10\\
0.999	9\\
0.896	8\\
0.791	7\\
0.684	6\\
0.575	5\\
0.464	4\\
0.351	3\\
0.236	2\\
0.119	1\\
};
\addplot [color=mycolor1, dashdotted, line width=2.0pt, forget plot]
  table[row sep=crcr]{%
15.264	53\\
0.288	1\\
};
\end{axis}
\end{tikzpicture}%

%% file: experiment_GC_Community_adaptive_T_7.tex
%
%
\definecolor{mycolor1}{rgb}{1.00000,0.00000,1.00000}%
\begin{tikzpicture}

\begin{axis}[%
width=0.876\figurewidth,
height=\figureheight,
at={(0\figurewidth,0\figureheight)},
scale only axis,
point meta min=0,
point meta max=1,
axis on top,
xmin=0.0625,
xmax=15.0125,
xlabel style={font=\color{white!15!black}},
xlabel={$m_{\text{exp}}$},
xtick={3,9,15},
minor x tick num=1,
xticklabels={3000,9000,15000},
ymin=0.5,
ymax=60.5,
ylabel style={font=\color{white!15!black}},
axis background/.style={fill=white},
xlabel style={font=\tiny},ylabel style={font=\tiny},
xticklabel style = {font=\tiny}, yticklabel style = {font=\tiny}, 
colormap={mymap}{[1pt] rgb(0pt)=(0.2422,0.1504,0.6603); rgb(1pt)=(0.2444,0.1534,0.6728); rgb(2pt)=(0.2464,0.1569,0.6847); rgb(3pt)=(0.2484,0.1607,0.6961); rgb(4pt)=(0.2503,0.1648,0.7071); rgb(5pt)=(0.2522,0.1689,0.7179); rgb(6pt)=(0.254,0.1732,0.7286); rgb(7pt)=(0.2558,0.1773,0.7393); rgb(8pt)=(0.2576,0.1814,0.7501); rgb(9pt)=(0.2594,0.1854,0.761); rgb(11pt)=(0.2628,0.1932,0.7828); rgb(12pt)=(0.2645,0.1972,0.7937); rgb(13pt)=(0.2661,0.2011,0.8043); rgb(14pt)=(0.2676,0.2052,0.8148); rgb(15pt)=(0.2691,0.2094,0.8249); rgb(16pt)=(0.2704,0.2138,0.8346); rgb(17pt)=(0.2717,0.2184,0.8439); rgb(18pt)=(0.2729,0.2231,0.8528); rgb(19pt)=(0.274,0.228,0.8612); rgb(20pt)=(0.2749,0.233,0.8692); rgb(21pt)=(0.2758,0.2382,0.8767); rgb(22pt)=(0.2766,0.2435,0.884); rgb(23pt)=(0.2774,0.2489,0.8908); rgb(24pt)=(0.2781,0.2543,0.8973); rgb(25pt)=(0.2788,0.2598,0.9035); rgb(26pt)=(0.2794,0.2653,0.9094); rgb(27pt)=(0.2798,0.2708,0.915); rgb(28pt)=(0.2802,0.2764,0.9204); rgb(29pt)=(0.2806,0.2819,0.9255); rgb(30pt)=(0.2809,0.2875,0.9305); rgb(31pt)=(0.2811,0.293,0.9352); rgb(32pt)=(0.2813,0.2985,0.9397); rgb(33pt)=(0.2814,0.304,0.9441); rgb(34pt)=(0.2814,0.3095,0.9483); rgb(35pt)=(0.2813,0.315,0.9524); rgb(36pt)=(0.2811,0.3204,0.9563); rgb(37pt)=(0.2809,0.3259,0.96); rgb(38pt)=(0.2807,0.3313,0.9636); rgb(39pt)=(0.2803,0.3367,0.967); rgb(40pt)=(0.2798,0.3421,0.9702); rgb(41pt)=(0.2791,0.3475,0.9733); rgb(42pt)=(0.2784,0.3529,0.9763); rgb(43pt)=(0.2776,0.3583,0.9791); rgb(44pt)=(0.2766,0.3638,0.9817); rgb(45pt)=(0.2754,0.3693,0.984); rgb(46pt)=(0.2741,0.3748,0.9862); rgb(47pt)=(0.2726,0.3804,0.9881); rgb(48pt)=(0.271,0.386,0.9898); rgb(49pt)=(0.2691,0.3916,0.9912); rgb(50pt)=(0.267,0.3973,0.9924); rgb(51pt)=(0.2647,0.403,0.9935); rgb(52pt)=(0.2621,0.4088,0.9946); rgb(53pt)=(0.2591,0.4145,0.9955); rgb(54pt)=(0.2556,0.4203,0.9965); rgb(55pt)=(0.2517,0.4261,0.9974); rgb(56pt)=(0.2473,0.4319,0.9983); rgb(57pt)=(0.2424,0.4378,0.9991); rgb(58pt)=(0.2369,0.4437,0.9996); rgb(59pt)=(0.2311,0.4497,0.9995); rgb(60pt)=(0.225,0.4559,0.9985); rgb(61pt)=(0.2189,0.462,0.9968); rgb(62pt)=(0.2128,0.4682,0.9948); rgb(63pt)=(0.2066,0.4743,0.9926); rgb(64pt)=(0.2006,0.4803,0.9906); rgb(65pt)=(0.195,0.4861,0.9887); rgb(66pt)=(0.1903,0.4919,0.9867); rgb(67pt)=(0.1869,0.4975,0.9844); rgb(68pt)=(0.1847,0.503,0.9819); rgb(69pt)=(0.1831,0.5084,0.9793); rgb(70pt)=(0.1818,0.5138,0.9766); rgb(71pt)=(0.1806,0.5191,0.9738); rgb(72pt)=(0.1795,0.5244,0.9709); rgb(73pt)=(0.1785,0.5296,0.9677); rgb(74pt)=(0.1778,0.5349,0.9641); rgb(75pt)=(0.1773,0.5401,0.9602); rgb(76pt)=(0.1768,0.5452,0.956); rgb(77pt)=(0.1764,0.5504,0.9516); rgb(78pt)=(0.1755,0.5554,0.9473); rgb(79pt)=(0.174,0.5605,0.9432); rgb(80pt)=(0.1716,0.5655,0.9393); rgb(81pt)=(0.1686,0.5705,0.9357); rgb(82pt)=(0.1649,0.5755,0.9323); rgb(83pt)=(0.161,0.5805,0.9289); rgb(84pt)=(0.1573,0.5854,0.9254); rgb(85pt)=(0.154,0.5902,0.9218); rgb(86pt)=(0.1513,0.595,0.9182); rgb(87pt)=(0.1492,0.5997,0.9147); rgb(88pt)=(0.1475,0.6043,0.9113); rgb(89pt)=(0.1461,0.6089,0.908); rgb(90pt)=(0.1446,0.6135,0.905); rgb(91pt)=(0.1429,0.618,0.9022); rgb(92pt)=(0.1408,0.6226,0.8998); rgb(93pt)=(0.1383,0.6272,0.8975); rgb(94pt)=(0.1354,0.6317,0.8953); rgb(95pt)=(0.1321,0.6363,0.8932); rgb(96pt)=(0.1288,0.6408,0.891); rgb(97pt)=(0.1253,0.6453,0.8887); rgb(98pt)=(0.1219,0.6497,0.8862); rgb(99pt)=(0.1185,0.6541,0.8834); rgb(100pt)=(0.1152,0.6584,0.8804); rgb(101pt)=(0.1119,0.6627,0.877); rgb(102pt)=(0.1085,0.6669,0.8734); rgb(103pt)=(0.1048,0.671,0.8695); rgb(104pt)=(0.1009,0.675,0.8653); rgb(105pt)=(0.0964,0.6789,0.8609); rgb(106pt)=(0.0914,0.6828,0.8562); rgb(107pt)=(0.0855,0.6865,0.8513); rgb(108pt)=(0.0789,0.6902,0.8462); rgb(109pt)=(0.0713,0.6938,0.8409); rgb(110pt)=(0.0628,0.6972,0.8355); rgb(111pt)=(0.0535,0.7006,0.8299); rgb(112pt)=(0.0433,0.7039,0.8242); rgb(113pt)=(0.0328,0.7071,0.8183); rgb(114pt)=(0.0234,0.7103,0.8124); rgb(115pt)=(0.0155,0.7133,0.8064); rgb(116pt)=(0.0091,0.7163,0.8003); rgb(117pt)=(0.0046,0.7192,0.7941); rgb(118pt)=(0.0019,0.722,0.7878); rgb(119pt)=(0.0009,0.7248,0.7815); rgb(120pt)=(0.0018,0.7275,0.7752); rgb(121pt)=(0.0046,0.7301,0.7688); rgb(122pt)=(0.0094,0.7327,0.7623); rgb(123pt)=(0.0162,0.7352,0.7558); rgb(124pt)=(0.0253,0.7376,0.7492); rgb(125pt)=(0.0369,0.74,0.7426); rgb(126pt)=(0.0504,0.7423,0.7359); rgb(127pt)=(0.0638,0.7446,0.7292); rgb(128pt)=(0.077,0.7468,0.7224); rgb(129pt)=(0.0899,0.7489,0.7156); rgb(130pt)=(0.1023,0.751,0.7088); rgb(131pt)=(0.1141,0.7531,0.7019); rgb(132pt)=(0.1252,0.7552,0.695); rgb(133pt)=(0.1354,0.7572,0.6881); rgb(134pt)=(0.1448,0.7593,0.6812); rgb(135pt)=(0.1532,0.7614,0.6741); rgb(136pt)=(0.1609,0.7635,0.6671); rgb(137pt)=(0.1678,0.7656,0.6599); rgb(138pt)=(0.1741,0.7678,0.6527); rgb(139pt)=(0.1799,0.7699,0.6454); rgb(140pt)=(0.1853,0.7721,0.6379); rgb(141pt)=(0.1905,0.7743,0.6303); rgb(142pt)=(0.1954,0.7765,0.6225); rgb(143pt)=(0.2003,0.7787,0.6146); rgb(144pt)=(0.2061,0.7808,0.6065); rgb(145pt)=(0.2118,0.7828,0.5983); rgb(146pt)=(0.2178,0.7849,0.5899); rgb(147pt)=(0.2244,0.7869,0.5813); rgb(148pt)=(0.2318,0.7887,0.5725); rgb(149pt)=(0.2401,0.7905,0.5636); rgb(150pt)=(0.2491,0.7922,0.5546); rgb(151pt)=(0.2589,0.7937,0.5454); rgb(152pt)=(0.2695,0.7951,0.536); rgb(153pt)=(0.2809,0.7964,0.5266); rgb(154pt)=(0.2929,0.7975,0.517); rgb(155pt)=(0.3052,0.7985,0.5074); rgb(156pt)=(0.3176,0.7994,0.4975); rgb(157pt)=(0.3301,0.8002,0.4876); rgb(158pt)=(0.3424,0.8009,0.4774); rgb(159pt)=(0.3548,0.8016,0.4669); rgb(160pt)=(0.3671,0.8021,0.4563); rgb(161pt)=(0.3795,0.8026,0.4454); rgb(162pt)=(0.3921,0.8029,0.4344); rgb(163pt)=(0.405,0.8031,0.4233); rgb(164pt)=(0.4184,0.803,0.4122); rgb(165pt)=(0.4322,0.8028,0.4013); rgb(166pt)=(0.4463,0.8024,0.3904); rgb(167pt)=(0.4608,0.8018,0.3797); rgb(168pt)=(0.4753,0.8011,0.3691); rgb(169pt)=(0.4899,0.8002,0.3586); rgb(170pt)=(0.5044,0.7993,0.348); rgb(171pt)=(0.5187,0.7982,0.3374); rgb(172pt)=(0.5329,0.797,0.3267); rgb(173pt)=(0.547,0.7957,0.3159); rgb(175pt)=(0.5748,0.7929,0.2941); rgb(176pt)=(0.5886,0.7913,0.2833); rgb(177pt)=(0.6024,0.7896,0.2726); rgb(178pt)=(0.6161,0.7878,0.2622); rgb(179pt)=(0.6297,0.7859,0.2521); rgb(180pt)=(0.6433,0.7839,0.2423); rgb(181pt)=(0.6567,0.7818,0.2329); rgb(182pt)=(0.6701,0.7796,0.2239); rgb(183pt)=(0.6833,0.7773,0.2155); rgb(184pt)=(0.6963,0.775,0.2075); rgb(185pt)=(0.7091,0.7727,0.1998); rgb(186pt)=(0.7218,0.7703,0.1924); rgb(187pt)=(0.7344,0.7679,0.1852); rgb(188pt)=(0.7468,0.7654,0.1782); rgb(189pt)=(0.759,0.7629,0.1717); rgb(190pt)=(0.771,0.7604,0.1658); rgb(191pt)=(0.7829,0.7579,0.1608); rgb(192pt)=(0.7945,0.7554,0.157); rgb(193pt)=(0.806,0.7529,0.1546); rgb(194pt)=(0.8172,0.7505,0.1535); rgb(195pt)=(0.8281,0.7481,0.1536); rgb(196pt)=(0.8389,0.7457,0.1546); rgb(197pt)=(0.8495,0.7435,0.1564); rgb(198pt)=(0.86,0.7413,0.1587); rgb(199pt)=(0.8703,0.7392,0.1615); rgb(200pt)=(0.8804,0.7372,0.165); rgb(201pt)=(0.8903,0.7353,0.1695); rgb(202pt)=(0.9,0.7336,0.1749); rgb(203pt)=(0.9093,0.7321,0.1815); rgb(204pt)=(0.9184,0.7308,0.189); rgb(205pt)=(0.9272,0.7298,0.1973); rgb(206pt)=(0.9357,0.729,0.2061); rgb(207pt)=(0.944,0.7285,0.2151); rgb(208pt)=(0.9523,0.7284,0.2237); rgb(209pt)=(0.9606,0.7285,0.2312); rgb(210pt)=(0.9689,0.7292,0.2373); rgb(211pt)=(0.977,0.7304,0.2418); rgb(212pt)=(0.9842,0.733,0.2446); rgb(213pt)=(0.99,0.7365,0.2429); rgb(214pt)=(0.9946,0.7407,0.2394); rgb(215pt)=(0.9966,0.7458,0.2351); rgb(216pt)=(0.9971,0.7513,0.2309); rgb(217pt)=(0.9972,0.7569,0.2267); rgb(218pt)=(0.9971,0.7626,0.2224); rgb(219pt)=(0.9969,0.7683,0.2181); rgb(220pt)=(0.9966,0.774,0.2138); rgb(221pt)=(0.9962,0.7798,0.2095); rgb(222pt)=(0.9957,0.7856,0.2053); rgb(223pt)=(0.9949,0.7915,0.2012); rgb(224pt)=(0.9938,0.7974,0.1974); rgb(225pt)=(0.9923,0.8034,0.1939); rgb(226pt)=(0.9906,0.8095,0.1906); rgb(227pt)=(0.9885,0.8156,0.1875); rgb(228pt)=(0.9861,0.8218,0.1846); rgb(229pt)=(0.9835,0.828,0.1817); rgb(230pt)=(0.9807,0.8342,0.1787); rgb(231pt)=(0.9778,0.8404,0.1757); rgb(232pt)=(0.9748,0.8467,0.1726); rgb(233pt)=(0.972,0.8529,0.1695); rgb(234pt)=(0.9694,0.8591,0.1665); rgb(235pt)=(0.9671,0.8654,0.1636); rgb(236pt)=(0.9651,0.8716,0.1608); rgb(237pt)=(0.9634,0.8778,0.1582); rgb(238pt)=(0.9619,0.884,0.1557); rgb(239pt)=(0.9608,0.8902,0.1532); rgb(240pt)=(0.9601,0.8963,0.1507); rgb(241pt)=(0.9596,0.9023,0.148); rgb(242pt)=(0.9595,0.9084,0.145); rgb(243pt)=(0.9597,0.9143,0.1418); rgb(244pt)=(0.9601,0.9203,0.1382); rgb(245pt)=(0.9608,0.9262,0.1344); rgb(246pt)=(0.9618,0.932,0.1304); rgb(247pt)=(0.9629,0.9379,0.1261); rgb(248pt)=(0.9642,0.9437,0.1216); rgb(249pt)=(0.9657,0.9494,0.1168); rgb(250pt)=(0.9674,0.9552,0.1116); rgb(251pt)=(0.9692,0.9609,0.1061); rgb(252pt)=(0.9711,0.9667,0.1001); rgb(253pt)=(0.973,0.9724,0.0938); rgb(254pt)=(0.9749,0.9782,0.0872); rgb(255pt)=(0.9769,0.9839,0.0805)}
]
\addplot [forget plot] graphics [xmin=0.0625, xmax=15.0125, ymin=0.5, ymax=60.5] {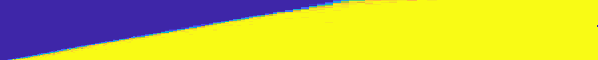};
\addplot [color=red, line width=1.0pt, forget plot]
  table[row sep=crcr]{%
3.600	60\\
3.599	59\\
3.596	58\\
3.591	57\\
3.584	56\\
3.575	55\\
3.564	54\\
3.551	53\\
3.536	52\\
3.519	51\\
3.500	50\\
3.479	49\\
3.456	48\\
3.431	47\\
3.404	46\\
3.375	45\\
3.344	44\\
3.311	43\\
3.276	42\\
3.239	41\\
3.200	40\\
3.159	39\\
3.116	38\\
3.071	37\\
3.024	36\\
2.975	35\\
2.924	34\\
2.871	33\\
2.816	32\\
2.759	31\\
2.700	30\\
2.639	29\\
2.576	28\\
2.511	27\\
2.444	26\\
2.375	25\\
2.304	24\\
2.231	23\\
2.156	22\\
2.079	21\\
2.000	20\\
1.919	19\\
1.836	18\\
1.751	17\\
1.664	16\\
1.575	15\\
1.484	14\\
1.391	13\\
1.296	12\\
1.199	11\\
1.100	10\\
0.999	9\\
0.896	8\\
0.791	7\\
0.684	6\\
0.575	5\\
0.464	4\\
0.351	3\\
0.236	2\\
0.119	1\\
};
\addplot [color=mycolor1, dashdotted, line width=2.0pt, forget plot]
  table[row sep=crcr]{%
8.640	60\\
0.144	1\\
};
\end{axis}
\end{tikzpicture}%

%% file: experiment_GC_ER_T_7_noisy.tex
%
%
\definecolor{mycolor1}{rgb}{0.00000,0.44700,0.74100}%
\definecolor{mycolor2}{rgb}{0.85000,0.32500,0.09800}%
\definecolor{mycolor3}{rgb}{0.92900,0.69400,0.12500}%
\definecolor{mycolor4}{rgb}{0.49400,0.18400,0.55600}%
\definecolor{mycolor5}{rgb}{0.46600,0.67400,0.18800}%
\definecolor{mycolor6}{rgb}{0.30100,0.74500,0.93300}%
\definecolor{mycolor7}{rgb}{0.63500,0.07800,0.18400}%
\definecolor{mycolor8}{rgb}{0.08000,0.39200,0.25100}%
\begin{tikzpicture}

\begin{axis}[%
width=0.951\figurewidth,
height=\figureheight,
at={(0\figurewidth,0\figureheight)},
scale only axis,
xmode=log,
xmin=0.1,
xmax=1e+16,
xminorticks=true,
xlabel style={font=\color{white!15!black}},
xlabel={Signal-to-noise ratio},
ymode=log,
ymin=1e-09,
ymax=10,
ytick = {1e-09,1e-07,1e-05,1e-03,1e-01,1e01},
yminorticks=true,
ylabel style={font=\color{white!15!black}},
ylabel={Avg. relative Frobenius error on A},
axis background/.style={fill=white},
legend style={at={(0.22,0.02)}, legend cell align=left, align=left, anchor=south, draw=white!15!black, font=\fontsize{5}{30}\selectfont,legend columns = 2},
legend image post style={xscale=0.375},
xlabel style={font=\tiny},ylabel style={font=\tiny},
]
\addplot [color=mycolor1, line width=1.0pt, mark size=2.5pt, mark=x, mark options={solid, mycolor1}]
 plot [error bars/.cd, y dir = both, y explicit]
 table[row sep=crcr, y error plus index=2, y error minus index=3]{%
0.1	3.02930920500858	0.129039085484446	0.103650214268009\\
1	1.63240589781862	0.0652984958963019	0.0724487061103132\\
10	1.30213338613181	0.0518026940283636	0.0387169884888121\\
100	1.234788410537	0.0346072524550702	0.0394924384148081\\
1000	1.21467294743606	0.0356174224970149	0.0328849043372661\\
10000	1.21939744986897	0.0397520284462392	0.0292429387566588\\
100000	1.22413839569786	0.0344077377255687	0.0544386901734901\\
1000000	1.21177963258228	0.0483779664182906	0.0251547178504468\\
10000000	1.20998788826519	0.0378001045328857	0.0285795115780423\\
100000000	1.2134137669045	0.0433422981858906	0.0395089821118644\\
1000000000	1.22337557204855	0.0280147226893483	0.0402810401562668\\
10000000000	1.21663879640195	0.0310801107088634	0.0366056458879109\\
100000000000	1.2255264445898	0.0363794610778096	0.049088534182341\\
1000000000000	1.22044487289244	0.0283142538435817	0.033733777532349\\
10000000000000	1.21517388537707	0.0405374768541411	0.032167402181551\\
100000000000000	1.21211289054559	0.0318924916025249	0.0279562212994475\\
1e+15	1.22149869878275	0.039336557410719	0.0330059610035995\\
1e+16	1.21348310930832	0.0441769736802926	0.0279033777427862\\
};
\addlegendentry{$\rho=1$}

\addplot [color=mycolor2, line width=1.0pt, mark size=2.5pt, mark=x, mark options={solid, mycolor2}]
 plot [error bars/.cd, y dir = both, y explicit]
 table[row sep=crcr, y error plus index=2, y error minus index=3]{%
0.1	3.71094163193117	0.305085744586949	0.239138297367544\\
1	1.90015793543791	0.119590180298884	0.0986255037801187\\
10	1.39685480992362	0.0633268230153246	0.0564719723097478\\
100	1.20085685238597	0.0703518242583507	0.0555820305945163\\
1000	1.18835720832517	0.0637962688580389	0.0668471437586038\\
10000	1.15779335605284	0.0642352644127109	0.0446937505080394\\
100000	1.20447112023602	0.0554391664247653	0.0731048511661361\\
1000000	1.16483363741781	0.0645747764254243	0.0616382937306035\\
10000000	1.17713869066941	0.082536638199604	0.0625610285922158\\
100000000	1.17863006448823	0.0519605970690666	0.0618137704414321\\
1000000000	1.18461393201279	0.0655690036986485	0.0578047932774344\\
10000000000	1.17530913321044	0.0667259776278797	0.0721920399299645\\
100000000000	1.1781637427629	0.0532263427470399	0.0564967013094486\\
1000000000000	1.17999393984061	0.0789589652344302	0.0720062956736414\\
10000000000000	1.16650311274858	0.0855486001682397	0.0496530134041513\\
100000000000000	1.17645112617917	0.0533876713362773	0.0535376255715982\\
1e+15	1.16960164824991	0.0520401685382792	0.0661089688858216\\
1e+16	1.18882560625881	0.0576234118958074	0.0629681006215266\\
};
\addlegendentry{$\rho=1.5$}

\addplot [color=mycolor3, line width=1.0pt, mark size=2.5pt, mark=x, mark options={solid, mycolor3}]
 plot [error bars/.cd, y dir = both, y explicit]
 table[row sep=crcr, y error plus index=2, y error minus index=3]{%
0.1	2.63713419298815	0.128978039692443	0.103514321379676\\
1	1.43313548221255	0.0609201478411889	0.0405901141111822\\
10	1.03070210336603	0.0438247601642912	0.0307758028675212\\
100	0.702550784057183	0.12859254472549	0.0810470549024385\\
1000	0.626990832285222	0.113703716534095	0.161745216751856\\
10000	0.53063216510315	0.134638944065949	0.50590094241159\\
100000	0.592787991940898	0.0977752088793844	0.159072721904674\\
1000000	0.548267182532452	0.134973763541295	0.114278926890124\\
10000000	0.5890388098707	0.105159879367807	0.147799789648328\\
100000000	0.603184621165194	0.0730809353459764	0.163787084673093\\
1000000000	0.620411798894434	0.119739769555058	0.153853042624193\\
10000000000	0.577591501935337	0.126839242700477	0.144535041856194\\
100000000000	0.61958056238627	0.101891126548316	0.110740182061962\\
1000000000000	0.619331022873378	0.0883065318736586	0.172605247432422\\
10000000000000	0.596261880648916	0.101174542270001	0.166592569839218\\
100000000000000	0.610167018884821	0.0972067020218096	0.180406749704645\\
1e+15	0.486701228661124	0.174063984886727	0.0824488583484248\\
1e+16	0.590853736799528	0.0929370760363593	0.16114288430766\\
};
\addlegendentry{$\rho=2$}

\addplot [color=mycolor4, line width=1.0pt, mark size=2.5pt, mark=x, mark options={solid, mycolor4}]
 plot [error bars/.cd, y dir = both, y explicit]
 table[row sep=crcr, y error plus index=2, y error minus index=3]{%
0.1	2.3599079625162	0.0840055197428833	0.0959950630365904\\
1	1.3202786398288	0.0720040822677264	0.0445139537518591\\
10	0.912081164415458	0.0340587095085431	0.0377768160466296\\
100	0.44874714172618	0.0787527721045951	0.244581643791164\\
1000	0.0670644882505602	0.356644380539432	0.00746862485391715\\
10000	0.0209782588556822	0.414723844902797	0.00179080028046061\\
100000	0.312539441008425	0.138140820283636	0.306480317011964\\
1000000	0.00194697227789618	0.378918788562886	0.000130739093827915\\
10000000	0.000637097734709788	0.440020758198096	5.02262958316357e-05\\
100000000	0.373264574959722	0.0672839392081878	0.373069786202249\\
1000000000	6.51719977517682e-05	0.432899588564835	4.90109023497259e-06\\
10000000000	0.193438638571257	0.250650716318949	0.193419428612718\\
100000000000	6.48131997285235e-06	0.434558050116539	4.91185699414777e-07\\
1000000000000	1.99769568090169e-06	0.417215721518012	1.23603677439544e-07\\
10000000000000	0.373543418298982	0.0765703908868582	0.373542808941499\\
100000000000000	2.0573250011465e-07	0.435256911874973	1.61225100244296e-08\\
1e+15	6.41839369686477e-08	0.4871275984394	4.80068478922975e-09\\
1e+16	2.11978870546727e-08	0.433813866510369	1.97594798527279e-09\\
};
\addlegendentry{$\rho=2.2$}

\addplot [color=mycolor5, line width=1.0pt, mark size=2.5pt, mark=x, mark options={solid, mycolor5}]
 plot [error bars/.cd, y dir = both, y explicit]
 table[row sep=crcr, y error plus index=2, y error minus index=3]{%
0.1	2.14370367921344	0.066173877778442	0.0857769453375252\\
1	1.24523141940309	0.0307055826392593	0.0327714938620267\\
10	0.762813095448714	0.0369694880932092	0.0412437510831879\\
100	0.165788945692407	0.0157181143228506	0.00744511846351634\\
1000	0.051647489379104	0.00240833986206225	0.00220761316703429\\
10000	0.0165266572276974	0.00085291199766075	0.000616471563113476\\
100000	0.00522364318100082	0.000257949131339026	0.000181612022616676\\
1000000	0.00164250387760723	9.13745498896002e-05	6.5287133581008e-05\\
10000000	0.00051265129300953	2.80372028217159e-05	1.78230744967483e-05\\
100000000	0.000166239007699666	0.348307918730715	6.21796008491016e-06\\
1000000000	5.06904490677724e-05	3.75081609460117e-06	1.05203263504942e-06\\
10000000000	1.631402978675e-05	9.48787638147492e-07	6.47286176289171e-07\\
100000000000	5.29527138654833e-06	2.39352358124831e-07	2.98567103726605e-07\\
1000000000000	1.66281322994245e-06	1.85025571273918e-07	8.62474199089516e-08\\
10000000000000	5.15143465860166e-07	2.19341362853138e-08	1.94081438616075e-08\\
100000000000000	1.64102892252269e-07	1.08943774011822e-08	5.34165075578709e-09\\
1e+15	5.10955734915473e-08	1.97723354751156e-09	1.59664088216255e-09\\
1e+16	1.63974277838016e-08	8.89948293966963e-10	7.26719956290424e-10\\
};
\addlegendentry{$\rho=2.4$}

\addplot [color=mycolor6, dashed, line width=1.0pt, mark size=2.5pt, mark=x, mark options={solid, mycolor6}]
 plot [error bars/.cd, y dir = both, y explicit]
 table[row sep=crcr, y error plus index=2, y error minus index=3]{%
0.1	1.98730371707751	0.0753311785811421	0.0932234402421208\\
1	1.15516744331066	0.0387961323773487	0.0271494666415042\\
10	0.651246649181002	0.0279042466726844	0.0612373078874584\\
100	0.14164289300809	0.0070064583116759	0.00395136194976481\\
1000	0.0444001096760455	0.00213632295877102	0.00139555668395071\\
10000	0.0142016351880916	0.000438140866791762	0.000674827257890486\\
100000	0.0044784505223741	0.000171652725264825	0.000198872938318427\\
1000000	0.00141498066538755	6.58821694710979e-05	4.5413454464916e-05\\
10000000	0.000447860794091283	1.4780459418975e-05	1.64261897543265e-05\\
100000000	0.000141679889348102	6.58464141524998e-06	3.15916266481223e-06\\
1000000000	4.54669285280896e-05	1.08893109086112e-06	1.67060638718949e-06\\
10000000000	1.42179281571111e-05	5.9403840184463e-07	4.69931883929481e-07\\
100000000000	4.50547750504852e-06	2.04843128700743e-07	1.53245833341177e-07\\
1000000000000	1.41727650774711e-06	4.36343009905483e-08	6.77069368413559e-08\\
10000000000000	4.4580646076869e-07	1.55880694552954e-08	8.6857267809187e-09\\
100000000000000	1.43014604734828e-07	5.68711656025511e-09	6.99547081339401e-09\\
1e+15	4.47263423797973e-08	1.37136617210326e-09	1.3354663857239e-09\\
1e+16	1.40672593003578e-08	3.72698284900596e-10	5.45633524797112e-10\\
};
\addlegendentry{$\rho=2.6$}

\addplot [color=mycolor7, dashed, line width=1.0pt, mark size=2.5pt, mark=x, mark options={solid, mycolor7}]
 plot [error bars/.cd, y dir = both, y explicit]
 table[row sep=crcr, y error plus index=2, y error minus index=3]{%
0.1	1.71613976488131	0.0549039437596364	0.0395862575861592\\
1	1.02728599085858	0.0255946450988929	0.0265379664765213\\
10	0.381016148432727	0.0899762903796356	0.0227052258841752\\
100	0.114996090051485	0.00257647359658049	0.00351645885460181\\
1000	0.0360438687974272	0.00100228913882519	0.00103162859860507\\
10000	0.011471728013758	0.00026412288942981	0.000367688248138986\\
100000	0.00365112129742082	0.000100334081878698	0.000113057812129941\\
1000000	0.00114780424730539	3.20583985473405e-05	3.66128339811562e-05\\
10000000	0.000364888263483472	9.05293076935618e-06	1.13010579458652e-05\\
100000000	0.000114501378102256	3.39197380829406e-06	2.31757366934694e-06\\
1000000000	3.66953472461072e-05	1.04271956228559e-06	1.02160345168779e-06\\
10000000000	1.14187717979459e-05	3.7703290523228e-07	4.13844134340086e-07\\
100000000000	3.602426115265e-06	9.03538192106608e-08	7.46670544108957e-08\\
1000000000000	1.15039416012764e-06	2.9297101774456e-08	4.0250165960027e-08\\
10000000000000	3.63704063011993e-07	9.97017850765071e-09	1.62392075568373e-08\\
100000000000000	1.13930147556089e-07	3.31106651317414e-09	4.42416366982108e-09\\
1e+15	3.6800470354844e-08	1.16363227343445e-09	1.13892729740662e-09\\
1e+16	1.15213113726999e-08	2.75394180705798e-10	3.22773235698355e-10\\
};
\addlegendentry{$\rho=3$}

\addplot [color=mycolor8, dashed, line width=1.0pt, mark size=2.5pt, mark=x, mark options={solid, mycolor8}]
 plot [error bars/.cd, y dir = both, y explicit]
 table[row sep=crcr, y error plus index=2, y error minus index=3]{%
0.1	1.43182995975668	0.0309226742075821	0.0278340479885932\\
1	0.807591919767541	0.013562218978079	0.0169434177194677\\
10	0.261172896256158	0.00611351273587712	0.00819940402226321\\
100	0.0819797976533206	0.00190528668752141	0.00237824205585492\\
1000	0.0258657188182658	0.000579003382867808	0.000735595422553775\\
10000	0.00817527141482902	0.000177114657139429	0.00016040575242858\\
100000	0.0025853496765473	7.01976943049717e-05	5.47460943973559e-05\\
1000000	0.000829782358530744	1.81142660538206e-05	2.14522490297967e-05\\
10000000	0.000264057866699778	5.25798294419429e-06	9.57120762259066e-06\\
100000000	8.18761622472167e-05	2.02842488913476e-06	1.96640087237391e-06\\
1000000000	2.6115489187136e-05	5.4705235647517e-07	6.05147591253189e-07\\
10000000000	8.09494793392826e-06	2.68208465959947e-07	1.39584615959934e-07\\
100000000000	2.61936753974062e-06	6.33809418870759e-08	4.1808967646698e-08\\
1000000000000	8.24694789522896e-07	2.36259790375556e-08	1.62028411764312e-08\\
10000000000000	2.62024628786724e-07	5.07702684868964e-09	7.64758483398897e-09\\
100000000000000	8.28022788802118e-08	2.09394923414694e-09	1.8173451160952e-09\\
1e+15	2.5784943303993e-08	5.06346624296448e-10	6.15483498061048e-10\\
1e+16	8.20647057963697e-09	2.4111363349593e-10	1.71603083261591e-10\\
};
\addlegendentry{$\rho=4$}

\end{axis}

\begin{axis}[%
width=1.227\figurewidth,
height=1.227\figureheight,
at={(-0.16\figurewidth,-0.135\figureheight)},
scale only axis,
xmin=0,
xmax=1,
ymin=0,
ymax=1,
axis line style={draw=none},
ticks=none,
axis x line*=bottom,
axis y line*=left,
legend style={legend cell align=left, align=left, draw=white!15!black},
xlabel style={font=\tiny},ylabel style={font=\tiny},
]
\end{axis}
\end{tikzpicture}%

%% file: experiment_GC_Community_T_7_noisy.tex
%
%
\definecolor{mycolor1}{rgb}{0.00000,0.44700,0.74100}%
\definecolor{mycolor2}{rgb}{0.85000,0.32500,0.09800}%
\definecolor{mycolor3}{rgb}{0.92900,0.69400,0.12500}%
\definecolor{mycolor4}{rgb}{0.49400,0.18400,0.55600}%
\definecolor{mycolor5}{rgb}{0.46600,0.67400,0.18800}%
\definecolor{mycolor6}{rgb}{0.30100,0.74500,0.93300}%
\definecolor{mycolor7}{rgb}{0.63500,0.07800,0.18400}%
\begin{tikzpicture}

\begin{axis}[%
width=0.951\figurewidth,
height=\figureheight,
at={(0\figurewidth,0\figureheight)},
scale only axis,
xmode=log,
xmin=0.1,
xmax=1e+16,
xminorticks=true,
xlabel style={font=\color{white!15!black}},
xlabel={Signal-to-noise ratio},
ymode=log,
ymin=1e-09,
ymax=10,
ytick = {1e-09,1e-07,1e-05,1e-03,1e-01,1e01},
yminorticks=true,
ylabel style={font=\color{white!15!black}},
ylabel={Avg. relative Frobenius error on A},
axis background/.style={fill=white},
legend style={at={(0.22,0.02)}, legend cell align=left, align=left, anchor=south, draw=white!15!black, font=\fontsize{5}{30}\selectfont,legend columns = 2},
legend image post style={xscale=0.375},
xlabel style={font=\tiny},ylabel style={font=\tiny},
]
\addplot [color=mycolor1, line width=1.0pt, mark size=2.0pt, mark=x, mark options={solid, mycolor1}]
 plot [error bars/.cd, y dir = both, y explicit]
 table[row sep=crcr, y error plus index=2, y error minus index=3]{%
0.1	3.1429388445265	0.133615522996425	0.0733820251169695\\
1	1.531984891679	0.0536615311246693	0.049889581688646\\
10	1.09208260682484	0.0299890365911311	0.0311742872945935\\
100	0.941245985664305	0.0298066965788948	0.0298123459437762\\
1000	0.896712280883357	0.0204627246578339	0.0257526825915395\\
10000	0.892181568992857	0.0286211623112494	0.0197073791294228\\
100000	0.890548835378834	0.0276531756298625	0.0331819478199598\\
1000000	0.896470408077916	0.018822831173624	0.0296693567423598\\
10000000	0.885143067239682	0.0260273402211291	0.0320995333393054\\
100000000	0.892159098993972	0.0267381622766022	0.0276529636396436\\
1000000000	0.889153946903945	0.0289433326070911	0.0379833656615912\\
10000000000	0.891136684027432	0.0241316790541251	0.0263093013050262\\
100000000000	0.887271498469274	0.0228859138395288	0.030470415661266\\
1000000000000	0.89714234539361	0.0246138107399415	0.0260423175307798\\
10000000000000	0.886005050662155	0.0207727588830384	0.0246628769308888\\
100000000000000	0.891523806499334	0.0216580329774403	0.0194891233039678\\
1e+15	0.899963770521641	0.0275577718976284	0.0252546746707208\\
1e+16	0.884854703437963	0.028534685413796	0.0196656967789162\\
};
\addlegendentry{$\rho=1$}

\addplot [color=mycolor2, line width=1.0pt, mark size=2.0pt, mark=x, mark options={solid, mycolor2}]
 plot [error bars/.cd, y dir = both, y explicit]
 table[row sep=crcr, y error plus index=2, y error minus index=3]{%
0.1	3.71863180728903	0.175792435244257	0.095391770043368\\
1	1.77273347748326	0.0653035137980451	0.0536799749766654\\
10	1.26319570339157	0.0349570797269498	0.045797935132841\\
100	0.947469693122649	0.0450504641626379	0.0289008420552761\\
1000	0.769750919805945	0.0426018500031666	0.0557982753615591\\
10000	0.770558869520795	0.0778906137797027	0.0638949501501259\\
100000	0.773128565444174	0.0650715297348661	0.0966295732878641\\
1000000	0.710128892531411	0.0913637931313451	0.106021348541033\\
10000000	0.781240638287922	0.0644028878248857	0.0898896880648481\\
100000000	0.77351972270484	0.0592638516478732	0.0945208828129381\\
1000000000	0.757314933868489	0.0657373479702171	0.0689280564540484\\
10000000000	0.733631532098476	0.0577050378073125	0.0539132529014323\\
100000000000	0.752843772273111	0.0546592024969319	0.0595683217022621\\
1000000000000	0.739488275816896	0.0617189491119012	0.0788703445825536\\
10000000000000	0.748710120308624	0.0578966224898647	0.0448556066317032\\
100000000000000	0.741417028721804	0.0855740286666316	0.0871045080968137\\
1e+15	0.762302078354214	0.0605641855951119	0.0683100094724224\\
1e+16	0.770297229080032	0.0603473287881581	0.0940047655959162\\
};
\addlegendentry{$\rho=1.4$}

\addplot [color=mycolor3, line width=1.0pt, mark size=2.0pt, mark=x, mark options={solid, mycolor3}]
 plot [error bars/.cd, y dir = both, y explicit]
 table[row sep=crcr, y error plus index=2, y error minus index=3]{%
0.1	3.12125304786145	0.0952165073289231	0.0950398841290463\\
1	1.47589541997193	0.0454552768932737	0.0484527504811159\\
10	0.991686093117024	0.0432458024881772	0.02769241196047\\
100	0.584094586836048	0.0424811444803558	0.0507053297265432\\
1000	0.122935137134382	0.0115457197443029	0.0073372831857125\\
10000	0.0396141135572519	0.00342496193463517	0.00284396634110094\\
100000	0.0127580142712045	0.00103131944782279	0.000953753910054655\\
1000000	0.00392436336915707	0.000336215026272861	0.000305356018067219\\
10000000	0.00128264255061659	0.000113836810908444	9.31350330738135e-05\\
100000000	0.000395160633022928	3.11401279593353e-05	3.03102223730539e-05\\
1000000000	0.00012318638215478	5.16730732165764e-06	1.01174748205687e-05\\
10000000000	3.85074628136711e-05	3.4693219493613e-06	1.92918431499725e-06\\
100000000000	1.25561101771738e-05	1.07924148518034e-06	9.1360855933081e-07\\
1000000000000	3.94225429724729e-06	4.20410535822198e-07	2.31063309557525e-07\\
10000000000000	1.26507151952595e-06	1.30897990617065e-07	9.455526774558e-08\\
100000000000000	3.87106342335071e-07	2.38302603668758e-08	2.48453741347891e-08\\
1e+15	1.25066424071345e-07	1.01427963201293e-08	6.9199550565317e-09\\
1e+16	3.8676424939278e-08	2.57251226171949e-09	1.34324099435079e-09\\
};
\addlegendentry{$\rho=1.6$}

\addplot [color=mycolor4, line width=1.0pt, mark size=2.0pt, mark=x, mark options={solid, mycolor4}]
 plot [error bars/.cd, y dir = both, y explicit]
 table[row sep=crcr, y error plus index=2, y error minus index=3]{%
0.1	2.64061021570974	0.082148256598555	0.0754616040459473\\
1	1.23613672877174	0.0363387433500519	0.020656428440021\\
10	0.771387504210358	0.0248441604992177	0.0378134941503374\\
100	0.264723192211832	0.0181278505092496	0.0162151245805277\\
1000	0.0799964939671726	0.00378733389942466	0.00509536528465883\\
10000	0.0251102348659054	0.0011374654362115	0.00170866316598849\\
100000	0.00781143368182183	0.000458596846833107	0.000311755478199387\\
1000000	0.00244253683422182	0.000159169860706214	9.62925681141213e-05\\
10000000	0.000783633676818528	3.22803612560496e-05	2.50151317110044e-05\\
100000000	0.000243797338636827	1.31956755571951e-05	1.30141953064272e-05\\
1000000000	7.93575235177137e-05	3.26631017137881e-06	4.4848603842809e-06\\
10000000000	2.4420443266534e-05	1.48791772826553e-06	1.02348696883769e-06\\
100000000000	7.8546810172625e-06	4.56653164034053e-07	2.37155807179105e-07\\
1000000000000	2.47438407111826e-06	1.76538965520026e-07	1.00283194151977e-07\\
10000000000000	7.87252097791311e-07	4.52309186655364e-08	3.42217992528824e-08\\
100000000000000	2.48568287609366e-07	1.49529313717604e-08	1.10473719991069e-08\\
1e+15	7.82936153538612e-08	4.42931380077938e-09	3.23153979181019e-09\\
1e+16	2.49281457961596e-08	9.01990188660645e-10	1.22234245705219e-09\\
};
\addlegendentry{$\rho=1.8$}

\addplot [color=mycolor5, line width=1.0pt, mark size=2.0pt, mark=x, mark options={solid, mycolor5}]
 plot [error bars/.cd, y dir = both, y explicit]
 table[row sep=crcr, y error plus index=2, y error minus index=3]{%
0.1	2.29370074684041	0.044835376334396	0.0567388826035149\\
1	1.06311080985791	0.0254756641299476	0.0314491285427663\\
10	0.592503217865684	0.0231097578181586	0.02400169281736\\
100	0.177213516786429	0.00774926055011721	0.00698343063594251\\
1000	0.057338413461879	0.00177974290065301	0.00239051353558362\\
10000	0.0181828873342048	0.000551431729088073	0.000768699690962112\\
100000	0.00567876735727627	0.000298262543763857	0.000240275039173257\\
1000000	0.00179160456136289	9.28863587831017e-05	6.68229287808476e-05\\
10000000	0.000565071985451963	2.45850234414412e-05	2.2090422284548e-05\\
100000000	0.000180040780088698	8.39053881491293e-06	7.60257510292516e-06\\
1000000000	5.71873896724055e-05	1.58675787248294e-06	2.6480468882472e-06\\
10000000000	1.76770870227975e-05	7.31900330418366e-07	4.29623245036216e-07\\
100000000000	5.61994272012268e-06	2.26697999128197e-07	1.74290110070646e-07\\
1000000000000	1.79584948509848e-06	8.112422236465e-08	6.4174196886983e-08\\
10000000000000	5.61096339399172e-07	2.10061673542071e-08	2.15340537955125e-08\\
100000000000000	1.79285995332427e-07	6.33167980842993e-09	5.8648388088649e-09\\
1e+15	5.65949926007473e-08	2.59295844153964e-09	2.13035708473874e-09\\
1e+16	1.79190443283537e-08	5.0772360258241e-10	6.5863104904347e-10\\
};
\addlegendentry{$\rho=2$}

\addplot [color=mycolor6, dashed, line width=1.0pt, mark size=2.0pt, mark=x, mark options={solid, mycolor6}]
 plot [error bars/.cd, y dir = both, y explicit]
 table[row sep=crcr, y error plus index=2, y error minus index=3]{%
0.1	1.65658676856759	0.0145903266484573	0.0210427110176115\\
1	0.636913218681065	0.0079313407699737	0.0103376613847541\\
10	0.236304144708717	0.00339761783471573	0.00430064564839619\\
100	0.0760227119807575	0.00130156625079041	0.00136570537119333\\
1000	0.0240865786163572	0.000403310607875931	0.000470678815972233\\
10000	0.00762224188402695	0.000147752166686352	0.000156890929112625\\
100000	0.00242188569322916	5.03085707712239e-05	5.26015309642671e-05\\
1000000	0.000766090376364649	1.32097541383155e-05	1.63163648613418e-05\\
10000000	0.000240692357121091	4.88165347508224e-06	4.75430066110277e-06\\
100000000	7.6355072957779e-05	1.82246830266421e-06	1.54343394045934e-06\\
1000000000	2.40736705679531e-05	5.10556855649428e-07	4.1596264606939e-07\\
10000000000	7.62096687962888e-06	1.45842067673317e-07	1.45384547622166e-07\\
100000000000	2.4112222447019e-06	4.5404220887644e-08	3.82977932354676e-08\\
1000000000000	7.58677836882826e-07	1.66064737969039e-08	1.45012362558214e-08\\
10000000000000	2.41483807752798e-07	5.56517285897508e-09	5.51113885448075e-09\\
100000000000000	7.59458922878421e-08	1.89609339797335e-09	1.4989994216764e-09\\
1e+15	2.43020928664384e-08	3.46536705427164e-10	5.48931028164578e-10\\
1e+16	7.57623625338746e-09	1.55386316945797e-10	1.42873317280729e-10\\
};
\addlegendentry{$\rho=3$}

\addplot [color=mycolor7, dashed, line width=1.0pt, mark size=2.0pt, mark=x, mark options={solid, mycolor7}]
 plot [error bars/.cd, y dir = both, y explicit]
 table[row sep=crcr, y error plus index=2, y error minus index=3]{%
0.1	1.74863447464391	0.0188824080456014	0.0174167648819754\\
1	0.585654703308413	0.00474739105993116	0.00502780817496751\\
10	0.190828952999539	0.0015199901310505	0.00193514542167966\\
100	0.0606281198240169	0.000813802349838737	0.000594039179548547\\
1000	0.0191092102388016	0.000220859559945403	0.000168888625614198\\
10000	0.00602896012797335	5.77022442528122e-05	5.83254746943651e-05\\
100000	0.00190497542740493	1.95751962380804e-05	1.37218117784808e-05\\
1000000	0.000607382072289186	6.22953345938339e-06	6.63191114989425e-06\\
10000000	0.000191439314631611	1.47949791526471e-06	2.65771500224761e-06\\
100000000	6.06050549541805e-05	7.69774603977957e-07	3.93513758474004e-07\\
1000000000	1.92092158558269e-05	1.60657260012138e-07	2.40311041522936e-07\\
10000000000	6.08656023447646e-06	6.27728854089613e-08	6.67981572896962e-08\\
100000000000	1.90486697217221e-06	1.77987763127386e-08	2.04044040429151e-08\\
1000000000000	6.06723755805189e-07	8.81603463130918e-09	7.05399148002455e-09\\
10000000000000	1.92055100343647e-07	1.83888591594559e-09	2.30497250872894e-09\\
100000000000000	6.04491384041392e-08	8.5732772401654e-10	4.84006416290847e-10\\
1e+15	1.9054860906733e-08	1.92589265002439e-10	2.34903895043619e-10\\
1e+16	6.03938327121778e-09	6.8365420810884e-11	5.70051367524304e-11\\
};
\addlegendentry{$\rho=4$}

\end{axis}
\end{tikzpicture}%